%% file: arxiv.tex
\documentclass[10pt,a4paper]{article}

\usepackage[T1]{fontenc}
\usepackage[utf8]{inputenc}
\usepackage{lmodern}

\usepackage{authblk}
\usepackage{geometry}
\usepackage{amsmath,amssymb,amsthm,mathtools}
\usepackage{amsfonts}
\usepackage{graphicx}
\usepackage{xcolor}
\usepackage{array,booktabs,colortbl,tabularx,adjustbox}
\usepackage{multirow,eepic}
\usepackage{algorithm}
\usepackage{algpseudocode}
\usepackage{xspace}
\usepackage{tikz}
\usetikzlibrary{positioning,automata,arrows,calc,arrows.meta,backgrounds}
\usepackage[hidelinks]{hyperref}
\hypersetup{
  pdftitle={A Minimum Exact Disjunctive Basis for Two-Robot Solvability Profiles},
  pdfauthor={Tsukasa Ninomiya, Shota Naito, Naoki Kitamura, Yuichi Sudo, Koichi Wada}
}
\usepackage{caption}

\definecolor{imposs}{RGB}{200,40,40}
\definecolor{newcell}{RGB}{0,90,180}
\definecolor{lipicsGray}{gray}{0.45}

\newtheorem{theorem}{Theorem}[section]
\newtheorem{lemma}[theorem]{Lemma}
\newtheorem{corollary}[theorem]{Corollary}
\newtheorem{proposition}[theorem]{Proposition}

\theoremstyle{definition}
\newtheorem{definition}[theorem]{Definition}

\theoremstyle{remark}
\newtheorem{remark}[theorem]{Remark}

\theoremstyle{plain}

\makeatletter
\def\senbun#1(#2)#3({\@senbun(#2)(}
\def\@senbun(#1,#2)(#3,#4){%
   \@tempdima#1\p@ \advance\@tempdima#3\p@
   \divide\@tempdima\tw@
   \@tempdimb#2\p@ \advance\@tempdimb#4\p@
   \divide\@tempdimb\tw@
   \edef\@senbun@temp{\noexpand\qbezier(#1,#2)%
      (\strip@pt\@tempdima,\strip@pt\@tempdimb)(#3,#4)}%
  \@senbun@temp}
\makeatother
\newcommand{\ASY}{\textnormal{\textsc{Asynch}}\xspace}
\newcommand{\FSY}{\textnormal{\textsc{Fsynch}}\xspace}
\newcommand{\SSY}{\textnormal{\textsc{Ssynch}}\xspace}
\newcommand{\RSY}{\textnormal{\textsc{Rsynch}}\xspace}

\newcommand{\LU}{\ensuremath{\mathcal{LUMI}}\xspace}
\newcommand{\FS}{\ensuremath{\mathcal{FST\!A}}\xspace}
\newcommand{\FC}{\ensuremath{\mathcal{FCOM}}\xspace}
\newcommand{\OB}{\ensuremath{\mathcal{OBLOT}}\xspace}
\newcommand{\Look}{\ensuremath{\mathit{Look}}\xspace}
\newcommand{\Compute}{\ensuremath{\mathit{Compute}}\xspace}
\newcommand{\Move}{\ensuremath{\mathit{Move}}\xspace}
\newcommand{\LCM}{\ensuremath{\mathit{LCM}}\xspace}

\newcommand{\clipout}[1]{\clip (-100,-100) rectangle (1700,1000) #1;}

\newcommand{\Sol}{\mathsf{Sol}}
\newcommand{\Prof}{\mathsf{Prof}}
\newcommand{\Real}{\mathsf{Real}}
\newcommand{\Up}{\operatorname{Up}}
\newcommand{\Ant}{\operatorname{Ant}}
\newcommand{\Min}{\operatorname{Min}}
\newcommand{\Max}{\operatorname{Max}}
\newcommand{\width}{\operatorname{width}}

\newcommand{\N}{{\rm I\kern-.22em N}} 
\newcommand{\Z}{{\sf Z\kern-.42em Z}} 
\newcommand{\R}{{\rm I\kern-.22em R}}

\newcommand{\LK}{{\mathit{Look}}}
\newcommand{\CP}{{\mathit{Comp}}}
\newcommand{\M}{{\mathit{Move}}}

\newcounter{Codeline}

\newcommand{\lcolor}{\mathit{color}\xspace}

\newcommand{\des}{\mathit{des}\xspace}

\newcommand{\Stop}{\texttt{Stop}}

\newcommand{\MoveSet}[1]{\mathrm{MoveSet}(#1)}
\newcommand{\moves}[1]{\mathrm{moves}(#1)}
\newcommand{\Stopped}[1]{\mathrm{Stopped}(#1)}

\newcommand{\Dyad}{\mathbb{Z}[1/2]}

\usetikzlibrary{backgrounds}
\usepackage{graphicx}
\usepackage{tikz}
\usetikzlibrary{automata,positioning,arrows.meta,calc}
\newcommand{\upa}{\mathord{\uparrow}}

\newcommand{\VSS}{\textnormal{\text{VSS}}\xspace}
\newcommand{\dis}[2]{\mathit{dis}(#1,#2)}

\newcommand{\SRDV}{\textnormal{\text{S-RDV}}\xspace}

\newcommand{\DMSD}{\textnormal{\text{DMSD}}\xspace}
\newcommand{\Sab}{\ensuremath{\Sigma_{ab}}}

\newcommand{\CFSM}{\textnormal{\text{CFSM}}\xspace}

\title{A Minimum Exact Disjunctive Basis for Two-Robot Solvability Profiles}
\date{}

\author[1]{Tsukasa Ninomiya}
\author[1]{Shota Naito}
\author[2]{Naoki Kitamura}
\author[1]{Yuichi Sudo}
\author[1]{Koichi Wada}

\affil[1]{Hosei University, Tokyo, Japan}
\affil[2]{The University of Osaka, Osaka, Japan}

\begin{document}
 
\maketitle

\begin{abstract}
We completely characterize the solvability profiles of tasks for two
autonomous mobile robots in the plane under variable disorientation,
chirality, and rigid movements. The four standard capability models
and the schedulers considered here have a known quotient order
$\mathcal Q$ of thirteen computational classes. A task's
\emph{solvability profile} is the set of classes in which it is
solvable, and is necessarily an upset of $\mathcal Q$. We prove that
every upset is realizable and construct a minimum exact disjunctive
basis of thirteen concrete tasks.

The basis contains one task for each principal upset. For every
antichain of $\mathcal Q$, the disjunction of the corresponding tasks
has profile exactly equal to the union of their profiles. This
exactness is not implied by the individual profiles: an algorithm
may satisfy different disjuncts in different executions. We establish
it by constructing, for every algorithm in each maximal class outside
the target union, a single admissible fair execution violating all
selected predicates simultaneously. Consequently, the realizable
profiles form the entire finite distributive lattice $\Up(\mathcal Q)$:
there are exactly $54$ profiles, counted by the antichain polynomial
$1+13z+24z^2+16z^3$, and each has a canonical disjunctive representation
using at most three principal tasks.

We also identify six tasks whose solvability signatures recover every
inclusion and incomparability in $\mathcal Q$, and prove that six is
minimum even among arbitrary separating tasks. Their profiles are
precisely the principal upsets generated by the join-irreducible
elements of $\mathcal Q$. In contrast, any task family generating
every profile by exact disjunctions must contain all thirteen
principal profiles. Thus our exact disjunctive basis has minimum
size and is unique at the profile level.
\end{abstract}

\section{Introduction}
\label{sec:introduction}

\subsection{Computational landscapes of mobile robots} 
The \Look--\Compute--\Move (\LCM) model is a standard abstraction for
studying distributed computation by autonomous mobile robots. Robots
are anonymous, identical points in the Euclidean plane, each repeatedly
observing the configuration, computing a destination, and moving toward
it. The classical model, $\OB$, provides neither persistent memory nor
explicit communication~\cite{SY}. The luminous model, $\LU$, equips
robots with persistent visible lights, providing both
resources~\cite{DFPSY}. Two intermediate models isolate their roles:
$\FS$ robots observe only their own lights and have finite private
memory, whereas $\FC$ robots observe only other robots' lights and have
finite communication without private memory~\cite{FSW19}.

The activation scheduler supplies a second computational resource.
Under full synchrony, all robots execute a cycle in every round;
under semi-synchrony, an arbitrary nonempty subset is activated.
Restricted synchrony limits consecutive activations, while asynchrony
permits arbitrary finite delays between phases. Partially atomic
variants distinguish delayed computation from observation during
movement. Cross-model studies have characterized how these timing
assumptions interact with memory and communication, refining the
synchronous landscape and extending it to asynchronous and
energy-constrained settings~\cite{FSW19,BuFKPSW21,FSSW23,BuFKPSW22}.

\noindent\textbf{Related computational-landscape studies:}
The computational-landscape viewpoint has been developed along several
orthogonal dimensions, including synchrony, memory, communication and
energy constraints~\cite{DFPSY,FSW19,BuFKPSW21,BuFKPSW22,FSSW23}, and
more recently beyond the \FSY--\SSY--\ASY axis: sequential schedulers
activating one robot per round~\cite{SequentialTCS,SequentialMFCS},
robots on graphs~\cite{GraphLandscape}, and limited-visibility robots,
where visibility itself is a computational
resource~\cite{VisibilityLandscape}.

\noindent\textbf{Why two robots are structurally exceptional:} 
Two robots form the smallest nontrivial population, but their
computational landscape is not obtained merely by restricting
separators designed for larger populations. Their initial geometry
is a single segment, with no third robot to provide an asymmetric
neighborhood or a richer geometric reference. Some separations
therefore require new witnesses, while some distinctions between
internal capabilities disappear.

Kitamura, Sudo, and Wada~\cite{KSW-2-robot} established the complete
pairwise computational landscape for two robots under variable
disorientation, chirality, and rigid movements. In particular, they
proved the full-synchrony collapse
$\FS^F\equiv\FC^F\equiv\LU^F$, including its color- and time-preserving
simulation, and introduced Dyadic-Multiple Stopped-Distance (DMSD)
to separate atomic from unrestricted asynchrony. We use these known
results as ingredients. 


\noindent\textbf{From a Venn diagram to a profile lattice:} 
Naito, Ninomiya, and Wada~\cite{NNW-SSS2025} developed this global perspective
through a Venn-diagram representation of solvable task sets.
A problem's position records all models that can solve it, exposing
interactions that pairwise comparisons do not capture. Their study
populated additional regions of the general landscape but left other
regions unresolved. We complete this regional classification for the
two-robot specialization, after identifying the settings that become
computationally equivalent for two robots.

Let $\mathcal Q$ be the resulting thirteen-class order, with $X\le Y$
when every task solvable in $X$ is solvable in $Y$. The
\emph{solvability profile} of a task $P$ is
\(
  \Prof(P)=\{X\in\mathcal Q:P\in\Sol(X)\}.
\)
Every profile is an upset of $\mathcal Q$. Thus the finite distributive
lattice $\Up(\mathcal Q)$ describes all membership patterns consistent
with the model order. A Venn region with pattern $U$ consists of tasks
whose profile is exactly $U$; the region is populated only if such a
task exists. Counting abstract upsets therefore does not establish
that all regions are populated. The central question is whether every
upset can be realized by a concrete two-robot task.

A task is \emph{principal at $X$} if its profile is the principal upset
$\upa X=\{Y\in\mathcal Q:X\le Y\}$.
Every upset is the union of the principal upsets generated
by its minimal elements, suggesting a task basis indexed by model
classes. However, disjoining task specifications need not merely take
the union of their profiles. An algorithm can potentially satisfy
different disjuncts in different executions without solving any fixed
disjunct. Consequently, constructing one principal task per class is
not enough: exact realization by disjunction requires a separate
common-failure argument.

\subsection{Contributions} 
Under variable disorientation, chirality, and rigid movements, we
construct a minimum exact disjunctive basis for the complete
solvability-profile landscape of two robots. All basis tasks share
the initial-instance space of two distinct positions. Our results
combine an optimal principal representation of the known model order
with a complete realization theorem for its profile lattice.

\noindent\textbf{An optimal principal representation of the model order:} 
We identify six tasks
\(
  \mathcal B_6=
  \{\mathrm{EPG}^{*},\mathrm{SRO},\mathrm{MCv},
    \CFSM,\VSS,\DMSD\}
\)
whose solvability signatures recover every inclusion and
incomparability among the thirteen classes. SRO is inherited from
earlier landscape studies~\cite{FSW19}, MCv from the asynchronous
landscape~\cite{KSW-2-robot}, and \DMSD from the complete two-robot
characterization~\cite{KSW-2-robot}. The other three are new or revised
here: $\mathrm{EPG}^{*}$ (Perpetual Expansion Preserving Gravity)
replaces Center of Gravity Expansion~\cite{FSW19}, whose coordinatewise
floor refers to the adversary's local coordinate systems and therefore
does not define a predicate on trajectories under variable
disorientation; \CFSM is a collision-free refinement of Single
Move~\cite{KSW-2-robot}; and \VSS (Visit Swap and Stop) is introduced for
the present construction. We prove that six tasks are necessary,
even when separating tasks need not be principal. Moreover,
$\mathcal Q$ is a lattice with six nonleast join-irreducibles, and
these tasks realize exactly the principal profiles generated by them.
Those profiles form the unique minimum principal separating family.

\noindent\textbf{A principal task for every class:} 
We extend $\mathcal B_6$ to a family
$\mathcal B=\{B_X:X\in\mathcal Q\}$ of thirteen tasks satisfying
$\Prof(B_X)=\upa X$ for every $X\in\mathcal Q$.
Three tasks come from earlier landscape studies (SRO, MCv and Cv),
three from SSS~2025 (HET, ZCC and VEC)~\cite{NNW-SSS2025}, and \DMSD
from the complete two-robot characterization~\cite{KSW-2-robot}.
The remaining six are new or revised here:
EPG and EPG$^{*}$, the trajectory-only revisions of
Center of Gravity Expansion and its perpetual variant;
CFSM, a collision-free refinement of Single Move;
OSPO, a strengthened version of Oscillating Points~\cite{DFPSY}
requiring exactly one nonzero Move per robot per phase;
S-RDV (Swap Rendezvous); and VSS.

\noindent\textbf{Exact disjunctions via common-failure executions:} 
For every antichain $A$ of $\mathcal Q$, we prove
\(
  \Prof\left(\bigvee_{X\in A}B_X\right)
  =\bigcup_{X\in A}\upa X.
\)
The substantive step is the reverse inclusion. For each maximal class
outside the target union and each algorithm in that class, we
construct a single admissible fair execution violating all selected
task predicates. Separate counterexamples to the individual tasks
would not suffice. We organize these \emph{common-failure executions}
into families sharing expansion, rotation, or exchange obstructions,
combining compatible finite violations with appropriate fair
continuations. The resulting exactness is a property of the chosen
tasks and their composition, not a consequence of the model order
alone.

\noindent\textbf{Complete realization and minimum basis size:} 
Every upset $U$ is realized by the canonical disjunction
$B_U=\bigvee_{X\in\Min(U)}B_X$, so
$\mathrm{Real}(\mathcal Q)=\Up(\mathcal Q)$. The antichain polynomial
is $1+13z+24z^2+16z^3$. Hence there are exactly $53$ nonempty profiles:
$13$ principal profiles, $24$ with two minimal classes, and $16$ with
three. Including the empty profile gives $54$. Every nonempty profile
has an exact canonical representation using at most three principal
tasks, and three are sometimes necessary. Any task family supplying
exact disjunctive representations of every profile must contain all
thirteen principal profiles. Our basis therefore has minimum size
and is unique at the profile level among minimum exact generating
families.


%
%

\section{Models, Tasks and Solvability Profiles}
\subsection{Robots and Internal Capabilities}
\label{subsec:robots-capabilities}

We consider exactly two anonymous, identical, and homogeneous robots
$R=\{a,b\}$ moving in the Euclidean plane $\mathbb R^2$.
The names $a$ and $b$ are used only for analysis; for $r\in R$,
$\bar r$ denotes the other robot. The position of $r$ at time
$t\ge0$ is denoted by $r(t)$.

Each robot has a local coordinate system centred at its current position.
Two attributes of such a system may differ between the robots, namely the
\emph{orientation} of its axes and its \emph{unit distance}, and both may
also change from one activation to the next; only the handedness is fixed,
since the robots share chirality. Thus a robot's local view is obtained from
the global frame by a translation, a rotation, and a positive scaling, with
the rotation angle and the scaling factor chosen independently for each robot
and each activation (\emph{variable disorientation}, VD). In particular, no
robot can rely on a common direction, on a common unit of length, or on the
persistence of either across its own activations.
When activated, a robot executes a \Look--\Compute--\Move (\LCM) cycle:
it takes an instantaneous snapshot in its local coordinates, computes
a destination by a deterministic common algorithm, and moves
continuously along the corresponding segment to that destination.
Movements are rigid. We assume no general multiplicity-detection
capability.

We consider four standard capability models.
In $\OB$, robots are oblivious and have no explicit communication.
In $\LU$, each robot has a persistent finite-valued visible light,
providing both memory and communication. The intermediate models
$\FS$ and $\FC$ retain only one of these capabilities:
an $\FS$ robot observes only its own light, which acts as finite
private memory, whereas an $\FC$ robot observes only the peer's light,
providing finite communication but no readable private state.
Multiple finite lights, when convenient, are identified with one
light whose color set is their Cartesian product.

For the two-robot setting, let $p_r$ be the position of $\bar r$
in the current local coordinates of $r$, and let $c_r,c_{\bar r}$
be their current colors. The information available at a Look is
\[
\begin{aligned}
  \operatorname{Obs}_{\OB}(r)&=p_r, &
  \operatorname{Obs}_{\FS}(r)&=(p_r,c_r),\\
  \operatorname{Obs}_{\FC}(r)&=(p_r,c_{\bar r}), &
  \operatorname{Obs}_{\LU}(r)&=(p_r,c_r,c_{\bar r}).
\end{aligned}
\]
These definitions also apply when $p_r=0$: in the fixed two-robot
setting the robots can therefore detect coincidence, and an externally
visible peer light remains observable there. This does not amount to
general multiplicity detection. In $\FC$, the robot never reads its
own light; at Compute it may either keep its current color or replace
it by a color determined solely by its observation.\footnote{All
previously established inclusions, equivalences, and separations used
in this paper remain valid for exactly two robots under this observation
convention and the initial-instance and movement assumptions adopted
here~\cite{KSW-2-robot}.}

\subsection{Schedulers and Executions}
\label{subsec:schedulers-events}

All schedulers are fair: each robot is activated infinitely often.
In synchronous executions, time is divided into rounds and each round
activates a nonempty subset of the two robots. Under \FSY{} both robots
are activated in every round, whereas \SSY{} allows any nonempty
activation set. Under \RSY{}, an execution is either fully synchronous
throughout, or has a finite fully synchronous prefix followed by
alternating singleton activations.

Under \ASY{}, the \Look, \Compute, and \Move{} phases of different
robots may overlap arbitrarily, subject to their local phase order,
finite operation durations, and fairness. Movements are continuous
and rigid.

We also use partially atomic subclasses of \ASY{}, distinguished by
which phase of one robot is protected from a \Look{} of the other:
$LC$-atomic protects the Look--Compute interval,
$M$-atomic protects the Move interval, and $CM$-atomic protects the
Compute--Move interval. The strongest variant, \LCM-atomic \ASY{},
protects the whole cycle and is equivalent to \SSY{}.
Their formal timing definitions are given in
Appendix~\ref{app:model-conventions-schduler}.

For the models considered here,
\(
  S \;\ge\; A_{LC} \;\ge\; A,
  \qquad
  S \;\ge\; A_{CM} \;\ge\; A_M \;\ge\; A .
\)
Moreover, $A_M$ and $A_{CM}$ are equivalent for \OB{} and \FS{},
and also for \FC{}~\cite{FSSW23}; hence $A_{CM}$ is omitted from
the quotient landscape.

For compactness, we write
\(
  F,\ R,\ S,\ A,\ A_{LC},\ A_M
\)
for \FSY{}, \RSY{}, \SSY{}, \ASY{}, $LC$-atomic \ASY{}, and
$M$-atomic \ASY{}, respectively.

Throughout the paper, unless explicitly stated otherwise, we assume
variable disorientation, chirality, and rigid movements.
\subsection{Tasks, Solvability, and Profiles}
\label{subsec:problems-relationships}

Let
\(
  \mathcal M=\{\LU,\FC,\FS,\OB\},
  \;
  \mathcal S=\{F,R,S,A,A_{LC},A_M\}.
\)

\paragraph*{Tasks and executions.}
A task $P=(\mathcal I_P,\Phi_P)$ consists of a set of admissible
initial geometric configurations and a temporal predicate on the
robots' geometric motion execution, independent of internal colors,
states, and local coordinate systems.
All basis tasks use the common initial-instance space
\(
  \mathcal I_\star
  =\{(p_a,p_b)\in(\mathbb R^2)^2:p_a\ne p_b\},
\)
and any auxiliary geometric objects are determined from the initial
positions.

We use $\Move_r(t)$ to mean that $r$ is inside a nonzero Move at time
$t$, that is, that the Move phase containing $t$ has a destination
distinct from the position from which it starts. Throughout such a phase
the robot is in motion at every instant: it travels along the segment
without ever pausing, and, movements being rigid, it is not halted before
the destination either. A robot is therefore at rest at time $t$ precisely
when $t$ lies in no nonzero Move; we write $\Stop_r(t)$ for this, so that
$\Stop_r(t)\iff\lnot\Move_r(t)$, and a Move whose destination coincides
with the current position leaves $r$ at rest throughout.
We write $\moves{r}$ for the number of its nonzero Moves and
$\Stopped{r}$ for eventual permanent stationarity.
Their formal definitions, together with the phase-time conventions
used in the appendix proofs, are given in
Appendix~\ref{app:model-conventions-predicate}.

For a deterministic algorithm $\mathcal A$ in model $M$ and scheduler
$K$, let $\operatorname{Exec}_{M,K}(\mathcal A,I)$ be the set of all
admissible fair, time-divergent executions from $I$, including all
permitted choices of schedules, local coordinate systems, phase
timings, and movement speeds. Both robots start in the designated
initial state or color, with no pending cycle.

An algorithm $\mathcal A\in\operatorname{Alg}(M)$ solves
$P=(\mathcal I_P,\Phi_P)$ under $(M,K)$ if
\(
  \forall I\in\mathcal I_P\;
  \forall\rho\in\operatorname{Exec}_{M,K}(\mathcal A,I):
  \pi_{\mathrm{geom}}(\rho)\models\Phi_P,
\)
where $\pi_{\mathrm{geom}}(\rho)$ denotes the geometric
motion execution obtained from $\rho$ by forgetting
internal states, colors, and local coordinate systems.
The task $P$ is solvable under $(M,K)$ if such an
algorithm exists:
\(
  P\in\Sol(M,K)
  \iff
  \exists\mathcal A\in\operatorname{Alg}(M):
  \mathcal A\text{ solves }P\text{ under }(M,K).
\)


We abbreviate $\pi_{\mathrm{geom}}(\rho)\models\Phi_P$ by
$\rho\models\Phi_P$, and write $\Phi_i$ for $\Phi_{P_i}$ when the
indexing family of tasks is clear from the context.

For a setting $\mu=(M,K)$, also written $M^K$, we write
$\Sol(\mu)=\Sol(M,K)$.

\subsection{Computational Order, Solvability Profiles and Disjunction of Tasks }
\label{subsec:computational-profiles}

For settings $\mu,\nu$, define
\(
  \mu\ge\nu
  \iff
  \Sol(\mu)\supseteq\Sol(\nu),
  \;
  \mu\equiv\nu
  \iff
  \Sol(\mu)=\Sol(\nu).
\)
Strict comparison and incomparability have their usual meanings.

Let $\mathfrak M=\mathcal M\times\mathcal S$ be the set of settings,
and let
\(
  \mathcal Q=\mathfrak M/{\equiv}
\)
be the quotient poset of computational equivalence classes, ordered by
\(
  X\le Y
  \iff
  \Sol(X)\subseteq\Sol(Y).
\)

For $S\subseteq\mathcal Q$, let $\Min(S)$ and $\Max(S)$ denote its
minimal and maximal elements, and let $\Ant(\mathcal Q)$ be the family
of antichains.

For a task $P$, its \emph{solvability profile} is
\(
  \Prof(P)=\{X\in\mathcal Q:P\in\Sol(X)\}.
\)
Every solvability profile is an upset of $\mathcal Q$.
For $X\in\mathcal Q$, write
\(
  \upa X=\{Y\in\mathcal Q:X\le Y\}.
\)
A task is \emph{principal at $X$} if
$\Prof(P)=\upa X$.

A family
\(
  \mathcal B=\{B_X:X\in\mathcal Q\}
\)
is a \emph{principal task basis} if
$\Prof(B_X)=\upa X$ for every $X$.
It is an \emph{exact disjunctive basis} if, for every
$A\in\Ant(\mathcal Q)$,
\(
  \Prof\left(\bigvee_{X\in A}B_X\right)
  =
  \bigcup_{X\in A}\upa X.
\)

\begin{proposition}
\label{prop:profile-is-upset}
For every task $P$, $\Prof(P)$ is an upset of $\mathcal Q$.
\end{proposition}

The family of abstract profiles is therefore the finite distributive
lattice
\(
\Up(\mathcal Q)
\)
under union and intersection.  We write
\(
  \Real(\mathcal Q)
  =\{\Prof(P):P\text{ is a two-robot task}\}
  \subseteq\Up(\mathcal Q).
\)

\begin{proposition}
\label{prop:canonical-upset-decomposition}
For every $U\in\Up(\mathcal Q)$,
\(
  U=\bigcup_{X\in\Min(U)}\upa X,
\)
and $\Min(U)$ is the unique antichain giving an irredundant
representation of $U$ by principal upsets.
Consequently, antichains and upsets are in bijection.
\end{proposition}

\begin{lemma}[Boundary certificate]
\label{lem:boundary-certificate}
Let $P$ be a task and $U\in\Up(\mathcal Q)$. Then $\Prof(P)=U$
if and only if
\[
\begin{aligned}
  P&\in\Sol(X)
    &&\text{for every }X\in\Min(U),\\
  P&\notin\Sol(Y)
    &&\text{for every }Y\in\Max(\mathcal Q\setminus U).
\end{aligned}
\]
\end{lemma}

The antichain polynomial is
\(
  \mathsf A_{\mathcal Q}(z)
  =\sum_{A\in\Ant(\mathcal Q)}z^{|A|},
\)
so $\mathsf A_{\mathcal Q}(1)=|\Up(\mathcal Q)|$.


For tasks $P_i=(\mathcal I_\star,\Phi_i)$ over the common initial
space, define
\(
  \bigvee_{i\in I}P_i
  =
  \left(\mathcal I_\star,\bigvee_{i\in I}\Phi_i\right).
\)

\begin{proposition}
\label{prop:disjunction-easy-direction}
For every finite family,
\(
  \bigcup_{i\in I}\Prof(P_i)
  \subseteq
  \Prof\left(\bigvee_{i\in I}P_i\right).
\)
\end{proposition}

The disjunction is \emph{exact} when equality holds.
The reverse inclusion is not automatic; Section~\ref{sec:exact-disjunctions}
establishes it for the antichain-indexed disjunctions of our basis.
Proofs of the elementary results above are given in
Appendix~\ref{app:order-preliminaries}.

\section{The Two-Robot Quotient and Minimum Separation}\label{sec:quotient}
\subsection{Basic Inclusions and Known Equivalences}

\begin{proposition}
\label{prop:basic-inclusions}
For every scheduler $K$ and model $M$,
\[
  \LU^K\ge\FS^K\ge\OB^K,\qquad
  \LU^K\ge\FC^K\ge\OB^K,
\]
and
\(
  M^F\ge M^R\ge M^S,\;
  M^S\ge M^{A_{LC}}\ge M^A,\;
  M^S\ge M^{A_{CM}}\ge M^{A_M}\ge M^A.
\)
\end{proposition}


\begin{proposition}
\label{prop:atomicity-equivalences}
\(
  \FS^A\equiv\FS^{A_{LC}},\quad
  \OB^A\equiv\OB^{A_{LC}},\quad\\
  \FS^{A_M}\equiv\FS^{A_{CM}},\quad
  \OB^{A_M}\equiv\OB^{A_{CM}}.
\)
\end{proposition}

We also use the known equivalences
\(
  \LU^R\equiv\LU^A\equiv\FC^R
  \;\text{\cite{DFPSY,FSW19,KSW-2-robot}},
\)
and
\(
  \FC^A\equiv\FC^{A_M},
  \;
  \FC^{A_{LC}}\equiv\FC^S
  \;\text{\cite{KSW-2-robot}}.
\)
The elementary proofs of the two propositions are deferred to
Appendix~\ref{app:basic-inclusion}.

The following two-robot full-synchrony collapse follows from the
stronger color- and round-preserving simulation of
\cite{KSW-2-robot}, recalled in Appendix~\ref{app:fsynch-no-cost-simulation}.


%

\begin{theorem}[Computational collapse under full synchrony]
\label{thm:equ-FCOMF-by-FSTAF}
For exactly two robots under the assumptions above,
\(
  \LU^F\equiv\FC^F\equiv\FS^F.
\)
\end{theorem}
The proof is deferred to Appendix~\ref{app:fsynch-no-cost-simulation}.

\subsection{The Thirteen Classes and Their Signatures}

Combining the basic capability and scheduler inclusions with the
known equivalences and the two-robot full-synchrony collapse above yields
the thirteen computational classes considered below.


To distinguish these classes, we use the six principal tasks\\
\(
  \mathcal B_6=
  \{\mathrm{EPG}^{*},\mathrm{SRO},\mathrm{MCv},
    \CFSM,\VSS,\DMSD\}.
\)
Three tasks deserve explicit specification here.
$\mathrm{EPG}^{*}$ and $\VSS$ are introduced in this paper, while
$\CFSM$ is a collision-free refinement of Single Move.
The remaining separators SRO~\cite{FSW19}, MCv~\cite{KSW-2-robot}, and DMSD~\cite{KSW-2-robot} are taken from earlier
landscape results; their specifications and references are recalled in
Appendix~\ref{app:geometric-tasks}.

\noindent\textbf{Perpetual Expansion Preserving Gravity
($\mathrm{EPG}^{*}$):} 
Let
\(
  c_0=\frac{a(0)+b(0)}2,
\)
and let $C_i=(a_i,b_i)$ be the sequence of jointly stationary
configurations with consecutive repetitions suppressed.
The task requires this sequence to be infinite and, for every $i\ge0$,
\(
  \frac{a_{i+1}+b_{i+1}}2=c_0,
  \;
  b_{i+1}-a_{i+1}=2(b_i-a_i).
\)

\noindent\textbf{Collision-Free Single Move ($\CFSM$):} 
CFSM strengthens Single Move~\cite{FSSW23}: each robot must perform
exactly one nonzero Move, and the robots must never collide. Thus
\(
  \moves{a}=\moves{b}=1,
  \;
  a(t)\ne b(t)\;\text{for all }t\ge0.
\)

\noindent\textbf{Visit Swap and Stop ($\VSS$):} 
Let $a_0=a(0)$ and $b_0=b(0)$, and set
\(
  A_0=2a_0-b_0,\; B_0=2b_0-a_0 .
\)
VSS requires the prescribed exchange phase to reach one of
\(
  (B_0,b_0),\; (a_0,A_0),\; (B_0,A_0),
\)
after which the robots must eventually remain permanently stationary
at one of the four designated configurations determined by
$A_0,a_0,b_0,B_0$. Every Move ends in $\{A_0,a_0,b_0,B_0\}$; the complete
move-count and terminal conditions are given in Appendix~\ref{app:vss}.

Table~\ref{tab:six-separator-signatures} records, for each class,
which of these six tasks are solvable. Thus each row is a six-bit
solvability signature of that class. Since the thirteen signatures
are distinct, $\mathcal B_6$ separates all thirteen computational
classes.

\begin{lemma}[Bounds for the six separators]
\label{lem:six-bounds}
Under the assumptions of this paper, the following hold.
\renewcommand{\labelenumi}{\textcolor{lipicsGray}{\sffamily\bfseries\upshape(\theenumi)}}
\begin{enumerate}
\item\label{lem:six-cge-star-bounds}
  $\mathrm{EPG}^{*}\in\Sol(\OB^F)$ and
  $\mathrm{EPG}^{*}\notin\Sol(\LU^R)$.
\item\label{lem:six-sro-bounds}
  $\mathrm{SRO}\in\Sol(\OB^R)$ and
  $\mathrm{SRO}\notin\Sol(\FC^S)\cup\Sol(\FS^S)$.
\item\label{lem:six-mcv-bounds}
  $\mathrm{MCv}\in\Sol(\OB^S)$ and
  $\mathrm{MCv}\notin\Sol(\FC^A)\cup\Sol(\FS^{A_M})$.
\item\label{lem:six-cfsm-bounds}
  $\CFSM\in\Sol(\FS^A)$ and
  $\CFSM\notin\Sol(\OB^F)\cup\Sol(\FC^S)$.
\item\label{lem:six-vss-bounds}
  $\VSS\in\Sol(\FC^A)$ and
  $\VSS\notin\Sol(\OB^F)\cup\Sol(\FS^R)$.
\item\label{lem:six-dmsd-bounds}
  $\DMSD\in\Sol(\OB^{A_M})$ and
  $\DMSD\notin\Sol(\FS^A)$.
\end{enumerate}
\end{lemma}

\begin{proof}
Items~(1)--(6) are proved in
Appendices~\ref{app:epgs}, \ref{app:sro}, \ref{app:mcv},
\ref{app:cfs}, \ref{app:vss} and~\ref{app:dmsd}, respectively.
\end{proof}

\begin{theorem}[The thirteen-class order from six signatures]
\label{thm:six-signature-landscape}
Write $\sigma(\mu)=\{P\in\mathcal B_6: P\in\Sol(\mu)\}$. The settings
under consideration form exactly thirteen computational equivalence
classes, represented by the rows of
Table~\ref{tab:six-separator-signatures}, and for any two
representatives $\mu,\nu$,
\[
  \mu\ge\nu
  \quad\Longleftrightarrow\quad
  \sigma(\mu)\supseteq\sigma(\nu).
\]
The certificate is given in Appendix~\ref{app:signature-order-certificate}.
\end{theorem}

\begin{theorem}[Six principal separators]
\label{thm:six-principal-separators}
The tasks in \( \mathcal B_6\)
are principal, with
\[
\begin{aligned}
  \Prof(\mathrm{EPG}^{*})&=\upa\OB^F,
  &\Prof(\mathrm{SRO})&=\upa\OB^R,\\
  \Prof(\mathrm{MCv})&=\upa\OB^S,
  &\Prof(\CFSM)&=\upa\FS^A,\\
  \Prof(\VSS)&=\upa\FC^A,
  &\Prof(\DMSD)&=\upa\OB^{A_M}.
\end{aligned}
\]
\end{theorem}

\begin{proof}
For each task, the positive and maximal negative bounds stated above,
together with upward closure of solvability, determine exactly the
corresponding principal upset.
\end{proof}

\subsection{Minimum Separation}

\begin{theorem}
\label{thm:min-separators}
Any family of tasks separating all thirteen classes of
$\mathcal Q$ contains at least six tasks. Hence
$\mathcal B_6$ is minimum.
\end{theorem}

\begin{proof}
The quotient contains the chain
\[
  \OB^A<\OB^{A_M}<\OB^S<\OB^R<\FS^R<\LU^R<\LU^F.
\]
Along a chain, the solvability set of each task is an upper interval
and hence can create at most one cut. Distinguishing the seven
elements therefore requires at least six tasks.
Table~\ref{tab:six-separator-signatures} shows that $\mathcal B_6$ achieves this
bound.
\end{proof}

\subsection{The Minimum Principal Separator Family}

Among principal profiles, the minimum separating family is unique:
its generators are precisely the join-irreducible elements of
$\mathcal Q$.

\begin{proposition}
\label{prop:unique-principal-separators}
The unique minimum family of principal upsets separating all elements
of $\mathcal Q$ is
\(
 \{
 \upa\OB^F,\,
 \upa\OB^R,\,
 \upa\OB^S,\,
 \upa\FC^A,\,
 \upa\FS^A,\,
 \upa\OB^{A_M}
 \}.
\)

Equivalently, the profiles of the six tasks in $\mathcal B_6$
form the unique minimum principal separating family.
\end{proposition}
The lattice-theoretic characterization of $\mathcal Q$, including
its join-irreducible elements and width, is given in
Appendix~\ref{app:lattice-structure}.

\begin{table}[t]
  \centering
  \small
  \setlength{\tabcolsep}{5pt}
  \renewcommand{\arraystretch}{1.2}
  \caption{Six-task solvability signatures of the thirteen
computational classes. Each row gives the six-bit signature with
respect to $\mathcal B_6$. An entry $1$ indicates solvability and $0$ indicates impossibility.
   Each representative denotes its entire computational equivalence
    class.}
  \label{tab:six-separator-signatures}
  \begin{tabular}{@{}l*{6}{c}@{}}
    \toprule
    \textbf{Representative}
      & $\mathrm{EPG}^{*}$
      & $\mathrm{SRO}$
      & $\mathrm{MCv}$
      & $\CFSM$
      & $\VSS$
      & $\DMSD$ \\
    \midrule
    $\mathrm{LUMI}^{F}$
      & 1 & 1 & 1 & 1 & 1 & 1 \\
    $\mathrm{OBLOT}^{F}$
      & 1 & 1 & 1 & 0 & 0 & 1 \\
    $\mathrm{LUMI}^{R}$
      & 0 & 1 & 1 & 1 & 1 & 1 \\
    $\mathrm{FSTA}^{R}$
      & 0 & 1 & 1 & 1 & 0 & 1 \\
    $\mathrm{OBLOT}^{R}$
      & 0 & 1 & 1 & 0 & 0 & 1 \\
    $\mathrm{FCOM}^{S}$
      & 0 & 0 & 1 & 0 & 1 & 1 \\
    $\mathrm{FCOM}^{A}$
      & 0 & 0 & 0 & 0 & 1 & 1 \\
    $\mathrm{FSTA}^{S}$
      & 0 & 0 & 1 & 1 & 0 & 1 \\
    $\mathrm{FSTA}^{A_{M}}$
      & 0 & 0 & 0 & 1 & 0 & 1 \\
    $\mathrm{FSTA}^{A}$
      & 0 & 0 & 0 & 1 & 0 & 0 \\
    $\mathrm{OBLOT}^{S}$
      & 0 & 0 & 1 & 0 & 0 & 1 \\
    $\mathrm{OBLOT}^{A_{M}}$
      & 0 & 0 & 0 & 0 & 0 & 1 \\
    $\mathrm{OBLOT}^{A}$
      & 0 & 0 & 0 & 0 & 0 & 0 \\
    \bottomrule
  \end{tabular}
\end{table}

\section{A Principal Task Basis for the Thirteen Classes}
\label{sec:Principal-task-basis}

Section~\ref{sec:quotient} showed that six principal tasks suffice,
and are necessary, to separate the thirteen computational classes.
We now extend this family to one task $B_X$ for every
$X\in\mathcal Q$ and prove that each $B_X$ is principal at $X$:
\(
  (\Prof(B_X)=\upa X ).
\)
Thus the thirteen tasks realize all principal upsets of the quotient.
The required solvability and impossibility bounds are summarized in
Table~\ref{tab:PP13}; full task specifications and proofs are deferred
to the appendices.

\input{pp13_four_columns}

\subsection{Principal Tasks and Boundary Certificates}
\label{subsec:principal-boundary-certificates}

For $X\in\mathcal Q$, let
\(
  \mathcal N_X=\Max(\mathcal Q\setminus\upa X).
\)
By Lemma~\ref{lem:boundary-certificate}, a task $B_X$ is principal at $X$
if and only if $B_X\in\Sol(X)$ and $B_X\notin\Sol(Y)$
for every $Y\in\mathcal N_X$.
For principal profiles these negative bounds may be witnessed
separately; simultaneous common-failure executions are needed only
for disjunctions and are treated in
Section~\ref{sec:exact-disjunctions}.

\subsection{The Principal Task Basis}
\label{subsec:basis-proof-obligations}

Write $\mathcal X_6=\{\OB^F,\OB^R,\OB^S,\FS^A,\FC^A,\OB^{A_M}\}$ for the
generators of $\mathcal B_6$ and
$\mathcal X_7=\mathcal Q\setminus\mathcal X_6$.
Table~\ref{tab:PP13} gives the complete assignment and the corresponding
boundary certificates.

\paragraph*{Boundary results for the seven additional tasks.}
All thirteen tasks use $\mathcal I_\star$.
Their specifications and boundary proofs are collected in
Appendices~\ref{app:geometric-tasks} and~\ref{app:new-tasks},
with references in Table~\ref{tab:PP13}.
For $\mathrm{Cv}$, the entry ``Trivial'' refers only to the
empty negative boundary of the least class $\OB^A$;
the Go-to-midpoint algorithm is proved correct in
Appendix~\ref{app:cv}.

\noindent\textbf{Expansion Preserving Gravity ($\mathrm{EPG}$):} 
EPG is the one-shot version of $\mathrm{EPG}^{*}$: preserving the
initial center of gravity at joint stops, the robots double the
initial segment once and then stop forever.
Its formal specification and bounds are given in
Appendix~\ref{app:epg}.

\noindent\textbf{Oscillating Points with One Move per Robot per Phase
($\mathrm{OSPO}$):}
OSPO is a strengthened version of Oscillating Points~\cite{DFPSY}.
Starting from distinct positions, the robots must alternate
approach and retreat forever without ever colliding.
The first approach phase starts at time~$0$, and every phase
is delimited by joint stops.
Within each phase, the distance is monotone and changes strictly
between its endpoints, and each robot completes exactly one
nonzero $\M$.
The phase boundaries must tend to infinity.
The full specification and boundary proofs are given in
Appendix~\ref{app:osp}.

\begin{lemma}[Bounds for $\mathrm{OSPO}$]
\label{lem:ospo-bounds}
Under the assumptions of this paper,
\(
  \mathrm{OSPO}\in\Sol(\FS^R),
\)
whereas
\(
  \mathrm{OSPO}
  \notin
  \Sol(\OB^F)
  \cup
  \Sol(\FS^S)
  \cup
  \Sol(\FC^S).
\)
\end{lemma}

\begin{proof}
The solvability algorithm in $\FS^R$ and the three impossibility
arguments are given in Appendix~\ref{app:osp}.
\end{proof}
\noindent\textbf{Swap--Rendezvous ($\mathrm{S\text{-}RDV}$):} 
The two robots must achieve the common exchange first phase of
Appendix~\ref{app:shared-exchange-phase}, namely form one of the three
configurations in which their order on the initial line has been
reversed, and must then gather at a single point and stay there, never
leaving the segment $[A,B]$ determined by the initial positions.
Its exchange phase is shared with $\VSS$; their subsequent
requirements differ. The full specification is
Definition~\ref{def:s-rdv} in Appendix~\ref{app:s-rdv}.


\begin{lemma}[Bounds for $\mathrm{S\text{-}RDV}$]
\label{lem:s-rdv-bounds}
Under the assumptions of this paper,
\(
  \mathrm{S\text{-}RDV}\in\Sol(\FC^S),
\)
whereas
\(
  \mathrm{S\text{-}RDV}
  \notin
  \Sol(\OB^F)
  \cup
  \Sol(\FS^R)
  \cup
  \Sol(\FC^A).
\)
\end{lemma}

\begin{proof}
The solvability algorithm in $\FC^S$ and the three impossibility
arguments are given in
Appendix~\ref{app:s-rdv-bounds}.
\end{proof}

\subsection{Principal-Basis Theorem}

\begin{theorem}[A principal task basis for the thirteen classes]
\label{thm:principal-task-basis}
Let $\mathcal B=\{B_X:X\in\mathcal Q\}$ be the family specified in
Table~\ref{tab:PP13}. Under the model assumptions and observation
convention of this paper, every task $B_X$ has initial-instance space
$\mathcal I_\star$ and satisfies
\(
  \Prof(B_X)=\upa X.
\)
Consequently, $\mathcal B$ is a principal task basis consisting of
thirteen tasks and extending the six-task separating family
$\mathcal B_6$.
\end{theorem}

\begin{proof}
If $X\in\mathcal X_6$, the identity $\Prof(B_X)=\upa X$ is exactly
one of the six statements of
Theorem~\ref{thm:six-principal-separators}.

If $X\in\mathcal X_7$, the task-specific results referenced in
Table~\ref{tab:PP13}, including
Corollary~\ref{cor:epg_profile}, establish solvability at $X$ and
impossibility at every class in $\mathcal N_X$. Applying
Lemma~\ref{lem:boundary-certificate} gives
$\Prof(B_X)=\upa X$. This also covers $X=\OB^A$, where
$\mathcal N_X=\varnothing$.

Since $\mathcal Q=\mathcal X_6\mathbin{\dot\cup}\mathcal X_7$ and
all tasks use $\mathcal I_\star$, the conclusion follows.
\end{proof}

The six-task family determines the computational order, whereas the
thirteen-task basis realizes every principal upset. These are distinct
roles: the minimality of the former does not make the seven additional
principal profiles redundant. Nor do the individual identities above
establish exactness of task disjunctions. Section~\ref{sec:representative-common-failure} presents representative common-failure constructions,
and Section~\ref{sec:exact-disjunctions} establishes the required property for every
antichain-indexed family from $\mathcal B$, yielding
\(
  \Prof\left(\bigvee_{X\in A}B_X\right)
  =\bigcup_{X\in A}\upa X
  \;\text{for every }A\in\Ant(\mathcal Q).
\)

\section{Representative Common-Failure Construction}
\label{sec:representative-common-failure}

Section~\ref{sec:Principal-task-basis} established each principal profile
individually. Refuting a disjunction requires one admissible fair execution
on which all selected predicates fail. We fix an arbitrary deterministic
algorithm $\mathcal A$ and an initial instance $\iota\in\mathcal I_\star$;
all predicates refer to this same instance, and execution blocks preserve
the reached states and pending operations without reset.

For the $\FS$ constructions below, \emph{normalization} means presenting
the peer at $e=(1,0)$ at every Look with distinct positions, using the
rotations and rescalings permitted by VD. The destination in these
normalized coordinates and the next private state then depend only on
the current private state. Finite safety violations persist under
extension; liveness violations require a suitable fair infinite continuation.

\subsection{Disjunctions with Perpetual Expansion}
\label{subsec:cf-expansion}

We group the pairs $\mathrm{EPG}^{*}\vee P$ with
\(
  P\in\mathcal C_{\mathrm{exp}}
  :=\{\mathrm{HET},\mathrm{OSPO},\SRDV,\mathrm{ZCC},\VSS,\mathrm{VEC},\CFSM\}.
\)
Let $C_0,C_1,\ldots$ be the joint-stop configurations in chronological
order with consecutive repetitions removed, and put
$c_0=\tfrac12(a(0)+b(0))$: $\mathrm{EPG}^{*}$ asks
this sequence to be infinite and every step of it to keep the center of gravity
at $c_0$ while doubling the segment. Keeping the center of gravity is exactly
what rules out a solo expansion, and this is the component shared by all seven
pairs.

\begin{lemma}
\label{lem:solo-drift}
Let $\rho$ be an admissible execution with joint-stop sequence $(C_i)_{i\ge0}$.
\begin{enumerate}
\item If one of the two robots completes no nonzero $\M$ between the joint stop
  realizing $C_i$ and the one realizing $C_{i+1}$, then
  \(
    \tfrac12(a_{i+1}+b_{i+1})\neq\tfrac12(a_i+b_i).
  \)
\item Let $K\in\{R,S,A_{LC},A_M,A\}$. For every algorithm there is a fair
  execution in which the $\LCM$ cycles of the two robots are pairwise disjoint
  in time. In such an execution either no nonzero $\M$ ever occurs, in which
  case $(C_i)$ has the single term $C_0$ and the perpetuity condition fails, or
  $\tfrac12(a_1+b_1)\neq c_0$ and the expansion condition fails at $i=0$.
  In both cases $\Phi_{\mathrm{EPG}^{*}}$ fails, and so does
  $\Phi_{\mathrm{EPG}}$.
\end{enumerate}
\end{lemma}

For (1), only one endpoint changes between two distinct joint-stop
configurations, so their midpoints differ. For (2), alternate complete
cycles of $a$ and $b$, assigning each cycle unit duration. The execution
is fair and time-divergent, belongs to $R$ and $S$, and respects every
atomicity constraint because cycles do not overlap. Either no robot
moves and progress fails, or the first nonzero $\M$ changes the midpoint
by (1). Formal timing details are in Appendix~\ref{app:solo-drift}.

\paragraph*{Combining the obstruction with another failure.}
A solo first nonzero $\M$ suffices to violate $\mathrm{EPG}^{*}$; later
cycles need not remain serialized. To refute $\mathrm{EPG}^{*}\vee P$,
we must also violate $P$ on that same execution, from the state actually
reached. An admissible fair tail preserves finite safety violations,
whereas a liveness failure requires a specifically constructed fair
continuation. The next proposition supplies a common execution for
both exchange tasks.

\begin{proposition}[Perpetual expansion and exchange tasks]
\label{prop:cf-expansion-exchange}
For every deterministic algorithm in $\FS^R$, there is an admissible
fair execution that violates $\mathrm{EPG}^{*}$,
$\mathrm{S\text{-}RDV}$, and $\VSS$ simultaneously.
In particular, for each
$E\in\{\mathrm{S\text{-}RDV},\VSS\}$,
\[
  \mathrm{EPG}^{*}\vee E\notin\Sol(\FS^R).
\]
\end{proposition}

\begin{proof}
Use coordinates on the initial line with $a_0=d$, $b_0=2d$, $d>0$,
so the outer points are $0,3d$ and $c_0=3d/2$. Activate $a,b$
alternately, completing each cycle before the next one starts; this is
an admissible fair $R$ execution with an empty fully synchronous prefix.
Under normalized observations, the two robots traverse the same private
state sequence until the first nonzero $\M$.

If no such $\M$ occurs, all three tasks fail their progress requirements.
Otherwise $a$ performs the first nonzero $\M$ alone, changing the center
of gravity at the next joint stop and violating $\mathrm{EPG}^{*}$.
The shared exchange phase requires one of $(3d,2d),(d,0),(3d,0)$ with
cumulative nonzero-move counts $(1,0),(0,1),(1,1)$, respectively.
Thus this first solo $\M$ of $a$ must land at $3d$: any other landing
would require another nonzero $\M$ of $a$ to reach an admissible
exchange configuration, exceeding its allowed count.

If $a$ lands at $3d$, then $b$, at its next activation, has the same
private state as $a$ had before that $\M$ and receives the same normalized
input. Hence it performs the same normalized displacement, which in its
current physical frame is twice its current distance to the peer,
and therefore lands at $2d+2d=4d$. This lies outside the segment $[0,3d]$ required
by $\SRDV$ and outside the allowed endpoint set $\{0,d,2d,3d\}$ of $\VSS$.
In either case both exchange tasks fail irreversibly. Continue with fair
alternating activations; all three violations persist.
\end{proof}

The remaining companions in $\mathcal C_{\mathrm{exp}}$ use the same solo
expansion obstruction, with task- and class-specific prefixes or fair tails.
These applicability and continuation arguments are recorded in
Appendix~\ref{app:cf-expansion-family}. The complete boundary coverage and the
resulting profile equalities are assembled in
Section~\ref{sec:exact-disjunctions}; the obstruction alone does not supply
those companion-task arguments.

\subsection{Shrinking Rotation and Exchange Tasks}
\label{subsec:cf-sro-exchange}

Let $\mathcal E=\{\SRDV,\VSS\}$. The two tasks share the exchange phase, so a construction that defeats the
shared phase defeats both at once.

\begin{proposition}[Shrinking rotation and exchange tasks]
\label{prop:cf-sro-exchange-family}
For every deterministic algorithm in $\FS^S$ there is an admissible fair
execution that violates $\mathrm{SRO}$ and both tasks in $\mathcal E$
simultaneously. Consequently
$\mathrm{SRO}\vee E\notin\Sol(\FS^S)$ for every $E\in\mathcal E$.
\end{proposition}

\begin{proof}[Proof sketch]
We first recall the geometric obstruction. SRO permits a clockwise
$90^\circ$ rotation of the segment, or a clockwise $45^\circ$ rotation
with scale $1/\sqrt2$. With the peer fixed at $(0,0)$ and the moving
robot initially at $(-1,0)$, the four possible two-step choices place
the moving robot at one of
\(
  (1,0),\; (1/2,1/2),\; (0,1/2).
\)
All lie outside the square with the initial segment as diagonal,
contrary to SRO's containment condition. Thus two consecutive solo
nonzero moves defeat SRO (Lemma~\ref{lem:cf-sro-two-moves}).

Use the same initial positions and normalized observations as in
Proposition~\ref{prop:cf-expansion-exchange}. If the initial state orbit
contains no nonzero output, a fair alternating execution never moves
and violates all three tasks. Otherwise activate $a$ alone until its
first nonzero $\M$; $S$ permits any finite number of solo activations.
If $a$ lands at $3d$, the segment reverses without changing length,
violating SRO. Activating $b$ through its preceding zero-output states
then reproduces the move to $4d$, violating both exchange tasks.

For any other landing, both exchange tasks have already failed.
Unless SRO also fails immediately, the robots remain distinct.
If $a$ has another nonzero output, execute it with $b$ fixed and apply
the two-move obstruction. Otherwise let $b$ perform its corresponding
first nonzero $\M$. Either this violates SRO, or both robots now follow
the same zero-output state suffix. Fair alternating activations then
leave the configuration stationary forever, violating SRO's perpetuity.
In the finite-violation cases, append a fair alternating suffix.
Appendix~\ref{app:cf-sro-family} gives the complete argument.
\end{proof}

This covers the entire maximal negative boundary of $\mathrm{SRO}\vee\VSS$ and
the $\FS^S$ boundary of $\mathrm{SRO}\vee\SRDV$; the latter pair also has the
boundary class $\FC^A$, treated in Appendix~\ref{app:cf-sro-family}.






\section{Exact Disjunctions via Common-Failure Executions}
\label{sec:exact-disjunctions}

Section~\ref{sec:Principal-task-basis} established the individual
principal profiles, and Section~\ref{sec:representative-common-failure}
exhibited representative executions violating several task predicates
simultaneously. We now combine those constructions with the remaining
boundary cases in the appendix to prove exactness for every
antichain-indexed disjunction. The argument separates the execution
constructions from the order-theoretic step that identifies the
resulting profiles.

\subsection{Reduction to Maximal Unsolvable Classes}
\label{subsec:exactness-boundary-reduction}

For every $A\in\Ant(\mathcal Q)$, write
\(
  B_A=\bigvee_{X\in A}B_X,
  \;
  U_A=\bigcup_{X\in A}\upa X,
  \;
  \mathcal N_A=\Max(\mathcal Q\setminus U_A).
\)
Principality immediately gives $U_A\subseteq\Prof(B_A)$: a class in
$U_A$ can solve one of the disjuncts. The converse does not follow
from the individual impossibility results. It requires one execution
on which all the selected predicates fail.

\begin{lemma}[Exactness from boundary common failure]
\label{lem:boundary-common-failure}
Let $A\in\Ant(\mathcal Q)$. Suppose that, for every $Y\in\mathcal N_A$
and every deterministic algorithm $\mathcal A$ in $Y$, there are an
initial instance $\iota\in\mathcal I_\star$ and an admissible fair
execution $\rho$ of $\mathcal A$ from $\iota$ such that
\(
  \rho\not\models\Phi_{B_X}
  \;\text{for every }X\in A,
\)
where $\Phi_{B_X}$ denotes the temporal predicate of the principal
task $B_X$ of Table~\ref{tab:PP13}.
We then say that $\{B_X:X\in A\}$ has the \emph{common-failure
property} in $Y$. Under this hypothesis, $\Prof(B_A)=U_A$.
\end{lemma}

\begin{proof}
If $Z\in U_A$, then $Z$ solves some $B_X$ with $X\in A$, and the same
algorithm solves $B_A$.

Conversely, the assumed executions show that $B_A\notin\Sol(Y)$
for every $Y\in\mathcal N_A$. Every $Z\notin U_A$ lies below some
such $Y$, by finiteness of $\mathcal Q$. If $B_A$ were solvable in
$Z$, it would also be solvable in $Y$, a contradiction. Hence
$\Prof(B_A)\subseteq U_A$.
\end{proof}

The witness instance and execution may depend on $A$, $Y$, and
$\mathcal A$, but not on the choice of a disjunct within $A$.
Different boundary classes may use different executions. Thus, the
remaining proof obligations are indexed by pairs $(A,Y)$ with
$Y\in\mathcal N_A$, rather than by all classes outside $U_A$.

\subsection{Coverage of Two- and Three-Element Antichains}
\label{subsec:common-failure-coverage}

The thirteen-class order has twenty-four two-element antichains and
sixteen three-element antichains, with no larger antichains.
Table~\ref{tab:cf-family-coverage} groups them by the construction
families introduced in
Section~\ref{sec:representative-common-failure}. For compactness, set
\[
\begin{aligned}
  \mathcal E
    &=\{\mathrm{S\text{-}RDV},\VSS\},\\
  \mathcal P_3
    &=\{\mathrm{ZCC},\mathrm{VEC},\CFSM\},
  &\mathcal P_4&=\{\mathrm{OSPO}\}\cup\mathcal P_3.
\end{aligned}
\]
We retain $\mathcal C_{\mathrm{exp}}$ from
Section~\ref{subsec:cf-expansion}. Each displayed task family determines
an antichain through the principal generators in
Table~\ref{tab:PP13}. The groups are disjoint and exhaustive.

\begin{table}[t]
  \centering
  \small
  \setlength{\tabcolsep}{4pt}
  \renewcommand{\arraystretch}{1.15}
  \caption{Coverage of nonprincipal profiles by construction family.
    The last column counts the boundary obligations $(A,Y)$, namely
    $\sum_A|\mathcal N_A|$ within the corresponding row.
    Full boundaries and construction references are listed in
    Tables~\ref{tab:PP13-join2} and~\ref{tab:PP13-join3}.}
  \label{tab:cf-family-coverage}
  \begin{tabular}{@{}clrr@{}}
    \toprule
    $|A|$ & Disjunction family & No. of $A$ & No. of $(A,Y)$ \\
    \midrule
    2 & $\mathrm{EPG}^{*}\vee P$, $P\in\mathcal C_{\mathrm{exp}}$
      & 7 & 16 \\
    2 & $\mathrm{SRO}\vee P$, $P\in\mathcal E\cup\mathcal P_3$
      & 5 & 8 \\
    2 & $E\vee P$, $E\in\mathcal E$, $P\in\mathcal P_4$
      & 8 & 18 \\
    2 & $\mathrm{MCv}\vee P$,
        $P\in\{\VSS,\mathrm{VEC},\CFSM\}$
      & 3 & 4 \\
    2 & $\DMSD\vee\CFSM$
      & 1 & 1 \\
    \cmidrule(lr){2-4}
      & Two-element total & 24 & 47 \\
    \midrule
    3 & $\mathrm{EPG}^{*}\vee E\vee P$,
        $E\in\mathcal E$, $P\in\mathcal P_4$
      & 8 & 18 \\
    3 & $\mathrm{SRO}\vee E\vee P$,
        $E\in\mathcal E$, $P\in\mathcal P_3$
      & 6 & 13 \\
    3 & $\mathrm{MCv}\vee\VSS\vee P$,
        $P\in\{\mathrm{VEC},\CFSM\}$
      & 2 & 3 \\
    \cmidrule(lr){2-4}
      & Three-element total & 16 & 34 \\
    \bottomrule
  \end{tabular}
\end{table}

\paragraph*{Coverage.}
For the expansion pairs, Lemma~\ref{lem:solo-drift} supplies the solo-expansion
obstruction, combined with a compatible prefix or fair tail for the companion
task; Proposition~\ref{prop:cf-expansion-exchange} discharges the two $\FS^R$
cases directly. For shrinking rotation,
Proposition~\ref{prop:cf-sro-exchange-family} discharges the two $\FS^S$ cases.
The remaining pairs use the exchange obstruction of
Appendix~\ref{app:shared-exchange-phase}, delayed movements against $\mathrm{MCv}$,
or an observation during movement against $\DMSD$, as recorded row by row in
Table~\ref{tab:PP13-join2}. For three-element antichains, a two-task
construction is reused wherever one exists and the third predicate is checked on
that same execution, with every earlier violation preserved; the cases whose
maximal boundary belongs to no constituent pair, namely the $\OB^S$ and
$\OB^{A_M}$ cases, get their own constructions. Table~\ref{tab:PP13-join3}
indexes all thirty-four.

\begin{proposition}[Common failure at all nonprincipal boundaries]
\label{prop:all-common-failure-boundaries}
For every $A\in\Ant(\mathcal Q)$ with $2\le |A|\le3$ and every
$Y\in\mathcal N_A$, the family $\{B_X:X\in A\}$ has the common-failure
property in $Y$.
\end{proposition}

\begin{proof}
The five two-task rows of Table~\ref{tab:cf-family-coverage} exhaust the
twenty-four two-element antichains and the three remaining rows exhaust the
sixteen three-element ones. Table~\ref{tab:PP13-join2} and
Table~\ref{tab:PP13-join3} name, for each listed $Y\in\mathcal N_A$, the
construction that violates every predicate of $A$ on one admissible fair
execution.
\end{proof}

\subsection{Exact Disjunctive Basis}
\label{subsec:exact-disjunctive-basis}

\begin{theorem}[An exact disjunctive basis]
\label{thm:exact-disjunctive-basis}
The principal task basis $\mathcal B=\{B_X:X\in\mathcal Q\}$ is an
exact disjunctive basis. For every $A\in\Ant(\mathcal Q)$,
\(
  \Prof\left(\bigvee_{X\in A}B_X\right)
  =\bigcup_{X\in A}\upa X.
\)
\end{theorem}

\begin{proof}
For $A=\varnothing$, both sides are empty by the convention for the
empty disjunction. For $|A|=1$, the claim is
Theorem~\ref{thm:principal-task-basis}. For $2\le |A|\le3$, apply
Proposition~\ref{prop:all-common-failure-boundaries} and
Lemma~\ref{lem:boundary-common-failure}. By
Lemma~\ref{lem:width-three} these cases exhaust all antichains.
\end{proof}


\section{Complete Realization and a Minimum Exact Disjunctive Basis}
\label{sec:complete-realization}


\subsection{Complete Realization and Minimality}
\label{subsec:complete-realization-canonical-form}
\label{subsec:minimum-exact-disjunctive-basis}

\begin{theorem}[Realization and canonical disjunctive normal form]
\label{thm:profile-realization}
For every nonempty upset $U\in\Up(\mathcal Q)$, define
\(
  B_U\coloneqq\bigvee_{X\in\Min(U)}B_X.
\)
Then $\Prof(B_U)=U$, and $\Min(U)$ is the unique irredundant set of
principal generators of $U$; the expression above is therefore the
canonical disjunctive normal form of $U$ at the profile level. With
$B_\varnothing=(\mathcal I_\star,\mathsf{false})$ the same identity
holds for $U=\varnothing$.
\end{theorem}

\begin{proof}
$\Min(U)$ is an antichain, so
Theorem~\ref{thm:exact-disjunctive-basis} and
Proposition~\ref{prop:canonical-upset-decomposition} give
$\Prof(B_U)=\bigcup_{X\in\Min(U)}\upa X=U$. If
$U=\bigcup_{X\in S}\upa X$ irredundantly, then every $Y\in\Min(U)$
lies in some $\upa X$ with $X\in S\subseteq U$, and minimality
forces $X=Y$; any further generator lies above a minimal element and is
redundant. Hence $S=\Min(U)$.
\end{proof}

Theorem~\ref{thm:profile-realization} realizes every upset by a task over
$\mathcal I_\star$; together with upward closure, this gives
%
\begin{equation}
  \Real(\mathcal Q)=\Up(\mathcal Q).
  \label{eq:complete-profile-realization}
\end{equation}
Every task profile thus has the canonical representative
$B_{\Prof(P)}$. This is equality of solvability profiles, not of task
specifications.

The matching lower bound is a property of the profile lattice alone and
does not assume that a competing family consists of principal tasks.

\begin{theorem}[Minimality of the principal-upset basis]
\label{thm:minimum-profile-basis}
Let $\mathcal G\subseteq\Up(\mathcal Q)$ be a family whose unions
generate every element of $\Up(\mathcal Q)$. Then
$\upa X\in\mathcal G$ for every $X\in\mathcal Q$. Consequently
$\{\upa X:X\in\mathcal Q\}$ is the unique least union-generating
family of profiles, of cardinality thirteen.
\end{theorem}

\begin{proof}
Write $\upa X=\bigcup_{G\in\mathcal G_X}G$ with
$\mathcal G_X\subseteq\mathcal G$. Some $G\in\mathcal G_X$ contains $X$,
and being an upset it contains $\upa X$; the displayed equality
gives the reverse inclusion, so $G=\upa X$. Conversely the
principal upsets generate every upset by canonical decomposition, and
they are distinct.
\end{proof}

\begin{corollary}[Minimum size and uniqueness at the profile level]
\label{cor:minimum-exact-task-basis}
Any task family $\mathcal D$ supplying an exact
representation\footnote{A finite disjunction from $\mathcal D$ is an
\emph{exact representation} of $U$ when the selected subfamily
$\mathcal D_U\subseteq\mathcal D$ satisfies
$\Prof(\bigvee_{P\in\mathcal D_U}P)
 =\bigcup_{P\in\mathcal D_U}\Prof(P)=U$; the empty disjunction is
allowed, and only the representations used to generate the profiles
need be exact.} of every upset of $\mathcal Q$ contains at least
thirteen tasks. The family $\mathcal B$ attains this bound, and if
$|\mathcal D|=13$ its task profiles are exactly
$\{\upa X:X\in\mathcal Q\}$, each occurring once.
\end{corollary}

\begin{proof}
The profiles $\{\Prof(P):P\in\mathcal D\}$ generate every upset by
union, so Theorem~\ref{thm:minimum-profile-basis} forces all thirteen
principal upsets to occur; $\mathcal B$ attains the bound by
Theorem~\ref{thm:profile-realization}.
\end{proof}

Uniqueness here concerns the profiles, not the predicates realizing
them: a task may be replaced by another with the same profile, though
this need not preserve exactness of disjunctions. Nor does the lower
bound claim minimality among families whose disjunctions may have
profiles strictly larger than the union of their components.

\subsection{Structural Consequences}
\label{subsec:complete-realization-consequences}

By the antichain--upset correspondence and
Theorem~\ref{thm:profile-realization}, the antichain polynomial
$1+13z+24z^{2}+16z^{3}$ counts the realizable profiles: thirteen
principal, twenty-four with two minimal classes, sixteen with three,
and the empty one, so
$|\Real(\mathcal Q)|=|\Up(\mathcal Q)|=54$. Every nonempty profile
therefore has an exact canonical expression using at most three tasks
from $\mathcal B$, and three are sometimes necessary: for the
three-element antichain $A_0=\{\OB^F,\FS^R,\FC^S\}$ the upset
$U_0=\bigcup_{X\in A_0}\upa X$ has exactly these three minimal
elements and canonical representative
\(
  B_{U_0}=\mathrm{EPG}^{*}\vee\mathrm{OSPO}\vee\SRDV,
\)
and no union of two principal profiles equals $U_0$ by uniqueness of
the irredundant decomposition. This bounds the number of
\emph{principal} terms, not the number of arbitrary tasks needed to
realize a given profile.
Figure~\ref{fig:complete-profile-venn} displays the $53$ nonempty profiles,
labelled by their canonical principal-task representations.

\input{STACS-Venn-Diagram}


\section{Conclusion}
We completely characterize the two-robot solvability-profile
landscape: all $54$ upsets of $\mathcal Q$ are realizable.
Minimum separation requires six principal tasks, whereas minimum
exact disjunctive generation requires thirteen.
%

\section*{Acknowledgment}
This work was supported by JSPS KAKENHI Grant Numbers
JP20KK0232, JP25K15057, 
JP25K03078, 
JP25K03079, 
JP26K02865, 
and by JST FOREST Program JPMJFR226U.

\newpage
\bibliographystyle{alpha}
\bibliography{referencesorg}

\clearpage
\appendix
%
%
\noindent
{\Large
{\bf Appendix}}
\section{Model Conventions and Properties of the Lattice}

\subsection{Formal Definition of Schedulers}\label{app:model-conventions-schduler}
\input{STACS-scheduler-old}

\subsection{Formal Definition of Temporal Predicates}\label{app:model-conventions-predicate}

A colored configuration at time $t$ is
\(
  \mathcal C(t)
  =\bigl((a(t),\mathrm{color}_a(t)),
         (b(t),\mathrm{color}_b(t))\bigr),
\)
where the ordering is used only for analysis.
Its geometric and color components are $(a(t),b(t))$ and
$(\mathrm{color}_a(t),\mathrm{color}_b(t))$, respectively.
In $\OB$, the color set is a singleton.

For a fixed execution, we use the following predicates,
suppressing its name from the notation.
For each robot $r\in\{a,b\}$, its \emph{moving-time set} and
\emph{moving predicate} are
\[
  \MoveSet{r}
  :=\bigcup_{\substack{i\ge1:\\
       r(t_B(r,i))\ne r(t_E(r,i))}}
       \bigl(t_B(r,i),t_E(r,i)\bigr),
  \qquad
  \Move_r(t)
  \quad:\Longleftrightarrow\quad
  t\in\MoveSet{r},
\]
the union being taken over the Move operations of $r$ whose starting
and ending positions differ. A robot is stationary exactly when it is
not inside such a Move:
\[
  \Stop_r(t)
  \quad:\Longleftrightarrow\quad
  t\notin\MoveSet{r}.
\]
Since $\MoveSet{r}$ is a union of open intervals, the endpoints of
nonzero Move intervals are included among the stopping times.

Writing the open set $\MoveSet{r}$ as the disjoint union
of its maximal open intervals, we define the number of
nonzero moves by
\[
  \MoveSet{r}
  =\bigsqcup_{k=1}^{N_r}I_r^{(k)},
  \;
  \moves{r}:=N_r\in\mathbb Z_{\ge0}\cup\{\infty\},
\]
and, for $t\ge0$, the number of nonzero Moves of $r$ \emph{completed}
by time $t$ by
\[
  m_r(t):=\bigl|\{\,k\le N_r:\ \sup I_r^{(k)}\le t\,\}\bigr| .
\]
Thus $m_r$ is non-decreasing with $m_r(0)=0$, and
$\moves{r}=\lim_{t\to\infty}m_r(t)$.
The strict separation between successive Move operations
ensures that each such interval corresponds to one nonzero
Move operation.
Permanent stopping is expressed by
\[
  \Stopped{r}
  \quad:\Longleftrightarrow\quad
  \exists T\ge0\ \forall t\ge T:\ \Stop_r(t).
\]

Finally, the \emph{joint stops} of the execution are the times at which
neither robot is moving:
\[
  \Sab
  :=\{\,t\ge0:\ \Stop_a(t)\wedge\Stop_b(t)\,\}
  =[0,\infty)\setminus\bigl(\MoveSet{a}\cup\MoveSet{b}\bigr).
\]
This set is closed in $[0,\infty)$ and contains $0$, since no cycle is
pending at time~$0$. The geometric configuration $(a(t),b(t))$ is
constant on each maximal subinterval of $\Sab$. Listing these
configurations in chronological order and suppressing consecutive
repetitions gives the \emph{joint-stop sequence} $C_0,C_1,\ldots$
of the execution, with $C_i=(a_i,b_i)$ and $d_i=b_i-a_i$.

\subsection{Proofs of Results in Section~\ref{subsec:computational-profiles}}\label{app:order-preliminaries}
\begin{proof}[Proof of Proposition~\ref{prop:profile-is-upset}]
Let $X\in\Prof(P)$ and $X\preceq Y$. Then
\[
  P\in\Sol(X)\subseteq\Sol(Y),
\]
and hence $Y\in\Prof(P)$.
\end{proof}

%

\begin{proof}[Proof of Proposition~\ref{prop:canonical-upset-decomposition}]
Since $\mathcal{Q}$ is finite, every element of $U$ lies
above a minimal element of $U$.
Thus $U=\bigcup_{X\in\Min(U)}\upa X$.
Conversely, let $A$ be an antichain with
$U=\bigcup_{X\in A}\upa X$. Every element of $A$ is minimal
in $U$, since $A$ is an antichain, and every minimal element of $U$
lies above, hence equals, some element of $A$. Thus $\Min(U)=A$.
The representation is irredundant because each
$X\in\Min(U)$ belongs to no other principal upset
in the union.
\end{proof}

\begin{proof}[Proof of Lemma~\ref{lem:boundary-certificate}]
Necessity is immediate. For sufficiency, every $Z\in U$ lies above
some $X\in\Min(U)$, since $\mathcal Q$ is finite. Thus $P\in\Sol(Z)$
by Proposition~\ref{prop:profile-is-upset}. Every $Z\notin U$ lies
below some $Y\in\Max(\mathcal Q\setminus U)$. If $P$ were solvable in
$Z$, it would also be solvable in $Y$, since $Z\le Y$, contradicting
the assumed impossibility at $Y$. Hence $\Prof(P)=U$.
\end{proof}
\begin{proof}[Proof of Proposition~\ref{prop:disjunction-easy-direction}]
If a model class solves $P_j$, the same algorithm always satisfies
$\Phi_j$ and therefore always satisfies $\bigvee_{i\in I}\Phi_i$.
\end{proof}

%
\section{Model Comparison}
\label{app:model-conparison}

\subsection{Basic inclusions}\label{app:basic-inclusion}

\begin{proof}[Proof of Proposition~\ref{prop:basic-inclusions}]
The cross-model inclusions follow by ignoring capabilities:
a $\LU$ robot can execute an $\FS$ or $\FC$ algorithm,
and either intermediate model can execute an $\OB$ algorithm.

The scheduler inclusions follow by restricting admissible
executions. If every schedule and timing choice permitted
by $K_1$ is also permitted by $K_2$, then any algorithm
correct under $K_2$ remains correct under $K_1$.
Thus,
\(
  \Sol(M,K_2)\subseteq\Sol(M,K_1).
\)
Applying this observation to the scheduler restrictions
in the model definitions proves the stated inclusions.
\end{proof}

\begin{proof}[Proof of Proposition~\ref{prop:atomicity-equivalences}]
These equivalences follow directly from the observation capabilities
of $\FS$ and $\OB$ and from the definitions of the corresponding
atomic scheduler classes.
\end{proof}

\subsection{Correctness of the Color- and Round-Preserving Simulation}
\label{app:fsynch-no-cost-simulation}


The following color- and round-preserving full-synchrony simulation
for two robots was established in~\cite{KSW-2-robot}.

\begin{theorem}[\cite{KSW-2-robot}]
\label{thm:fsynch-no-cost-simulation}
\label{lemma:no-cost-sim}
Let $\mathcal A$ be a deterministic two-robot $\LU^F$ algorithm with
a palette $C$ of $k$ colors and a designated initial color $c_0$.
There are algorithms $\mathcal A_S$ in $\FS^F$ and $\mathcal A_C$ in
$\FC^F$, both using the palette $C$, such that, for every initial
geometric configuration $I$ and every execution $\rho$ of either
simulator from $I$, there is an admissible execution $\widehat\rho$
of $\mathcal A$ from $I$ with
\(
  \pi_{\mathrm{geom}}(\rho)
  =\pi_{\mathrm{geom}}(\widehat\rho).
\)
The matching execution has the same round boundaries and movement
timings. Each simulator round implements exactly one round of
$\mathcal A$, without auxiliary colors or initialization rounds.
\end{theorem}


\begin{proof}
Write
\[
  \mathcal A(p,c_r,c_{\bar r})=(q,c').
\]
Fix $e=(0,1)$. For $p\neq0$, let $N_p$ be the
orientation-preserving similarity mapping $p$ to $e$,
and let $N_0=\mathrm{id}$.

\paragraph*{Construction.}
Both simulators start with $c_0$. In each round, a robot with peer position $p$ reads the available color $c$: its own color in $\FS^F$, or its peer's color in $\FC^F$.
It computes
\[
  (q,c')=\mathcal A(N_{p}(p),c,c),
\]
sets its state or external light to $c'$, and moves to $N_{p}^{-1}(q)$ in its actual local coordinates.
This uses only permitted observations and the palette $C$; $N_p$ is computed from the current snapshot and
 is not stored between rounds.

 \paragraph*{Correctness.}
Fix an arbitrary execution $\rho$ of either simulator.
We construct a matching execution $\widehat{\rho}$ of the
$\LU^F$ algorithm $\mathcal A$ by induction over rounds.

At the beginning of round $j$, we maintain the following invariant:
the corresponding robots in $\rho$ and $\widehat{\rho}$ occupy the
same positions, and the two robots have a common color $c_j$.
The invariant holds initially by the designated common initial color.

Let $\psi_{r,j}$ be the map from global coordinates to robot $r$'s
actual local coordinates at its Look in round $j$, and let
$p_{r,j}$ be the observed peer position in these coordinates.
In the corresponding execution of $\mathcal A$, choose the virtual
local coordinate system
\[
  \widehat{\psi}_{r,j}
  =
  N_{p_{r,j}}\circ\psi_{r,j}.
\]
Because $N_{p_{r,j}}$ is an orientation-preserving similarity,
$\widehat{\psi}_{r,j}$ is admissible under variable disorientation
and preserves chirality.

By the induction hypothesis, the geometric configuration and the
two colors at the beginning of the round coincide in the two
executions. Hence, under the virtual coordinate system,
each source robot receives exactly the normalized input used by
the simulator. In particular, when the robots are at distinct
positions, both receive
\[
  (e,c_j,c_j),
\]
whereas at coincidence both receive
\[
  (0,c_j,c_j).
\]
Therefore, determinism of $\mathcal A$ implies that both corresponding
computations produce the same next color $c_{j+1}$.

Let $q_j$ be the destination computed by $\mathcal A$ in the
normalized coordinates. The global destination of robot $r$ in
$\widehat{\rho}$ is
\[
  \widehat{\psi}_{r,j}^{-1}(q_j)
  =
  \psi_{r,j}^{-1}
  \bigl(
    N_{p_{r,j}}^{-1}(q_j)
  \bigr),
\]
which is exactly the destination used by the simulator.

We now choose the same round boundaries, phase timings, and progress
along each rigid movement segment in $\widehat{\rho}$ as in $\rho$.
Thus the two executions have identical geometric trajectories
throughout round $j$, and the common-color invariant is preserved
for round $j+1$.

Induction over all rounds therefore yields an admissible infinite
execution $\widehat{\rho}$ of $\mathcal A$ satisfying
\[
  \pi_{\mathrm{geom}}(\rho)
  =
  \pi_{\mathrm{geom}}(\widehat{\rho}).
\]
Moreover, each simulator round implements exactly one round of
$\mathcal A$, and no auxiliary colors or initialization rounds are
introduced. Hence the simulator uses the same palette and preserves
the round and movement timings claimed in the theorem.
\end{proof}

\paragraph*{Scope of the construction.}
The correspondence constructed above maps every simulator execution
to an admissible execution of the original $\LU^F$ algorithm with
the same geometric trajectory. The virtual local coordinate systems
of the source execution need not coincide with those of the simulator.
The construction also relies on the designated common initial color;
it does not address arbitrary initial color assignments.

\begin{proof}[Proof of Theorem~\ref{thm:equ-FCOMF-by-FSTAF}]
The quantitative simulation theorem (Theorem~\ref{thm:fsynch-no-cost-simulation})
 shows that every
$\LU^F$ algorithm can be simulated in both $\FS^F$ and $\FC^F$,
preserving the geometric execution.
Hence
\[
  \Sol(\LU^F)\subseteq\Sol(\FS^F)
  \cap\Sol(\FC^F).
\]
The reverse inclusions follow immediately because $\FS$ and $\FC$
are submodels of $\LU$.
\end{proof}

\section{Definitions and Boundary Proofs of the Geometric Basis Tasks}\label{app:geometric-tasks}
\input{Definitions-Boundary}

\section{The New and Revised Tasks}\label{app:new-tasks}
\input{New-Tasks}

\section{Auxiliary Execution Constructions}\label{app:auxiliary-executions}
\input{Auxiliary-Execution-Constructions}

\section{Common-Failure Executions for Two-Task Disjunction}
\label{app:cf-two-task-disjunctions}
\input{CFE-Two-Task}

\section{Common-Failure Executions for Three-Task Disjunctions}\label{app:cf-three-task-disjunctions}
\input{CFE-Three-Task}

\section{Order-Theoretic Certificates and Enumeration}\label{app:lattice-structure}

\subsection{A Certificate for the Signature Order}
\label{app:signature-order-certificate}

Let $\mathcal R$ be the thirteen representative settings in
Table~\ref{tab:six-separator-signatures}. Define
\[
\begin{aligned}
  \mathcal C_\sigma(\mu)=\{\nu\in\mathcal R:\;&
    \sigma(\nu)\subsetneq\sigma(\mu),\\
  &\nexists\xi\in\mathcal R:
    \sigma(\nu)\subsetneq\sigma(\xi)\subsetneq\sigma(\mu)\}.
\end{aligned}
\]
Thus $\nu\in\mathcal C_\sigma(\mu)$ means that the signature of
$\nu$ is immediately below the signature of $\mu$ among the thirteen
signatures. This definition uses only the six-column table, not an
assumed computational comparison between its rows.

\begin{table}[t]
  \centering
  \small
  \setlength{\tabcolsep}{7pt}
  \renewcommand{\arraystretch}{1.15}
  \caption{The eighteen covering pairs of the signature-inclusion
    order. Each pair $\mu,\nu$ in a row satisfies $\mu\ge\nu$
    by previously established inclusions and equivalences.}
  \label{tab:signature-cover-certificate}
  \begin{tabular}{@{}ll@{}}
    \toprule
    Representative $\mu$ & $\mathcal C_\sigma(\mu)$ \\
    \midrule
    $\LU^F$       & $\{\OB^F,\LU^R\}$ \\
    $\OB^F$       & $\{\OB^R\}$ \\
    $\LU^R$       & $\{\FS^R,\FC^S\}$ \\
    $\FS^R$       & $\{\OB^R,\FS^S\}$ \\
    $\OB^R$       & $\{\OB^S\}$ \\
    $\FC^S$       & $\{\FC^A,\OB^S\}$ \\
    $\FC^A$       & $\{\OB^{A_M}\}$ \\
    $\FS^S$       & $\{\FS^{A_M},\OB^S\}$ \\
    $\FS^{A_M}$   & $\{\FS^A,\OB^{A_M}\}$ \\
    $\FS^A$       & $\{\OB^A\}$ \\
    $\OB^S$       & $\{\OB^{A_M}\}$ \\
    $\OB^{A_M}$   & $\{\OB^A\}$ \\
    $\OB^A$       & $\varnothing$ \\
    \bottomrule
  \end{tabular}
\end{table}

Reading the thirteen binary rows gives
Table~\ref{tab:signature-cover-certificate}. Every displayed
comparison follows from Proposition~\ref{prop:basic-inclusions}
and transitivity, except for the following two comparisons, which
also use the previously established equivalences:
\[
  \LU^R\equiv\FC^R\ge\FC^S,
  \qquad
  \FC^A\equiv\FC^{A_M}\ge\OB^{A_M}.
\]
No strictness or principality is assumed in these derivations.

For any proper inclusion $\sigma(\mu)\supsetneq\sigma(\nu)$,
finiteness supplies a chain from $\sigma(\mu)$ to $\sigma(\nu)$
whose successive inclusions are covering relations among the
thirteen signatures. The table certifies the corresponding
computational inclusion at each step. Transitivity therefore gives
$\mu\ge\nu$. Equality of two signatures occurs only for the same
row. This proves the reverse implication used in
Theorem~\ref{thm:six-signature-landscape} without relying on a
previously established complete comparison table.


\subsection{Meet, Join, and Join-Irreducibility}
\label{app:model-lattice-certificate}

This subsection certifies the lattice claims made after the
computational order has been established by
Theorem~\ref{thm:six-signature-landscape}. It is not used to prove
that theorem: the earlier signature-order certificate relies only
on the covering comparisons and the task-specific boundary results.

\begin{lemma}[The computational lattice and its join-irreducibles]
\label{lem:computational-lattice}
The poset $\mathcal Q$ is a finite lattice with least element $\OB^A$
and greatest element $\LU^F$. Its nonleast join-irreducible elements
are exactly
\[
  J(\mathcal Q)=\{\OB^F,\OB^R,\OB^S,\FS^A,\FC^A,\OB^{A_M}\}.
\]
\end{lemma}

Let $\mathcal S_\sigma=\{\sigma(X):X\in\mathcal Q\}$. The family contains
$\varnothing=\sigma(\OB^A)$ and
$\mathcal B_6=\sigma(\LU^F)$. For every incomparable pair, the
intersection of its signatures is the signature of the class in
the meet column of Table~\ref{tab:model-lattice-operations}.
For a comparable pair $X\le Y$, the intersection is $\sigma(X)$.
These checks exhaust all pairs, so $\mathcal S_\sigma$ is closed under
intersection.

For any $T\subseteq\mathcal B_6$, define
\[
  \operatorname{cl}(T)
  =\bigcap\{S\in\mathcal S_\sigma:T\subseteq S\}.
\]
The family being intersected is nonempty because it contains
$\mathcal B_6$. Finiteness and intersection closure give
$\operatorname{cl}(T)\in\mathcal S_\sigma$. It is the least signature
containing $T$. Therefore the lattice operations are characterized by
\begin{equation}
\begin{aligned}
  \sigma(X\wedge_{\mathcal Q}Y)
     &=\sigma(X)\cap\sigma(Y),\\
  \sigma(X\vee_{\mathcal Q}Y)
     &=\operatorname{cl}\bigl(\sigma(X)\cup\sigma(Y)\bigr).
\end{aligned}
\label{eq:signature-lattice-operations}
\end{equation}
The join column of the table is obtained by this closure operation.
For comparable $X\le Y$, one has $X\wedge_{\mathcal Q}Y=X$ and
$X\vee_{\mathcal Q}Y=Y$, so those pairs need not be listed.

\begin{table}[t]
\centering
\small
\setlength{\tabcolsep}{7pt}
\renewcommand{\arraystretch}{1.1}
\caption{Meet and join for all twenty-four incomparable pairs.
The order and signatures are those already established in the main text.}
\label{tab:model-lattice-operations}
\begin{tabular}{@{}llll@{}}
\toprule
$X$ & $Y$ & $X\wedge_{\mathcal Q}Y$ & $X\vee_{\mathcal Q}Y$\\
\midrule
$\OB^F$ & $\LU^R$ & $\OB^R$ & $\LU^F$ \\
$\OB^F$ & $\FS^R$ & $\OB^R$ & $\LU^F$ \\
$\OB^F$ & $\FC^S$ & $\OB^S$ & $\LU^F$ \\
$\OB^F$ & $\FC^A$ & $\OB^{A_M}$ & $\LU^F$ \\
$\OB^F$ & $\FS^S$ & $\OB^S$ & $\LU^F$ \\
$\OB^F$ & $\FS^{A_M}$ & $\OB^{A_M}$ & $\LU^F$ \\
$\OB^F$ & $\FS^A$ & $\OB^A$ & $\LU^F$ \\
$\FS^R$ & $\FC^S$ & $\OB^S$ & $\LU^R$ \\
$\FS^R$ & $\FC^A$ & $\OB^{A_M}$ & $\LU^R$ \\
$\OB^R$ & $\FC^S$ & $\OB^S$ & $\LU^R$ \\
$\OB^R$ & $\FC^A$ & $\OB^{A_M}$ & $\LU^R$ \\
$\OB^R$ & $\FS^S$ & $\OB^S$ & $\FS^R$ \\
$\OB^R$ & $\FS^{A_M}$ & $\OB^{A_M}$ & $\FS^R$ \\
$\OB^R$ & $\FS^A$ & $\OB^A$ & $\FS^R$ \\
$\FC^S$ & $\FS^S$ & $\OB^S$ & $\LU^R$ \\
$\FC^S$ & $\FS^{A_M}$ & $\OB^{A_M}$ & $\LU^R$ \\
$\FC^S$ & $\FS^A$ & $\OB^A$ & $\LU^R$ \\
$\FC^A$ & $\FS^S$ & $\OB^{A_M}$ & $\LU^R$ \\
$\FC^A$ & $\FS^{A_M}$ & $\OB^{A_M}$ & $\LU^R$ \\
$\FC^A$ & $\FS^A$ & $\OB^A$ & $\LU^R$ \\
$\FC^A$ & $\OB^S$ & $\OB^{A_M}$ & $\FC^S$ \\
$\FS^{A_M}$ & $\OB^S$ & $\OB^{A_M}$ & $\FS^S$ \\
$\FS^A$ & $\OB^S$ & $\OB^A$ & $\FS^S$ \\
$\FS^A$ & $\OB^{A_M}$ & $\OB^A$ & $\FS^{A_M}$ \\
\bottomrule
\end{tabular}
\end{table}

\paragraph*{Identifying the join-irreducibles.}
In any finite lattice, a nonleast element $j$ with a unique lower
cover $j_*$ is join-irreducible. Indeed, every $x<j$ lies below
$j_*$, and hence the join of two elements strictly below $j$ is
also at most $j_*$. Conversely, if $j$ has two distinct lower
covers $u,v$, they are incomparable. Their join lies below $j$,
but cannot lie strictly below $j$ without contradicting that $u$
and $v$ are lower covers. Thus $j=u\vee v$ is not join-irreducible.

The covering relations in
Table~\ref{tab:signature-cover-certificate} give the six elements
and lower covers in Table~\ref{tab:model-join-irreducibles}.
For explicit certificates excluding every other nonleast element,
Table~\ref{tab:model-lattice-operations} gives
\[
\begin{aligned}
  \FS^{A_M}&=\FS^A\vee_{\mathcal Q}\OB^{A_M},\\
  \FS^S&=\FS^A\vee_{\mathcal Q}\OB^S,\\
  \FS^R&=\FS^A\vee_{\mathcal Q}\OB^R,\\
  \FC^S&=\FC^A\vee_{\mathcal Q}\OB^S,\\
  \LU^R&=\FS^R\vee_{\mathcal Q}\FC^S,\\
  \LU^F&=\OB^F\vee_{\mathcal Q}\LU^R.
\end{aligned}
\]
In each equality, both joinands are strictly smaller than the
left-hand side. The least element $\OB^A$ is excluded from
$J(\mathcal Q)$ by definition. This proves the complete list in
Lemma~\ref{lem:computational-lattice}.

\begin{table}[t]
\centering
\small
\setlength{\tabcolsep}{9pt}
\renewcommand{\arraystretch}{1.15}
\caption{Join-irreducible model classes, their unique lower covers,
and the principal tasks already established at those classes.}
\label{tab:model-join-irreducibles}
\begin{tabular}{@{}lll@{}}
\toprule
$j$ & $j_*$ & Task with profile $\upa j$\\
\midrule
$\OB^F$ & $\OB^R$ & $\mathrm{EPG}^{*}$ \\
$\OB^R$ & $\OB^S$ & $\mathrm{SRO}$ \\
$\OB^S$ & $\OB^{A_M}$ & $\mathrm{MCv}$ \\
$\FS^A$ & $\OB^A$ & $\CFSM$ \\
$\FC^A$ & $\OB^{A_M}$ & $\VSS$ \\
$\OB^{A_M}$ & $\OB^A$ & $\DMSD$ \\
\bottomrule
\end{tabular}
\end{table}

\paragraph*{Join-irreducible generation.}
For completeness, in any finite lattice $L$ with least element $0$,
\begin{equation}
  x=\bigvee\{j\in J(L):j\le x\},
  \label{eq:finite-lattice-join-generation}
\end{equation}
where the empty join is $0$. To prove this, induct on the number of
elements at most $x$. For $x=0$ it is immediate, and for
join-irreducible $x$ the right-hand side includes $x$. Otherwise
$x=u\vee v$ for two strictly smaller elements. The induction
hypothesis expresses each of them as a join of join-irreducibles
at most $x$, proving the claim.
Consequently, if $x\nleq y$, some $j\in J(L)$ satisfies
$j\le x$ and $j\nleq y$. This is the sufficiency argument in
Proposition~\ref{prop:unique-principal-separators}.

For necessity, let $j\in J(Q)$ and let $j_*$ be its unique
lower cover. A principal upset $\uparrow x$ distinguishes
$j$ from $j_*$ only if $x\le j$ and $x\not\le j_*$.
Since every $x<j$ lies below $j_*$, this forces $x=j$.
Thus every principal separating family contains
$\uparrow j$ for every $j\in J(Q)$, proving uniqueness.

\paragraph*{Order embedding versus lattice embedding.}
The six-task signatures preserve meet, but their joins are not
in general set unions. For example,
\[
  \FS^A\vee_{\mathcal Q}\FC^A=\LU^R,
\]
whose signature properly contains
$\sigma(\FS^A)\cup\sigma(\FC^A)$.
No distributivity of the model lattice is asserted. In fact,
\[
\begin{aligned}
  \OB^R\wedge_{\mathcal Q}
       (\FS^A\vee_{\mathcal Q}\FC^A)&=\OB^R,\\
  (\OB^R\wedge_{\mathcal Q}\FS^A)
       \vee_{\mathcal Q}(\OB^R\wedge_{\mathcal Q}\FC^A)
       &=\OB^{A_M}.
\end{aligned}
\]
This does not affect distributivity of the upset lattice
$\Up(\mathcal Q)$, whose operations are union and intersection.

\subsection{Antichain and Maximal-Complement Enumeration}
\label{app:antichain-enumeration}

All counts in this subsection are determined solely by the finite
order of the thirteen rows. They certify the number of abstract
profiles and boundary obligations, not the existence of algorithms
or of common-failure executions.

\begin{proposition}[Antichains and maximal-complement counts]
\label{prop:antichain-boundary-enumeration}
The antichain polynomial of $\mathcal Q$ is
\[
  \mathsf A_{\mathcal Q}(z)
  =\sum_{A\in\Ant(\mathcal Q)}z^{|A|}
  =1+13z+24z^{2}+16z^{3},
\]
so $\Up(\mathcal Q)$ has $54$ elements, of which $53$ are nonempty.
For an antichain $A$ put
$U_A=\bigcup_{X\in A}\upa X$ and
$\mathcal N_A=\Max(\mathcal Q\setminus U_A)$. Then
\[
  \sum_{\substack{A\in\Ant(\mathcal Q)\\|A|=2}}|\mathcal N_A|=47,
  \qquad
  \sum_{\substack{A\in\Ant(\mathcal Q)\\|A|=3}}|\mathcal N_A|=34.
\]
\end{proposition}

\begin{lemma}[Width of the computational lattice]
\label{lem:width-three}
The width of $\mathcal Q$ is three.
\end{lemma}

\begin{proof}
The thirteen classes are partitioned into the three chains
\[
\begin{aligned}
  C_1&:\quad \OB^A<\OB^{A_M}<\OB^S<\OB^R<\OB^F<\LU^F,\\
  C_2&:\quad \FS^A<\FS^{A_M}<\FS^S<\FS^R,\\
  C_3&:\quad \FC^A<\FC^S<\LU^R.
\end{aligned}
\]
Each $C_i$ is a chain in the order established by
Theorem~\ref{thm:six-signature-landscape}, and the three cover all
thirteen classes. An antichain meets each chain at most once, so its
size is at most three. Conversely $\{\OB^F,\FS^R,\FC^S\}$ is a
three-element antichain, since no two of its signatures are comparable.
Hence $\width(\mathcal Q)=3$.
\end{proof}

\begin{lemma}[Deletion recurrence for the antichain polynomial]
\label{lem:antichain-deletion-recurrence}
For a finite poset $P$ and $x\in P$, let
$\mathsf A_P(z)=\sum_{A\in\Ant(P)}z^{|A|}$. Then
\[
  \mathsf A_P(z)
  =\mathsf A_{P\setminus\{x\}}(z)
   +z\,\mathsf A_{P\setminus(\downarrow_Px\cup\upa_Px)}(z),
  \qquad \mathsf A_{\varnothing}(z)=1,
\]
where deletions carry the induced order.
\end{lemma}

\begin{proof}
Antichains not containing $x$ are precisely those of
$P\setminus\{x\}$. If an antichain contains $x$, its other elements
must be incomparable with $x$, and form an antichain of
$P\setminus(\downarrow_Px\cup\upa_Px)$.
Conversely, adjoining $x$ to any such antichain is valid and increases
its size by one, accounting for the factor $z$.
\end{proof}

Number the classes $q_1,\ldots,q_{13}$ in the row order of
Table~\ref{tab:six-separator-signatures}, and let $S_i=\sigma(q_i)$.
Thus
\[
\begin{aligned}
(q_1,\ldots,q_{13})=(&\LU^F,\OB^F,\LU^R,\FS^R,\OB^R,
                    \FC^S,\FC^A,\\
                  &\FS^S,\FS^{A_M},\FS^A,\OB^S,\OB^{A_M},\OB^A).
\end{aligned}
\]
For $k\in\{2,3\}$, let $a_k(i)$ be the number of $k$-element
antichains whose smallest row index is $i$, and let $b_k(i)$ be the
sum of $|\mathcal N_A|$ over those antichains.
The term ``smallest'' here refers to the numerical row index, not
to a least element in the computational order.

Here is an explicit counting rule using the signature table. Put
\[
  I_i=\{j>i:S_i\nsubseteq S_j\text{ and }S_j\nsubseteq S_i\}.
\]
Then
\[
  a_2(i)=|I_i|,\qquad
  a_3(i)=\bigl|\{(j,k):j,k\in I_i,\ j<k,\ q_j\parallel q_k\}\bigr|.
\]
Thus the incomparable-pair list in
Table~\ref{tab:model-lattice-operations} also lists exactly the
possible pairs used by this counting rule. Alternatively, the
polynomial follows by applying
Lemma~\ref{lem:antichain-deletion-recurrence} to the same order.

For the boundary counts, $Y$ belongs to $\mathcal N_A$ exactly when
\begin{align*}
 &\sigma(X)\nsubseteq\sigma(Y)\quad\text{for every }X\in A,\\
 &\text{for every }Z>Y\text{ there is }X\in A
      \text{ with }\sigma(X)\subseteq\sigma(Z).
\end{align*}
The first condition puts $Y$ outside $U_A$, and the second makes it
maximal there. In the second condition it suffices to test the upper
covers of $Y$, because $U_A$ is an upset and $\mathcal Q$ is finite.
Applying these rules gives Table~\ref{tab:antichain-enumeration-certificate}.

\begin{table}[t]
\centering
\small
\setlength{\tabcolsep}{7pt}
\renewcommand{\arraystretch}{1.15}
\caption{Finite enumeration certificate, grouped by the smallest
row index. The omitted row indices have zero in all four columns.}
\label{tab:antichain-enumeration-certificate}
\begin{tabular}{@{}clrrrr@{}}
\toprule
$i$ & $q_i$ & $a_2(i)$ & $a_3(i)$ & $b_2(i)$ & $b_3(i)$\\
\midrule
2 & $\OB^F$ & 7 & 8 & 16 & 18 \\
4 & $\FS^R$ & 2 & 0 & 5 & 0 \\
5 & $\OB^R$ & 5 & 6 & 8 & 13 \\
6 & $\FC^S$ & 3 & 0 & 8 & 0 \\
7 & $\FC^A$ & 4 & 2 & 6 & 3 \\
9 & $\FS^{A_M}$ & 1 & 0 & 2 & 0 \\
10 & $\FS^A$ & 2 & 0 & 2 & 0 \\
\midrule
 & Total & 24 & 16 & 47 & 34\\
\bottomrule
\end{tabular}
\end{table}

\begin{proof}[Proof of Proposition~\ref{prop:antichain-boundary-enumeration}]
There is one empty antichain and there are thirteen singleton
antichains. Lemma~\ref{lem:width-three} excludes antichains of size
four or more, and the columns $a_2$ and $a_3$ of
Table~\ref{tab:antichain-enumeration-certificate} give $24$ and $16$.
Thus $\mathsf A_{\mathcal Q}(z)=1+13z+24z^2+16z^3$, and by the
antichain--upset bijection of
Proposition~\ref{prop:canonical-upset-decomposition} the number of
upsets is $1+13+24+16=54$. The columns $b_2$ and $b_3$ give the totals
$47$ and $34$ for the maximal-complement obligations.
\end{proof}


\end{document}

%% file: pp13_four_columns.tex
\begin{table}[!htbp]
\centering
\small
  \footnotesize
  \setlength{\tabcolsep}{4pt}
  \renewcommand{\arraystretch}{1.0}
\caption{Thirteen principal tasks and their boundary proof obligations.
The third column lists the maximal classes outside $\uparrow X$. 
}
\label{tab:PP13}
\begin{tabular}{@{}>{\raggedright\arraybackslash}p{\dimexpr.29\linewidth-1.5\tabcolsep\relax}>{\raggedright\arraybackslash}p{\dimexpr.26\linewidth-1.5\tabcolsep\relax}>{\raggedright\arraybackslash}p{\dimexpr.31\linewidth-1.5\tabcolsep\relax}>{\raggedright\arraybackslash}p{\dimexpr.14\linewidth-1.5\tabcolsep\relax}@{}}
\toprule
\textbf{Minimal profile class $X$}
& \textbf{Principal task $B_X$}
& $\Max(\mathcal Q\setminus\uparrow X)$
& \textbf{Proof reference} \\
\midrule
$\LU^F$ & $\text{EPG}$ & $\{\OB^F,\allowbreak \LU^R\}$ & Cor.~\ref{cor:epg_profile} \\
\addlinespace[2pt]
$\OB^F$ & $\mathrm{EPG}^{*}$ & $\{\LU^R\}$ & Thm.~\ref{thm:six-principal-separators}
\\
\addlinespace[2pt]
$\LU^R$ & $\text{HET}$ & $\{\OB^F,\allowbreak \FS^R,\allowbreak \FC^S\}$ & \cite{NNW-SSS2025} \\
\addlinespace[2pt]
$\FS^R$ & $\text{OSPO}$ & $\{\OB^F,\allowbreak \FC^S,\allowbreak \FS^S\}$ & Lem.~\ref{lem:ospo-bounds} \\
\addlinespace[2pt]
$\OB^R$ & $\text{SRO}$ & $\{\FC^S,\allowbreak \FS^S\}$ & Thm.~\ref{thm:six-principal-separators}
\\
\addlinespace[2pt]
$\FC^S$ & $\text{\SRDV}$ & $\{\OB^F,\allowbreak \FS^R,\allowbreak \FC^A\}$ & Lem.~\ref{lem:s-rdv-bounds} \\
\addlinespace[2pt]
$\FS^S$ & $\text{ZCC}$ & $\{\OB^F,\allowbreak \FC^S,\allowbreak \FS^{A_M}\}$ & \cite{NNW-SSS2025},App.~\ref{app:cv} \\
\addlinespace[2pt]
$\OB^S$ & $\text{MCv}$ & $\{\FC^A,\allowbreak \FS^{A_M}\}$ & Thm.~\ref{thm:six-principal-separators}
\\
\addlinespace[2pt]
$\FC^A$ & $\text{\VSS}$ & $\{\OB^F,\allowbreak \FS^R\}$ & Thm.~\ref{thm:six-principal-separators} \\
\addlinespace[2pt]
$\FS^{A_M}$ & $\text{VEC}$ & $\{\OB^F,\allowbreak \FC^S,\allowbreak \FS^A\}$ & \cite{NNW-SSS2025} \\
\addlinespace[2pt]
$\OB^{A_M}$ & $\text{\DMSD}$ & $\{\FS^A\}$ & Thm.~\ref{thm:six-principal-separators} \\
\addlinespace[2pt]
$\FS^A$ & $\text{\CFSM}$ & $\{\OB^F,\allowbreak \FC^S\}$ & Thm.~\ref{thm:six-principal-separators} \\
\addlinespace[2pt]
$\OB^A$ & $\text{Cv}$ & $\varnothing$ & Trivial, App.~\ref{app:cv} \\
\bottomrule
\end{tabular}
\end{table}

%% file: STACS-Venn-Diagram.tex
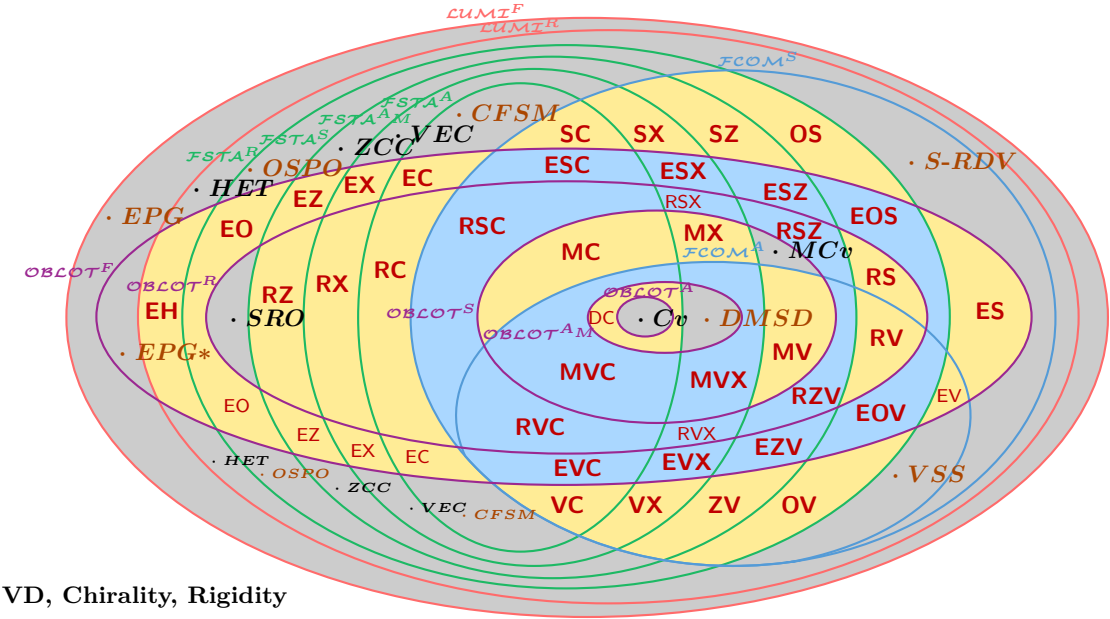
\begin{figure}[!htbp]
  \centering
  \definecolor{letRed}{RGB}{192,0,0}
  \definecolor{letOrg}{RGB}{170,70,0}
  \definecolor{ellRed}{RGB}{255,107,107}
  \definecolor{ellGreen}{RGB}{29,185,100}
  \definecolor{ellBlue}{RGB}{84,155,215}
  \definecolor{ellPurple}{RGB}{155,40,145}
  \def\eRa{(799,450)  ellipse [x radius=764, y radius=440]}
  \def\eRb{(830,449)  ellipse [x radius=690, y radius=421]}
  \def\eGa{(768,449)  ellipse [x radius=563, y radius=399]}
  \def\eGb{(748,450)  ellipse [x radius=446, y radius=383]}
  \def\eGc{(721,450)  ellipse [x radius=338, y radius=365]}
  \def\eGd{(701,450)  ellipse [x radius=238, y radius=344]}
  \def\eBa{(1013,451) ellipse [x radius=473, y radius=364]}
  \def\eBb{(1099.66,808.84)
    arc[start angle=72.15, end angle=-195.31, x radius=377.3, y radius=225.6]
    arc[start angle=146.17, end angle=79.44, x radius=473, y radius=364] -- cycle}
  \def\ePa{(765,449)  ellipse [x radius=686, y radius=247]}
  \def\ePb{(769,450)  ellipse [x radius=529, y radius=200]}
  \def\ePc{(901,449)  ellipse [x radius=263, y radius=156]}
  \def\ePd{(913,450)  ellipse [x radius=113, y radius=52]}
  \def\ePe{(884,449)  ellipse [x radius=41,  y radius=29]}

  \newcommand{\fillall}{\fill[dotfill] (-100,-100) rectangle (1700,1000);}

  \begin{tikzpicture}[
      x=0.009cm, y=-0.009cm, line width=0.8pt,
      even odd rule,  
      let/.style={font=\sffamily\bfseries\small, text=letRed},   
      slet/.style={font=\sffamily\scriptsize, text=letRed},      
      pt/.style={font=\small},                                   
      Opt/.style={font=\small, text=letOrg},                     
      spt/.style={font=\tiny},                                   
      Ospt/.style={font=\tiny, text=letOrg},                     
      nm/.style={font=\small},                                   
       dotfill/.style={fill=black!20},                            
      ac2fill/.style={fill={rgb,255:red,255;green,236;blue,150}}, 
      ac3fill/.style={fill={rgb,255:red,170;green,215;blue,255}}, 
    ]

    \begin{scope}\clip \eRa; \clipout{\eRb}\clipout{\eGa}\clipout{\eBa}\clipout{\ePa} \fillall\end{scope}
    \begin{scope}\clip \eRa; \clip \eRb; \clipout{\eGa}\clipout{\eBa}\clipout{\ePa} \fillall\end{scope}
    \begin{scope}\clip \eRa; \clip \eRb; \clip \eGa; \clipout{\eGb}\clipout{\eBa}\clipout{\ePa} \fillall\end{scope}
    \begin{scope}\clip \eRa; \clip \eRb; \clip \eGb; \clipout{\eGc}\clipout{\eBa}\clipout{\ePa} \fillall\end{scope}
    \begin{scope}\clip \eRa; \clip \eRb; \clip \eGc; \clipout{\eGd}\clipout{\eBa}\clipout{\ePa} \fillall\end{scope}
    \begin{scope}\clip \eRa; \clip \eRb; \clip \eGd; \clipout{\eBa}\clipout{\ePa} \fillall\end{scope}
    \begin{scope}\clip \eRa; \clip \ePa; \clipout{\eRb}\clipout{\eGa}\clipout{\eBa}\clipout{\ePb} \fillall\end{scope}
    \begin{scope}\clip \eRa; \clip \eRb; \clip \eGa; \clip \ePb; \clipout{\eGb}\clipout{\eBa}\clipout{\ePc} \fillall\end{scope}
    \begin{scope}\clip \eGd; \clip \eBb; \clip \ePe; \fillall\end{scope}
    \begin{scope}\clip \eRa; \clip \eRb; \clip \eGc; \clip \eBb; \clip \ePd; \clipout{\eGd}\clipout{\ePe} \fillall\end{scope}
    \begin{scope}\clip \eRa; \clip \eRb; \clip \eGb; \clip \eBa; \clip \ePc; \clipout{\eGc}\clipout{\eBb}\clipout{\ePd} \fillall\end{scope}
    \begin{scope}\clip \eRa; \clip \eRb; \clip \eBa; \clipout{\eGa}\clipout{\eBb}\clipout{\ePa} \fillall\end{scope}
    \begin{scope}\clip \eRa; \clip \eRb; \clip \eBb; \clipout{\eGa}\clipout{\ePa} \fillall\end{scope}

    \begin{scope}\clip \eRb; \clip \ePa; \clipout{\eGa}\clipout{\eBa}\clipout{\ePb} \fill[ac2fill] (-100,-100) rectangle (1700,1000);\end{scope}
    \begin{scope}\clip \eRb; \clip \eGa; \clip \ePa; \clipout{\eGb}\clipout{\eBa}\clipout{\ePb} \fill[ac2fill] (-100,-100) rectangle (1700,1000);\end{scope}
    \begin{scope}\clip \eRb; \clip \eBa; \clip \ePa; \clipout{\eGa}\clipout{\eBb}\clipout{\ePb} \fill[ac2fill] (-100,-100) rectangle (1700,1000);\end{scope}
    \begin{scope}\clip \eRb; \clip \eGb; \clip \ePa; \clipout{\eGc}\clipout{\eBa}\clipout{\ePb} \fill[ac2fill] (-100,-100) rectangle (1700,1000);\end{scope}
    \begin{scope}\clip \eRb; \clip \eBb; \clip \ePa; \clipout{\eGa}\clipout{\ePb} \fill[ac2fill] (-100,-100) rectangle (1700,1000);\end{scope}
    \begin{scope}\clip \eRb; \clip \eGc; \clip \ePa; \clipout{\eGd}\clipout{\eBa}\clipout{\ePb} \fill[ac2fill] (-100,-100) rectangle (1700,1000);\end{scope}
    \begin{scope}\clip \eRb; \clip \eGd; \clip \ePa; \clipout{\eBa}\clipout{\ePb} \fill[ac2fill] (-100,-100) rectangle (1700,1000);\end{scope}
    \begin{scope}\clip \eRb; \clip \eGa; \clip \eBa; \clipout{\eGb}\clipout{\eBb}\clipout{\ePa} \fill[ac2fill] (-100,-100) rectangle (1700,1000);\end{scope}
    \begin{scope}\clip \eRb; \clip \eGa; \clip \eBb; \clipout{\eGb}\clipout{\ePa} \fill[ac2fill] (-100,-100) rectangle (1700,1000);\end{scope}
    \begin{scope}\clip \eRb; \clip \eGa; \clip \eBa; \clip \ePb; \clipout{\eGb}\clipout{\eBb}\clipout{\ePc} \fill[ac2fill] (-100,-100) rectangle (1700,1000);\end{scope}
    \begin{scope}\clip \eRb; \clip \eGb; \clip \ePb; \clipout{\eGc}\clipout{\eBa}\clipout{\ePc} \fill[ac2fill] (-100,-100) rectangle (1700,1000);\end{scope}
    \begin{scope}\clip \eRb; \clip \eGa; \clip \eBb; \clip \ePb; \clipout{\eGb}\clipout{\ePc} \fill[ac2fill] (-100,-100) rectangle (1700,1000);\end{scope}
    \begin{scope}\clip \eRb; \clip \eGc; \clip \ePb; \clipout{\eGd}\clipout{\eBa}\clipout{\ePc} \fill[ac2fill] (-100,-100) rectangle (1700,1000);\end{scope}
    \begin{scope}\clip \eRb; \clip \eGd; \clip \ePb; \clipout{\eBa}\clipout{\ePc} \fill[ac2fill] (-100,-100) rectangle (1700,1000);\end{scope}
    \begin{scope}\clip \eRb; \clip \eGb; \clip \eBa; \clipout{\eGc}\clipout{\eBb}\clipout{\ePa} \fill[ac2fill] (-100,-100) rectangle (1700,1000);\end{scope}
    \begin{scope}\clip \eRb; \clip \eGc; \clip \eBa; \clipout{\eGd}\clipout{\eBb}\clipout{\ePa} \fill[ac2fill] (-100,-100) rectangle (1700,1000);\end{scope}
    \begin{scope}\clip \eRb; \clip \eGd; \clip \eBa; \clipout{\eBb}\clipout{\ePa} \fill[ac2fill] (-100,-100) rectangle (1700,1000);\end{scope}
    \begin{scope}\clip \eRb; \clip \eGb; \clip \eBb; \clipout{\eGc}\clipout{\ePa} \fill[ac2fill] (-100,-100) rectangle (1700,1000);\end{scope}
    \begin{scope}\clip \eRb; \clip \eGb; \clip \eBb; \clip \ePc; \clipout{\eGc}\clipout{\ePd} \fill[ac2fill] (-100,-100) rectangle (1700,1000);\end{scope}
    \begin{scope}\clip \eRb; \clip \eGc; \clip \eBa; \clip \ePc; \clipout{\eGd}\clipout{\eBb}\clipout{\ePd} \fill[ac2fill] (-100,-100) rectangle (1700,1000);\end{scope}
    \begin{scope}\clip \eRb; \clip \eGd; \clip \eBa; \clip \ePc; \clipout{\eBb}\clipout{\ePd} \fill[ac2fill] (-100,-100) rectangle (1700,1000);\end{scope}
    \begin{scope}\clip \eRb; \clip \eGc; \clip \eBb; \clipout{\eGd}\clipout{\ePa} \fill[ac2fill] (-100,-100) rectangle (1700,1000);\end{scope}
    \begin{scope}\clip \eRb; \clip \eGd; \clip \eBb; \clipout{\ePa} \fill[ac2fill] (-100,-100) rectangle (1700,1000);\end{scope}
    \begin{scope}\clip \eRb; \clip \eGd; \clip \eBb; \clip \ePd; \clipout{\ePe} \fill[ac2fill] (-100,-100) rectangle (1700,1000);\end{scope}

    \begin{scope}\clip \eRb; \clip \eGa; \clip \eBa; \clip \ePa; \clipout{\eGb}\clipout{\eBb}\clipout{\ePb} \fill[ac3fill] (-100,-100) rectangle (1700,1000);\end{scope}
    \begin{scope}\clip \eRb; \clip \eGa; \clip \eBb; \clip \ePa; \clipout{\eGb}\clipout{\ePb} \fill[ac3fill] (-100,-100) rectangle (1700,1000);\end{scope}
    \begin{scope}\clip \eRb; \clip \eGb; \clip \eBa; \clip \ePa; \clipout{\eGc}\clipout{\eBb}\clipout{\ePb} \fill[ac3fill] (-100,-100) rectangle (1700,1000);\end{scope}
    \begin{scope}\clip \eRb; \clip \eGc; \clip \eBa; \clip \ePa; \clipout{\eGd}\clipout{\eBb}\clipout{\ePb} \fill[ac3fill] (-100,-100) rectangle (1700,1000);\end{scope}
    \begin{scope}\clip \eRb; \clip \eGd; \clip \eBa; \clip \ePa; \clipout{\eBb}\clipout{\ePb} \fill[ac3fill] (-100,-100) rectangle (1700,1000);\end{scope}
    \begin{scope}\clip \eRb; \clip \eGb; \clip \eBb; \clip \ePa; \clipout{\eGc}\clipout{\ePb} \fill[ac3fill] (-100,-100) rectangle (1700,1000);\end{scope}
    \begin{scope}\clip \eRb; \clip \eGc; \clip \eBb; \clip \ePa; \clipout{\eGd}\clipout{\ePb} \fill[ac3fill] (-100,-100) rectangle (1700,1000);\end{scope}
    \begin{scope}\clip \eRb; \clip \eGd; \clip \eBb; \clip \ePa; \clipout{\ePb} \fill[ac3fill] (-100,-100) rectangle (1700,1000);\end{scope}
    \begin{scope}\clip \eRb; \clip \eGb; \clip \eBa; \clip \ePb; \clipout{\eGc}\clipout{\eBb}\clipout{\ePc} \fill[ac3fill] (-100,-100) rectangle (1700,1000);\end{scope}
    \begin{scope}\clip \eRb; \clip \eGc; \clip \eBa; \clip \ePb; \clipout{\eGd}\clipout{\eBb}\clipout{\ePc} \fill[ac3fill] (-100,-100) rectangle (1700,1000);\end{scope}
    \begin{scope}\clip \eRb; \clip \eGd; \clip \eBa; \clip \ePb; \clipout{\eBb}\clipout{\ePc} \fill[ac3fill] (-100,-100) rectangle (1700,1000);\end{scope}
    \begin{scope}\clip \eRb; \clip \eGb; \clip \eBb; \clip \ePb; \clipout{\eGc}\clipout{\ePc} \fill[ac3fill] (-100,-100) rectangle (1700,1000);\end{scope}
    \begin{scope}\clip \eRb; \clip \eGc; \clip \eBb; \clip \ePb; \clipout{\eGd}\clipout{\ePc} \fill[ac3fill] (-100,-100) rectangle (1700,1000);\end{scope}
    \begin{scope}\clip \eRb; \clip \eGd; \clip \eBb; \clip \ePb; \clipout{\ePc} \fill[ac3fill] (-100,-100) rectangle (1700,1000);\end{scope}
    \begin{scope}\clip \eRb; \clip \eGc; \clip \eBb; \clip \ePc; \clipout{\eGd}\clipout{\ePd} \fill[ac3fill] (-100,-100) rectangle (1700,1000);\end{scope}
    \begin{scope}\clip \eRb; \clip \eGd; \clip \eBb; \clip \ePc; \clipout{\ePd} \fill[ac3fill] (-100,-100) rectangle (1700,1000);\end{scope}

    \draw[ellRed] \eRa;
    \draw[ellRed] \eRb;
    \draw[ellGreen] \eGa;
    \draw[ellGreen] \eGb;
    \draw[ellGreen] \eGc;
    \draw[ellGreen] \eGd;
    \draw[ellBlue] \eBa;
    \draw[ellBlue] \eBb;
    \draw[ellPurple] \ePa;
    \draw[ellPurple] \ePb;
    \draw[ellPurple] \ePc;
    \draw[ellPurple] \ePd;
    \draw[ellPurple] \ePe;

    \node[spt, text=ellRed] at (650,0) {{\boldmath$\LU^{F}$}};
    \node[spt, text=ellRed] at (700,20) {{\boldmath$\LU^{R}$}};

    \node[spt, text=ellBlue] at (1050,70) {{\boldmath$\FC^{S}$}};
    \node[spt, text=ellBlue] at (1000,350) {{\boldmath$\FC^{A}$}};

    \node[spt, text=ellGreen] at (550,130) {{\boldmath$\FS^{A}$}};
    \node[spt, text=ellGreen] at (475,155) {{\boldmath$\FS^{A_{M}}$}};
    \node[spt, text=ellGreen] at (370,185) {{\boldmath$\FS^{S}$}};
    \node[spt, text=ellGreen] at (265,210) {{\boldmath$\FS^{R}$}};

    \node[spt, text=ellPurple] at (40,380)  {{\boldmath$\OB^{F}$}};
    \node[spt, text=ellPurple] at (190,400) {{\boldmath$\OB^{R}$}};
    \node[spt, text=ellPurple] at (570,440) {{\boldmath$\OB^{S}$}};
    \node[spt, text=ellPurple] at (730,470) {{\boldmath$\OB^{A_{M}}$}};
    \node[spt, text=ellPurple] at (890,410) {{\boldmath$\OB^{A}$}};

    \node[Opt] at (150,300)  {\hyperref[app:epg]{\boldmath$\cdot\ EPG$}};
    \node[pt] at (280,260)  {\hyperref[app:het]{\boldmath$\cdot\ HET$}};
    \node[Opt] at (370,230)  {\hyperref[app:osp]{\boldmath$\cdot\ OSPO$}};
    \node[pt] at (490,200)  {\hyperref[app:zcc]{\boldmath$\cdot\ ZCC$}};
    \node[pt] at (575,180)  {\hyperref[app:vec]{\boldmath$\cdot\ VEC$}};
    \node[Opt] at (683,150)  {\hyperref[app:cfs]{\boldmath$\cdot\ CFSM$}};
    \node[Opt] at (180,500)  {\hyperref[app:epgs]{\boldmath$\cdot\ EPG*$}};
    \node[pt] at (330,450)  {\hyperref[app:sro]{\boldmath$\cdot\ SRO$}};
    \node[pt] at (910,450)  {\hyperref[app:cv]{\boldmath$\cdot\ Cv$}};
    \node[Opt] at (1050,450) {\hyperref[app:dmsd]{\boldmath$\cdot\ DMSD$}};
    \node[pt] at (1130,350) {\hyperref[app:mcv]{\boldmath$\cdot\ MCv$}};
    \node[Opt] at (1300,680) {\hyperref[app:vss]{\boldmath$\cdot\ VSS$}};
    \node[Opt] at (1350,220) {\hyperref[app:s-rdv]{\boldmath$\cdot\ S\textbf{-}RDV$}};

    \node[spt] at (290,660) {\hyperref[app:het]{\boldmath$\cdot\ HET$}};
    \node[Ospt] at (370,680) {\hyperref[app:osp]{\boldmath$\cdot\ OSPO$}};
    \node[spt] at (470,700) {\hyperref[app:zcc]{\boldmath$\cdot\ ZCC$}};
    \node[spt] at (580,730) {\hyperref[app:vec]{\boldmath$\cdot\ VEC$}};
    \node[Ospt] at (670,740) {\hyperref[app:cfs]{\boldmath$\cdot\ CFSM$}};
    
    \node[let]  at (175,440)  {\hyperref[ssub:f1-1]{EH}};
    \node[let]  at (285,320)  {\hyperref[ssub:f1-2]{EO}};
    \node[slet] at (285,580)  {\hyperref[ssub:f1-2]{EO}};
    \node[let]  at (1390,440) {\hyperref[ssub:f1-3]{ES}};
    \node[let]  at (390,276)  {\hyperref[ssub:f1-4]{EZ}};
    \node[slet] at (390,622)  {\hyperref[ssub:f1-4]{EZ}};
    \node[slet]  at (1330,565) {\hyperref[ssub:f1-5]{EV}};
    \node[let]  at (465,255)  {\hyperref[ssub:f1-6]{EX}};
    \node[slet] at (470,645)  {\hyperref[ssub:f1-6]{EX}};
    \node[let]  at (550,245)  {\hyperref[ssub:f1-7]{EC}};
    \node[slet] at (550,655)  {\hyperref[ssub:f1-7]{EC}};
    
    \node[let]  at (1230,390) {\hyperref[ssub:f2-1]{RS}};
    \node[let]  at (345,416)  {\hyperref[ssub:f2-2]{RZ}};
    \node[let]  at (1238,480) {\hyperref[ssub:f2-3]{RV}};
    \node[let]  at (425,400)  {\hyperref[ssub:f2-4]{RX}};
    \node[let]  at (510,380)  {\hyperref[ssub:f2-5]{RC}};

    \node[let]  at (1120,180)  {\hyperref[ssub:f3-1]{OS}};
    \node[let]  at (1110,725)  {\hyperref[ssub:f3-2]{OV}};
    \node[let]  at (1000,180)  {\hyperref[ssub:f3-3]{SZ}};
    \node[let]  at (890,180)  {\hyperref[ssub:f3-4]{SX}};
    \node[let]  at (782,180)  {\hyperref[ssub:f3-5]{SC}};
    \node[let]  at (1000,725) {\hyperref[ssub:f3-6]{ZV}};
    \node[let]  at (885,725)  {\hyperref[ssub:f3-7]{VX}};
    \node[let]  at (770,725)  {\hyperref[ssub:f3-8]{VC}};

    \node[let]  at (1100,500) {\hyperref[ssub:f4-1]{MV}};
    \node[let]  at (970,325)  {\hyperref[ssub:f4-2]{MX}};
    \node[let]  at (790,354)  {\hyperref[ssub:f4-3]{MC}};

    \node[slet] at (820,450)  {\hyperref[ssub:f5]{DC}};

    \node[let]  at (1220,300) {\hyperref[ssub:g1-1]{EOS}};
    \node[let]  at (1230,590) {\hyperref[ssub:g1-2]{EOV}};
    \node[let]  at (1090,265) {\hyperref[ssub:g1-3]{ESZ}};
    \node[let]  at (940,235)  {\hyperref[ssub:g1-4]{ESX}};
    \node[let]  at (770,225)  {\hyperref[ssub:g1-5]{ESC}};    
    \node[let]  at (1080,640) {\hyperref[ssub:g1-6]{EZV}};
    \node[let]  at (945,662)  {\hyperref[ssub:g1-7]{EVX}};
    \node[let]  at (785,670)  {\hyperref[ssub:g1-8]{EVC}};
    
    \node[let]  at (1110,322) {\hyperref[ssub:g2-1]{RSZ}};
    \node[slet] at (940,280)  {\hyperref[ssub:g2-2]{RSX}};
    \node[let]  at (645,315)  {\hyperref[ssub:g2-3]{RSC}};
    \node[let]  at (1135,565) {\hyperref[ssub:g2-4]{RZV}};
    \node[slet] at (960,620)  {\hyperref[ssub:g2-5]{RVX}};
    \node[let]  at (730,610)  {\hyperref[ssub:g2-6]{RVC}};
    
    \node[let]  at (992,542)  {\hyperref[ssub:g3-1]{MVX}};
    \node[let]  at (800,530)  {\hyperref[ssub:g3-2]{MVC}};
    \node[font=\small] at (150,860) {$\textbf{VD, Chirality, Rigidity}$};
  \end{tikzpicture}
  \caption{%
  Complete solvability-profile diagram for the thirteen classes.
Colors distinguish principal, two-minimal, and three-minimal
profiles; labels give their canonical principal-task
representations (e.g., EO and EOS).
The task codes are
$E=\mathrm{EPG}^{*}$,
$R=\mathrm{SRO}$,
$M=\mathrm{MCv}$,
$C=\mathrm{CFSM}$,
$V=\mathrm{VSS}$,
$D=\mathrm{DMSD}$,
$G=\mathrm{EPG}$,
$H=\mathrm{HET}$,
$O=\mathrm{OSPO}$,
$S=\mathrm{S\text{-}RDV}$,
$Z=\mathrm{ZCC}$,
$X=\mathrm{VEC}$, and
$Q=\mathrm{Cv}$.
}
  \label{fig:complete-profile-venn} 
\end{figure}

%% file: STACS-scheduler-old.tex
We define the schedulers considered in this paper for the two-robot
system $R=\{a,b\}$.  All schedulers are assumed to be fair: each robot
performs infinitely many \LCM{} cycles.

For a robot $r\in R$ and $i\ge 1$, let
$t_L(r,i)$, $t_C(r,i)$, $t_B(r,i)$, and $t_E(r,i)$ denote the times at
which $r$ performs its $i$-th \Look, \Compute, begins its $i$-th \Move, and
ends its $i$-th \Move, respectively.  We assume
\(
t_L(r,i) < t_C(r,i) < t_B(r,i) < t_E(r,i) < t_L(r,i+1).
\)
The \Look{} and \Compute{} operations are instantaneous, while every
\Move{} operation has finite duration.  Movements are continuous and rigid:
if $r$ computes destination $p_E$ from its current position $p_B$, then it
moves along the segment $[p_B,p_E]$ during
$[t_B(r,i),t_E(r,i)]$ and reaches $p_E$ at time $t_E(r,i)$.

A synchronous execution is described by an infinite activation sequence
$E_1,E_2,\ldots$, where each
\(
E_j\in\bigl\{\{a\},\{b\},\{a,b\}\bigr\}
\)
is the set of robots activated in round $j$.  Robots in $E_j$ perform one
complete \LCM{} cycle synchronously, while robots not in $E_j$ remain
inactive.

The fully synchronous scheduler \FSY{} is the case where
$E_j=\{a,b\}$ for all $j$.  The semi-synchronous scheduler \SSY{} allows
any nonempty activation set in each round, subject only to fairness.  The
restricted synchronous scheduler \RSY{} consists of \FSY or an optional finite
prefix of fully synchronous rounds followed by alternating singleton
activations:
\(
\{a\},\{b\},\{a\},\{b\},\ldots
\quad\text{or}\quad
\{b\},\{a\},\{b\},\{a\},\ldots .
\)

In the asynchronous scheduler \ASY{}, there are no rounds.  The adversary
chooses the times $t_L,t_C,t_B,t_E$ and the duration of each \Move{}
operation arbitrarily, subject only to the ordering constraint above,
finite duration of operations, and fairness.

We also consider atomic subclasses of \ASY{}.  Since there are only two
robots, these subclasses can be defined by forbidding one robot from
performing a \Look{} operation during certain intervals of the other
robot's cycle.  For all $r\in R$ and all $i,j\ge 1$, we define:

\begin{itemize}
\item \textbf{$LC$-atomic \ASY}: 
\(
t_L(\overline r,j)\notin (t_L(r,i),t_C(r,i)].
\)

\item \textbf{$M$-atomic \ASY}: 
\(
t_L(\overline r,j)\notin [t_B(r,i),t_E(r,i)] .
\)

\item \textbf{$CM$-atomic \ASY}: 
\(
t_L(\overline r,j)\notin [t_C(r,i),t_E(r,i)] .
\)

\item \textbf{\LCM-atomic \ASY}: 
\(
t_L(\overline r,j)\notin (t_L(r,i),t_E(r,i)] .
\)
\end{itemize}

The \LCM-atomic \ASY{} scheduler is equivalent to \SSY{}.  In the
\OB{} and \FS{} models, $M$-atomic \ASY{} and $CM$-atomic \ASY{} are
equivalent because robots cannot observe the other robot's color.  The
same equivalence is known for \FC{} robots~\cite{FSSW23}.  Hence we focus
on $M$-atomic \ASY{} and omit $CM$-atomic \ASY{} from the landscape.


%% file: Definitions-Boundary.tex
\label{app:known-separator-details}
\label{app:additional-principal-tasks}

This appendix collects the formal specifications of the geometric basis
tasks and the boundary proofs cited in
Section~\ref{subsec:basis-proof-obligations}.

\subsection{Shrinking Rotation (SRO)}
\label{app:sro}

This subsection records the formal specification of $\mathrm{SRO}$ and
proves Lemma~\ref{lem:six-bounds}\,(\ref{lem:six-sro-bounds}).

\begin{definition}[Shrinking Rotation]
\label{def:sro}
Two robots $a,b$ start from distinct positions. Let $\Sab$ be the set
of times at which neither robot is inside a nonzero $\M$ operation,
list the configurations at the times of $\Sab$ in chronological order,
and remove consecutive repetitions; call the resulting sequence
$C_0,C_1,C_2,\dots$ with $C_i=(a_i,b_i)$. Put $d_i:=b_i-a_i$, let
$Q_i$ be the square with diagonal $a_ib_i$, interior included, and set
$Q_{-1}:=\mathbb R^2$. Let $R_\phi$ denote the counterclockwise
rotation by $\phi$. The first condition (perpetuity) requires the
sequence $(C_i)$ to be infinite. The second condition (rotation)
requires that for every $i$ the vector $d_{i+1}$ be the clockwise
$90^\circ$ rotation of $d_i$ ($\mathbf{C1}$), or its clockwise
$45^\circ$ rotation scaled by $1/\sqrt2$ ($\mathbf{C2}$), and that
$a_{i+1},b_{i+1}\in Q_{i-1}$ ($\mathbf{C3}$).
\begin{align*}
\Phi_{\mathrm{SRO}}\equiv\big[\ &\{\,(C_i)_{i\ge0}\ \text{is infinite}\,\}\\
 &\text{and}\ \{\,\forall i\ge0:\
 \bigl(d_{i+1}=R_{-90^\circ}d_i\ \text{or}\
 d_{i+1}=\tfrac{1}{\sqrt2}R_{-45^\circ}d_i\bigr)\\
 &\qquad\wedge\bigl(a_{i+1},b_{i+1}\in Q_{i-1}\bigr)\,\}\ \big]
\end{align*}
\end{definition}

Here $\Sab$ is the set of joint stops of Definition~\ref{def:six-dmsd}.
The reference~\cite{FSW19} does not make the sequence $(C_i)$ explicit.
We define $(C_i)$ as the configurations at the times at which neither
robot is inside a nonzero $\M$, so that the intermediate configurations
in which one robot has landed while the other is still in flight are
excluded. With this convention each term of $(C_i)$ has both robots at
their initial positions or at the landing points of their last completed
nonzero $\M$ operations, and the sequence does not depend on how the
adversary distributes the movement speed inside a $\M$ operation.
In the proofs below we write $c_1(\gamma)$ and $c_2(\gamma)$ for the
solo and simultaneous contraction ratios of Lemma~\ref{lem:ratios};
these are not to be confused with the configurations $C_1,C_2$ of the
sequence $(C_i)$ or with the conditions $\mathbf{C1},\mathbf{C2}$.

\subsubsection*{Impossibility}

\begin{lemma}[{\cite[Lemma~23]{FSW19}}]
\label{lem:sro-not-fc-s}
$\mathrm{SRO}\notin\Sol(\FC^S)$.
\end{lemma}

\begin{proof}
Run a fully synchronous prefix, which is admissible under $\SSY$. By
Lemma~\ref{lem:rigidA2} the two robots always carry the same color. If
$\lambda_\gamma=0$ for every color $\gamma$ occurring in that common
sequence, this fair execution contains no nonzero $\M$, the
configuration never changes, the sequence $(C_i)$ reduces to $C_0$
alone, and the first condition fails. Otherwise let $\gamma$ be the
first color with $\lambda_\gamma>0$ and truncate the prefix at the
round in which both robots carry $\gamma$; the configuration is still
the initial one.

Now activate $a$ alone twice in a row (Lemma~\ref{lem:solo}; under
$\SSY$ the robot that is not activated has no pending operation), and
write $C_i$ for the configuration at the beginning of this interval. The
visible state of an $\FC$ robot is the color of its peer, and $b$ is
not activated, so that color is constant; by
Lemma~\ref{lem:single-rule}(b) the displacement rule of $a$ is the
single pair $(\lambda_\gamma,\theta_\gamma)$, and $\lambda_\gamma>0$
makes every activation a nonzero $\M$. If $c_1(\gamma)=0$ and $a$
reaches the position of $b$, then $d=0$ in that configuration and
neither $\mathbf{C1}$ nor $\mathbf{C2}$ can hold, so the second
condition fails. Otherwise $d>0$ is preserved. For the first move to
satisfy the second condition it must satisfy the rotation condition,
and with the peer at rest only two displacements do so, one of type
$\mathbf{C1}$ and one of type $\mathbf{C2}$
(Lemma~\ref{lem:cf-sro-two-moves}). The rule being single, the second
move is of the same type, and among the four computations of that
lemma the two relevant ones both give $a_{i+2}\notin Q_i$, so
$\mathbf{C3}$ fails.
\end{proof}

\begin{lemma}[{\cite[Lemma~23]{FSW19}}]
\label{lem:sro-not-fs-s}
$\mathrm{SRO}\notin\Sol(\FS^S)$.
\end{lemma}

\begin{proof}
Under $\SSY$ the robot that is not activated has no pending operation,
so the hypotheses of Lemma~\ref{lem:solo} hold at every round
boundary. The visible state of an $\FS$ robot is its own color, so in
the normalized frame of Lemma~\ref{lem:norm}(1) the color sequence of
each robot is $c_0,g(c_0),g^2(c_0),\dots$ independently of the
schedule, and it is eventually periodic with preperiod $p$ and period
$q$ (Lemma~\ref{lem:pumpB2}). The displacement of each activation is
likewise determined by the color through
$(\lambda_\gamma,\theta_\gamma)$.

\emph{The periodic part contains a color with $\lambda_\gamma>0$.}
Activating $a$ alone $p+2q$ times (Lemma~\ref{lem:solo}) makes $a$
perform at least two nonzero $\M$ operations while $b$ stays at rest;
write $C_i$ for the configuration at the beginning of this interval.
Zero $\M$ operations do not advance the sequence $(C_i)$, so these
correspond to two consecutive terms of $(C_i)$. If at some point
$c_1(\gamma)=0$ and $a$ reaches the position of $b$, then $d=0$ in
that configuration and neither $\mathbf{C1}$ nor $\mathbf{C2}$ can
hold, so the second condition fails. Otherwise $d>0$ is preserved and
the same robot performs two consecutive nonzero $\M$ operations with
its peer at rest, so by Lemma~\ref{lem:cf-sro-two-moves} the rotation
condition and $\mathbf{C3}$ cannot both be satisfied and the second
condition fails.

Suppose now that the periodic part contains no color with
$\lambda_\gamma>0$.
After the preperiod the color sequence stays in
the periodic part, where $\lambda_\gamma=0$, so each robot performs at
most $p$ nonzero $\M$ operations whatever the schedule. Hence the
configuration changes at only finitely many times, the sequence
$(C_i)$ is finite, and the first condition fails.
\end{proof}

\subsubsection*{Solvability}

\begin{lemma}[{\cite[Lemma~2]{BuFKPSW22}}]
\label{lem:sro-in-ob-r}
$\mathrm{SRO}\in\Sol(\OB^R)$.
\end{lemma}
\begin{proof}
Each robot moves to the point obtained by rotating its own position
clockwise by $90^\circ$ about the midpoint of the two robots. If both
robots are activated simultaneously, the segment is rotated by
$90^\circ$ and $\mathbf{C1}$ holds; if only one robot is activated,
the segment is scaled by $1/\sqrt2$ and rotated by $45^\circ$, and
$\mathbf{C2}$ holds.

An $\RSY$ execution is either fully synchronous throughout,
which is the case treated in~\cite[Lemma~22]{FSW19}, or a finite
fully synchronous prefix followed by alternating singleton
activations (Lemma~\ref{lem:altA}).
Fully synchronous rounds preserve the square.
If such a round is followed by a singleton round,
both new positions lie in the same square.
For two consecutive alternating singleton rounds, consider the
square immediately before the first round: the first mover
reaches one of its vertices, and the second reaches its centre.
Thus $a_{i+1},b_{i+1}\in Q_{i-1}$ for every $i\ge1$;
the case $i=0$ is vacuous since $Q_{-1}=\mathbb{R}^2$.
Hence $\mathbf{C3}$ holds as well.

The destination is determined by the midpoint and the rotation
alone and requires no unit of length, so it is computable
in each robot's own local frame.
\end{proof}

Lemma~\ref{lem:sro-not-fc-s}, Lemma~\ref{lem:sro-not-fs-s} and
Lemma~\ref{lem:sro-in-ob-r} together prove Lemma~\ref{lem:six-bounds}\,(\ref{lem:six-sro-bounds}).

\subsection{Monotone Convergence (MCv)}
\label{app:mcv}

This subsection records the formal specification of $\mathrm{MCv}$ and
proves Lemma~\ref{lem:six-bounds}\,(\ref{lem:six-mcv-bounds}).

\begin{definition}[Monotone Convergence]
\label{def:mcv}
Two robots $a,b$ start from distinct positions, and we write
$d(t):=\|a(t)-b(t)\|$. The first condition (convergence) requires the
two robots to converge to a common point $l$. The second condition
(monotonicity) requires their distance never to increase.
\[
  \Phi_{\mathrm{MCv}}
  \equiv
  \bigl[\ \{\,\exists l\in\mathbb R^2:\
    \lim_{t\to\infty}a(t)=\lim_{t\to\infty}b(t)=l\,\}
  \ \text{and}\
    \{\,\forall t,t'\ge0:\ t\le t'\Rightarrow d(t')\le d(t)\,\}\ \bigr]
\]
\end{definition}

Monotonicity is required at every instant of the continuous execution,
not only at the joint stops. This is Monotone Line Convergence
of~\cite{FSSW23} with the two conditions confining the robots to the
segment joining their initial positions and forbidding overtaking
removed.

\subsubsection*{Impossibility}

\begin{lemma}[{\cite[Theorem~8]{KSW-2-robot}}, {\cite[Lemma~13]{FSSW23}}]
\label{lem:mcv_not_fc_a}
$\mathrm{MCv}\notin\Sol(\FC^{A})$.
\end{lemma}
\begin{proof}
Let $\mathcal A$ be any deterministic algorithm of $\FC$ under $\ASY$.
Run fully synchronous cycles with normalized frames.
By Lemma~\ref{lem:rigidA2}, the two robots carry the same color
at every round boundary.
If every output along this execution is zero, the resulting
fair, time-divergent execution keeps the robots at their
distinct initial positions, violating convergence.
 
Otherwise, stop the zero-Move prefix just before the first
round with a nonzero output.
Let $\gamma$ be the common color at this point, so that
$\lambda_\gamma>0$, and put $c:=c_1(\gamma)$.
 
Let $b$ perform its $\LK$, and delay its $\CP$.
This finite delay is permitted under $\ASY$ by Lemma~\ref{lem:suspend}.
The saved snapshot determines the nonzero displacement
\[
  v=\lambda_\gamma R_{\theta_\gamma}(w_{\mathrm{old}}),
\]
where $w_{\mathrm{old}}$ is the vector from $b$ to $a$
at that $\LK$.
Although $b$ has not yet performed its $\CP$, determinism
ensures that it will compute this displacement from the
saved snapshot when its $\CP$ is eventually performed.
 
Throughout the delay, both the position and the visible
color $\gamma$ of $b$ remain unchanged.
Since an $\FC$ robot does not read its own color, we may
therefore activate $a$ for any prescribed finite number
of complete cycles, using normalized frames and the same
nonzero displacement rule
$(\lambda_\gamma,\theta_\gamma)$ in each cycle.
 
If $c>1$, a single such Move increases the distance
between its starting and ending times, violating monotonicity.
If $c=1$, the starting and ending positions of this nonzero
Move are distinct points at the same distance from the
stationary robot $b$.
The midpoint of the movement segment is strictly closer
to $b$ than either endpoint.
Thus the distance increases between the time $a$ passes
that midpoint and the completion of its Move, again
violating monotonicity.
 
It remains to consider $0\le c<1$.
Choose a finite $k\ge1$ such that, after $k$ complete
cycles of $a$,
\[
  \|w_k\|=c^k\|w_{\mathrm{old}}\|<\frac{\|v\|}{2},
\]
where $w_k$ is the vector from $b$ to $a$ at that time.
For $c=0$, take $k=1$; otherwise every cycle starts
with distinct robot positions, so the normalized rule
can be applied repeatedly.
 
Keep $a$ stationary and let $b$ complete its $\CP$
using the saved snapshot, update its color, and perform
the resulting Move with displacement $v$.
The distance immediately before that Move is $\|w_k\|$,
whereas at its completion it satisfies
\[
  \|w_k-v\|
  \ge \|v\|-\|w_k\|
  > \|w_k\|.
\]
Hence monotonicity fails.
 
In every finite-violation case, complete any pending cycle
and append an arbitrary admissible fair, time-divergent
continuation.
The violation persists, so no such algorithm solves
$\mathrm{MCv}$.
\end{proof}

\begin{lemma}[{\cite[Theorem~8]{KSW-2-robot}}, {\cite[Lemma~20]{FSSW23}}]
\label{lem:mcv_not_fs_am}
$\mathrm{MCv}\notin\Sol(\FS^{A_M})$.
\end{lemma}
 
\begin{proof}
We prove the following stronger statement: from any configuration
with distinct robot positions, a common internal color, and no
pending operation, there is an admissible fair, time-divergent
continuation violating $\Phi_{\mathrm{MCv}}$.

Use normalized frames whenever the robots are distinct.
At a configuration with common color $s$, write
\[
  z_s=\lambda_s e^{\mathrm{i}\theta_s},
  \qquad c_s=|1-z_s|.
\]
We advance the common color by activating $a$ and then $b$,
one complete cycle at a time, except in the special case
specified below.

If $z_s=0$, both cycles are zero moves and both robots reach
color $g(s)$.

If $c_s=0$, then $z_s=1$. Instead, let both robots Look and
Compute before either moves, and let them execute their moves
simultaneously with equal fractional progress. They exchange
positions, so their distance becomes zero and then positive
again. Monotonicity fails.

If $z_s\ne0$ and $c_s\ge1$, a solo move of $a$ already violates
monotonicity. For $c_s>1$ its landing distance is larger than
its starting distance. For $c_s=1$, the endpoints are distinct
points at the same distance from the stationary peer; the
interior of the movement segment contains points strictly
closer to that peer, so the distance increases before landing.

It remains to consider $0<c_s<1$. Execute the two singleton
cycles. If either move increases the distance at any time,
monotonicity has already failed. Otherwise both robots finish
at distinct positions with common color $g(s)$.

Continue in this way through the transient and one period of
the common normalized color orbit. Unless a violation has
already occurred, every nonzero rule in the period has
$0<c_s<1$. If every rule in the period is zero, an infinite
fair continuation leaves the robots at distinct fixed
positions, so convergence fails.

Otherwise, the product $\Gamma$ of the solo distance ratios
over one period satisfies $0<\Gamma<1$. Advance both robots
to a common periodic color with a nonzero rule. Let $b$ Look
and Compute a nonzero displacement $v$, and delay its Move
until after a prescribed finite number of complete cycles
of $a$. This is admissible under $A_M$ by Corollary~E.6.
During the delay, $b$ remains fixed, while the private color
of $a$ follows its periodic orbit. Consequently, after
sufficiently many periods the vector $w$ from $b$ to $a$
satisfies $\|w\|<\|v\|/2$.

Keep $a$ stationary and let $b$ complete its pending Move.
The distance increases, since
\[
  \|w-v\|\ge\|v\|-\|w\|>\|w\|.
\]
After every finite violation, complete any pending cycle
and append an arbitrary fair, time-divergent continuation.
The construction uses neither the designated initial color
nor the original positions, provided that its starting
configuration satisfies the three stated conditions.
\end{proof}

\subsubsection*{Solvability}

\begin{lemma}[{\cite[Theorem~8]{KSW-2-robot}}, {\cite[Lemma~4]{FSSW23}}]
\label{lem:mcv_in_ob_s}
$\mathrm{MCv}\in\Sol(\OB^{S})$.
\end{lemma}
 
\begin{proof}
Each robot moves to the midpoint of the two observed positions. Under
$\SSY$ every round is atomic, so it suffices to follow the
configuration at the round boundaries. If only one robot is activated
in a round, it lands at the midpoint and the distance is halved; if
both are activated, they land on the same midpoint and the distance
becomes $0$. In either case the distance never increases, so the
monotonicity condition holds, and $d$ is at least halved in every
round, so $d(t)\to0$.
 
Moreover each landing point lies on the segment joining the two current
positions, so the segments $[a(t),b(t)]$ form a nested family of
compact sets whose diameters tend to $0$. Their intersection is a
single point $l$, and both robots converge to it, which is the
convergence condition.
\end{proof}
 
Lemma~\ref{lem:mcv_not_fc_a}, Lemma~\ref{lem:mcv_not_fs_am} and
Lemma~\ref{lem:mcv_in_ob_s} together prove Lemma~\ref{lem:six-bounds}\,(\ref{lem:six-mcv-bounds}).

\subsection{Hexagonal Edge Traversal (HET)}
\label{app:het}

$\mathrm{HET}$~\cite{NNW-SSS2025} places the two robots at opposite
vertices of a regular hexagon and requires each of them to move to an
adjacent vertex along an edge and back, alternating left and right
forever.

\begin{definition}[Hexagonal Edge Traversal]
\label{def:het}
Two robots $a,b$ start from distinct positions. Let $v_0,\dots,v_5$ be
the vertices of the regular hexagon having $a(0)$ and $b(0)$ as
opposite vertices, labelled counterclockwise from $v_0=a(0)$ and
$v_3=b(0)$, the indices being read modulo $6$. The cycles of the two
robots are
\[
  w_a:=(v_0,v_{-1},v_0,v_1),
  \qquad
  w_b:=(v_3,v_2,v_3,v_4).
\]
The first condition (perpetuity) requires each robot to perform
infinitely many nonzero $\M$ operations. The second condition (edge
traversal) requires each robot $r$, at every time at which it is not
moving, to stand at the vertex of its own cycle indexed by the number
$m_r(t)$ of nonzero $\M$ operations it has performed so far.
\begin{align*}
\Phi_{\mathrm{HET}}\equiv\big[\
  &\{\,\forall r\in\{a,b\}:\ \moves{r}=\infty\,\}\\
  &\text{and}\
  \{\,\forall r\in\{a,b\},\ \forall t\ge0:\
     \lnot\operatorname{Move}_r(t)\Rightarrow
     r(t)=w_r\bigl(m_r(t)\bmod4\bigr)\,\}\ \big]
\end{align*}
\end{definition}

The reference~\cite{NNW-SSS2025} states the same requirement by
exhibiting the time intervals spent at each vertex and the monotone
progress along each edge. Since movements here are rigid and follow the
segment joining the two endpoints, we have rewritten it by prescribing
the sequence of landing points of the nonzero $\M$ operations. The
hexagon is uniquely determined by $a(0)$ and $b(0)$, and the two cycles
are carried onto each other by the rotation by $\pi$ about the
midpoint.

\subsubsection*{Impossibility}

\begin{lemma}[{\cite[Lemma~8]{NNW-SSS2025}}]
\label{lem:het_not_ob_f}
$\mathrm{HET}\notin\Sol(\OB^F)$.
\end{lemma}

\begin{proof}
Let $\mathcal A$ be any $\OB$ algorithm and consider the fully
synchronous execution of Lemma~\ref{lem:rigidA1}, in which the two
robots apply one and the same pair $(\lambda,\theta)$ in every round.
If $\lambda=0$ then $\moves{a}=0$ and the perpetuity condition fails.

Assume $\lambda>0$. If either robot ever leaves its prescribed cycle,
the edge-traversal condition fails. Otherwise, normalizing the
circumradius of the hexagon to $1$, the configuration at the beginning
of round~$1$ is $(v_0,v_3)$ and at the beginning of round~$2$ it is
$(v_{-1},v_2)$. In round~$1$ the robot $a$ sees its peer at
$v_3-v_0$, of length $2$, and must perform the displacement
$v_{-1}-v_0$, of length $1$, so the required pair is
$(\lambda,\theta)=(1/2,+60^\circ)$. In round~$2$ it sees its peer at
$v_2-v_{-1}$, again of length $2$, and must perform the displacement
$v_0-v_{-1}$, of length $1$, so the required pair is
$(1/2,-60^\circ)$; the same two pairs are required of $b$ by the
symmetry of the two cycles. The two rounds therefore demand opposite
rotations relative to the peer, which a single pair $(\lambda,\theta)$
cannot deliver, so the edge-traversal condition fails in one of the two
rounds.
\end{proof}

\begin{lemma}[{\cite[Lemma~7]{NNW-SSS2025}}]
\label{lem:het_not_fs_r}
$\mathrm{HET}\notin\Sol(\FS^R)$.
\end{lemma}

\begin{proof}
The adversary chooses the alternating singleton activation
$\{a\},\{b\},\{a\},\dots$ with an empty fully synchronous prefix
(Lemma~\ref{lem:altA}).

\emph{First stage.} The configuration is unchanged until the first
nonzero $\M$. The hexagon and the two cycles of
Definition~\ref{def:het} are carried onto each other by the rotation by
$\pi$ about the midpoint combined with exchanging $a$ and $b$, and this
rotation preserves chirality. The local input of an $\FS$ robot is the
pair consisting of the relative position of the peer and its own state,
and does not contain the state of the peer, so in the normalized frame
of Lemma~\ref{lem:norm}(1) both local inputs have the form
$(e,\text{own state})$ and are isomorphic. The two robots start from the
same initial state $s_0$, so $a$ in round $2k-1$ and $b$ in round $2k$
both pass from $s_{k-1}$ to $s_k$. If no state occurring in
$s_0,s_1,\dots$ produces a nonzero $\M$, then $\moves{a}=0$ and the
perpetuity condition fails. Let $j$ be the least index that produces a
nonzero $\M$; then $a$ performs the first nonzero $\M$ in round $2j+1$,
from state $s_j$.

\emph{Second stage.} If the landing point of $a$ is not $v_{-1}$, the
edge-traversal condition $a(t)=w_a(1)=v_{-1}$ fails at the time
$m_a=1$. So assume $a$ lands at $v_{-1}$. Writing $R$ for the
circumradius of the hexagon, the edge length is $R$ and the diagonal
length is $2R$, so the constant of Lemma~\ref{lem:norm}(1) is
\[
  \lambda_{s_j}=\frac{\|v_{-1}-v_0\|}{\|v_3-v_0\|}=\frac12 .
\]

\emph{Third stage.} The robot $b$, activated in round $2j+2$, has not
yet moved and is in state $s_j$. The robots have no common unit of
length, so in the normalized frame the local input of $b$ is the same
$(e,s_j)$ as that of $a$ in the second stage: the change of the
observed distance from $2R$ to $\|v_{-1}-v_3\|=\sqrt3\,R$ does not
appear in the local input. Hence $b$ applies the same
$\lambda_{s_j}=\tfrac12$ to its own peer vector $v_{-1}-v_3$, and the
length of its displacement is $\tfrac{\sqrt3}{2}R$. The landing point
required of $b$ by the edge-traversal condition is $w_b(1)=v_2$, with
$\|v_2-v_3\|=R$. Since $\tfrac{\sqrt3}{2}R\neq R$, the robot $b$ cannot
land at $v_2$ and the edge-traversal condition fails.
\end{proof}

\begin{lemma}[{\cite[Lemma~9]{NNW-SSS2025}}]
\label{lem:het_not_fc_s}
$\mathrm{HET}\notin\Sol(\FC^S)$.
\end{lemma}

\begin{proof}
Run a fully synchronous prefix, which is admissible under $\SSY$. By
Lemma~\ref{lem:rigidA2} the two robots always carry the same color. If
$\lambda_\gamma=0$ for every color $\gamma$ occurring in that common
sequence, this fair execution contains no nonzero $\M$, so
$\moves{a}=0$ and the perpetuity condition fails. Otherwise let
$\gamma$ be the first color with $\lambda_\gamma>0$ and truncate the
prefix at the round in which both robots carry $\gamma$; no nonzero
$\M$ has occurred, so the configuration is still the initial one.

Now activate $a$ alone repeatedly (Lemma~\ref{lem:solo}; under $\SSY$
the robot that is not activated has no pending operation). The visible
state of an $\FC$ robot is the color of its peer, and $b$ is not
activated, so that color stays $\gamma$; by
Lemma~\ref{lem:single-rule}(b) the displacement rule of $a$ is the
single pair $(\lambda_\gamma,\theta_\gamma)$, and $\lambda_\gamma>0$
makes every activation a nonzero $\M$. The cycle
$w_a=(v_0,v_{-1},v_0,v_1)$ requires the two successive steps
$v_0\to v_{-1}$ and $v_{-1}\to v_0$. With $b$ standing still at $v_3$,
the first of them approaches the peer, since
$\|v_{-1}-v_3\|=\sqrt3\,R<2R=\|v_0-v_3\|$, and the second retreats from
it. By Lemma~\ref{lem:single-rule} a single pair
$(\lambda_\gamma,\theta_\gamma)$ cannot realize both, so the
edge-traversal condition fails.
\end{proof}

\subsubsection*{Solvability}

\begin{lemma}[{\cite[Lemma~10]{NNW-SSS2025}}]
\label{lem:het_in_lu_r}
$\mathrm{HET}\in\Sol(\LU^R)$.
\end{lemma}

\begin{proof}
Each robot records its position in its own cycle with a light of four
colors. Since a $\LU$ robot sees both lights, it learns from the color
of its peer which vertex the peer currently occupies, so even under the
alternating singleton activations of an $\RSY$ execution it can
reconstruct the hexagon from the observed relative position and move to
the next vertex of its cycle. The destination is determined by the
relative position of the peer and by chirality alone and requires no
unit of length, so it is computable in each robot's own local frame.
\end{proof}

Lemma~\ref{lem:het_not_ob_f}, Lemma~\ref{lem:het_not_fs_r},
Lemma~\ref{lem:het_not_fc_s} and Lemma~\ref{lem:het_in_lu_r} supply the
boundary certificate of Table~\ref{tab:PP13} for $\mathrm{HET}$, which
by Lemma~\ref{lem:boundary-certificate} gives
$\Prof(\mathrm{HET})=\upa\LU^R$ as asserted in
Theorem~\ref{thm:principal-task-basis}.

\subsection{Zig-Zag Collision-Free Convergence (ZCC)}
\label{app:zcc}

$\mathrm{ZCC}$~\cite{NNW-SSS2025} requires the two robots to converge
to a single point without overtaking each other, while each of them
perpetually repeats a zig-zag motion that alternately increases and
decreases its distance from its own initial position.

\begin{definition}[Zig-Zag Collision-Free Convergence]
\label{def:zcc}
Two robots $a,b$ start from distinct positions. The first condition
(convergence) requires the two robots to converge to a point $l$. The
second condition (collision freedom) requires each robot never to go,
as seen from its own initial position, farther than its peer. The third
condition (zig-zag) requires, with $f_r(t):=\|r(t)-r(0)\|$, that for
each robot $r$ there be times $0=t_{-1}<t_0<t_1<\cdots$ at which $f_r$
takes maxima $f_r(t_{2i})$ and minima $f_r(t_{2i+1})$ alternately, the
minima increasing.
\begin{align*}
\Phi_{\mathrm{ZCC}}\equiv{}
 &\bigl[\,\{\,\exists l\in\mathbb R^2:\
   \lim_{t\to\infty}a(t)=\lim_{t\to\infty}b(t)=l\,\}\\
 &\quad\text{and}\ \{\,\forall r\in\{a,b\},\ \forall t\ge0:\
     f_r(t)\le\|\bar r(t)-r(0)\|\,\}\\
 &\quad\text{and}\ \{\,\forall r\in\{a,b\},\
   \exists\,0=t_{-1}<t_0<t_1<\cdots,\\
 &\quad\phantom{\text{and}\ \{\,}\forall i\ge0:\
   f_r(t_{2i-1})<f_r(t_{2i+1})<f_r(t_{2i})\,\}\,\bigr]
\end{align*}
\end{definition}

There are two differences from~\cite{NNW-SSS2025}. The quantification
in the first condition is $\forall\epsilon\ge0$ there, but including
$\epsilon=0$ would require the two robots to coincide at $l$ and stay
there from some time on, which is incompatible with the third
condition; we therefore use $\forall\epsilon>0$. Also, that paper
imposes the third condition on one robot only, whereas we impose it on both.
Along a single Move the position of $r$ traverses a segment, so
$f_r$ is a convex function of the movement progress and hence
non-increasing and then non-decreasing in time; outside the Move
intervals $f_r$ is constant. Each nonzero Move therefore contributes
at most one turning point of $f_r$, and since executions are
time-divergent, only finitely many Moves occur in any bounded time
interval. Consequently the times $t_i$ of the zig-zag condition
satisfy $t_i\to\infty$.

The three conditions differ in nature. Collision freedom is a safety
condition: it is violated as soon as
$f_r(t)>\|\bar r(t)-r(0)\|$ at a single time $t$, and no extension can
repair this. The zig-zag condition is a liveness condition: beyond any
$T$ the sequence $(t_i)$ still provides a strictly decreasing step
$f_r(t_{2i})>f_r(t_{2i+1})$ and a strictly increasing one
$f_r(t_{2i+1})<f_r(t_{2i+2})$. Hence, if for some $T$ the restriction
of $f_r$ to $[T,\infty)$ is monotone for one of the two robots, then
$\Phi_{\mathrm{ZCC}}$ fails. In particular, the monotonicity of $f_r$
along a finite solo activation interval does not by itself establish
failure.

\subsubsection*{Impossibility}

\begin{lemma}
\label{lem:zcc_not_ob_f}
$\mathrm{ZCC}\notin\Sol(\OB^F)$.
\end{lemma}

\begin{proof}
An $\OB^F$ robot applies the same rule in every round, so it cannot
switch between the two directions that the zig-zag condition demands.
Consider the fully synchronous execution of Lemma~\ref{lem:rigidA1},
let $m$ be the midpoint and $v:=a(0)-m$. The difference vector is
mapped in every round by the same similarity $I-2\lambda R_\theta$, of
ratio $c_2$ and rotation angle $\varphi$, and the midpoint is invariant
(Lemma~\ref{lem:ratios}(ii)), so $a(k)-m=c_2^{\,k}R_\varphi^{\,k}v$ and
$b(k)-m=-(a(k)-m)$.

If $\lambda=0$ the robots never move and $f_r\equiv0$.

If $c_2=0$ the two robots coincide at $m$ at the end of round~$1$, and
by Lemma~\ref{lem:keep-coincident} the coincidence is maintained. While
they coincide the observed peer position is $0$. If the local output is
then zero, the robots never move again and $f_r$ is constant from
round~$1$ on. Otherwise the adversary chooses their common global
displacement (Lemma~\ref{lem:norm}(2)); choosing in every round one and
the same direction $w$ with $\langle w,v\rangle<0$ makes
$f_a(t)=\|a(t)-a(0)\|$ strictly increasing from round~$1$ on.

In both cases the restriction of $f_a$ to a tail is monotone, so the
zig-zag condition fails by the criterion stated after
Definition~\ref{def:zcc}.

Assume $\lambda\neq0$ and $c_2\neq0$. If
$\varphi\not\equiv0\pmod{2\pi}$, then by Lemma~\ref{lem:ratios}(iii)
applied to the nonzero angle $\varphi$ there is $k\ge1$ with
$\cos(k\varphi)\le-\tfrac12<0$, whence
$\langle v,a(k)-m\rangle=c_2^{\,k}\cos(k\varphi)\|v\|^2<0$. Since
$b(k)-m=-(a(k)-m)$ and $a(0)=m+v$,
\[
  \|a(k)-a(0)\|^2-\|b(k)-a(0)\|^2
  =-4\,\langle v,a(k)-m\rangle>0 ,
\]
so at the end of round $k$ we have $f_a>\|b(k)-a(0)\|$ and the
collision-freedom condition fails. This is a safety violation at a
finite time.

If $\varphi\equiv0$, then $R_\varphi=I$ and $a(k)=m+c_2^{\,k}v$, so all
the positions of $a$ lie on the ray from $m$ through $a(0)$ and advance
monotonically along it, towards $m$ when $c_2<1$ and away from $a(0)$
when $c_2>1$; the value $c_2=1$ would force $\lambda=0$. Since
movements are rigid along segments, $f_a(t)=|\,c_2^{\,\cdot}-1\,|\,\|v\|$
interpolates monotonically within each round, so $f_a$ is monotone on
the whole execution and the zig-zag condition fails by the same
criterion. This is a liveness violation, witnessed by the infinite
fully synchronous execution itself.
\end{proof}

\begin{lemma}[{\cite[Lemma~26]{NNW-SSS2025}}]
\label{lem:zcc_not_fc_s}
$\mathrm{ZCC}\notin\Sol(\FC^S)$.
\end{lemma}
 
\begin{proof}
We prove a continuation property that will also be used in the
common-failure arguments. Fix the original initial positions
$a_0,b_0$. Suppose that a finite prefix ends with no pending
operation and with the robots lying in $[a_0,b_0]$ in their
original order. Their current colors need not be equal.
There is a fair, time-divergent continuation violating
$\Phi_{\mathrm{ZCC}}$, with $a_0,b_0$ kept fixed.
 
First consider a configuration in which the robots coincide.
Activate them alternately, one complete cycle at a time.
If every output is zero, both positions remain constant forever,
so the zig-zag condition fails by the criterion stated after
Definition~\ref{def:zcc}. Otherwise, at the first nonzero
output, choose the local frame so that the moving robot increases
its distance from its own original initial position; this is
possible by Lemma~\ref{lem:norm}(2).
Its peer remains at the common position, so the collision-freedom
inequality is violated. Append an arbitrary fair tail.
We call this construction the \emph{coincidence continuation}.
 
Now suppose that the robots are distinct. Use alternating
singleton activations as the default schedule, with normalized
frames at every $\LK$. For the currently activated robot $r$, put
\[
  \sigma=1-\lambda e^{\mathrm{i}\theta}.
\]
If $\sigma\in(0,1]$, let the cycle complete and continue with the
default schedule. Such a cycle either leaves $r$ stationary or
moves it towards its peer without reaching it. Thus both robots
remain in the initial segment in their original order.
 
At the first output with $\sigma\notin(0,1]$, freeze the peer
and repeatedly activate $r$. Its peer's color is fixed, so the
same normalized rule is used until a coincidence occurs.
If $\sigma=0$, the first cycle reaches the peer and the
coincidence continuation above applies.
 
Assume $\sigma\ne0$, and write
$\sigma=c e^{\mathrm{i}\psi}$.
At the start of this repetition, let $\bar r$ denote the peer and
\[
  L=\|\bar r-r(0)\|,\qquad h=\|\bar r-r\|.
\]
The maintained order gives $L\ge h>0$. After $j$ repetitions,
\[
  \|r_j-r(0)\|^2
  =L^2+c^{2j}h^2-2Lc^j h\cos(j\psi).
\]
If $\psi\not\equiv0\pmod{2\pi}$, some finite $j$ satisfies
$\cos(j\psi)\le0$ (Lemma~\ref{lem:ratios}(iii)), so this quantity
exceeds $L^2$.
If $\psi\equiv0$, then necessarily $c>1$, and for sufficiently
large finite $j$ we again have $\|r_j-r(0)\|>L$.
In either case the stationary peer remains at distance $L$
from $r(0)$, so the collision-freedom condition fails.
Append an arbitrary fair tail.
 
Finally, if no output outside $(0,1]$ ever occurs, the default
alternating execution is fair and each
$f_r(t)=\|r(t)-r(0)\|$ is non-decreasing throughout.
The zig-zag condition therefore fails.
 
All infinite executions can be scheduled to be time-divergent,
and every exceptional repetition above is finite.
Applying this continuation property to the initial configuration
proves the claim.
\end{proof}
 
\begin{lemma}[{\cite[Lemmas~24--25]{NNW-SSS2025}}]
\label{lem:zcc_not_fs_am}
$\mathrm{ZCC}\notin\Sol(\FS^{A_M})$.
\end{lemma}
 
\begin{proof}
We prove a continuation statement with the original initial
positions $a_0,b_0$ kept fixed. Suppose that no operation is
pending and that either the robots coincide, or their current
normalized color orbits have the same eventual cycle.
There is an admissible fair, time-divergent continuation
violating $\Phi_{\mathrm{ZCC}}$.
 
If a safety violation has already occurred, append any fair
tail. If the robots coincide, use the coincidence continuation
from the proof of Lemma~\ref{lem:zcc_not_fc_s}. Hence assume that
they are distinct and that safety has not yet failed.
 
If a state with solo ratio $c_1=0$ is reachable along the
normalized orbit of either robot, activate that robot alone
until it first lands on its stationary peer. All preceding
$\LK$s see distinct positions. At the resulting coincidence
neither robot has a pending operation, so the same coincidence
continuation applies.
 
We may therefore assume that all future solo ratios are
positive. Let $\Gamma>0$ be their product over the common
eventual color cycle. Consider the fair execution $E$ of
alternating singleton activations under normalized frames.
 
If $\Gamma\ge1$, then, after both color sequences have entered
their periodic parts, every block of one full period of each
robot multiplies the distance by $\Gamma^2\ge1$.
The distance thus has a positive subsequence bounded away
from zero, and convergence fails.
 
Assume $0<\Gamma<1$. If $E$ already violates
$\Phi_{\mathrm{ZCC}}$, use $E$. Otherwise $E$ satisfies
both safety and convergence.
Put
\[
  e=\frac{b_0-a_0}{\|b_0-a_0\|},
  \qquad x_r(t)=\langle r(t),e\rangle .
\]
Adding the squared safety inequalities gives
\[
  2\langle a(t)-b(t),b_0-a_0\rangle\le0,
\]
hence $x_a(t)\le x_b(t)$ at every time.
 
The execution $E$ must contain a $\M$ of $a$ with positive
displacement along $e$, or a $\M$ of $b$ with negative
displacement along $e$. Indeed, otherwise $x_a$ would be
non-increasing and $x_b$ non-decreasing. Convergence would
then require both projections to be equal and constant.
Equality of the projections, together with safety, forces
both squared safety inequalities to be equalities.
During a singleton $\M$, the moving robot would consequently
have to remain on the intersection of two fixed circles
centered at $a_0$ and $b_0$. This finite set admits no
nonconstant continuous movement. Thus the distinct robots
could not converge, a contradiction.
Fix the first such $\M$ in $E$.
 
Suppose first that the selected $\M$ is a $\M$ of $b$
with displacement $v$ satisfying $\langle v,e\rangle<0$.
Follow $E$ through its $\LK$ and $\CP$, but delay the start
of its $\M$. While $b$ stays fixed, activate $a$ through any
remaining transient part of its color orbit and then for
sufficiently many complete periods. Since $\Gamma<1$, the vector
$w$ from $b$ to $a$ can be made to satisfy
\[
  \|w\|<-\langle v,e\rangle/2.
\]
If safety fails during this finite interval, we are done.
Otherwise keep $a$ stationary and release the pending $\M$
of $b$. At its completion,
\[
  x_a-x_b=\langle w-v,e\rangle>0,
\]
contradicting the necessary projection order.
 
If the selected $\M$ is a $\M$ of $a$ with
$\langle v,e\rangle>0$, delay that $\M$ and instead activate $b$
through any remaining transient part of its color orbit and then
for sufficiently many complete periods, bringing it arbitrarily
close to the fixed position of $a$. Choose the resulting
vector $w$ from $a$ to $b$ with
$\|w\|<\langle v,e\rangle/2$. Releasing the $\M$ of $a$ gives
\[
  x_a-x_b=\langle v-w,e\rangle>0,
\]
again violating safety.
 
All delays are finite and occur before the delayed $\M$
starts, as permitted by $A_M$ and Corollary~\ref{cor:stale}.
No moving robot is stopped inside its $\M$.
After a finite violation, complete any pending operation
and append an arbitrary fair, time-divergent tail.
The initial configuration satisfies the continuation
hypotheses, which proves the stated impossibility.
\end{proof}

\subsubsection*{Solvability}

\begin{algorithm}
\caption{$\mathrm{ZCC}$ for $\FS$ under $\SSY$ (robot $r$)}
\label{alg:ZCC_FS_S}
\textbf{Model and scheduler:} $\FS$, $\SSY$. \\
\textbf{State:} $r.\mathit{status}\in\{\mathrm{A},\mathrm{B}\}$,
initially $\mathrm{A}$. \\
\textbf{$\LK$:} observe the peer $\overline r$ and obtain the distance
$D=\|\overline r-r\|$.\\
\textbf{$\CP$:}
\begin{algorithmic}[1]
\If{$r.\mathit{status}=\mathrm{A}$}
  \State $r.\des\gets$ the point at distance $D/4$ from $r$ in the direction of $\overline r$
  \Comment{approach}
  \State $r.\mathit{status}\gets\mathrm{B}$
\Else
  \State $r.\des\gets$ the point at distance $D/8$ from $r$ in the direction away from $\overline r$
  \Comment{retreat}
  \State $r.\mathit{status}\gets\mathrm{A}$
\EndIf
\end{algorithmic}
\textbf{$\M$:} move to $r.\des$.
\end{algorithm}

\begin{lemma}
\label{lem:zcc_in_fs_s}
$\mathrm{ZCC}\in\Sol(\FS^{S})$, as witnessed by
Algorithm~\ref{alg:ZCC_FS_S}.
\end{lemma}

\begin{proof}
The algorithm of~\cite[Lemma~27]{NNW-SSS2025}, which approaches by
$D/2$ and retreats by $D/4$, makes the two robots coincide at the
midpoint when they are activated simultaneously, after which every $\M$
is zero and the zig-zag condition fails. We therefore approach by $D/4$
and retreat by $D/8$. Every movement takes place on the line through
$a(0)$ and $b(0)$, so we put a coordinate $x$ on that line with
$x_a(0)=0$ and $x_b(0)=1$, and write $D:=x_b-x_a$.

\emph{(i) The order is preserved.} Under $\SSY$ each round is atomic,
so the destinations of a round are computed from the configuration at
its beginning. The factor by which $D$ changes in one round is $1/2$
when both approach, $3/4$ when only one approaches, $7/8$ when one
approaches and the other retreats, $9/8$ when only one retreats, and
$5/4$ when both retreat; all of these are positive. Moreover the two
movement segments of a round are disjoint, since $a$ never goes beyond
$x_a+D/4$ while $b$ never goes below $x_b-D/4=x_a+3D/4$. Hence
$x_a(t)<x_b(t)$ at every time, movements included.

\emph{(ii) No robot leaves the outer side of its own initial position.}
Let $D_1$ be the distance at the $\LK$ of some approach of $a$ and
$D_2$ the distance at the $\LK$ of its next retreat. After the approach
round the distance is at most $7D_1/8$, and while $a$ is not activated
the robot $b$ alternates approach and retreat, so the distance
increases by a factor of at most $9/8$. Hence
$D_2\le\tfrac{63}{64}D_1$, and the retreat amount $D_2/8$ is smaller
than the preceding approach amount $D_1/4$. Therefore the position of
$a$ after each retreat is larger than its position before the preceding
approach, and inductively $x_a(t)\ge0$. Symmetrically $x_b(t)\le1$.

\emph{(iii) The collision-freedom condition.} By~(i) and~(ii) we have
$0\le x_a(t)<x_b(t)\le1$, so
$f_a(t)=x_a(t)<x_b(t)=\|b(t)-a(0)\|$, and symmetrically for $b$.

\emph{(iv) The convergence condition.} Give weight $1$ to the state
$\mathrm{A}$ and weight $6/5$ to the state $\mathrm{B}$, and put
$H:=D\cdot w_a\cdot w_b$. The factor by which $H$ changes in one round
is $\tfrac34\cdot\tfrac65=\tfrac9{10}$ when only one approaches,
$\tfrac98\cdot\tfrac56=\tfrac{15}{16}$ when only one retreats,
$\tfrac12\cdot\tfrac{36}{25}=\tfrac{18}{25}$ when both approach,
$\tfrac54\cdot\tfrac{25}{36}=\tfrac{125}{144}$ when both retreat, and
$\tfrac78$ when one approaches and the other retreats; all of these are
at most $\tfrac{15}{16}$. At least one robot is activated in every
round, and $w_a,w_b\in[1,6/5]$, so $D\to0$ geometrically. Each
displacement is at most $D/4$, so the total length travelled by each
robot is finite and both converge to one and the same point $l$.

\emph{(v) The zig-zag condition.} By~(i) we have $D>0$, so every
activation produces a nonzero $\M$, and by fairness each robot performs
infinitely many approaches and retreats alternately. By~(ii) we have
$x_a\ge0$, so an approach of $a$ increases $f_a$ and a retreat
decreases it. Let $t_{k-1}$ be the landing time of the $k$-th $\M$ of
$a$, with $t_{-1}:=0$. Then $[t_{2i-1},t_{2i}]$ is an approach and
$[t_{2i},t_{2i+1}]$ a retreat, and each $\M$ progresses monotonically
towards its destination, so $f_a$ takes the maxima $f_a(t_{2i})$ and
the minima $f_a(t_{2i+1})$ alternately. By~(ii) the position after each
retreat is larger than the position after the previous retreat, so
$f_a(t_{2i+1})>f_a(t_{2i-1})$. The same holds for $b$.

Finally the displacement is a fixed ratio of the observed distance, so
it requires no unit of length, and the state is private, so the
algorithm is admissible in $\FS$.
\end{proof}

Lemma~\ref{lem:zcc_not_ob_f}, Lemma~\ref{lem:zcc_not_fc_s},
Lemma~\ref{lem:zcc_not_fs_am} and Lemma~\ref{lem:zcc_in_fs_s} supply
the boundary certificate of Table~\ref{tab:PP13} for $\mathrm{ZCC}$,
which by Lemma~\ref{lem:boundary-certificate} gives
$\Prof(\mathrm{ZCC})=\upa\FS^{S}$ as asserted in
Theorem~\ref{thm:principal-task-basis}.

\subsection{Vertical Configuration (VEC)}
\label{app:vec}

$\mathrm{VEC}$~\cite{NNW-SSS2025} requires the two robots to transform
the segment joining them by the same rotations as in $\mathrm{SRO}$
and, once a configuration orthogonal to the initial segment has been
reached, to stop there forever.

\begin{definition}[Vertical Configuration]
\label{def:vec}
Two robots $a,b$ start from distinct positions. The sequence of
configurations $C_i=(a_i,b_i)$, the vectors $d_i:=b_i-a_i$, the squares
$Q_i$ with $Q_{-1}:=\mathbb R^2$, and the rotations $R_\phi$ are
defined as in Definition~\ref{def:sro}. A vector $u$ is \emph{orthogonal
to $d_0$}, written $u\perp d_0$, if $u\neq0$ and $\langle u,d_0\rangle=0$.
The first condition (attainment) requires $d_i\perp d_0$ for some $i$.
The second condition (transition) requires that for every $i$ the
vector $d_{i+1}$ be equal to $d_i$, or its clockwise $90^\circ$
rotation, or its clockwise $45^\circ$ rotation scaled by $2^{\mp1/2}$,
and that $a_{i+1},b_{i+1}\in Q_{i-1}$. The third condition (stopping)
requires that if at a time at which both robots are at rest the segment
is orthogonal to $d_0$, then neither robot ever moves again.
\begin{align*}
\Phi_{\mathrm{VEC}}\equiv\big[\
 &\{\,\exists i\ge0:\ d_i\perp d_0\,\}\\
 &\text{and}\ \Bigl\{\,\forall i\ge0:\
   d_{i+1}\in\{d_i,\ R_{-90^\circ}d_i,\
     2^{-1/2}R_{-45^\circ}d_i,\ 2^{1/2}R_{-45^\circ}d_i\}\\
 &\qquad\wedge\bigl(a_{i+1},b_{i+1}\in Q_{i-1}\bigr)\,\Bigr\}\\
 &\text{and}\ \{\,\forall t\in\Sab:\
   b(t)-a(t)\perp d_0\Rightarrow
   \forall t'\ge t:\ a(t')=a(t)\ \wedge\ b(t')=b(t)\,\}\ \big]
\end{align*}
\end{definition}

\subsubsection*{Impossibility}

\begin{lemma}
\label{lem:vec_not_ob_f}
$\mathrm{VEC}\notin\Sol(\OB^F)$.
\end{lemma}

\begin{proof}
Once an $\OB^F$ robot moves it keeps moving in every round and cannot
stop when the orthogonal configuration is reached. Consider the fully
synchronous execution of Lemma~\ref{lem:rigidA1}. If $\lambda=0$ the
robots do not move, the sequence reduces to $C_0$, and $d_0\not\perp
d_0$, so the attainment condition fails. If $c_2=0$ the two robots
coincide in round~$1$ and $d_1=0$, while all the transitions permitted
by the second condition are nonzero, so that condition fails. If
$c_2>0$ both robots perform a nonzero $\M$ in every round, so even if
$d_k\perp d_0$ holds at the end of some round they move again in the
next one and the stopping condition fails; and if no such round exists
the attainment condition fails.
\end{proof}

\begin{lemma}[{\cite[Lemma~21]{NNW-SSS2025}}]
\label{lem:vec_not_fc_s}
$\mathrm{VEC}\notin\Sol(\FC^S)$.
\end{lemma}

\begin{proof}
Run a fully synchronous prefix, which is admissible under $\SSY$. By
Lemma~\ref{lem:rigidA2} the two robots always carry the same color. If
$\lambda_\gamma=0$ for every color $\gamma$ occurring in that sequence,
this fair execution contains no nonzero $\M$, the sequence reduces to
$C_0$, and $d_0\not\perp d_0$, so the attainment condition fails.
Otherwise let $\gamma$ be the first color with $\lambda_\gamma>0$ and
truncate the prefix at the round in which both robots carry $\gamma$;
the configuration is still the initial one.

Now activate $a$ alone repeatedly (Lemma~\ref{lem:solo}; under $\SSY$
the robot that is not activated has no pending operation). The color of
$b$ stays $\gamma$, so by Lemma~\ref{lem:single-rule}(b) the robot $a$
applies the same pair $(\lambda_\gamma,\theta_\gamma)$ at every
activation. With the peer at rest, a single nonzero $\M$ can satisfy
the transition condition only by making $d_{i+1}$ equal to
$R_{-90^\circ}d_i$ or to $2^{\mp1/2}R_{-45^\circ}d_i$, and each of
these applies the same rotation to $d$ at every activation. Hence $a$
reaches a configuration with $d_i\perp d_0$ in at most two moves, and
since $\lambda_\gamma>0$ it moves again at the very next activation, so
the stopping condition fails.
\end{proof}

\begin{lemma}[{\cite[Lemma~22]{NNW-SSS2025}}]
\label{lem:vec_not_fs_a}
$\mathrm{VEC}\notin\Sol(\FS^A)$.
\end{lemma}

\begin{proof}
By Proposition~\ref{prop:atomicity-equivalences} we have
$\FS^A\equiv\FS^{A_{LC}}$, so it suffices to exhibit an admissible
$A_{LC}$ execution. Run a fully synchronous prefix, which is admissible
under $A_{LC}$. By Lemma~\ref{lem:rigidA2} the two robots always carry
the same state. If $\lambda_s=0$ for every state $s$ occurring in that
common sequence, this fair execution contains no nonzero $\M$, the
sequence reduces to $C_0$, and $d_0\not\perp d_0$, so the attainment
condition fails. Otherwise let $s$ be the first state with
$\lambda_s>0$ and truncate the prefix at the round in which both robots
carry $s$; no nonzero $\M$ has occurred, so the configuration is still
the initial one.

Activate $a$ alone. It performs a nonzero $\M$ from $p_B:=a_0$ to some
$p_E\neq p_B$. The robot $b$ has finished its cycle and has no pending
operation, so the hypotheses of Lemma~\ref{lem:mid} are met and, for
each $\mu\in(0,1)$, there is an execution in which $b$ observes $a$ at
$p_\mu=(1-\mu)p_B+\mu p_E$. Let the $\M$ of $b$ begin before the $\M$
of $a$ ends; then the two $\M$ intervals overlap, so no time of $\Sab$
occurs between the two landings and the configuration after both
landings is $C_1$. In state $s$ the robot $b$ lands, by
Lemma~\ref{lem:norm}(1), at $b_0+\lambda_sR_{\theta_s}(p_\mu-b_0)$, so
\[
  d_1(\mu)=b_0+\lambda_sR_{\theta_s}(p_\mu-b_0)-p_E
\]
is an affine function of $\mu$ whose slope
$\lambda_sR_{\theta_s}(p_E-p_B)$ is nonzero because $\lambda_s>0$.
Hence a suitable coordinate $h(\mu)$ of $d_1$ is continuous and
non-constant on $(0,1)$. For $i=0$ we have $Q_{-1}=\mathbb R^2$, so the
transition condition only requires
$d_1\in\{d_0,R_{-90^\circ}d_0,2^{-1/2}R_{-45^\circ}d_0,
2^{1/2}R_{-45^\circ}d_0\}$, and the corresponding values of $h$ are
only four. By Corollary~\ref{cor:continuum} there is a $\mu$ avoiding
those four values, and for it the transition condition fails.
\end{proof}

\subsubsection*{Solvability}

\begin{algorithm}
\caption{$\mathrm{VEC}$ for $\FS$ under $A_M$ (robot $r$)}
\label{alg:VEC_FS_AM}
\textbf{Model and scheduler:} $\FS$, $A_M$. \\
\textbf{State:} $r.\mathit{status}\in\{1,2\}$, initially $1$. \\
\textbf{$\LK$:} observe the peer $\overline r$.\\
\textbf{$\CP$:}
\begin{algorithmic}[1]
\If{$r.\mathit{status}=1$}
  \State $r.\des\gets$ the image of $r$ under the clockwise $90^\circ$
  rotation about the midpoint of $r$ and $\overline r$
  \State $r.\mathit{status}\gets2$
\Else
  \State $r.\des\gets$ the current position of $r$
  \Comment{no further movement}
\EndIf
\end{algorithmic}
\textbf{$\M$:} move to $r.\des$.
\end{algorithm}

\begin{lemma}[{\cite[Lemma~23]{NNW-SSS2025}}]
\label{lem:vec_in_fs_am}
$\mathrm{VEC}\in\Sol(\FS^{A_M})$, as witnessed by
Algorithm~\ref{alg:VEC_FS_AM}.
\end{lemma}

\begin{proof}
Since $A_M$ forbids a $\LK$ of one robot during a $\M$ of the other,
every configuration observed at a $\LK$ is one of the terms of the
sequence $(C_i)$.

Suppose first that the $\LK$ of $b$ precedes the landing of $a$. Then
both robots take as destinations the two remaining vertices of the
square $Q_0$. If there is a time at which both robots are at rest
between the two landings, the intermediate configuration has
$d_1=2^{-1/2}R_{-45^\circ}d_0$, and the next one has
$d_2=R_{-90^\circ}d_0=2^{1/2}R_{-45^\circ}d_1$; if there is no such
time, the two landings produce a single term with
$d_1=R_{-90^\circ}d_0$.

Suppose instead that the $\LK$ of $b$ follows the landing of $a$. Then
$d_1=2^{-1/2}R_{-45^\circ}d_0$, and $b$ moves to a vertex of $Q_1$,
giving $d_2=2^{-1/2}R_{-45^\circ}d_1$.

In every case, each transition has one of the permitted forms.
In the first case, the final positions are the two remaining
vertices of $Q_0$; in the second case, they are a vertex
and the centre of $Q_0$.
Thus, when an intermediate configuration exists, both positions
of $C_2$ belong to $Q_0$, as required.
The containment condition for the first transition is vacuous
since $Q_{-1}=\mathbb{R}^2$.
The final
segment is orthogonal to $d_0$, which is the attainment condition, and
both robots are then in state $2$ and never move again; no earlier
segment is orthogonal to $d_0$, since the only intermediate one is at
$45^\circ$. Hence the stopping condition holds as well.
\end{proof}

Lemma~\ref{lem:vec_not_ob_f}, Lemma~\ref{lem:vec_not_fc_s},
Lemma~\ref{lem:vec_not_fs_a} and Lemma~\ref{lem:vec_in_fs_am} supply
the boundary certificate of Table~\ref{tab:PP13} for $\mathrm{VEC}$,
which by Lemma~\ref{lem:boundary-certificate} gives
$\Prof(\mathrm{VEC})=\upa\FS^{A_M}$ as asserted in
Theorem~\ref{thm:principal-task-basis}.

\subsection{Dyadic-Multiple Stopped-Distance (DMSD)}
\label{app:dmsd}

\begin{definition}[Dyadic-Multiple Stopped-Distance]
\label{def:six-dmsd}
Let $D_0=\|a(0)-b(0)\|>0$, and let
\(
  \mathbb Z[1/2]
  =\{m/2^k\mid m\in\mathbb Z,\ k\in\mathbb Z_{\ge0}\}
\)
be the set of dyadic rationals. Using the moving predicate of
Section~\ref{subsec:problems-relationships}, define the set of \emph{joint stops}
\(
  \Sigma_{ab}
  =\{t\ge0\mid
       \lnot\operatorname{Move}_a(t)\land\lnot\operatorname{Move}_b(t)\},
\)
the set of times at which neither robot is inside a nonzero $\M$
operation. Every execution must satisfy
\begin{align}
  &\forall t\in\Sigma_{ab}:\quad
      \frac{\|a(t)-b(t)\|}{D_0}\in\mathbb Z[1/2],
      \tag{D1}\label{eq:dmsd-dyadic}\\
  &\exists t\in\Sigma_{ab}:\quad
      \|a(t)-b(t)\|<D_0.
      \tag{D2}\label{eq:dmsd-progress}
\end{align}
\end{definition}

\subsubsection*{Impossibility}

\begin{lemma}
\label{lem:dmsd_not_fs_a}
$\mathrm{DMSD}\notin\Sol(\FS^{A})$.
\end{lemma}

\begin{proof}
The argument extends~\cite[Lemma~6]{KSW-2-robot} to a general
displacement $(\lambda,\theta)$ including a rotation, and to the
internal states of $\FS$. By
Proposition~\ref{prop:atomicity-equivalences} we have
$\FS^{A}\equiv\FS^{A_{LC}}$, so it suffices to exhibit, for an
arbitrary $\FS$ algorithm $\mathcal{A}$, an admissible $A_{LC}$
execution that violates $\mathrm{DMSD}$.

In the normalized frame of Lemma~\ref{lem:norm}(1), a robot in state
$s$ has displacement $\lambda_s R_{\theta_s}(w)$, where $w$ is the
vector to its peer, and its next state is determined by $s$ alone. The
local coordinate system may be chosen anew at each activation, and we
supply this frame independently at every activation below. Normalize
$a(0)=0$, $b(0)=e_1$ and $D_0=1$.

\emph{First stage.} Run a fully synchronous prefix, which is admissible
under $A_{LC}$. By Lemma~\ref{lem:rigidA2} the two robots always carry
the same state; let $s_0,s_1,\dots$ be that common state sequence. If
$\lambda_s=0$ for every state $s$ occurring in the sequence, this fair
execution contains no nonzero $\M$, so $a$ and $b$ stay at their
initial positions and $\|a(t)-b(t)\|=D_0$ for every $t\in\Sab$, which
violates (D2). Otherwise let $s$ be the first state with
$\lambda_s>0$ and truncate the prefix at the round in which both robots
carry $s$; write $\lambda:=\lambda_s>0$ and $R:=R_{\theta_s}$. No
nonzero $\M$ has occurred, so the robots are at their initial positions
and the ratio at every earlier joint stop is $1\in\Dyad$.

\emph{Second stage.} Activate $b$ alone. It performs a nonzero $\M$ to
\[
  p_E=e_1-\lambda Re_1\neq b(0),
\]
and for $\mu\in(0,1)$ we put
$p_\mu:=(1-\mu)b(0)+\mu p_E=e_1-\mu\lambda Re_1$. The robot $a$ has
finished its cycle and has no pending operation, so the hypotheses of
Lemma~\ref{lem:mid} are met and there is an execution in which $a$
performs a $\LK$ at a time when $b$ is at $p_\mu$. In state $s$ the
robot $a$ then lands at
\[
  \lambda Rp_\mu=\lambda Re_1-\mu\lambda^2R^2e_1 .
\]
Let both $\M$ operations complete, delay the next $\LK$ of each robot
by a finite amount, and append a fair tail. Throughout that delay
neither robot is inside a nonzero $\M$, so the whole interval is
contained in $\Sab$, and $m_b\ge1$ there.

\emph{Third stage.} On that interval the distance equals
\[
  h(\mu)=\|c+\mu w\|,\qquad
  c:=e_1-2\lambda Re_1,\qquad
  w:=\lambda^2R^2e_1\neq0 .
\]
Indeed $\|c+\mu w\|^2$ is a quadratic polynomial in $\mu$ with leading
coefficient $\|w\|^2>0$, so $h$ is continuous and non-constant on
$(0,1)$. By Corollary~\ref{cor:continuum} there is a $\mu$ with
$h(\mu)/D_0\notin\Dyad$, and the corresponding execution violates (D1).
Since (D1) is of the form $\forall t\in\Sab$, a single violating time
suffices, and the joint stops produced in the first two stages need not
be inspected.
\end{proof}

\begin{corollary}[{\cite[Lemma~6]{KSW-2-robot}}]
\label{cor:dmsd_not_ob_a}
$\mathrm{DMSD}\notin\Sol(\OB^{A})$.
\end{corollary}

\begin{proof}
In $\mathcal Q$ the class $\OB^{A}$ lies below $\FS^{A}$, so this
follows from Lemma~\ref{lem:dmsd_not_fs_a} and
Proposition~\ref{prop:profile-is-upset}.
\end{proof}

\subsubsection*{Solvability}

\begin{lemma}[{\cite[Lemma~5]{KSW-2-robot}}]
\label{lem:dmsd_in_ob_am}
$\mathrm{DMSD}\in\Sol(\OB^{A_M})$, as witnessed by the Go-to-midpoint
algorithm.
\end{lemma}

\begin{proof}
Each robot moves to the midpoint of its own position and the observed
position of its peer. By
Proposition~\ref{prop:atomicity-equivalences} we have
$\OB^{A_{CM}}\equiv\OB^{A_M}$, so it is enough to argue under $A_M$.

\emph{Observed positions.} Under $A_M$ a $\LK$ of one robot never falls
inside the closed interval $[t_B,t_E]$ of a $\M$ operation of the
other, so every observed position is an initial position or the landing
point of a completed nonzero $\M$. Moreover a robot is itself at rest
at its own $\LK$, since $t_E(r,k-1)<t_L(r,k)<t_B(r,k)$. Consequently
every $\LK$ time belongs to $\Sab$.

\emph{Condition (D1).} We claim that at every time the position of each
robot is of the form $(1-\alpha)a(0)+\alpha b(0)$ with
$\alpha\in\Dyad$, whenever that robot is not inside a $\M$ operation.
This holds initially with $\alpha=0$ and $\alpha=1$. A robot lands at
the midpoint of its own position and an observed position, and by the
previous paragraph both are of the stated form; the midpoint of
$(1-\alpha)a(0)+\alpha b(0)$ and $(1-\beta)a(0)+\beta b(0)$ is the
point with parameter $(\alpha+\beta)/2\in\Dyad$. Since the destination
is computed from the snapshot taken at the $\LK$, a delayed $\M$ does
not change the landing point (Corollary~\ref{cor:stale}), so the claim
is preserved even when a snapshot is stale. Hence for $t\in\Sab$ we get
$\|a(t)-b(t)\|=|\beta-\alpha|\,D_0$ with $\beta-\alpha\in\Dyad$, which
is (D1).

\emph{Condition (D2).} Every landing point is a midpoint of two points
of the segment $[a(0),b(0)]$, so by induction both robots always lie in
that segment. By fairness some robot, say $a$, is activated first; at
its $\LK$ the two robots are still at their initial positions, so the
observed distance is $D_0>0$ and $a$ performs a nonzero $\M$, landing
at the midpoint $m$ of $a(0)$ and $b(0)$. Let $t:=t_L(a,2)$ be the next
$\LK$ of $a$, which exists by fairness and satisfies
$t_E(a,1)<t<t_B(a,2)$, so $a(t)=m$; by the first paragraph
$t\in\Sab$. Since $b(t)$ lies in $[a(0),b(0)]$ and $m$ is the midpoint
of that segment,
\[
  \|a(t)-b(t)\|\le \tfrac{D_0}{2}<D_0 ,
\]
which is (D2).
\end{proof}

Lemma~\ref{lem:dmsd_not_fs_a} and Lemma~\ref{lem:dmsd_in_ob_am}
together prove Lemma~\ref{lem:six-bounds}\,(\ref{lem:six-dmsd-bounds}).

\subsection{Convergence (Cv)}
\label{app:cv}

$\mathrm{Cv}$ requires nothing but that the two robots converge to a
single point. It occupies the row of the least class of $\mathcal Q$,
so its profile is the whole of $\mathcal Q$.

\begin{definition}[Convergence]
\label{def:cv}
Two robots $a,b$ start from distinct positions and converge to a point
$l$.
\[
  \Phi_{\mathrm{Cv}}
  \equiv
  \{\,\exists l\in\mathbb R^2:\
    \lim_{t\to\infty}a(t)=\lim_{t\to\infty}b(t)=l\,\}
\]
\end{definition}

Since $\OB^A$ is the least element of $\mathcal Q$, we have
$\mathcal N_{\OB^A}=\Max(\mathcal Q\setminus\upa\OB^A)=\varnothing$,
so the negative half of the condition in
Lemma~\ref{lem:boundary-certificate}, applied with
$U=\upa\OB^A=\mathcal Q$, is vacuous for $\mathrm{Cv}$.
This is what the entry ``Trivial'' of Table~\ref{tab:PP13} records. It
concerns only the impossibility side: the positive half still has to be
established, and we do so here.

\subsubsection*{Solvability}

\begin{algorithm}
\caption{$\mathrm{Cv}$ for $\OB$ under $\ASY$ (robot $r$)}
\label{alg:Cv_OB_A}
\textbf{Model and scheduler:} $\OB$, $\ASY$. \\
\textbf{$\LK$:} observe the peer $\overline r$ at the local coordinate
$v$.\\
\textbf{$\CP$:}
\begin{algorithmic}[1]
\State $r.\des\gets v/2$
\Comment{the midpoint of the two observed positions}
\end{algorithmic}
\textbf{$\M$:} move to $r.\des$.
\end{algorithm}

\begin{lemma}[{\cite{CP}}]
\label{lem:cv_in_ob_a}
$\mathrm{Cv}\in\Sol(\OB^{A})$, as witnessed by
Algorithm~\ref{alg:Cv_OB_A}.
\end{lemma}

\begin{proof}
Algorithm~\ref{alg:Cv_OB_A} is the two-robot case of the gravitational
algorithm of~\cite{CP}: each robot moves to the center of
gravity of the configuration it has observed, which for two robots is
the midpoint of the two observed positions. The destination is
determined by the midpoint alone, so it requires neither a unit of
length nor chirality, and it is computed from the current snapshot
alone, so the algorithm is admissible in $\OB$.

All the motion takes place on the line through $a(0)$ and $b(0)$, since
the midpoint of two points of that line lies on it. Let $H(t)$ be the
convex hull, an interval of that line, of the current positions of the
two robots together with 
the destinations determined by all Looks whose cycles
have not yet completed at time $t$, including cycles
awaiting Compute
We claim that
$H$ is non-increasing. Indeed, the destination determined by a \Look is the
midpoint of two positions occupied at that time, both of which belong
to $H$ at that time, so the destination belongs to $H$ as well; and
every movement is a rigid motion along a segment whose two endpoints
belong to $H$, so it never leaves $H$. In particular
$H(t)\subseteq H(0)=[a(0),b(0)]$ for every $t$, which is the
containment hypothesis of~\cite{CP}.

Executions are fair, so each robot performs infinitely many $\LCM$
cycles, and the adversary's freedom in $\ASY$ is exactly the one
allowed there: the destination carried into a $\M$ may have been
computed from a stale snapshot, and the two robots may be in flight
simultaneously. Under these hypotheses the gravitational algorithm
converges~\cite{CP}, so there is a point $l$ with
$a(t),b(t)\to l$, which is $\Phi_{\mathrm{Cv}}$.
\end{proof}

Lemma~\ref{lem:cv_in_ob_a} is the positive half of the boundary
certificate for $\mathrm{Cv}$; the negative half is empty, as noted
above. By Lemma~\ref{lem:boundary-certificate} we obtain
$\Prof(\mathrm{Cv})=\upa\OB^{A}=\mathcal Q$, as asserted in
Theorem~\ref{thm:principal-task-basis}.

%% file: New-Tasks.tex
\subsection{Expansion Preserving Gravity (EPG)}
\label{app:epg}

This subsection and Appendix~\ref{app:epgs} read their requirements off the
same sequence of configurations, the one already used for $\mathrm{SRO}$ in
Definition~\ref{def:sro} and for $\mathrm{VEC}$ in Definition~\ref{def:vec}.
Let
\[
  \Sab:=\{\,t\ge0 \mid \forall r\in\{a,b\}:\lnot\M_r(t)\,\}
\]
be the set of \emph{joint stops}, the times at which neither robot is inside a
nonzero $\M$ operation. List the configurations taken at the times of $\Sab$ in
chronological order, remove consecutive repetitions, and write the resulting
sequence as $C_0,C_1,C_2,\dots$ with $C_i=(a_i,b_i)$; thus
$C_0=\bigl(a(0),b(0)\bigr)$. Put
\[
  c_0:=\tfrac12\bigl(a(0)+b(0)\bigr),
  \qquad
  d_i:=b_i-a_i .
\]
Since $\MoveSet{r}$ collects only the open intervals of the nonzero $\M$
operations of $r$, at every $t\in\Sab$ each robot stands either at its initial
position or at the landing point of its last completed nonzero $\M$; in
particular $0\in\Sab$ and the sequence is well defined.

\begin{remark}[Preserving the center of gravity forces simultaneity]
\label{rem:epg-forces-simultaneity}
Suppose $\tfrac12(a_{i+1}+b_{i+1})=c_0=\tfrac12(a_i+b_i)$. Then
$a_{i+1}+b_{i+1}=a_i+b_i$, so $b_{i+1}=b_i$ would force $a_{i+1}=a_i$ and hence
$C_{i+1}=C_i$. Consecutive terms of the sequence being distinct, every step of
the sequence therefore moves \emph{both} robots. Equivalently, a nonzero $\M$
performed by one robot while the other is stationary across the surrounding
joint stops displaces the center of gravity by half that displacement.
\end{remark}

\begin{definition}[Expansion Preserving Gravity]
\label{def:epg}
Two robots $a,b$ start from distinct positions. The first condition (one
expansion) requires the joint-stop sequence to consist of exactly the two terms
$C_0$ and $C_1$. The second condition (expansion preserving gravity) requires
the step from $C_0$ to $C_1$ to keep the center of gravity at $c_0$ and to
double the segment. The third condition (stopping) requires both robots to
remain stationary from some time on.
\begin{align*}
  \Phi_{\mathrm{EPG}}\equiv\bigl[\ 
    &\{\,(C_i)_{i\ge0}\ \text{has exactly the two terms}\ C_0,C_1\,\}\\
    &\text{and}\ \bigl\{\,\tfrac12(a_1+b_1)=c_0\ \text{and}\ d_1=2d_0\,\bigr\}\\
    &\text{and}\ \{\,\forall r\in\{a,b\}:\Stopped{r}\,\}\ \bigr]
\end{align*}
\end{definition}

The two equalities determine $C_1$ outright: $a_1=c_0-d_0$ and $b_1=c_0+d_0$,
that is, each robot ends at twice its initial distance from $c_0$ 
each robot ends at twice its initial distance from $c_0$. 
The predicate mentions only the trajectory and $\gamma(0)$: no
floor function, no reference to the local coordinate systems, and no reference
to activation counts, since the terms of $(C_i)$ are determined by the
configurations at the joint stops and not by how many cycles produced them.

\begin{lemma}[Bounds for $\mathrm{EPG}$]
\label{lem:epg-bounds}
Under the assumptions of this paper,
\[
  \mathrm{EPG}\in\Sol(\LU^F),
  \qquad
  \mathrm{EPG}\notin\Sol(\OB^F)\cup\Sol(\LU^R).
\]
\end{lemma}

\begin{proof}
The positive statement is Lemma~\ref{lem:epg_in_lu_f}, and the two
negative ones are Lemmas~\ref{lem:epg_not_ob_f}
and~\ref{lem:epg_not_lu_r} below.
\end{proof}

\subsubsection*{Solvability}

\begin{algorithm}[t]
\caption{$\mathrm{EPG}$ for $\FS$ under \FSY{} (robot $r$)}
\label{alg:epg}
\textbf{Model and scheduler:} $\FS$ with $C=\{0,1\}$ and initial color $0$,
\FSY.\\
\textbf{$\LK$:} observe the peer $\bar r$ at the local coordinate $v$, and read
$s:=r.\lcolor$.\\
\textbf{$\CP$:}
\begin{algorithmic}[1]
  \If{$s=0$}
    \State $r.\des \gets -v/2$ \Comment{away from the observed midpoint, by the observed distance}
    \State $r.\lcolor \gets 1$
  \Else
    \State $r.\des \gets 0$
  \EndIf
\end{algorithmic}
\textbf{$\M$:} move to $r.\des$.
\end{algorithm}

\begin{lemma}
\label{lem:epg_in_lu_f}
$\mathrm{EPG}\in\Sol(\LU^{F})$, as witnessed by Algorithm~\ref{alg:epg} in
$\FS^{F}$.
\end{lemma}

\begin{proof}
Both robots start with color $0$, so in round~$1$ both execute the first
branch. A robot at $p$ observing its peer at the local coordinate $v$ lands at
$p+L^{-1}(-v/2)=p-\tfrac12(p_{\bar r}-p_{r})$, that is, it moves away from the
peer by half the observed distance; the displacement is determined by the
snapshot alone and needs neither a unit of length nor chirality, so
Algorithm~\ref{alg:epg} is admissible in $\FS$. The two displacements are
$\tfrac12(a_0-b_0)$ and $\tfrac12(b_0-a_0)$, opposite and of equal norm, so
after round~$1$
\[
  \tfrac12(a_1+b_1)=\tfrac12(a_0+b_0)=c_0,
  \qquad
  d_1=2d_0 ,
\]
which is the second condition. Both displacements are nonzero because
$d_0\neq0$, and under \FSY{} the two $\M$ operations occupy the same round, so
no time of $\Sab$ lies strictly between the two boundaries of round~$1$ and
$C_1=(a_1,b_1)$ is the term of the sequence following $C_0$.

Both robots end round~$1$ with color $1$ and thereafter execute the second
branch, so every later $\M$ is zero. The configuration never changes again, the
joint-stop sequence consists of exactly $C_0$ and $C_1$, and both robots are
permanently stationary. Hence $\Phi_{\mathrm{EPG}}$ holds, and $\FS^F$ is the class $\LU^F$ by Theorem~\ref{thm:equ-FCOMF-by-FSTAF}.
\end{proof}

\subsubsection*{Impossibility}

\begin{lemma}
\label{lem:epg_not_ob_f}
$\mathrm{EPG}\notin\Sol(\OB^{F})$.
\end{lemma}

\begin{proof}
Let $\mathcal A$ be any deterministic algorithm of $\OB$ under \FSY. An $\OB$
robot has no color and no memory, so its output at an activation is a function
of the observed local position $v$ of the peer alone; write it $f(v)$. The
adversary may give a robot a frame in which the peer appears at any prescribed
nonzero local position (Lemma~\ref{lem:norm}(1)). Fix frames for round~$1$ and
let $v_a,v_b$ be the resulting local inputs.

If $C_1$, the configuration at the end of round~$1$, violates
$\tfrac12(a_1+b_1)=c_0$ or $d_1=2d_0$, the second condition fails, provided
$C_1\neq C_0$. If $C_1=C_0$, both $f(v_a)$ and $f(v_b)$ are zero; repeating the
same frames in every round leaves the configuration unchanged forever, the
joint-stop sequence has the single term $C_0$, and the first condition fails.

So assume $C_1\neq C_0$ and the second condition holds at $i=0$. Then
$d_1=2d_0\neq0$, so in round~$2$ the adversary may choose a frame for $a$ with
$L'_a(b_1-a_1)=v_a$ and one for $b$ with $L'_b(a_1-b_1)=v_b$. Both robots
receive the local input they received in round~$1$, and by determinism they
output $f(v_a)$ and $f(v_b)$ again. At least one of these is nonzero, since
$C_1\neq C_0$, so the configuration at the end of round~$2$ differs from $C_1$
and contributes a third term to the joint-stop sequence. The first condition
fails.
\end{proof}

\begin{lemma}
\label{lem:epg_not_lu_r}
$\mathrm{EPG}\notin\Sol(\LU^{R})$.
\end{lemma}

\begin{proof}
Apply Lemma~\ref{lem:solo-drift}(2) with $K=R$.
\end{proof}

\begin{corollary}
\label{cor:epg_profile}
$\Prof(\mathrm{EPG})=\upa\LU^{F}$.
\end{corollary}

\begin{proof}
Table~\ref{tab:PP13} gives
$\Max(\mathcal Q\setminus\upa\LU^{F})=\{\OB^{F},\LU^{R}\}$. Apply
Lemma~\ref{lem:boundary-certificate} with
Lemma~\ref{lem:epg_in_lu_f} on the positive side and
Lemmas~\ref{lem:epg_not_ob_f} and~\ref{lem:epg_not_lu_r} on the negative side.
\end{proof}

\subsection{Perpetual Expansion Preserving Gravity (\texorpdfstring{EPG$^{*}$}{EPG*})}
\label{app:epgs}

\begin{definition}[Perpetual Expansion Preserving Gravity]
\label{def:epgs}
Two robots $a,b$ start from distinct positions, and $\Sab$, the sequence
$(C_i)_{i\ge0}$, $c_0$ and $d_i$ are as in Appendix~\ref{app:epg}. The first
condition (perpetuity) requires the sequence to be infinite. The second
condition (expansion preserving gravity) requires every step of the sequence to
keep the center of gravity at $c_0$ and to double the segment.
\begin{align*}
  \Phi_{\mathrm{EPG}^{*}}\equiv\bigl[\ 
    &\{\,(C_i)_{i\ge0}\ \text{is infinite}\,\}\\
    &\text{and}\ \Bigl\{\,\forall i\ge0:\ 
       \tfrac12(a_{i+1}+b_{i+1})=c_0\ \text{and}\ d_{i+1}=2d_i\,\Bigr\}\ \bigr]
\end{align*}
\end{definition}

The two conditions give $C_i=\bigl(c_0-2^{i-1}d_0,\ c_0+2^{i-1}d_0\bigr)$ for
$i\ge1$, so the configuration at the $i$-th joint stop of the sequence is
determined by $\gamma(0)$ alone. $\mathrm{EPG}$ is the one-shot version:
Definition~\ref{def:epg} keeps the second condition at $i=0$ and replaces
perpetuity by a single expansion followed by permanent stopping.

\begin{remark}[Relation to $\mathrm{CGE}$ and $\mathrm{CGE}^{*}$]
\label{rem:epg-vs-cge}
Center of Gravity Expansion and its perpetual variant
\cite{BuFKPSW22,FSW19} prescribe each landing point by a coordinatewise floor,
which makes the expansion irregular and moves the center of gravity. That
formulation cannot be carried over to the present model. Evaluated on a lattice
common to the two robots, as in the global formula of \cite{BuFKPSW22}, the
destination is not computable under variable disorientation, so the profile
would be empty; evaluated on each robot's own lattice, the predicate refers to
the local frames chosen by the adversary and is therefore no longer a predicate
on the trajectory, contrary to
Section~\ref{subsec:problems-relationships}. We keep the exact doubling, drop
the floor, and recover the same boundary by requiring the center of gravity to
be preserved at every joint stop. By
Remark~\ref{rem:epg-forces-simultaneity} that requirement rules out a solo
expansion, which is the role the floor played in \cite{FSW19}.
\end{remark}

\subsubsection*{Solvability}

\begin{algorithm}[t]
\caption{$\mathrm{EPG}^{*}$ for $\OB$ under \FSY{} (robot $r$)}
\label{alg:epgs}
\textbf{Model and scheduler:} $\OB$, \FSY.\\
\textbf{$\LK$:} observe the peer $\bar r$ at the local coordinate $v$.\\
\textbf{$\CP$:}
\begin{algorithmic}[1]
  \State $r.\des \gets -v/2$ \Comment{away from the observed midpoint, by the observed distance}
\end{algorithmic}
\textbf{$\M$:} move to $r.\des$.
\end{algorithm}

\begin{lemma}
\label{lem:epgs_in_ob_f}
$\mathrm{EPG}^{*}\in\Sol(\OB^{F})$, as witnessed by Algorithm~\ref{alg:epgs}.
\end{lemma}

\begin{proof}
The destination is computed from the current snapshot alone and requires
neither a unit of length nor chirality, so the algorithm is admissible in
$\OB$. We show by induction on $k\ge0$ that at the start of round $k+1$ the two
robots are stationary with no pending operation in the configuration
\[
  \Gamma_k:=\bigl(c_0+2^{k}(a_0-c_0),\ c_0+2^{k}(b_0-c_0)\bigr).
\]
This holds for $k=0$. Assume it for $k$ and let $\delta:=2^{k}d_0\neq0$ be the
distance. Under \FSY{} both robots are activated in round $k+1$ and compute
from the same configuration; each moves away from the peer by $\delta/2$, so
the two displacements are opposite and of equal norm. The midpoint is therefore
unchanged and the difference vector is doubled, which is $\Gamma_{k+1}$, and
both $\M$ operations are nonzero since $\delta\neq0$.

The two $\M$ operations of a round occupy the same interval, so no time of
$\Sab$ lies strictly between consecutive round boundaries at which a Move
occurs, and the joint-stop sequence is exactly $(\Gamma_k)_{k\ge0}$. It is
infinite, consecutive terms are distinct because $\delta\neq0$, and every step
satisfies $\tfrac12(a_{k+1}+b_{k+1})=c_0$ and $d_{k+1}=2d_k$. Hence
$\Phi_{\mathrm{EPG}^{*}}$ holds.
\end{proof}

\subsubsection*{Impossibility}

\begin{lemma}
\label{lem:epgs_not_lu_r}
$\mathrm{EPG}^{*}\notin\Sol(\LU^{R})$.
\end{lemma}

\begin{proof}
Apply Lemma~\ref{lem:solo-drift}(2) with $K=R$.
\end{proof}

Lemma~\ref{lem:epgs_not_lu_r} is the negative half of
Lemma~\ref{lem:six-bounds}\,(\ref{lem:six-cge-star-bounds}), whose positive half is
Lemma~\ref{lem:epgs_in_ob_f}. Since
$\Max(\mathcal Q\setminus\upa\OB^{F})=\{\LU^{R}\}$,
Lemma~\ref{lem:boundary-certificate} gives
$\Prof(\mathrm{EPG}^{*})=\upa\OB^{F}$, as already recorded in
Theorem~\ref{thm:six-principal-separators}.

\subsection{Oscillating Points with One Move per Robot per Phase (OSPO)}
\label{app:osp}

We introduce $\mathrm{OSPO}$ (Oscillating Points with One Move
per Robot per Phase), a strengthened version of
$\mathrm{OSP}$~\cite{DFPSY}.
The first approach phase starts at time~0, all phase boundaries
are joint stops, and each robot completes exactly one nonzero
Move in each phase.


\begin{definition}[Oscillating Points with One Move
per Robot per Phase]
\label{def:osp}
Two robots $a,b$ start from distinct positions, and $d(t)$ is their
distance as in Definition~\ref{def:mcv}.
For $r\in\{a,b\}$, let $m_r(t)$ denote the number of nonzero
Move operations completed by $r$ by time $t$.

The task $\mathrm{OSPO}$ requires the robots to remain at distinct
positions throughout the execution and to alternate approach and
retreat forever, starting with approach at time $0$.
Each phase is delimited by joint stops, and each robot completes
exactly one nonzero Move during that phase.
More precisely, there is a strictly increasing sequence
$(t_j)_{j\ge0}$ of joint stops, with $t_0=0$ and $t_j\to\infty$,
such that on each interval $[t_{2i},t_{2i+1}]$ the distance is
monotonically non-increasing and decreases overall, and on each
interval $[t_{2i+1},t_{2i+2}]$ it is monotonically non-decreasing
and increases overall.
The move-count condition is
$m_r(t_{j+1})-m_r(t_j)=1$ for every $j\ge0$ and $r\in\{a,b\}$.

\begin{align*}
\Phi_{\mathrm{OSPO}}\equiv{}
  &\{\,\forall t\ge0:\ d(t)>0\,\}\\
  &{}\land\exists\,(t_j)_{j\ge0}\subseteq\Sab:\ \Bigl[
    t_0=0,\quad t_j\to\infty,\\
  &\qquad \forall j\ge0:\ t_j<t_{j+1},\\
  &\qquad \forall j\ge0\ \forall r\in\{a,b\}:
       m_r(t_{j+1})-m_r(t_j)=1,\\
  &\qquad \forall i\ge0:\ \Bigl(
       d(t_{2i+1})<d(t_{2i})\\
  &\qquad\qquad {}\land
       \forall h,h'\in[t_{2i},t_{2i+1}]:
       \bigl(h<h'\Rightarrow d(h')\le d(h)\bigr)\\
  &\qquad\qquad {}\land d(t_{2i+2})>d(t_{2i+1})\\
  &\qquad\qquad {}\land
       \forall h,h'\in[t_{2i+1},t_{2i+2}]:
       \bigl(h<h'\Rightarrow d(h')\ge d(h)\bigr)
       \Bigr)\Bigr].
\end{align*}
\end{definition}

The specification $\Phi_{\mathrm{OSPO}}$ combines safety constraints
with a requirement of infinite progress.
Collision freedom is violated as soon as $d(t)=0$ for some $t$.
The move-count requirement also admits finite witnesses of failure:
since $t_0=0$ and each phase contains exactly one completed nonzero
Move by each robot, every satisfying execution must have
$m_a(t_j)=m_b(t_j)=j$ and hence
$|m_a(t)-m_b(t)|\le1$ for every $t\ge0$.
Thus, if one robot completes two nonzero Moves before the other
completes its first, no extension can satisfy $\Phi_{\mathrm{OSPO}}$.
After either a collision or such a move-count violation, any
admissible fair continuation witnesses failure.
By contrast, monotonicity of $d$ along a finite solo activation
interval does not by itself establish failure: such an interval
may lie within one phase, subject to the move-count requirement.
Requiring phase boundaries to be joint stops also prevents a
decrease and subsequent increase within a single simultaneous
Move from being counted as two phases, since there is no
intervening joint stop.

For infinite executions, the divergence $t_j\to\infty$ provides
another sufficient criterion for failure.
If, for some $T$, the restriction of $d$ to
$\Sab\cap[T,\infty)$ is monotone (non-increasing or non-decreasing),
then $\Phi_{\mathrm{OSPO}}$ fails.
Indeed, any witnessing sequence would eventually have both
a strictly decreasing phase and a strictly increasing phase
with endpoints in $\Sab\cap[T,\infty)$, contradicting that
monotonicity.
This includes the case where $d$ is eventually constant
at joint stops.

\subsubsection*{Impossibility}

\begin{lemma}
\label{lem:osp_not_ob_f}
$\mathrm{OSPO}\notin\Sol(\OB^F)$.
\end{lemma}

\begin{proof}
Fix any deterministic $\OB$ algorithm and take the fully synchronous
execution of Lemma~\ref{lem:rigidA1}, with normalized local frames at
every $\LK$ while the robots are distinct and with time-divergent
round timings.

If the two robots ever occupy the same point, the collision-freedom
condition of $\Phi_{\mathrm{OSPO}}$ fails at a finite time. This covers
$c_2=0$, where they coincide at the end of round~$1$, and also the
case in which they meet while travelling inside a round.

Otherwise $d(t)>0$ throughout. By Lemma~\ref{lem:rigidA1} every joint
stop of the execution carries one of the round-boundary distances
$d(k)=c_2^{\,k}d(0)$, so the restriction of $d$ to $\Sab$ is monotone.
The turning times required by $\Phi_{\mathrm{OSPO}}$ are joint stops,
so by the criterion stated after Definition~\ref{def:osp} their
distances cannot alternate between strict decrease and strict increase.
This admissible fair execution violates $\Phi_{\mathrm{OSPO}}$.
\end{proof}

\begin{lemma}
\label{lem:osp_not_fc_s}
$\mathrm{OSPO}\notin\Sol(\FC^S)$.
\end{lemma}

\begin{proof}
Fix any deterministic $\FC$ algorithm.
Use normalized local frames and activate both robots synchronously
as long as their computed displacements are zero.
Their colors remain equal during this prefix.

If no nonzero displacement ever occurs, continuing synchronously
gives an admissible fair execution with no movement, which violates
$\mathrm{OSPO}$.
Otherwise, let $\gamma$ be the common color immediately before
the first round that would produce nonzero displacements.
Instead of activating both robots in that round, activate $a$
alone for two complete consecutive cycles, keeping $b$ inactive.

If a collision occurs, $\mathrm{OSPO}$ already fails.
Otherwise, $b$ remains at its initial position with color $\gamma$.
At both Looks of $a$, normalization presents the same relative
position and the same peer color.
Since an $\FC$ robot cannot read its own color, it computes
the same nonzero normalized displacement in both cycles.
At the completion time $T$ of the second Move,
\[
  m_a(T)=2,\qquad m_b(T)=0.
\]

In any execution satisfying $\mathrm{OSPO}$, the first phase
starts at $t_0=0$ and ends at a joint stop $t_1$ with
$m_a(t_1)=m_b(t_1)=1$.
Since $b$ has not moved by time $T$, such a $t_1$ would have
to satisfy $t_1>T$; but then $m_a(t_1)\ge2$, a contradiction.
Thus the finite prefix already makes $\mathrm{OSPO}$ impossible.
Append any admissible fair, time-divergent $\SSY$ continuation.
\end{proof}

\begin{lemma}
\label{lem:osp_not_fs_s}
$\mathrm{OSPO}\notin\Sol(\FS^S)$.
\end{lemma}

\begin{proof}
Fix any deterministic $\FS$ algorithm and use normalized
local frames whenever the robots are distinct.
Under these observations, a robot's next private state and
normalized displacement depend only on its current private state.
Let
\[
  s_0,s_1,s_2,\ldots
\]
be the resulting state sequence from the designated initial state,
and let $\lambda_j\ge0$ be the displacement magnitude associated
with state $s_j$.

Suppose first that this sequence contains at least two nonzero
displacements. Activate $a$ alone until it completes its second
nonzero Move, while keeping $b$ inactive.
This takes finitely many cycles.
If a collision occurs, the collision-freedom condition fails.
Otherwise, at the completion time $T$,
\[
  m_a(T)=2,\qquad m_b(T)=0.
\]
No endpoint $t_1$ of the first phase can then satisfy
$m_a(t_1)=m_b(t_1)=1$: it cannot occur by time $T$ because
$b$ has not moved, and it cannot occur later because
$a$ has already moved twice.
The violation persists under any admissible fair,
time-divergent continuation.

Suppose instead that the state sequence contains at most one
nonzero displacement.
Choose an infinite fully synchronous execution with normalized
frames. If the robots collide, $\mathrm{OSPO}$ fails.
Otherwise, each robot follows the state sequence above and
completes at most one nonzero Move.
This is incompatible with infinitely many phases, each requiring
one nonzero Move from each robot.
Hence this fair, time-divergent execution also violates
$\mathrm{OSPO}$.
\end{proof}

\subsubsection*{Solvability}

\begin{lemma}
\label{lem:osp_in_fs_r}
$\mathrm{OSPO}\in\Sol(\FS^R)$.
\end{lemma}

\begin{proof}
\emph{The algorithm.}
Each robot keeps two private states, \texttt{approach} and
\texttt{retreat}, initially \texttt{approach}. Upon an activation, let
$p$ be its own position and $q$ the observed position of its peer. In
state \texttt{approach} its destination is $p+\tfrac14(q-p)$, in state
\texttt{retreat} it is $p-\tfrac14(q-p)$; in both cases it then
switches to the other state. The destination is a fixed fraction of
the observed displacement, so it depends neither on the orientation
nor on the unit distance of the local frame and is admissible under
variable disorientation. The two robots are always at distinct
positions, as shown below, so $q\neq p$ and every $\M$ is nonzero;
hence every activation increases $m_r$ by exactly one.

\emph{One round.}
Let $D>0$ be the distance at the start of a round, and for each robot
$r$ let $u_r\in[0,1]$ be the fraction of its movement segment covered
so far, set to $0$ before its $\M$ begins and to $1$ after its
completion; each $u_r$ is continuous and non-decreasing, and
$u_a+u_b$ runs from $0$ to the number of robots activated in the
round. If both robots are activated
and share their state, the distance during the round is
\[
  D\Bigl(1-\frac{u_a+u_b}{4}\Bigr)\ \text{(approach)},
  \qquad
  D\Bigl(1+\frac{u_a+u_b}{4}\Bigr)\ \text{(retreat)} ,
\]
since both displacements are collinear with the segment. Here
$u_a+u_b\in[0,2]$, so the approach factor lies in $[\tfrac12,1]$: the
distance is monotone in the required direction, stays positive, and
ends at $D/2$ or $3D/2$. If a single robot is activated, the same
expressions hold with $u_a+u_b\in[0,1]$, the approach factor lies in
$[\tfrac34,1]$, and the final distance is $3D/4$ or $5D/4$.

\emph{Phases.}
If the execution is fully synchronous forever, both robots always
share their state and successive rounds alternate between approach and
retreat; take each round as one phase.

Otherwise let $\ell\ge0$ be the length of the fully synchronous prefix,
after which the activations are alternating singletons
(Lemma~\ref{lem:altA}). Both robots share their state throughout the
prefix, and round $k\le\ell$ is an approach round exactly when $k$ is
odd; at the end of the prefix their common state is \texttt{approach}
exactly when $\ell$ is even. Partition the singleton tail into
consecutive pairs of rounds. In a pair the two robots are activated
once each; the second one has not yet switched, so both use the state
held in common at the start of the pair, and at the end of the pair
both hold the switched state. Take each such pair as one phase.
Consequently the first tail phase is an approach phase exactly when
$\ell$ is even, which for $\ell\ge1$ is the type opposite to round
$\ell$; the tail phases alternate, so the phase types alternate across
the whole execution, beginning with approach.

\emph{Within a phase.}
A phase of the first kind is a single round, treated above. A phase of
the second kind is two singleton rounds of the same type. Each of them
is monotone in the same direction and they agree at the round
boundary, so the distance is monotone on the whole phase; starting
from $D$ it ends at $\tfrac{9}{16}D$ in an approach phase and at
$\tfrac{25}{16}D$ in a retreat phase, so it also changes strictly.
In every case the distance stays positive, so $d(t)>0$ for all $t\ge0$.

\emph{The times $t_j$.}
Take $t_0=0$ and let $t_{j+1}$ be the end of the $(j+1)$-st phase.
Each phase ends at a round boundary, where neither robot is inside a
Move, so $t_j\in\Sab$ for every $j$; the same holds at $t_0=0$, where
no cycle is pending. Each phase contains exactly one nonzero $\M$ of
each robot, so
\(
  m_r(t_{j+1})-m_r(t_j)=1
\)
for every $j\ge0$ and $r\in\{a,b\}$. The phases exhaust the rounds of
the execution and each contains at least one round, so $t_j\to\infty$
by time divergence. Together with the monotonicity, strictness and
positivity established above, and with the first phase being an
approach phase, all conditions of Definition~\ref{def:osp} hold.
\end{proof}

Lemma~\ref{lem:osp_not_ob_f}, Lemma~\ref{lem:osp_not_fc_s},
Lemma~\ref{lem:osp_not_fs_s} and Lemma~\ref{lem:osp_in_fs_r} supply the
boundary certificate of Table~\ref{tab:PP13} for $\mathrm{OSPO}$, which
by Lemma~\ref{lem:boundary-certificate} gives
$\Prof(\mathrm{OSPO})=\upa\FS^R$ as asserted in
Theorem~\ref{thm:principal-task-basis}.

\subsection{Collision-Free Single Move (CFSM)}
\label{app:cfs}

\begin{definition}[Collision-Free Single Move]
\label{def:six-cfsm}
Two robots $a$ and $b$, initially at distinct positions, must each
perform exactly one nonzero Move and remain stationary before that
Move begins and after it ends. They must never collide, including
during their movements. With $\operatorname{moves}(r)$ denoting the
number of nonzero Moves of robot $r$, the specification is
\[
  \Phi_{\CFSM}
  \equiv
  \bigl(\operatorname{moves}(a)=1\bigr)
  \land
  \bigl(\operatorname{moves}(b)=1\bigr)
  \land
  \bigl(\forall t\ge0:\ a(t)\ne b(t)\bigr).
\]
\end{definition}

\subsubsection*{Impossibility}

\begin{lemma}
\label{lem:cfsm_not_ob_f}
$\CFSM \notin\Sol(\OB^{F})$.
\end{lemma}

\begin{proof}
Let $\mathcal{A}$ be any $\OB$ algorithm and consider the fully
synchronous execution of Lemma~\ref{lem:rigidA1}. Since an $\OB$ robot
has no visible state and the normalized local input of each robot is
the same in every round, the two robots apply one and the same pair
$(\lambda,\theta)$ in every round, and the configuration is updated by
the simultaneous contraction ratio $c_2=c_2(\lambda,\theta)$ of
Lemma~\ref{lem:ratios}(ii).

If $\lambda=0$, no robot ever performs a nonzero $\M$, so
$\moves{a}=\moves{b}=0$ and the single-move condition fails. If
$c_2=0$, that is $(\lambda,\theta)=(1/2,0)$, both robots land on the
common midpoint in round~$1$, so $a(t)=b(t)$ at that time and the
collision-free condition fails. If $c_2>0$, the distance stays
positive, hence $\lambda>0$ makes every round a nonzero $\M$ for both
robots and $\moves{a}=\moves{b}=\infty$, so the single-move condition
fails again. The execution is fair in all three cases.
\end{proof}

\begin{lemma}
\label{lem:cfsm_not_fc_s}
$\CFSM \notin\Sol(\FC^{S})$.
\end{lemma}

\begin{proof}
Let $\mathcal{A}$ be any deterministic $\FC$ algorithm.

\emph{First stage.} Run a fully synchronous prefix, which is admissible
under $\SSY$. By Lemma~\ref{lem:rigidA2} the two robots always carry
the same color, and we write $\gamma_0,\gamma_1,\dots$ for that common
color sequence. If $\lambda_\gamma=0$ for every color $\gamma$
occurring in the sequence, then this fair execution contains no nonzero
$\M$ and $\moves{a}=0$, so the single-move condition fails. Otherwise
let $\gamma$ be the first color of the sequence with
$\lambda_\gamma>0$ and truncate the prefix at the round in which both
robots carry $\gamma$. Nothing has moved so far, so the robots are
still at their initial positions and the distance is $d>0$.

\emph{Second stage.} From that configuration, activate $a$ alone
repeatedly (Lemma~\ref{lem:solo}; under $\SSY$ the robot that is not
activated has no pending operation, and the activations of $b$ are
placed far enough in the future to keep the execution fair). The robot
$b$ is never activated, so its color remains $\gamma$, and the visible
state of $a$, which is the color of its peer, is constant along the
interval. By Lemma~\ref{lem:single-rule}(b) the robot $a$ therefore
applies the same pair $(\lambda_\gamma,\theta_\gamma)$ at every
activation. Since $\lambda_\gamma>0$ and the distance is positive, the
first activation is a nonzero $\M$. If it lands on $b(0)$, the two
robots occupy the same point and the collision-free condition fails.
Otherwise the distance remains positive, so the second activation is a
nonzero $\M$ as well and $\moves{a}\ge2$, which violates the
single-move condition.
\end{proof}

\subsubsection*{Solvability}

Algorithm~\ref{alg:CFSM_FS_A} uses one private bit to make each robot
move exactly once. The move is directed \emph{away} from the observed
peer, so the two trajectories remain on opposite sides of the initial
segment, even under $\ASY$.

\begin{algorithm}
\caption{$\CFSM$ for $\FS$ under $\ASY$ (robot $r$)}
\label{alg:CFSM_FS_A}
\textbf{Model and scheduler:} $\FS$, $\ASY$. \\
\textbf{State:} $r.\mathit{status}\in\{0,1\}$, initially $0$. \\
\textbf{$\LK$:} observe the peer $\overline r$ and obtain the distance
$D=\|\overline r-r\|$.\\
\textbf{$\CP$:}
\begin{algorithmic}[1]
\If{$r.\mathit{status}=0$}
  \State $r.\des\gets$ the point at distance $D/2$ from $r$ in the direction away from $\overline r$
  \State $r.\mathit{status}\gets1$
\Else
  \State $r.\des\gets$ the current position of $r$
  \Comment{no further movement}
\EndIf
\end{algorithmic}
\textbf{$\M$:} move to $r.\des$.
\end{algorithm}

\begin{lemma}
\label{lem:cfsm_in_fs_a}
$\CFSM \in \Sol(\FS^{A})$, as witnessed by
Algorithm~\ref{alg:CFSM_FS_A}.
\end{lemma}

\begin{proof}
Put a coordinate on the line through $a(0)$ and $b(0)$ so that
$b(0)<a(0)$. We claim the invariants
\[
  a(t)\ge a(0)\quad\text{and}\quad b(t)\le b(0)\qquad(\forall t\ge0),
\]
which give $b(t)\le b(0)<a(0)\le a(t)$ and hence the collision-free
condition at every instant, during movements included.

Each robot moves at most once, namely at the activation at which its
status is still $0$; at that moment it stands at its initial position.
Let $p$ be the position of $b$ recorded by the $\LK$ of $a$. If $b$ has
not yet moved then $p=b(0)$; if $b$ is moving or has already landed
then $p\le b(0)$ by the invariant for $b$. In both cases $p<a(0)$, so
$D=a(0)-p>0$ and $a$ moves to
\[
  a(0)+\tfrac12\bigl(a(0)-p\bigr)>a(0),
\]
and it stays there because all its later $\M$ operations are zero. The
argument for $b$ is symmetric. The two inductions use only the fact,
valid from time $0$ on, that neither trajectory has ever crossed the
initial position of the other, so they are not circular.

For the single-move condition, fairness guarantees that each robot is
activated at least once. At its first activation the observed distance
is positive, so its displacement $D/2$ is nonzero; every later
activation produces a zero $\M$. Hence $\moves{r}=1$ for both robots.
\end{proof}

Lemma~\ref{lem:cfsm_not_ob_f}, Lemma~\ref{lem:cfsm_not_fc_s} and
Lemma~\ref{lem:cfsm_in_fs_a} together prove Lemma~\ref{lem:six-bounds}\,(\ref{lem:six-cfsm-bounds}).

\subsection{The Shared Exchange Phase of S-RDV and VSS}
\label{app:shared-exchange-phase}

\SRDV (Definition~\ref{def:s-rdv}) and \VSS
(Definition~\ref{def:six-vss}) differ in their terminal requirements,
one demanding rendezvous and the other permanent stopping, but they
impose the same initial instance and the same exchange phase. This
subsection isolates that common part and proves the four statements
that Appendices~\ref{app:s-rdv}, \ref{app:vss},
\ref{app:cf-two-task-disjunctions} and
\ref{app:cf-three-task-disjunctions} use.

\subsubsection*{The common setting}

Let $a_0=a(0)$ and $b_0=b(0)$ be the distinct initial positions, let
$d=\|b_0-a_0\|>0$, and let
\[
  A=2a_0-b_0,\qquad B=2b_0-a_0 ,
\]
so that $A,a_0,b_0,B$ occur in this order at equal spacing $d$ on the
initial line. Throughout this appendix we place a coordinate on that
line with
\[
  A=0,\qquad a_0=d,\qquad b_0=2d,\qquad B=3d,
\]
and write $P:=\{0,d,2d,3d\}$ for the set of designated points. The
exchange configurations are
\[
  C_1=(B,b_0),\qquad C_2=(a_0,A),\qquad C_3=(B,A).
\]
In the notation of Definition~\ref{def:six-vss} these are
$A=A_0$, $B=B_0$ and $C_i=\gamma_i$ for $i=1,2,3$, while
$\gamma_0=(a_0,b_0)$ is the initial configuration. We write $m_r(T)$ for
the number of nonzero $\M$ operations performed by $r$ up to time $T$.

\begin{figure}[t]
\centering
\begin{tikzpicture}[x=1.1cm,y=0.75cm]
  \foreach \k/\lab in {0/$A$,1/$a_0$,2/$b_0$,3/$B$}{
    \node[above] at (\k,0.15) {\lab};
  }
  \foreach \r in {0,1,2,3}{
    \draw (0,-\r) -- (3,-\r);
    \foreach \k in {0,1,2,3}{\fill (\k,-\r) circle (1.2pt);}
  }
  \node[left] at (-0.35,0)  {$\gamma_0$};
  \node[left] at (-0.35,-1) {$C_1$};
  \node[left] at (-0.35,-2) {$C_2$};
  \node[left] at (-0.35,-3) {$C_3$};
  \foreach \ya/\yb in {0/0}{}
  \node[below] at (1,-0.25) {$a$}; \node[below] at (2,-0.25) {$b$};
  \node[below] at (3,-1.25) {$a$}; \node[below] at (2,-1.25) {$b$};
  \node[below] at (1,-2.25) {$a$}; \node[below] at (0,-2.25) {$b$};
  \node[below] at (3,-3.25) {$a$}; \node[below] at (0,-3.25) {$b$};
\end{tikzpicture}
\caption{The initial configuration $\gamma_0$ and the three exchange
configurations $C_1$--$C_3$ admitted in the first phase.}
\label{fig:srdv}
\end{figure}
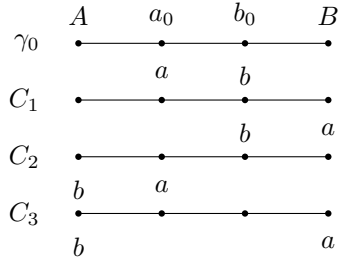

Both tasks require the following, which we refer to collectively as
condition $(\mathrm{P1})$.

\begin{itemize}
\item[$(\mathrm{P1}\text{-}a)$] In the initial configuration the points
  $A,a_0,b_0,B$ lie on a line in this order at equal spacing $d$.
\item[$(\mathrm{P1}\text{-}b)$] Every nonzero $\M$ ends at a point of
  $P$: before $T_1$ for \SRDV, and throughout the execution for \VSS.
\item[$(\mathrm{P1}\text{-}c)$] At some time $T_1$ one of $C_1,C_2,C_3$
  is formed, with the corresponding cumulative counts
  $\bigl(m_a(T_1),m_b(T_1)\bigr)$ equal to $(1,0)$, $(0,1)$ or $(1,1)$.
\end{itemize}

\subsubsection*{The first nonzero Move}

\begin{lemma}[A nonzero Move must occur]
\label{lem:p1_no_null}
In every execution, starting from the initial configuration, of a
deterministic algorithm that solves a task containing
$(\mathrm{P1})$, at least one robot performs a nonzero $\M$.
\end{lemma}

\begin{proof}
If neither robot ever performs a nonzero $\M$, then $m_a(T)=m_b(T)=0$
for every $T$. Each of the three clauses of $(\mathrm{P1}\text{-}c)$
requires $m_a(T_1)+m_b(T_1)\ge1$, so no valid $T_1$ exists, a
contradiction.
\end{proof}

\begin{lemma}[Uniqueness of the first Move]
\label{lem:p1_first_move}
In every execution satisfying $\SRDV$ or $\VSS$, each nonzero
$\M$ with earliest completion time lands at $B$ if performed
by $a$, and at $A$ if performed by $b$.
If the Look associated with that $\M$ observes the initial
geometric configuration, its normalized displacement is
$(\lambda,\theta)=(2,0)$.
\end{lemma}

\begin{proof}
Such a $\M$ exists by Lemma~\ref{lem:p1_no_null}.
Suppose that $a$ completes one at the earliest completion
time $t^\ast$, and let $p$ be its landing point.
Then $m_a(t^\ast)=1$; the peer may complete its first
nonzero $\M$ simultaneously.

Any valid exchange time $T_1$ satisfies $T_1\ge t^\ast$,
since both completed-move counts are zero before $t^\ast$.
As $m_a(T_1)\ge1$, the witness cannot be $C_2$.
The remaining possibilities, $C_1$ and $C_3$, both require
\[
  a(T_1)=B,\qquad m_a(T_1)=1.
\]

Suppose that $p\ne B$.
Then $T_1>t^\ast$, and reaching $B$ requires another
nonzero $\M$ of $a$.
In $\SRDV$, every movement remains in $[A,B]$.
In $\VSS$, every movement is a straight segment with both
endpoints in $P\subseteq[A,B]$.
In either case, the exterior endpoint $B$ cannot be reached
in the interior of a nonzero movement.
Thus reaching $B$ requires completing another nonzero $\M$,
giving $m_a(T_1)\ge2$, a contradiction.
Therefore $p=B$.
The argument for a first completion by $b$ is symmetric
and gives the landing point $A$.

Now assume that the Look associated with this $\M$ observes
the initial geometric configuration.
For a Move of $a$, the observed positions are $a_0,b_0$, and
\[
  B-a_0=2(b_0-a_0).
\]
Hence the displacement is twice the observed distance
in the direction of the peer, so its normalized form is
$(\lambda,\theta)=(2,0)$.
For a Move of $b$, the same conclusion follows from
$A-b_0=2(a_0-b_0)$.
\end{proof}

The proof uses the fact that every Move is contained in $[A,B]$.
For $\SRDV$, this is an explicit requirement.
For $\VSS$, it follows from straight-line motion and the
requirement that every Move ends in $P\subseteq[A,B]$,
since the initial positions also belong to $P$.


\subsubsection*{Replaying the first Move}

The next lemma is the obstruction used whenever an exchange task is
refuted in a model without peer-state information. It says that in
$\FS$ and $\OB$ the required first Move is necessarily imitated by the
peer, which carries the pair out of the admissible region.

\begin{lemma}[The peer replays the first Move]
\label{lem:p1_replay}
Let $M\in\{\FS,\OB\}$ and consider an execution that starts from the
initial configuration, consists of alternating singleton activations
(Lemma~\ref{lem:altA}), and is given normalized frames
(Lemma~\ref{lem:norm}(1)) up to and including the round following the
first nonzero $\M$. If that first nonzero $\M$ complies with
Lemma~\ref{lem:p1_first_move}, then the robot activated in the next
round receives an isomorphic local input, applies the same displacement,
and lands at $4d$ if it is $b$, and at $-d$ if it is $a$. In particular
the landing point lies outside $P$ and outside the segment $[A,B]$.
\end{lemma}

\begin{proof}
As long as no nonzero $\M$ has occurred, the configuration is unchanged,
so under alternating singleton activation the two robots traverse the
same sequence of states. In $\FS$ and $\OB$ the internal state of the
peer is not part of the local input, which consists of the relative
position of the peer together with the robot's own color in $\FS$ and of
the relative position alone in $\OB$.

Suppose the first nonzero $\M$ is performed by $a$. By
Lemma~\ref{lem:p1_first_move} it lands at $B=3d$, while $b$ is still at
$b_0=2d$, so the distance is again $d$. In the next round $b$ is
activated, and since the adversary may supply the normalized frame, the
local input of $b$ is isomorphic to the one that $a$ received. By
determinism $b$ applies the same displacement $(\lambda,\theta)=(2,0)$
and lands at $2d+2d=4d$. We have $4d\notin P$ and $4d\notin[A,B]=[0,3d]$.
The case in which $b$ moves first is symmetric: $b$ lands at $A=0$, $a$
is still at $d$, and $a$ lands at $d-2d=-d\notin P\cup[A,B]$.
\end{proof}

Lemma~\ref{lem:p1_replay} refutes the endpoint requirement
$(\mathrm{P1}\text{-}b)$ of VSS and the interval requirement of \SRDV in
a single step, and it is a safety violation, so it may be placed before
any liveness argument in a concatenated execution
(Lemma~\ref{lem:concat}(1)).

\subsection{Swap--Rendezvous (S-RDV)}
\label{app:s-rdv}
\label{app:s-rdv-bounds}

This subsection gives the formal specification of \SRDV, informally
described in Section~\ref{subsec:basis-proof-obligations}, and proves
Lemma~\ref{lem:s-rdv-bounds}. Recall from
Appendix~\ref{app:shared-exchange-phase} the auxiliary points
$A=2a_0-b_0$ and $B=2b_0-a_0$, the exchange configurations
$C_1=(B,b_0)$, $C_2=(a_0,A)$, $C_3=(B,A)$, the coordinate $A=0$,
$a_0=d$, $b_0=2d$, $B=3d$ on the initial line, and the set
$P=\{0,d,2d,3d\}$ of designated points.

\begin{definition}[Swap--Rendezvous]
\label{def:s-rdv}
An execution solves \SRDV if the following conditions hold.
\begin{enumerate}
\item Throughout the execution, both robots remain in the segment
      $[A,B]$.
\item There exists a time $T_1>0$ at which one of $C_1,C_2,C_3$ is
      formed, with the corresponding cumulative nonzero-move counts
      $(1,0)$, $(0,1)$, or $(1,1)$. Before $T_1$, every robot that is
      not moving is located at one of $A,a_0,b_0,B$.
\item There exists $T_2>T_1$ such that $a(t)=b(t)$ for every
      $t\ge T_2$.
\end{enumerate}
\end{definition}

Conditions~1 and~2 are condition $(\mathrm{P1})$ of
Appendix~\ref{app:shared-exchange-phase}, with
$(\mathrm{P1}\text{-}b)$ imposed before $T_1$, and condition~3 is the
terminal rendezvous requirement that distinguishes \SRDV from \VSS.
Consequently Lemmas~\ref{lem:p1_no_null}, \ref{lem:p1_first_move},
and Lemma~\ref{lem:p1_replay} all apply
to \SRDV.

\subsubsection*{Impossibility}
\begin{lemma}
\label{lem:srdv_not_fc_a}
$\mathrm{S\text{-}RDV}\notin\Sol(\FC^A)$.
\end{lemma}

\begin{proof}
Stable Rendezvous is impossible in $\FC^A$ under variable
disorientation, regardless of the number of colors and without
restricting the algorithm to the $\mathcal L$-class
~\cite{MSW-rendezvous}.

Suppose, for contradiction, that an algorithm $\mathcal A$ solves
$\mathrm{S\text{-}RDV}$ in $\FC^A$. Every initial configuration of
Rendezvous with distinct robot positions is an instance of
$\mathrm{S\text{-}RDV}$, since the auxiliary points $A$ and $B$ are
determined solely from the initial positions and are not landmarks
available to the robots. Moreover, every execution satisfying
$\mathrm{S\text{-}RDV}$ eventually satisfies
\[
  a(t)=b(t)
  \qquad\text{for all sufficiently large }t,
\]
which is exactly the stable Rendezvous requirement.
Hence $\mathcal A$ would solve Rendezvous in $\FC^A$, a contradiction.
\end{proof}

\begin{lemma}
\label{lem:srdv_not_fs_r}
$\SRDV\notin\Sol(\FS^{R})$.
\end{lemma}

\begin{proof}
Suppose an $\FS^{R}$ algorithm $\mathcal{A}$ solves \SRDV. The local
input of an $\FS$ robot is the pair consisting of the relative position
of the peer and its own state; the state of the peer is not included.
The adversary chooses the alternating singleton activation
$\{a\},\{b\},\{a\},\dots$ with an empty fully synchronous prefix
(Lemma~\ref{lem:altA}).

\emph{First stage.} The configuration is unchanged until the first
nonzero $\M$. The initial configuration is symmetric under the rotation
by $\pi$ about the midpoint combined with exchanging $a$ and $b$, and
this rotation preserves chirality, so in the normalized frame
(Lemma~\ref{lem:norm}(1)) both local inputs have the form
$(e,\text{own state})$ and are isomorphic. The two robots start from
the same initial state $s_0$, so $a$ in round $2k-1$ and $b$ in round
$2k$ both pass from $s_{k-1}$ to $s_k$. If no state occurring in
$s_0,s_1,\dots$ produces a nonzero $\M$, the configuration never
changes, contradicting Lemma~\ref{lem:p1_no_null}. Let $j$ be the least
index that produces a nonzero $\M$; then $a$ performs the first nonzero
$\M$ of the whole execution in round $2j+1$, from state $s_j$.

\emph{Second and third stages.} Since $\mathcal{A}$ is assumed to solve
\SRDV, Lemma~\ref{lem:p1_first_move} forces that $\M$ to land at
$B=3d$. The execution is an alternating singleton activation from the
initial configuration under normalized frames whose first nonzero $\M$
complies with Lemma~\ref{lem:p1_first_move}, which is exactly the
hypothesis of Lemma~\ref{lem:p1_replay}. Hence the robot $b$, activated
in the next round, receives an isomorphic local input, performs the
same displacement and lands at $4d$. Since $4d\notin[A,B]$,
condition~1 of Definition~\ref{def:s-rdv} is violated.
\end{proof}

\begin{lemma}
\label{lem:srdv_not_ob_f}
$\SRDV\notin\Sol(\OB^{F})$.
\end{lemma}

\begin{proof}
Suppose an $\OB^{F}$ algorithm $\mathcal{A}$ solves \SRDV. The local
input of an $\OB$ robot is its current snapshot alone. In round~$1$ the
local input of both robots is ``the peer at distance $d$''. A zero $\M$
would leave all later inputs unchanged and the configuration fixed
forever, so by Lemma~\ref{lem:p1_no_null} the $\M$ is nonzero, and by
Lemma~\ref{lem:p1_first_move} the robot $a$ lands at $B$ and $b$ lands
at $A$, forming $C_3$ at distance $3d$.

The local input in round~$2$, ``the peer at distance $3d$'', coincides
with the round-$1$ local input of the initial instance of initial
distance $3d$, namely $A'=0$, $a'(0)=3d$, $b'(0)=6d$, $B'=9d$, and an
$\OB$ robot cannot distinguish the two. In that instance the same
argument as above shows that the first $\M$ is nonzero, so by
Lemma~\ref{lem:p1_first_move} the output is twice the observed distance
in the direction of the peer. Applying the same local rule in round~$2$
of the actual execution, $a$ lands at $3d-6d=-3d$ and $b$ lands at
$0+6d=6d$. Both lie outside $[A,B]=[0,3d]$, which violates condition~1
of Definition~\ref{def:s-rdv}.
\end{proof}

\subsubsection*{Solvability}

\begin{algorithm}
\caption{\SRDV for $\FC$ under $\SSY$ (robot $r$)}
\label{alg:S_RDV_FC_S}
\textbf{Model and scheduler:} $\FC$, $\SSY$. \\
\textbf{Light:} $r.\lcolor\in\{0,1,2,3\}$, initially $0$. \\
\textbf{$\LK$:} observe the peer $\overline r$ and obtain the distance
$D=\|\overline r-r\|$ and the color $\overline r.\lcolor$.\\
\textbf{$\CP$:}
\begin{algorithmic}[1]
\If{$D=0$}
  \State $r.\des\gets$ the current position of $r$
  \Comment{rule $\mathrm{R}$: keep the rendezvous}
\ElsIf{$\overline r.\lcolor=0$}
  \State $r.\des\gets$ the point at distance $2D$ from $r$ in the direction of $\overline r$
  \Comment{rule $0$: jump over}
  \State $r.\lcolor\gets1$
\ElsIf{$\overline r.\lcolor=1$}
  \State $r.\des\gets$ the midpoint of $r$ and $\overline r$
  \Comment{rule $1$: to the midpoint}
  \State $r.\lcolor\gets2$
\ElsIf{$\overline r.\lcolor=2$}
  \State $r.\des\gets$ the current position of $r$
  \Comment{rule $2$: stay}
  \State $r.\lcolor\gets3$
\Else
  \State $r.\des\gets$ the current position of $\overline r$
  \Comment{rule $3$: to the peer}
  \State $r.\lcolor\gets1$
\EndIf
\end{algorithmic}
\textbf{$\M$:} move to $r.\des$.
\end{algorithm}

\begin{lemma}
\label{lem:srdv_in_fc_s}
$\SRDV\in\Sol(\FC^{S})$, as witnessed by
Algorithm~\ref{alg:S_RDV_FC_S}.
\end{lemma}
\begin{proof}
Each robot carries a light with color set $\{0,1,2,3\}$,
initially $0$. Apart from the coincidence test, every rule
depends on the color of the peer alone.

\emph{(i) Both robots always lie in $[A,B]$.}
Under $\SSY$, all active robots compute from the configuration
at the beginning of the round.
Rules $\mathrm{R}$ and $2$ do not move, while rules $1$ and $3$
choose the midpoint and the peer's position, respectively.
These destinations belong to $[A,B]$ whenever both current
positions do.

For rule $0$, the peer still has color $0$ and remains
at its initial position.
While that color remains $0$, successive applications of
rule $0$ alternate the observing robot between its initial
position and the opposite outer point.
The possible configurations and destinations are
\[
\begin{aligned}
  (a,b)&=(d,2d): &&a\to3d=B, &
  (a,b)&=(3d,2d): &&a\to d=a_0,\\
  (a,b)&=(d,2d): &&b\to0=A, &
  (a,b)&=(d,0): &&b\to2d=b_0.
\end{aligned}
\]
Thus every destination lies in $[A,B]$.
Since each Move follows a straight segment, both robots
remain in $[A,B]$ throughout the execution.

\emph{(ii) The first phase.}
Let $T_1$ be the end of the first round.
If both robots are activated, $C_3$ is formed with cumulative
nonzero-move counts $(1,1)$.
If only $a$ is activated, $C_1$ is formed with counts $(1,0)$;
if only $b$ is activated, $C_2$ is formed with counts $(0,1)$.
Moreover, up to $T_1$, every stationary robot is at its initial
position or at its first destination, all of which belong to $P$.
Hence $T_1$ is a valid witness for the exchange phase.

\emph{(iii) The second phase.}
By fairness, both robots are activated within finite time.
At the end of the first round by which this has happened,
the color pair is $(1,1)$, $(1,2)$, or $(2,1)$.
From $(1,1)$, simultaneous activation brings both robots
to their common midpoint; a singleton activation produces
$(1,2)$ or $(2,1)$.

Consider $(1,2)$, the other case being symmetric.
Until $a$ is activated, only $b$ moves, repeatedly halving
the distance while preserving the color pair.
By fairness, $a$ is eventually activated: it sees color $2$,
stays in place, and changes its color to $3$.
Whether or not $b$ is activated in that round, the resulting
color pair is $(3,2)$.
Until the next activation of $b$, robot $a$ remains stationary.
When $b$ is activated, it sees color $3$ and moves to $a$;
if $a$ is activated simultaneously, it sees color $2$ and stays.
Thus the robots occupy the same point at the end of this round.
Fairness ensures that this occurs at some finite time $T_2>T_1$.

\emph{(iv) Stability of the rendezvous.}
At the round boundary $T_2$, both robots have completed
their Moves and occupy the same point.
Every subsequent Look therefore observes coincidence,
so rule $\mathrm{R}$ applies and neither robot moves again.
Hence $a(t)=b(t)$ for every $t\ge T_2$.
\end{proof}

Lemma~\ref{lem:srdv_not_fc_a}, Lemma~\ref{lem:srdv_not_fs_r},
Lemma~\ref{lem:srdv_not_ob_f} and Lemma~\ref{lem:srdv_in_fc_s} together
prove Lemma~\ref{lem:s-rdv-bounds}.

\subsection{Visit Swap and Stop (VSS)}
\label{app:vss}

\begin{definition}[VSS]
\label{def:six-vss}
Let $a_0=a(0)$ and $b_0=b(0)$ be distinct, and define
\(
  A_0=2a_0-b_0,
  \;
  B_0=2b_0-a_0.
\)
The four points $A_0,a_0,b_0,B_0$ occur in this order at equal
spacing on the initial line. Write $\gamma(t)=(a(t),b(t))$ and set
\[
\begin{aligned}
  \gamma_0&=(a_0,b_0), & \gamma_1&=(B_0,b_0),\\
  \gamma_2&=(a_0,A_0), & \gamma_3&=(B_0,A_0).
\end{aligned}
\]
Every completed Move must end at one of the four designated points.
The execution must first reach one of $\gamma_1,\gamma_2,\gamma_3$,
with cumulative nonzero move counts $(1,0)$, $(0,1)$, or $(1,1)$,
respectively. After this exchange configuration has been formed,
the robots must eventually remain permanently stationary in a
configuration belonging to
$\{\gamma_0,\gamma_1,\gamma_2,\gamma_3\}$.
The restriction to designated points concerns Move endpoints;
straight-line travel between them, including passage through the
peer's position, is allowed.
\end{definition}

\subsubsection*{Impossibility}

\begin{lemma}
\label{lem:vss_not_fs_r}
$\VSS\notin\Sol(\FS^{R})$.
\end{lemma}

\begin{proof}
Use the alternating singleton execution of the proof of
Lemma~\ref{lem:srdv_not_fs_r}. Its first stage holds word for word,
since it uses only $(\mathrm{P1})$, so the first nonzero $\M$ of the
whole execution is performed by $a$ in round $2j+1$ under a normalized
frame. Since the algorithm is assumed to solve \VSS,
Lemma~\ref{lem:p1_first_move} forces that $\M$ to land at $B=3d$, and
Lemma~\ref{lem:p1_replay} then applies: the robot $b$, activated in the
next round, receives an isomorphic local input, performs the same
displacement and lands at $4d\notin P$. This contradicts the
requirement of Definition~\ref{def:six-vss} that every completed $\M$
end at one of the four designated points.
\end{proof}

\begin{lemma}
\label{lem:vss_not_ob_f}
$\VSS\notin\Sol(\OB^{F})$.
\end{lemma}

\begin{proof}
As in the proof of Lemma~\ref{lem:srdv_not_ob_f}, in round~$1$ the
robot $a$ lands at $B$ and $b$ lands at $A$, forming $C_3$, and in
round~$2$ the robots cannot distinguish their input from the round-$1$
input of the initial instance of initial distance $3d$, so $a$ lands at
$-3d$ and $b$ lands at $6d$. Neither point belongs to $P$, so the
endpoint requirement of Definition~\ref{def:six-vss} is violated.
\end{proof}

\subsubsection*{Solvability}

\begin{algorithm}
\caption{\VSS for $\FC$ under $\ASY$ (robot $r$)}
\label{alg:VSS_FC_A}
\textbf{Model and scheduler:} $\FC$, $\ASY$. \\
\textbf{Light:} $r.\lcolor\in\{0,1\}$, initially $0$. \\
\textbf{$\LK$:} observe the peer $\overline r$ and obtain the distance
$D=\|\overline r-r\|$ and the color $\overline r.\lcolor$.\\
\textbf{$\CP$:}
\begin{algorithmic}[1]
\If{$\overline r.\lcolor=0$}
  \State $r.\des\gets$ the point at distance $2D$ from $r$ in the direction of $\overline r$
  \Comment{rule $0$: jump over}
  \State $r.\lcolor\gets1$
\Else
  \State $r.\des\gets$ the current position of $r$
  \Comment{rule $1$: stop}
  \State $r.\lcolor\gets1$
\EndIf
\end{algorithmic}
\textbf{$\M$:} move to $r.\des$.
\end{algorithm}

\begin{lemma}
\label{lem:vss_in_fc_a}
$\VSS\in\Sol(\FC^{A})$, as witnessed by
Algorithm~\ref{alg:VSS_FC_A}.
\end{lemma}

\begin{proof}
Here $D$ is the distance recorded at the robot's own $\LK$.

\emph{(i) A robot of color $0$ is at its initial position and is not
moving.} The color of a robot is $0$ until its first $\CP$, and since
$t_L<t_C<t_B<t_E$, its first $\M$ starts after that $\CP$. Hence an
observation of a peer of color $0$ is an observation of a robot that
has not begun any $\M$. In particular the observations during a peer's
movement that $\ASY$ permits cannot affect the decision of rule $0$,
and no atomicity assumption is used below.

\emph{(ii) When rule $0$ fires, $D=d$ and the destination belongs to
$P$.} By~(i) the peer is at its initial position, and the observing
robot is itself not moving at its own $\LK$, so by induction it stands
at a landing point of $P$. The possible cases and landing points are
the four listed in part~(i) of the proof of
Lemma~\ref{lem:srdv_in_fc_s}, all of which belong to $P$. In
particular each firing of rule $0$ toggles the robot between its own
initial position and $B$, respectively between its own initial position
and $A$. The destination is computed from the snapshot taken at the
$\LK$, so delaying the $\M$ (Corollary~\ref{cor:stale}) does not change
the landing point.

\emph{(iii) The first phase.} Let $t$ be the earliest completion of a nonzero Move,
say by $a$. By (ii), $a(t)=B$.
If $b$ has not begun a nonzero Move, then $C_1$ holds at $t$.
Otherwise, $b$ is executing its first nonzero Move towards $A$
(or completes it at $t$), and its color is already $1$.
Hence $a$ remains at $B$ until $b$ completes this Move,
when $C_3$ holds with cumulative counts $(1,1)$.
The case in which $b$ completes first is symmetric.


\emph{(iv) Permanent stopping.} A robot applies rule $0$ only if its
$\LK$ took place before the first $\CP$ of its peer. That $\CP$ occurs
at some finite time $T$, and each cycle is finite, so each robot
performs only finitely many $\LK$ operations before $T$ and therefore
fires rule $0$ finitely often. Under either rule the activated robot
sets its color to $1$, and after the last of these finitely many
nonzero $\M$ operations has been completed, only rule $1$ fires and
neither robot moves again. Hence $\Stopped{a}$ and $\Stopped{b}$ hold.

\emph{(v) The terminal configuration.} By~(ii), a robot that fired rule
$0$ an even number of times is back at its initial position, and one
that fired it an odd number of times stands at $B$, respectively at
$A$. The four combinations are exactly $C_0,C_1,C_2,C_3$, and
Definition~\ref{def:six-vss} admits each of them as a terminal
configuration.
\end{proof}

Lemma~\ref{lem:vss_not_fs_r}, Lemma~\ref{lem:vss_not_ob_f} and
Lemma~\ref{lem:vss_in_fc_a} together prove Lemma~\ref{lem:six-bounds}\,(\ref{lem:six-vss-bounds}).

%% file: Auxiliary-Execution-Constructions.tex
\subsection{Normalization and Distance Updates}
\label{app:normalization}

Throughout this appendix we fix two robots $a,b$ at distinct positions
and write $d=\|a-b\|$ for their distance. We call the color-dependent
part of a robot's observation its \emph{visible state} and denote it by
$\gamma$. Thus $\OB$ has no visible state, the visible state of an $\FS$
robot is its own color, that of an $\FC$ robot is the color of its peer,
and that of a $\LU$ robot is the pair of both colors. Each lemma below
carries applicability conditions on the starting state, namely the
positions, the colors, and the pending operations, and on the scheduler
class. When a lemma is applied inside a construction, these conditions
are verified against the state actually reached by the preceding prefix.

\begin{lemma}[Normalization]
\label{lem:norm}
Fix a reference vector $e\neq0$ of the plane. For every algorithm
$\mathcal A$ there are constants $\lambda_\gamma\ge0$ and
$\theta_\gamma$, determined by $\mathcal A$ alone, and a map $g$ on
visible states, such that the following hold.
\begin{enumerate}
\item[(1)] Suppose that at an activation of a robot $r$ the peer is at
  global vector $w\neq0$ and the visible state of $r$ is $\gamma$.
  Then the adversary may choose the local coordinate system of $r$ so
  that the global displacement of $r$ is
  \[
    v=\lambda_\gamma R_{\theta_\gamma}(w),
  \]
  where $R_\theta$ denotes the rotation by angle $\theta$, and the
  color set by $r$ is $g(\gamma)$.
\item[(2)] If $w=0$, the local output is determined by $\gamma$ alone,
  but the adversary may choose the direction and the magnitude of the
  resulting global displacement. In particular, two robots with equal
  visible states can be given equal global displacements.
\end{enumerate}
Unless stated otherwise, we use this coordinate system, called the
\emph{normalized frame}.
\end{lemma}

\begin{proof}
(1) Place $r$ at the origin. By chirality there is a unique
orientation-preserving similarity $\psi$ of the plane that maps $w$ to
$e$, and variable disorientation admits $\psi$ as a local coordinate
system. The local input is then $(e,\gamma)$, so by determinism the
output is a fixed pair $(q_\gamma,g(\gamma))$. Writing
$q_\gamma=\lambda_\gamma R_{\theta_\gamma}(e)$ and observing that
$\psi^{-1}$ is a composition of a rotation and a positive scaling,
hence commutes with rotations, the global displacement is
\[
  \psi^{-1}(q_\gamma)
  =\lambda_\gamma R_{\theta_\gamma}\bigl(\psi^{-1}(e)\bigr)
  =\lambda_\gamma R_{\theta_\gamma}(w).
\]
(2) If $w=0$, the geometric part of the input is the origin in every
local coordinate system, so the output $q$ depends only on $\gamma$.
The global displacement $\psi^{-1}(q)$ changes direction and magnitude
with the choice of $\psi$, and choosing the same $\psi$ for two robots
with equal visible states makes the two displacements equal.
\end{proof}

\begin{lemma}[Contraction ratios]
\label{lem:ratios}
Put
\[
  c_1(\gamma):=\sqrt{1-2\lambda_\gamma\cos\theta_\gamma+\lambda_\gamma^{2}},
  \qquad
  c_2(\gamma):=\sqrt{1-4\lambda_\gamma\cos\theta_\gamma+4\lambda_\gamma^{2}}.
\]
Under normalized frames the following hold.
\begin{enumerate}
\item[(i)] If exactly one robot with visible state $\gamma$ completes a
  cycle while its peer is stationary, then $d\mapsto c_1(\gamma)\,d$.
\item[(ii)] If two robots, both with visible state $\gamma$, complete
  cycles simultaneously from the same snapshot, then
  $d\mapsto c_2(\gamma)\,d$ and the midpoint is invariant.
\end{enumerate}
Moreover
\[
  c_1(\gamma)=0\iff(\lambda_\gamma,\theta_\gamma)=(1,0),
  \qquad
  c_2(\gamma)=0\iff(\lambda_\gamma,\theta_\gamma)=(1/2,0).
\]
\begin{enumerate}
\item[(iii)] Identify the plane with the complex plane and $R_\theta$
  with multiplication by $e^{\mathrm i\theta}$. The transformations of
  the difference vector in (i) and (ii) are multiplication by
  \[
    \sigma_1(\gamma):=1-\lambda_\gamma e^{\mathrm i\theta_\gamma},
    \qquad
    \sigma_2(\gamma):=1-2\lambda_\gamma e^{\mathrm i\theta_\gamma},
  \]
  and $c_j(\gamma)=|\sigma_j(\gamma)|$. When $c_j(\gamma)\neq0$, put
  $\psi(\gamma):=\arg\sigma_1(\gamma)$ and
  $\varphi(\gamma):=\arg\sigma_2(\gamma)$. If
  $\psi\not\equiv0\pmod{2\pi}$, then there is $k\ge1$ with
  $\cos(k\psi)\le-\tfrac12$, and for every $\theta_0$ there is $k\ge1$
  with $\cos(\theta_0+k\psi)\le0$. The same statements hold for
  $\varphi$.
\end{enumerate}
\end{lemma}

\begin{proof}
Let $w:=b-a$ and $\hat u:=w/d$.

(i) By Lemma~\ref{lem:norm}(1) the new difference vector is
$d(\hat u-\lambda R_\theta\hat u)$, whose squared norm is
$d^{2}(1-2\lambda\cos\theta+\lambda^{2})$.

(ii) The displacement of $b$ is the negative of that of $a$, so the new
difference vector is $d(\hat u-2\lambda R_\theta\hat u)$ and $a+b$ is
unchanged. The conditions for vanishing follow immediately.

(iii) The first assertion is the computation in (i) and (ii) read as
multiplication of complex numbers, together with
$|1-\lambda e^{\mathrm i\theta}|^{2}=1-2\lambda\cos\theta+\lambda^{2}$
and
$|1-2\lambda e^{\mathrm i\theta}|^{2}=1-4\lambda\cos\theta+4\lambda^{2}$.

For the second assertion, suppose first that $\psi/2\pi$ is irrational.
Then $\{k\psi\bmod2\pi\}_{k\ge1}$ is dense in $[0,2\pi)$: by the
pigeonhole principle some integer multiple of $\psi$ is arbitrarily
close to $0$, and its own multiples then fill the circle. Both
assertions follow. Suppose next that $\psi/2\pi=p/q$ in lowest terms
with $q\ge2$. Then
$\{k\psi\bmod2\pi\}_{k\ge1}=\{2\pi j/q:0\le j<q\}$. If $q$ is even,
the value $j=q/2$ belongs to this set and $\cos\pi=-1$. If $q$ is odd,
and therefore $q\ge3$, the points are spaced by $2\pi/q$, so one of
them lies within $\pi/q$ of $\pi$, where
$\cos\le-\cos(\pi/q)\le-\cos(\pi/3)=-\tfrac12$. For the version with an
offset, the set $\{\theta_0+k\psi\bmod2\pi\}_{k\ge1}$ consists of $q$
equally spaced points with gap $2\pi/q\le\pi$, so the closed arc
$[\pi/2,3\pi/2]$, which has length $\pi$, contains at least one of
them, and there $\cos\le0$. The argument for $\varphi$ is identical.
\end{proof}

From now on we fix the reference vector of Lemma~\ref{lem:norm} to be
$e=(1,0)$. The choice is only for concreteness: the arguments below use
nothing about $e$ beyond the existence, for each visible state, of the
pair $(\lambda_\gamma,\theta_\gamma)$ supplied by that lemma, and the
displacement of an activation is $\lambda_\gamma R_{\theta_\gamma}(w)$
for the actual peer vector $w$.

\begin{remark}[The rule of $\mathrm{EPG}^{*}$ in a normalized frame]
\label{rem:epgs-normalized-rule}
Write $(\lambda^{\star},\theta^{\star}):=(\tfrac12,\pi)$, the
displacement that moves a robot away from its observed peer by half the
observed distance; this is the rule of Algorithm~\ref{alg:epgs}. In a
round in which both robots are activated from the same configuration
and apply the same rule $(\lambda,\theta)$, the midpoint is preserved
and the difference vector is multiplied by
$\sigma_2=1-2\lambda e^{\mathrm{i}\theta}$
(Lemma~\ref{lem:ratios}(ii),(iii)). The expansion condition
$d_{i+1}=2d_i$ of Definition~\ref{def:epgs} therefore holds at such a
step if and only if $\sigma_2=2$, that is if and only if
$(\lambda,\theta)=(\lambda^{\star},\theta^{\star})$. In that case
$\sigma_1=1-\lambda^{\star}e^{\mathrm{i}\theta^{\star}}=\tfrac32$, so a
solo activation multiplies the distance by $c_1^{\star}=\tfrac32$ and
$\psi^{\star}\equiv0$.
\end{remark}

\subsection{Admissible Execution Segments}
\label{app:admissible-segments}

The constructions of Appendices~\ref{app:cf-two-task-disjunctions}
and~\ref{app:cf-three-task-disjunctions} are concatenations of a small
number of execution segments. This subsection records each segment
together with the scheduler classes that admit it and the state it
requires at its beginning. Recall from
Section~\ref{subsec:schedulers-events} that we write
$F,R,S,A,A_{LC},A_M$ for the scheduler classes, and that the phase
times of the $i$-th cycle of a robot $r$ are denoted
$t_L(r,i)<t_C(r,i)<t_B(r,i)<t_E(r,i)$.

\begin{lemma}[Solo activation interval]
\label{lem:solo}
Let $K\in\{S,A_{LC},A_M,A\}$, and suppose that at the end of a finite
admissible prefix the robot $b$ has no pending operation. For every
$k\ge1$ there is an execution that extends this prefix by an interval
in which $a$ completes $k$ consecutive $\LCM$ cycles while $b$ is never
activated, followed by a fair tail. Throughout this interval the
position and the color of $b$ are unchanged, and $\Stop_b(t)$ holds at
every time $t$ of the interval.
\end{lemma}

\begin{proof}
Under $\SSY$ it suffices to take $\{a\}$ as the activation set of $k$
consecutive rounds. Under $\ASY$ and its subclasses, delay the next
$\LK$ of $b$ and place the $k$ cycles of $a$ one after another inside
that delay. The atomicity constraints only forbid a $\LK$ of one robot
during prescribed phases of the other robot's cycle, and $b$ executes
no cycle during the interval, so no constraint is violated. Appending a
tail that activates both robots infinitely often makes the execution
fair.
\end{proof}

The class $K=R$ is excluded. Indeed, in an $\RSY$ execution the same
robot cannot be activated twice in a row after the fully synchronous
prefix, by Lemma~\ref{lem:altA} below.

\begin{lemma}[Delay before a phase]
\label{lem:suspend}
Under $\ASY$ and its subclasses, the adversary may insert a finite
delay before any of the first three phase times of a cycle of a robot
$r$, that is, before $t_L(r,i)$, inside $(t_L(r,i),t_C(r,i))$, or
inside $(t_C(r,i),t_B(r,i))$, and place events of the other robot
inside that delay, within the limits imposed by the atomicity
constraints. No delay is inserted inside a $\M$ operation: once the
$\M$ of $r$ begins at $t_B(r,i)$, the robot proceeds to the destination
computed at $t_C(r,i)$ and the $\M$ ends at $t_E(r,i)$ without
interruption. The values used by $r$ in the later phase are those
obtained in the earlier one.
\end{lemma}

\begin{proof}
The times $t_L,t_C,t_B,t_E$ are chosen by the adversary subject only
to the ordering constraint and to finiteness of each phase
(Section~\ref{subsec:schedulers-events}); in particular the adversary
controls the duration $t_E(r,i)-t_B(r,i)$ of a $\M$ operation, but not
an interruption of it, which is why every delay used here is placed
before $t_B(r,i)$. The second assertion is the determinism of the
algorithm: the destination and the new color computed at $t_C(r,i)$ do
not depend on anything occurring afterwards.
\end{proof}

\begin{corollary}[Delayed move with a stale destination]
\label{cor:stale}
Let $K\in\{A_{LC},A_M,A\}$. The adversary can construct a fair
execution in which $a$ computes its destination at $t_C(a,i)$, is then
delayed, the robot $b$ completes at least one nonzero $\M$ during the
delay, and only afterwards does $a$ move to the destination it had
already computed. The delay is placed before the $\M$ of $a$ begins,
and the $\M$ itself, once begun, is carried out in full. Writing
$w_{\mathrm{old}}$ and $w_{\mathrm{new}}$ for the vector from $a$ to
$b$ at $t_C(a,i)$ and at $t_B(a,i)$ respectively, the difference vector
after the move of $a$ is
\[
  w_{\mathrm{new}}-\lambda_\gamma R_{\theta_\gamma}(w_{\mathrm{old}}).
\]
\end{corollary}

\begin{proof}
Apply Lemma~\ref{lem:suspend} with the delay placed between
$t_C(a,i)$ and $t_B(a,i)$, and place the whole cycle of $b$ inside it.
What $A_{LC}$ forbids is a $\LK$ in $(t_L(a,i),t_C(a,i)]$, and what
$A_M$ forbids is a $\LK$ in $[t_B(a,i),t_E(a,i)]$; the delay
interval $(t_C(a,i),t_B(a,i))$ meets neither, so the segment is
admissible under $A_{LC}$, $A_M$ and $A$ alike. The displacement
formula is Lemma~\ref{lem:norm}(1) applied to the snapshot that $a$
obtained at its $\LK$, which is released unchanged when the $\M$
begins.
\end{proof}

\begin{lemma}[Observation during movement]
\label{lem:mid}
Let $K\in\{A_{LC},A\}$, and suppose that $b$ performs a nonzero $\M$
from $p_B$ to $p_E\neq p_B$ while $a$ has no pending operation. For
every $\mu\in(0,1)$ there is a fair execution in which $a$ performs a
$\LK$ at a time when $b$ is at
\[
  p_\mu=(1-\mu)p_B+\mu p_E .
\]
\end{lemma}

\begin{proof}
The movement of $b$ is an uninterrupted continuous motion along the
segment $[p_B,p_E]$, whose duration the adversary may choose, so by the
intermediate value theorem there is
$\tau_\mu\in(t_B(b,i),t_E(b,i))$ with $b(\tau_\mu)=p_\mu$. Setting
$t_L(a,j)=\tau_\mu$ is unconstrained under $A$, and it is admissible
under $A_{LC}$ as well, because $\tau_\mu>t_C(b,i)$.
\end{proof}

\begin{lemma}[Alternating singleton activation]
\label{lem:altA}
Let $K=R$. For every $p\ge0$ there is a fair execution whose first $p$
rounds are fully synchronous and whose tail activates
$\{a\},\{b\},\{a\},\ldots$ or $\{b\},\{a\},\ldots$. In this tail,
\begin{enumerate}
\item[(i)] the distance $d$ is updated by the factor $c_1$ at every
  round;
\item[(ii)] the robot that is not activated remains stationary during
  that round;
\item[(iii)] if the displacement of the moving robot is $v$, then
  $a+b$ is translated by $v$;
\item[(iv)] the two colors may evolve independently.
\end{enumerate}
Once the alternating tail has begun, no fully synchronous round is
admissible.
\end{lemma}

\begin{proof}
This is the definition of $\RSY$ in
Section~\ref{subsec:schedulers-events}, and (i)--(iv) follow from the
fact that exactly one displacement is applied in each round of the
tail; for (i) this is Lemma~\ref{lem:ratios}(i).
\end{proof}

\subsection{Invisibility and State Orbits}
\label{app:invisibility}

Several constructions insert a segment that the other robot cannot
detect, so that its next activation reproduces the behavior it would
have had without the insertion. The following notion isolates exactly
what has to be preserved.

\begin{definition}[$r$-invisible segment]
\label{def:silent}
A finite execution segment $\sigma$ is \emph{$r$-invisible} if the
normalized local input that $r$ receives at its next activation is the
same before and after $\sigma$.
\end{definition}

\begin{lemma}[Invisible transitions]
\label{lem:silent}
Let $M\in\{\FS,\OB\}$. If $r$ is not activated during $\sigma$ and the
peer $\overline r$ occupies a position distinct from that of $r$ both
before and after $\sigma$, then $\sigma$ is $r$-invisible. Consequently,
however the internal state of $\overline r$ changes during $\sigma$,
the robot $r$ computes the same output at its next activation.
\end{lemma}

\begin{proof}
Neither $\operatorname{Obs}_{\FS}(r)=(p_r,c_r)$ nor
$\operatorname{Obs}_{\OB}(r)=p_r$ contains the color of the peer, so
only two components have to be preserved. The own color $c_r$ is
unchanged because $r$ is not activated. The relative position of the
peer may change by a similarity during $\sigma$, but as long as it is
nonzero it is mapped to $e$ in the normalized frame by
Lemma~\ref{lem:norm}(1). The last assertion is then the determinism of
the algorithm.
\end{proof}

\begin{lemma}[Period-wise pumping in $\FS$]
\label{lem:pumpB2}
Let $M=\FS$ and $K\in\{S,A_{LC},A_M,A\}$. Consider a solo activation
interval of $a$ (Lemma~\ref{lem:solo}) in which every $\LK$ of $a$
observes its peer at a position different from its own. Along such an
interval, the color of $a$ evolves as
\[
  c_0,\;g(c_0),\;g^{2}(c_0),\;\ldots,
\]
which is eventually periodic, say with preperiod $p$ and period $q$.
For every $n\ge1$ such that the first $p+nq$ activations of $a$
satisfy this condition, there is a fair execution in which $a$ is
activated alone $p+nq$ times; at the end of that interval the color
of $a$ is $\gamma_p:=g^{p}(c_0)$, and $b$, never being activated,
stays at its position throughout the interval. If moreover $a$ never
reaches the position of $b$ during the interval, then the inserted
interval is $b$-invisible.
\end{lemma}

\begin{proof}
The visible state of an $\FS$ robot is its own color. Every $\LK$ of
$a$ in the interval observes a nonzero relative position of its peer,
so by Lemma~\ref{lem:norm}(1) each activation updates the color by $g$.
The color sequence is therefore the orbit of $g$ and is eventually
periodic. After $p+nq$ steps the orbit is back at $\gamma_p$. The robot
$b$ is not activated, so its own color is unchanged, and the color of
$a$ is not part of the input of $b$. The relative position of $a$
changes by a similarity, but as long as it is nonzero the adversary can
present the same normalized input to the next activation of $b$; hence
the interval is $b$-invisible by Lemma~\ref{lem:silent}.
\end{proof}

\subsection{Synchronous Prefixes and Repetition}
\label{app:synchronous}

This subsection collects the segments in which the same rule is applied
repeatedly, either by both robots simultaneously or by one robot alone.

\begin{lemma}[Maintaining coincidence]
\label{lem:keep-coincident}
Suppose that $a(t_0)=b(t_0)$, that neither robot has a pending operation
at $t_0$, and that the scheduler class permits activating both robots
in every round from $t_0$ on. The last condition always holds under
$\FSY$; under $\SSY$ and the subclasses of $\ASY$ it holds for the full
activation set with aligned phase times; under $\RSY$ it holds only
while the execution is still in its fully synchronous prefix. Suppose
moreover that in each such round the two local outputs are either both
zero or both nonzero; this is in particular the case when the two
visible states are equal. Then the adversary can construct a fair
execution with $a(t)=b(t)$ for every $t\ge t_0$.
\end{lemma}

\begin{proof}
At a coincidence the geometric part of both local inputs is the origin
(Section~\ref{subsec:robots-capabilities}). By Lemma~\ref{lem:norm}(2)
the direction and the magnitude of the global displacement obtained
from a nonzero local output are chosen by the adversary, so two nonzero
outputs can be mapped to the same global displacement, while two zero
outputs produce no movement. Activating both robots in every round
therefore preserves the coincidence. If one output is zero and the
other is not, the coincidence is destroyed, so the condition on the
outputs cannot be omitted. In $\OB$ the output at a coincidence is
unique, hence that condition always holds.
\end{proof}

\begin{lemma}[Synchronous rigidity in $\OB^{F}$]
\label{lem:rigidA1}
Let $M=\OB$ and $K=F$, so that the displacement rule is a single pair
$(\lambda,\theta)$ by Lemma~\ref{lem:single-rule}(a), and put
$c_2:=c_2(\lambda,\theta)$. The adversary can construct a fair
execution in which, in every round $k$ with $d(k)>0$, the midpoint is
fixed at $m$ and $d(k)=c_2^{\,k}d(0)$, and in which exactly one of the
following occurs.
\begin{enumerate}
\item[(a)] $\lambda=0$: neither robot moves, and
  $\moves{a}=\moves{b}=0$.
\item[(b)] $c_2=0$: in round $1$ the two robots coincide at $m$, and
  the coincidence is maintained by Lemma~\ref{lem:keep-coincident}.
\item[(c)] $0<c_2<1$: the sequence $d(k)$ decreases strictly to $0$ but
  remains positive, and $\moves{a}=\moves{b}=\infty$.
\item[(d)] $c_2\ge1$: $d(k)\ge d(0)$ for every $k$, and
  $\moves{a}=\moves{b}=\infty$.
\end{enumerate}

Moreover the adversary lets the two robots start and complete their
Moves simultaneously, with the same fractional progress along their
segments. The two robots are then simultaneously stationary only
before and after the common Move interval of a round, so every joint
stop of the execution carries one of the round-boundary distances
$d(k)$. In particular the restriction of $d$ to $\Sab$ is monotone.

\end{lemma}

\begin{proof}
Give the normalized frame in every round, so that
Lemma~\ref{lem:ratios}(ii) applies to each round. In case~(b) the
execution is fully synchronous and the model is $\OB$, so the output
condition of Lemma~\ref{lem:keep-coincident} is satisfied. In
cases~(c) and~(d) a nonzero $\M$ occurs in every round, and since
$t_E(r,i)<t_L(r,i+1)$ these are counted as distinct moves.

Activating both robots in the same round with a common speed profile
is admissible under $\FSY$. In a round containing a nonzero $\M$ the
two Move intervals then coincide, so the instants at which both robots
are stationary are those preceding the common Move, where the distance
is $d(k)$, and those following it, where it is $d(k+1)$. In a round
containing no nonzero $\M$ the distance is constant and equal to both.
Hence $d$ takes only the values $d(k)=c_2^{\,k}d(0)$ on $\Sab$, and
that sequence is monotone in each of the four cases: constant in~(a),
constant at $0$ from round~$1$ on in~(b), strictly decreasing in~(c),
and non-decreasing in~(d).
\end{proof}

\begin{lemma}[Synchronous rigidity in the colored models]
\label{lem:rigidA2}
Let $M\in\{\LU,\allowbreak\FC,\allowbreak\FS\}$ and let $K$ be any
scheduler class that permits activating both robots in every round;
under $\RSY$ this means during the fully synchronous prefix. The
adversary can construct a fair execution in which, in every round $k$,
the two robots carry the same color $\gamma_k$, where $\gamma_0=c_0$,
and the distance satisfies $d(k+1)=c_2(\gamma_k)\,d(k)$; moreover, in
every round $k$ with $d(k)>0$ the midpoint is fixed at $m$.
While the robots occupy distinct positions, their common color
satisfies $\gamma_{k+1}=g(\gamma_k)$.
After a coincidence, the common color instead follows the transition
map for the zero geometric observation.
In either case the two colors remain equal.
Since coincidence is thereafter maintained, the common color sequence
$(\gamma_k)$ is eventually periodic.
\end{lemma}

\begin{proof}
Activate both robots in every round and give them the normalized
frames at every $\LK$ at which they are distinct. The designated
initial colors are equal, and equal colors give the two robots equal
visible states, so by Lemma~\ref{lem:norm} their next colors are equal
as well; while the robots are distinct, this common next color is
$g(\gamma_k)$. While the robots are distinct, the distance update and
the invariance of the midpoint are Lemma~\ref{lem:ratios}(ii). If
$c_2(\gamma_k)=0$, the robots coincide at $m$ at the end of round $k$;
from then on their visible states are still equal, so the coincidence
is preserved by Lemma~\ref{lem:keep-coincident} and $d(k+1)=0
=c_2(\gamma_k)\,d(k)$ holds trivially, whereas the common position may
move, since by Lemma~\ref{lem:norm}(2) the adversary chooses the
displacement at a coincidence. At every such round both robots observe
the zero geometric configuration with equal visible states, so they
apply the same transition map for that observation and their colors
remain equal. Before a coincidence the common color sequence is an
orbit of $g$; after it, the sequence is an orbit of the transition map
for the zero observation. Either is the orbit of a map on a finite set,
so $(\gamma_k)$ is eventually periodic.
\end{proof}

\begin{lemma}[Solo pumping in the models without private memory]
\label{lem:pumpB1}
Let $M\in\{\FC,\allowbreak\OB\}$ and $K\in\{S,A_{LC},A_M,A\}$. Let $\gamma_\star$
be the visible state of $a$ at the beginning of a solo activation
interval of Lemma~\ref{lem:solo}, and put
$\lambda:=\lambda_{\gamma_\star}$, $\theta:=\theta_{\gamma_\star}$ and
$c:=c_1(\gamma_\star)$. Assume that the robots occupy distinct
positions at the beginning of the interval. Throughout the interval
$\Stop_b(t)$ holds, so
the peer remains stationary. At every cycle whose $\LK$ observes
distinct positions, the vector from $a$ to $b$ is transformed by
$I-\lambda R_\theta$. Hence $d(j)=c^{\,j}d(0)$ up to and including the
first landing at the peer, if such a landing occurs. Moreover:
\begin{enumerate}
\item[(a)] if $\lambda=0$, then $a$ does not move during the interval;
  in particular, if $M=\OB$ then $b$ obeys the same rule, so the
  configuration stays at $\iota$ along every fair infinite execution;
\item[(b)] if $c=0$, the first activation makes $a$ reach its peer.
  The normalized rule does not determine subsequent outputs at
  coincidence; those must be treated separately;
\item[(c)] if $\lambda>0$ and $c>0$, every activation in the solo
  interval is a nonzero $\M$, and any prescribed finite number of such
  $\M$s can be scheduled.
\end{enumerate}
\end{lemma}

\begin{proof}
While $b$ is not activated, the visible state of $a$ is constant: the
model $\OB$ has no visible state, and in $\FC$ it is the color of $b$,
which does not change, since $a$ cannot observe its own color. Hence,
at every $\LK$ that observes distinct positions, the same pair
$(\lambda,\theta)$ is applied, and the transformation of the vector
from $a$ to $b$ and the distance formula follow from
Lemma~\ref{lem:norm}(1) and Lemma~\ref{lem:ratios}(i). In case~(c) the
distance stays positive, so every $\LK$ observes distinct positions and
every displacement has length $\lambda d(j)>0$. Concerning~(a), when
$M=\FC$ the color of $a$ may change during the interval, so the
subsequent behavior of $b$ is not determined by $\lambda$ alone; this
is why immobility of the whole configuration is asserted only for
$M=\OB$.
\end{proof}

\begin{lemma}[A single displacement rule]
\label{lem:single-rule}
In each of the following two cases, the displacement rule of the
activated robot is a single pair $(\lambda,\theta)$:
\begin{enumerate}
\item[(a)] any algorithm with $M=\OB$;
\item[(b)] an algorithm with $M=\FC$, together with an interval during
  which the peer is not activated.
\end{enumerate}
When these activations are performed one at a time, with
each cycle completed before the next begins, and every
$\LK$ observes distinct positions, successive landing
distances are multiplied by the same factor
$c_1(\lambda,\theta)$, so that $d(k)=c_1^{\,k}d(0)$ after $k$
activations. Their sequence is therefore monotone and cannot alternate
between strict decreases and strict increases. A landing at the peer
ends this normalized argument and must be handled separately. If
$\lambda=0$, the robot does not move at all.
\end{lemma}
 
\begin{proof}
By Lemma~\ref{lem:norm} the displacement is determined by the pair
$(\lambda_\gamma,\theta_\gamma)$ attached to the visible state $\gamma$.
In case~(a) the model $\OB$ has a unique visible state, so the family of
constants degenerates to a single pair. In case~(b) the visible state of
an $\FC$ robot is the color of its peer alone, since it cannot observe
its own color; as long as the peer is not activated that color is
unchanged, hence the visible state is constant
(Lemma~\ref{lem:pumpB1}). For such singleton cycles, the landing
distance is multiplied by $c_1(\lambda,\theta)$ by
Lemma~\ref{lem:ratios}(i), so the landing distances form a geometric
sequence and are monotone. If $\lambda=0$ the displacement is
identically zero.
\end{proof}

\begin{corollary}[The first nonzero rule along a synchronous prefix]
\label{cor:first-nonzero}
Let $M\in\{\LU,\allowbreak\FC,\allowbreak\FS\}$ with the designated common initial color
$c_0$, and let the execution begin with rounds activating both robots,
as in Lemma~\ref{lem:rigidA2}, so that the common color evolves by
$\gamma_0=c_0$, $\gamma_{k+1}=g(\gamma_k)$ and is eventually periodic.
Exactly one of the following holds.
\begin{enumerate}
\item[(a)] $\lambda_{\gamma_k}=0$ for every $k$. Then this fair
  execution contains no nonzero $\M$ at all, the configuration remains
  $\iota$ forever, and both robots are permanently stationary.
\item[(b)] There is a least $t$ with $\lambda_{\gamma_t}\neq0$. Then the
  configuration at the beginning of round $t$ is still $\iota$, both
  robots are stationary with no pending operation, and their common
  visible state is $\gamma_t$. Hence any segment of
  Appendix~\ref{app:admissible-segments} may be appended there. In
  particular, replacing round $t$ by a solo activation interval of $a$
  (Lemma~\ref{lem:solo}) leaves the local input of $a$ unchanged, so
  the rule applied at its first activation is
  $(\lambda_{\gamma_t},\theta_{\gamma_t})$ with
  $\lambda_{\gamma_t}\neq0$; the visible state then stays equal to
  $\gamma_t$ throughout the interval when $M=\FC$
  (Lemma~\ref{lem:pumpB1}), whereas for $M=\FS$ it follows the orbit of
  $g$ (Lemma~\ref{lem:pumpB2}).
\end{enumerate}
In case~(a) the joint-stop sequence of Definition~\ref{def:epg} has the
single term $C_0$, so the perpetuity condition of $\mathrm{EPG}^{*}$
fails. In case~(b) the configuration is unchanged until round $t$, so
the first activation of $a$ in the solo activation interval is the
first nonzero $\M$ of the execution and it is performed while $b$ is not
activated; Lemma~\ref{lem:solo-drift}(1) then makes the expansion
condition fail at $i=0$.
Since the color sequence is eventually periodic, the dichotomy is
decided by a finite check. For $M=\OB$ the same statement holds with
the single rule $(\lambda,\theta)$ of
Lemma~\ref{lem:single-rule}(a), case~(a) being $\lambda=0$.
\end{corollary}

\begin{proof}
By Lemma~\ref{lem:rigidA2} the two robots carry the same color
$\gamma_k$ in round $k$, so by Lemma~\ref{lem:norm}(1) both apply
$(\lambda_{\gamma_k},\theta_{\gamma_k})$, and the displacement is
nonzero exactly when $\lambda_{\gamma_k}\neq0$. In case~(a) no
displacement is ever nonzero, so the configuration never changes. In
case~(b) every round before $t$ applies a zero displacement, hence the
configuration at the beginning of round $t$ is still $\iota$; the robot
$a$ then has the same peer position and the same visible state as it
would under simultaneous activation, so its rule there is
$(\lambda_{\gamma_t},\theta_{\gamma_t})$. The finite check is the
eventual periodicity of the orbit of $g$ on a finite color set.
\end{proof}

\subsection{Preservation and Continuation of Violations}
\label{app:concatenation}

Every common-failure execution in
Appendices~\ref{app:cf-two-task-disjunctions}
and~\ref{app:cf-three-task-disjunctions} is a concatenation of the
segments collected above. The following lemma states when such a
concatenation is admissible and which violations survive it.

\begin{lemma}[Concatenation of segments]
\label{lem:concat}
The violations exhibited by a common-failure construction split into
two kinds. A \emph{safety violation}, namely landing at a forbidden
point, a forbidden overtaking, a forbidden coincidence, a forbidden
transition, or an expansion that moves the center of gravity, occurs at
a finite time and cannot be repaired by any continuation. A \emph{liveness
violation}, namely the failure of infinitely many alternations, of
perpetuity, of reaching, of convergence, or of stable gathering or
stopping, is decided only on an infinite execution. The following hold.
\begin{enumerate}
\item A finite segment that forces a safety violation may be followed
  by any continuation permitted by the scheduler class, in particular
  by a fair tail, and the violation is preserved.
\item Segments are concatenated starting from the state actually
  reached by the preceding prefix. Positions, colors, internal states,
  and pending operations are carried over and are not reset. A segment
  may be appended only if its applicability conditions hold in that
  state and if the concatenation is an execution of the scheduler class
  under consideration. In particular, no fully synchronous round may be
  placed after an alternating singleton tail under $\RSY$
  (Lemma~\ref{lem:altA}); a solo activation interval
  (Lemma~\ref{lem:solo}) requires the robot that is not activated to
  have no pending operation; and a delayed move
  (Corollary~\ref{cor:stale}) requires $K\in\{A_{LC},A_M,A\}$.
\item A liveness violation requires an explicit fair infinite
  continuation along which the violation persists; an arbitrary fair
  tail does not suffice. When a safety violation of one predicate is
  combined with a liveness violation of another, the segment forcing
  the safety violation is placed first, and the infinite continuation
  witnessing the liveness violation is built from the state that this
  segment leaves.
\end{enumerate}
Since the predicate of each task is evaluated relative to the fixed
initial instance $\iota$, the argument for a later segment remains
valid even when earlier segments have already changed the
configuration.
\end{lemma}

\begin{proof}
A safety violation is witnessed by a finite prefix, and every extension
of that prefix contains it; this gives~(1).

Claim~(2) records that a concatenation is an execution of the algorithm
in the given scheduler class exactly when each stage is admissible from
the state already reached. The conditions listed are precisely those
required by the corresponding segment lemmas.

For~(3), a liveness property is not decided by any finite prefix, so
an infinite continuation violating it has to be exhibited. Finally,
the task predicates refer to $\iota$ and to the trajectory, not to the
configuration at which a segment happens to begin; hence a change of
configuration by an earlier segment does not invalidate the argument
for a later one.
\end{proof}

\subsection{The Solo-Expansion Obstruction}
\label{app:solo-drift}

\begin{proof}[Proof of Lemma~\ref{lem:solo-drift}]
(1) Say $\bar r$ completes no nonzero $\M$, so $\bar r_{i+1}=\bar r_i$.
Consecutive terms of the sequence are distinct, hence $r_{i+1}\neq r_i$ and
\(
  a_{i+1}+b_{i+1}=(a_i+b_i)+(r_{i+1}-r_i)\neq a_i+b_i .
\)

(2) For $K=R$ take the execution of Lemma~\ref{lem:altA} with an empty fully
synchronous prefix; every round activates exactly one robot. For
$K\in\{S,A_{LC},A_M,A\}$ the adversary places the whole cycle of one robot after
the whole cycle of the other, alternately: while one robot runs its cycle the
other performs no operation, so no $\Look$ falls inside a $\M$ interval of the
peer and every atomicity constraint of
Section~\ref{subsec:schedulers-events} is met; both robots are activated
infinitely often, so the execution is fair.

In such an execution the boundaries between consecutive cycles are times of
$\Sab$. If no nonzero $\M$ ever occurs, the configuration never changes and
$(C_i)$ has the single term $C_0$; then the first condition of
$\mathrm{EPG}$ (exactly two terms) and the first condition of
$\mathrm{EPG}^{*}$ (an infinite sequence) both fail. Otherwise consider the
first cycle containing a nonzero $\M$, performed by $r$. Every earlier cycle
left the configuration unchanged, so the joint stop preceding that cycle
realizes $C_0$ and the joint stop following it realizes $C_1\neq C_0$, while
$\bar r$ ran no cycle in between. By (1),
$\tfrac12(a_1+b_1)\neq c_0$, so the second condition of both tasks fails at
$i=0$. This is a safety violation and is preserved by any fair continuation.
\end{proof}

\subsection{The Two-Nonzero-Move Obstruction for SRO}
\label{app:sro-two-moves}

We use the notation of Definition~\ref{def:sro}: $(C_i)_{i\ge0}$ is the
sequence of configurations at the times when neither robot is inside a
nonzero $\M$, with consecutive repetitions removed,
$d_i=b_i-a_i$, and $Q_i$ is the square with diagonal $a_ib_i$, interior
included, with $Q_{-1}=\mathbb R^2$. Condition $\mathbf{C1}$ is
$d_{i+1}=R_{-90^\circ}d_i$, condition $\mathbf{C2}$ is
$d_{i+1}=2^{-1/2}R_{-45^\circ}d_i$, and condition $\mathbf{C3}$ is
$a_{i+1},b_{i+1}\in Q_{i-1}$.

\begin{lemma}[Two solo nonzero moves defeat SRO]
\label{lem:cf-sro-two-moves}
Consider a situation in which only one robot moves while its peer
remains stationary. The displacement that can satisfy $\mathbf{C1}$ is
unique, namely the one rotating the segment by $90^\circ$ clockwise
about the peer, and the displacement that can satisfy $\mathbf{C2}$ is
also unique, namely the one rotating the moving robot by $90^\circ$
clockwise about the midpoint. No other displacement satisfies
$\mathbf{C1}$ or $\mathbf{C2}$. If the same robot performs two
consecutive nonzero $\M$ operations while its peer stays still, the
rotation condition and $\mathbf{C3}$ cannot both be satisfied.
Consequently, two solo nonzero moves of one robot violate
$\Phi_{\mathrm{SRO}}$, irrespective of the internal states. Zero $\M$
operations do not advance the sequence $(C_i)$.
\end{lemma}

\begin{proof}
After each nonzero $\M$ of the moving robot, both robots are stationary
until the next $\LK$, so the configurations reached after the two moves
are consecutive terms $C_{i+1},C_{i+2}$ of the sequence $(C_i)$.

Condition $\mathbf{C1}$ requires $\|d_{i+1}\|=\|d_i\|$ together with
$d_{i+1}=R_{-90^\circ}d_i$, so with the peer held fixed the new
position of the moving robot is uniquely determined. Condition
$\mathbf{C2}$ requires $d_{i+1}=2^{-1/2}R_{-45^\circ}d_i$ and
determines it uniquely in the same way. Both conditions are purely
geometric and do not depend on the internal states.

Exchanging the names of the robots replaces each $d_j$ by $-d_j$,
which leaves the conditions unchanged, and they are also invariant
under orientation-preserving similarities. We may therefore assume
that $a$ is the moving robot and that
\[
  b_i=(0,0),\qquad a_i=(-1,0),
\]
so that $d_i=(1,0)$ and
\[
  Q_i=\Bigl\{(x,y):\ \bigl|x+\tfrac12\bigr|+|y|\le\tfrac12\Bigr\}.
\]

If the first move satisfies $\mathbf{C1}$, then $a_{i+1}=(0,1)$, and
the second move gives $a_{i+2}=(1,0)$ under $\mathbf{C1}$ and
$a_{i+2}=(\tfrac12,\tfrac12)$ under $\mathbf{C2}$. If the first move
satisfies $\mathbf{C2}$, then $a_{i+1}=(-\tfrac12,\tfrac12)$, and the
second move gives $a_{i+2}=(\tfrac12,\tfrac12)$ under $\mathbf{C1}$ and
$a_{i+2}=(0,\tfrac12)$ under $\mathbf{C2}$. In each of the four cases
$a_{i+2}\notin Q_i$, so $\mathbf{C3}$ fails at $C_{i+2}$. This is the
computation of~\cite[Lemma~23, Fig.~5]{FSW19}.
\end{proof}

The lemma is applied in two ways. When the two solo nonzero moves are
available, it refutes $\Phi_{\mathrm{SRO}}$ outright. When only one
solo move is available, its first part still shows that a displacement
which is neither the $90^\circ$ clockwise rotation about the peer nor
the $90^\circ$ clockwise rotation about the midpoint already violates
the rotation condition. The degenerate case in which the moving robot
lands on its peer is covered by this observation, since $d_{i+1}=0$
satisfies neither $\mathbf{C1}$, which requires $\|d_{i+1}\|=\|d_i\|\neq0$,
nor $\mathbf{C2}$, whose right-hand side is nonzero.

\subsection{The Overtaking Construction}
\label{app:overtaking}

The delayed move of Corollary~\ref{cor:stale} becomes an obstruction
when the destination computed before the delay is long compared with
the distance that remains when the move is released. The following
lemma records the two forms in which this is used, one for
$\mathrm{MCv}$ and one for $\mathrm{ZCC}$.

\begin{lemma}[Overtaking]
\label{lem:overtake}
Let $K\in\{A_{LC},A_M,A\}$. Suppose that $b$ computes at its $\CP$ a
destination from its saved $\LK$ snapshot, with global displacement
$v\neq0$, that its $\M$ is then
delayed as in Corollary~\ref{cor:stale}, and that only $a$ is activated
during the delay. Let $w$ be the vector from $b$ to $a$ at the moment
at which the $\M$ of $b$ is released. Keep $a$ stationary while $b$
executes this pending $\M$.
\begin{enumerate}
\item If $\|w\|<\|v\|/2$, then the distance after the $\M$ of $b$ is
  larger than $\|w\|$. Hence the second condition of
  $\Phi_{\mathrm{MCv}}$, monotonicity, is violated.
\item If $b(0)$, the current position of $b$, and $a$ lie on a line in
  this order, if $v$ points towards $a$, and if $\|w\|<\|v\|$, then $b$
  lands beyond $a$, and at its landing time $t$ one has
  $\dis{b(0)}{b(t)}>\dis{b(0)}{a(t)}$. Hence the second condition of
  $\Phi_{\mathrm{ZCC}}$, collision freedom, is violated.
\end{enumerate}
The delay is placed between $t_C(b,i)$ and $t_B(b,i)$. By
Lemma~\ref{lem:suspend} and Corollary~\ref{cor:stale}, it is admissible
under $A_{LC}$, $A_M$ and $A$, and any prescribed finite number of
complete cycles of $a$ may be placed inside it. Throughout the delay
$b$ stays at its position, and its pending $\M$ uses the destination
computed at $t_C(b,i)$, independently of these cycles. When
$M\in\{\FC,\OB\}$ the visible state of $a$ is constant during the delay
(for $\FC$ it is the color that $b$ adopted at $t_C(b,i)$); when
$M=\FS$ the color of $a$ follows the orbit of $g$ as long as every
$\LK$ of $a$ observes distinct positions. Hence the distance formulas
of Lemmas~\ref{lem:pumpB1} and~\ref{lem:pumpB2} apply to these cycles,
although $b$ has a pending operation.
\end{lemma}
 
\begin{proof}
(1) After the $\M$ of $b$, the vector from $b$ to $a$ is $w-v$, and
\[
  \|w-v\|\ \ge\ \|v\|-\|w\|\ >\ \|w\| .
\]
 
(2) On the line we have
\[
  \dis{b(0)}{b+v}
  =\dis{b(0)}{b}+\|v\|
  >\dis{b(0)}{b}+\|w\|
  =\dis{b(0)}{a},
\]
which is the stated inequality at the landing time.
 
The admissibility of the delay and of the cycles of $a$ placed inside
it is Lemma~\ref{lem:suspend} together with Corollary~\ref{cor:stale};
the last assertion of Lemma~\ref{lem:suspend} shows that the pending
$\M$ of $b$ is unaffected by them. The distance formulas follow from
Lemma~\ref{lem:norm}(1) and Lemma~\ref{lem:ratios}(i), exactly as in
the proofs of Lemmas~\ref{lem:pumpB1} and~\ref{lem:pumpB2}, since $b$
does not move during the delay.
\end{proof}

\begin{remark}[The reachability hypothesis is not automatic]
\label{rem:overtake-reachability}
Neither $\|w\|<\|v\|/2$ nor $\|w\|<\|v\|$ holds by construction: they
require the solo activations of $a$ during the delay to bring $a$
close enough to $b$. When the solo contraction ratio of $a$ satisfies
$c<1$, we have $d(j)=c^{\,j}d(0)\to0$ by Lemma~\ref{lem:pumpB1}, and
the hypothesis can be met. For $M=\FS$ the relevant quantity is the
product $\Gamma=\prod_j c_1(\gamma_j)$ over one color period, by
Lemma~\ref{lem:pumpB2}, and the hypothesis can be met when
$\Gamma<1$. The reachability condition and, for part~(2), the required
collinearity and order must be verified in each application.
The continuation arguments for $\mathrm{MCv}$ and
$\mathrm{ZCC}$ are given directly in
Lemmas~\ref{lem:mcv_not_fs_am} and~\ref{lem:zcc_not_fs_am}.
\end{remark}

\subsection{Avoiding a Countable Set}
\label{app:continuum}

Lemma~\ref{lem:mid} lets the adversary choose, by a continuous
parameter, the point of the peer's movement at which an observation is
made. Several impossibility arguments combine this freedom with the
fact that the set of outcomes admitted by a task predicate is countable,
or even finite.

\begin{corollary}[Avoiding a countable set]
\label{cor:continuum}
Let $K\in\{A_{LC},A\}$, suppose that $b$ performs a nonzero $\M$, and
suppose that $a$ observes it with parameter $\mu\in(0,1)$ as in
Lemma~\ref{lem:mid}. Let $h:(0,1)\to\mathbb R$ assign to $\mu$ some
quantity of the resulting execution. If $h$ is continuous and
nonconstant, then for every countable set $G\subset\mathbb R$ the
adversary can choose $\mu$ with $h(\mu)\notin G$.
\end{corollary}

\begin{proof}
The image of a nonconstant continuous function on an interval contains
a nondegenerate interval, hence is uncountable, and therefore is not
contained in $G$.
\end{proof}

The corollary is used with the following two observation parameters.

\paragraph*{Impossibility of $\mathrm{VEC}$.}
In the proof that $\mathrm{VEC}\notin\Sol(\FS^{A})$
(Lemma~\ref{lem:vec_not_fs_a}), $h$ is a coordinate of the segment
determined by the two landing points, which varies continuously and
nonconstantly with $\mu$. The set $G$ consists of the finitely many
successor configurations admitted by the transition condition of
$\Phi_{\mathrm{VEC}}$, so a single $\mu$ avoids all of them. In
Lemma~\ref{lem:A2-6-FS-A} the same $h$ is used to avoid one further
value, namely the unique successor admitted by the expansion condition
of $\mathrm{EPG}^{*}$, and the countably many $\mu$ at which the two
robots coincide.

\paragraph*{Impossibility of $\mathrm{DMSD}$.}
In the proof that $\mathrm{DMSD}\notin\Sol(\FS^{A})$
(Lemma~\ref{lem:dmsd_not_fs_a}), the normalization $a(0)=0$,
$b(0)=e_1$, $D_0=1$ gives
\[
  h(\mu)=\|c+\mu w\|,
  \qquad
  c:=e_1-2\lambda Re_1,
  \qquad
  w:=\lambda^{2}R^{2}e_1\neq0 .
\]
Here $\|c+\mu w\|^{2}$ is a quadratic polynomial in $\mu$ whose leading
coefficient $\|w\|^{2}$ is positive, so $h$ is nonconstant on $(0,1)$.
The set $G$ is the countable set $D_0\,\Dyad$, and a value
$h(\mu)\notin G$ witnesses a joint stop at which the ratio is not
dyadic.

In both applications the exceptional set is countable, so several such
requirements can be imposed simultaneously: a finite or countable
union of countable sets is still countable, and a single $\mu$ avoids
all of them at once. This is what allows one execution to refute
several task predicates, as in
Appendices~\ref{app:cf-two-task-disjunctions}
and~\ref{app:cf-three-task-disjunctions}.

%% file: CFE-Two-Task.tex
\label{app:cf-other-pairs}


\noindent
Each row lists one two-element antichain $A$, given by the
corresponding disjunction $B_A$ of principal tasks from
Table~\ref{tab:PP13}, together with
$\mathcal N_A=\Max(\mathcal Q\setminus U_A)$ and, for each
$Y\in\mathcal N_A$, the statement exhibiting one admissible fair
execution of an arbitrary algorithm in $Y$ that violates both task
predicates. The four obligations carrying a reference to
Proposition~\ref{prop:cf-expansion-exchange} or
Proposition~\ref{prop:cf-sro-exchange-family} are proved in the main
text; each of them refutes three predicates on a single execution, which is what lets one statement discharge two obligations at once, namely those of $\mathrm{EPG}^{*}\vee \SRDV$ and $\mathrm{EPG}^{*} \vee \VSS$ for the first and those of $\mathrm{SRO}\vee \SRDV$ and $\mathrm{SRO}\vee \VSS$ for the second. The column sum is the number $47$ of two-task boundary obligations certified in Table~\ref{tab:antichain-enumeration-certificate} (Appendix~\ref{app:antichain-enumeration}) and recorded in Table~\ref{tab:cf-family-coverage}.

\clearpage
\begingroup
  \small
  \setlength{\tabcolsep}{4pt}
  \renewcommand{\arraystretch}{1.15}
  \captionsetup{hypcap=false}
  \captionof{table}{The twenty-four two-element antichains, their maximal
    unsolvable classes $\mathcal N_A$, and the common-failure execution
    that discharges each obligation $(A,Y)$. The rows are grouped by the
    construction families of
    Section~\ref{sec:representative-common-failure} and
    Table~\ref{tab:cf-family-coverage}. Four obligations are proved in the
    main text and are listed here only as references; all the others are
    proved in Appendix~\ref{app:cf-two-task-disjunctions}.}
  \label{tab:PP13-join2}
  \par\medskip
  \noindent
  \begin{tabularx}{\linewidth}{@{}l >{\raggedright\arraybackslash}X c@{}}
    \toprule
    \textbf{Disjunction $B_A$}
      & \textbf{$\mathcal N_A$ and proof reference}
      & \textbf{$|\mathcal N_A|$} \\
    \midrule
    \multicolumn{3}{@{}l}{\emph{Expansion family, $\mathrm{EPG}^{*}\vee P$ with
      $P\in\mathcal C_{\mathrm{exp}}$, Appendix~\ref{app:cf-expansion-family}}}\\
    $\mathrm{EPG}^{*}\vee\mathrm{HET}$
      & $\FC^{S}$ Lem.~\ref{lem:A2-1-FC-S};\quad
        $\FS^{R}$ Lem.~\ref{lem:A2-1-FS-R}
      & 2 \\
    $\mathrm{EPG}^{*}\vee\mathrm{OSPO}$
      & $\FC^{S}$ Lem.~\ref{lem:A2-2-FC-S};\quad
        $\FS^{S}$ Lem.~\ref{lem:A2-2-FS-S};\quad
        $\OB^{R}$ Lem.~\ref{lem:A2-2-OB-R}
      & 3 \\
    $\mathrm{EPG}^{*}\vee \SRDV$
      & $\FC^{A}$ Lem.~\ref{lem:A2-3-FC-A};\quad
        $\FS^{R}$ Prop.~\ref{prop:cf-expansion-exchange}
      & 2 \\
    $\mathrm{EPG}^{*}\vee\mathrm{ZCC}$
      & $\FC^{S}$ Lem.~\ref{lem:A2-4-FC-S};\quad
        $\FS^{A_M}$ Lem.~\ref{lem:A2-4-FS-AM};\quad
        $\OB^{R}$ Lem.~\ref{lem:A2-4-OB-R}
      & 3 \\
    $\mathrm{EPG}^{*}\vee \VSS$
      & $\FS^{R}$ Prop.~\ref{prop:cf-expansion-exchange}
      & 1 \\
    $\mathrm{EPG}^{*}\vee\mathrm{VEC}$
      & $\FC^{S}$ Lem.~\ref{lem:A2-6-FC-S};\quad
        $\FS^{A}$ Lem.~\ref{lem:A2-6-FS-A};\quad
        $\OB^{R}$ Lem.~\ref{lem:A2-6-OB-R}
      & 3 \\
    $\mathrm{EPG}^{*}\vee \CFSM$
      & $\FC^{S}$ Lem.~\ref{lem:A2-7-FC-S};\quad
        $\OB^{R}$ Lem.~\ref{lem:A2-7-OB-R}
      & 2 \\
    \cmidrule(lr){1-3}
    \multicolumn{3}{@{}l}{\emph{Shrinking-rotation family, Appendix~\ref{app:cf-sro-family}}}\\
    $\mathrm{SRO}\vee \SRDV$
      & $\FC^{A}$ Lem.~\ref{lem:A2-10-FC-A};\quad
        $\FS^{S}$ Prop.~\ref{prop:cf-sro-exchange-family}
      & 2 \\
    $\mathrm{SRO}\vee\mathrm{ZCC}$
      & $\FC^{S}$ Lem.~\ref{lem:A2-11-FC-S};\quad
        $\FS^{A_M}$ Lem.~\ref{lem:A2-11-FS-AM}
      & 2 \\
    $\mathrm{SRO}\vee \VSS$
      & $\FS^{S}$ Prop.~\ref{prop:cf-sro-exchange-family}
      & 1 \\
    $\mathrm{SRO}\vee\mathrm{VEC}$
      & $\FC^{S}$ Lem.~\ref{lem:A2-13-FC-S};\quad
        $\FS^{A}$ Lem.~\ref{lem:A2-13-FS-A}
      & 2 \\
    $\mathrm{SRO}\vee \CFSM$
      & $\FC^{S}$ Lem.~\ref{lem:A2-14-FC-S}
      & 1 \\
    \cmidrule(lr){1-3}
    \multicolumn{3}{@{}l}{\emph{Exchange family, Appendix~\ref{app:cf-exchange-pairs}}}\\
    $\mathrm{OSPO}\vee \SRDV$
      & $\FC^{A}$ Lem.~\ref{lem:A2-8-FC-A};\quad
        $\FS^{S}$ Lem.~\ref{lem:A2-8-FS-S};\quad
        $\OB^{F}$ Lem.~\ref{lem:A2-8-OB-F}
      & 3 \\
    $\mathrm{OSPO}\vee \VSS$
      & $\FS^{S}$ Lem.~\ref{lem:A2-9-FS-S};\quad
        $\OB^{F}$ Lem.~\ref{lem:A2-9-OB-F}
      & 2 \\
    $\SRDV \vee\mathrm{ZCC}$
      & $\FC^{A}$ Lem.~\ref{lem:A2-15-FC-A};\quad
        $\FS^{A_M}$ Lem.~\ref{lem:A2-15-FS-AM};\quad
        $\OB^{F}$ Lem.~\ref{lem:A2-15-OB-F}
      & 3 \\
    $\SRDV \vee\mathrm{VEC}$
      & $\FC^{A}$ Lem.~\ref{lem:A2-16-FC-A};\quad
        $\FS^{A}$ Lem.~\ref{lem:A2-16-FS-A};\quad
        $\OB^{F}$ Lem.~\ref{lem:A2-16-OB-F}
      & 3 \\
    $\SRDV \vee \CFSM$
      & $\FC^{A}$ Lem.~\ref{lem:A2-17-FC-A};\quad
        $\OB^{F}$ Lem.~\ref{lem:A2-17-OB-F}
      & 2 \\
    $\mathrm{ZCC}\vee \VSS$
      & $\FS^{A_M}$ Lem.~\ref{lem:A2-18-FS-AM};\quad
        $\OB^{F}$ Lem.~\ref{lem:A2-18-OB-F}
      & 2 \\
    $\VSS \vee\mathrm{VEC}$
      & $\FS^{A}$ Lem.~\ref{lem:A2-22-FS-A};\quad
        $\OB^{F}$ Lem.~\ref{lem:A2-22-OB-F}
      & 2 \\
    $\VSS \vee \CFSM$
      & $\OB^{F}$ Lem.~\ref{lem:A2-23-OB-F}
      & 1 \\
    \cmidrule(lr){1-3}
    \multicolumn{3}{@{}l}{\emph{Monotone-convergence family, Appendix~\ref{app:cf-mcv-pairs}}}\\
    $\mathrm{MCv}\vee \VSS$
      & $\FS^{A_M}$ Lem.~\ref{lem:A2-19-FS-AM}
      & 1 \\
    $\mathrm{MCv}\vee\mathrm{VEC}$
      & $\FC^{A}$ Lem.~\ref{lem:A2-20-FC-A};\quad
        $\FS^{A}$ Lem.~\ref{lem:A2-20-FS-A}
      & 2 \\
    $\mathrm{MCv}\vee \CFSM$
      & $\FC^{A}$ Lem.~\ref{lem:A2-21-FC-A}
      & 1 \\
    \cmidrule(lr){1-3}
    \multicolumn{3}{@{}l}{\emph{Dyadic family, Appendix~\ref{app:cf-dmsd-cfsm}}}\\
    $\DMSD \vee \CFSM$
      & $\OB^{A}$ Lem.~\ref{lem:A2-24-OB-A}
      & 1 \\
    \midrule
    \textbf{Total} & & \textbf{47} \\
    \bottomrule
  \end{tabularx}
\endgroup
\clearpage

\subsection{Disjunctions Containing \texorpdfstring{EPG$^{*}$}{EPG*}}
\label{app:cf-expansion-family}

This subsection treats the seven antichains
$\{\mathrm{EPG}^{*},P\}$ with
$P\in\mathcal C_{\mathrm{exp}}=
\{\mathrm{HET},\allowbreak\mathrm{OSPO},\allowbreak \SRDV,\allowbreak
\mathrm{ZCC},\allowbreak \VSS,\allowbreak\mathrm{VEC},\allowbreak \CFSM\}$,
which carry sixteen boundary obligations $(A,Y)$ in total. The two
$\FS^R$ obligations, for
$\mathrm{EPG}^{*}\vee \SRDV$ and
$\mathrm{EPG}^{*}\vee \VSS$, are proved in the main text by
Proposition~\ref{prop:cf-expansion-exchange} and are not repeated
here.

Every construction below obtains its $\mathrm{EPG}^{*}$ violation from
Lemma~\ref{lem:solo-drift} in one of two ways. If the execution
contains no nonzero $\M$ at all, the joint-stop sequence has the single
term $C_0$ and the perpetuity condition fails, which is a liveness
violation. If it does contain one, and the first one is performed by a
robot while its peer performs none, then the configuration has not
changed before it, so the joint stops surrounding it realize $C_0$ and
$C_1$, and Lemma~\ref{lem:solo-drift}(1) gives
$\tfrac12(a_1+b_1)\neq c_0$: the expansion condition fails at $i=0$,
which is a safety violation. Each proof only has to record which of the
two cases applies and to refute the companion task on the same
execution.

Throughout this subsection and the next, we use the normalized frames of
Lemma~\ref{lem:norm} and the contraction ratios of
Lemma~\ref{lem:ratios}, and we write
$(\lambda^{\star},\theta^{\star})=(\tfrac12,\pi)$ for the rule of
Algorithm~\ref{alg:epgs} in a normalized frame, as recorded in
Remark~\ref{rem:epgs-normalized-rule}: a round in which both robots are
activated from the same configuration and apply the same rule
$(\lambda,\theta)$ satisfies the expansion condition $d_{i+1}=2d_i$ if
and only if $(\lambda,\theta)=(\lambda^{\star},\theta^{\star})$, and a
solo activation under that rule multiplies the distance by
$c_1^{\star}=\tfrac32$ with $\psi^{\star}\equiv0$.

\subsubsection{\texorpdfstring{EPG$^{*}$ $\vee$ HET (EH)}{EPG* or HET}}
\label{ssub:f1-1}

\begin{lemma}
\label{lem:A2-1-FC-S}
$\mathrm{EH}\notin\Sol(\FC^{S})$.
\end{lemma}

\begin{proof}
Let $\mathcal A$ be any deterministic algorithm of $\FC$ under $\SSY$.
Normalize the edge length of the hexagon of $\iota$ to $1$, with
$a(0)=v_0$ and $b(0)=v_3$; the initial colors are common. The visible
state of an $\FC$ robot is the color of its peer, so while both robots
are activated in every round their colors stay equal and follow a
common sequence $c_0,g(c_0),\dots$ (Lemma~\ref{lem:rigidA2}), which is
eventually periodic since the palette is finite. During this prefix the
configuration stays at $\iota$.

If $\lambda_{c_j}=0$ for every $j$, this fair execution contains no
nonzero $\M$, so $\moves{a}=0$ and the perpetuity condition of
$\mathrm{HET}$ fails, while the joint-stop sequence has the single term
$C_0$ and the perpetuity condition of $\mathrm{EPG}^{*}$ fails as well
(Lemma~\ref{lem:solo-drift}(2)). Both are liveness violations of the
same fair infinite execution, so Lemma~\ref{lem:concat}(3) is
satisfied.

Otherwise let $t$ be the least $j$ with $\lambda_{c_j}\neq0$, replace
round $t$ by an interval in which $a$ alone is activated $k\ge2$ times
in a row (Lemma~\ref{lem:solo}), and append a fair tail. Since $b$ is
not activated, the visible state of $a$ is constantly $c_t$, so by
Lemma~\ref{lem:single-rule}(b) its rule is the single pair
$(\lambda,\theta)=(\lambda_{c_t},\theta_{c_t})$, and by
Lemma~\ref{lem:pumpB1} the distance satisfies $d(j)=c_1^{\,j}d(0)$.

All rounds before the interval left the configuration at $\iota$, so
the first activation of $a$ in the interval is the first nonzero $\M$
of the execution and it is performed while $b$ is not activated. By
Lemma~\ref{lem:solo-drift}(1) the center of gravity of the next term of
the joint-stop sequence is not $c_0$, so the expansion condition of
$\mathrm{EPG}^{*}$ fails at $i=0$; this is a safety violation and is
preserved by the tail.

For $\mathrm{HET}$, as $b$ stands still at $v_3$, the second condition
requires $a$ to be at $v_0,v_{-1},v_0$ after $0,1,2$ nonzero $\M$
operations, that is $d=2,\sqrt3,2$; this would force $c_1=\sqrt3/2$ and
$c_1=2/\sqrt3$ simultaneously, which is impossible. (If $c_1=0$ then
$a$ reaches $v_3\neq v_{-1}$ in one activation.) Hence the second
condition fails, and being a safety violation it is preserved by the
tail (Lemma~\ref{lem:concat}(1)).

In both cases a single execution violates both predicates, so the
common-failure property holds at $\FC^{S}$ and
Lemma~\ref{lem:boundary-common-failure} gives the claim.
\end{proof}

\begin{lemma}
\label{lem:A2-1-FS-R}
$\mathrm{EH}\notin\Sol(\FS^{R})$.
\end{lemma}

\begin{proof}
Let $\mathcal A$ be any deterministic algorithm of $\FS$ under $\RSY$.
Take the alternating singleton activation with an empty fully
synchronous prefix (Lemma~\ref{lem:altA}) followed by an arbitrary fair
tail, with edge length $1$, $a(0)=v_0$, $b(0)=v_3$, common initial
colors, and normalized frames up to and including the round following
the first nonzero $\M$. Every round of this execution activates exactly
one robot, so Lemma~\ref{lem:solo-drift}(2) applies to it. The initial
configuration is symmetric under the rotation by $\pi$ about the
midpoint combined with exchanging the robots, and in $\FS$ the state of
the peer does not appear in the local input
(Lemma~\ref{lem:silent}), so until the first nonzero $\M$ both robots
follow the same color sequence $\gamma_0,g(\gamma_0),\dots$
(Lemma~\ref{lem:pumpB2}).
 
If $\lambda_{g^j(\gamma_0)}=0$ for every $j$, no nonzero $\M$ ever occurs,
so $\moves{a}=0$ and the perpetuity condition of $\mathrm{HET}$ fails,
while the joint-stop sequence has the single term $C_0$ and the
perpetuity condition of $\mathrm{EPG}^{*}$ fails.
 
Otherwise let $j$ be least with $\lambda_{g^j(\gamma_0)}\neq0$. The first
nonzero $\M$ is the $(j+1)$-st activation of $a$, performed while $b$
is not activated, so by Lemma~\ref{lem:solo-drift}(1) the expansion
condition of $\mathrm{EPG}^{*}$ fails at $i=0$. Write $(\lambda,\theta)$
for the displacement of that activation, with $\lambda\neq0$, and
$a^{*}$ for its landing point. If $a^{*}=b(0)$, then
$a^{*}=v_3\neq v_{-1}$ violates the second condition of $\mathrm{HET}$.
If $a^{*}\neq b(0)$ and $a^{*}\neq v_{-1}$, the second condition fails
at once. If $a^{*}=v_{-1}$, the robot $b$ activated next has the same
visible state as $a$ had, so it applies the same $(\lambda,\theta)$;
since $\lambda R_\theta$ maps $v_3-v_0$, of length $2$, to
$v_{-1}-v_0$, of length $1$, we have $\lambda=1/2$, and the
displacement of $b$ has length $\tfrac12\|v_{-1}-v_3\|=\sqrt3/2$, which
differs from the $\|v_2-v_3\|=1$ required by the second condition
(Lemma~\ref{lem:het_not_fs_r}).
 
In each case a single execution violates both predicates, so
Lemma~\ref{lem:boundary-common-failure} gives the claim.
\end{proof}

\subsubsection{\texorpdfstring{EPG$^{*}$ $\vee$ OSPO (EO)}{EPG* or OSPO}}
\label{ssub:f1-2}

\begin{lemma}
\label{lem:A2-2-FC-S}
$\mathrm{EO}\notin\Sol(\FC^{S})$.
\end{lemma}

\begin{proof}
Let $\mathcal A$ be any deterministic algorithm of $\FC$ under $\SSY$.
Run fully synchronous rounds with normalized frames. By
Lemma~\ref{lem:rigidA2} the colors follow a common sequence
$c_0,g(c_0),\dots$, and the configuration stays at $\iota$ while every
$\M$ is zero.

If $\lambda_{c_j}=0$ for every $j$, continue forever. Then $d$ is
constant, so $\Phi_{\mathrm{OSPO}}$ has no strictly decreasing phase.
The joint-stop sequence is $(C_0)$, so the perpetuity condition of
$\mathrm{EPG}^{*}$ fails. This fair execution witnesses both liveness
violations.

Otherwise let $t$ be least with $\lambda_{c_t}\neq0$. Replace round $t$
by two consecutive solo activations of $a$ (Lemma~\ref{lem:solo}), and
append a fair tail. Since $b$ is inactive, the visible state of $a$
stays $c_t$, so by Lemma~\ref{lem:single-rule}(b) the robot $a$ applies
$(\lambda_{c_t},\theta_{c_t})$ with $\lambda_{c_t}\neq0$ both times.

The first of these activations is the first nonzero $\M$ of the
execution, and it is solo. Hence the expansion condition of
$\mathrm{EPG}^{*}$ fails at $i=0$ (Lemma~\ref{lem:solo-drift}(1)).

For OSPO: if $d$ vanishes during the interval, the collision-freedom
condition of $\Phi_{\mathrm{OSPO}}$ fails. Otherwise the second
activation is also a nonzero $\M$. At its completion $m_a=2$ and
$m_b=0$, which no execution satisfying $\Phi_{\mathrm{OSPO}}$ admits
(as in Lemma~\ref{lem:osp_not_fc_s}).

In the latter case, all these violations are safety violations, preserved by the tail
(Lemma~\ref{lem:concat}(1)). Lemma~\ref{lem:boundary-common-failure}
gives the claim.
\end{proof}

\begin{lemma}
\label{lem:A2-2-FS-S}
$\mathrm{EO}\notin\Sol(\FS^{S})$.
\end{lemma}

\begin{proof}
Let $\mathcal A$ be any deterministic algorithm of $\FS$ under $\SSY$.
Give normalized frames at every $\LK$ at which the robots are distinct.
Then each robot advances along $c_0,g(c_0),\dots$, one step per
activation and independently of the schedule, and its displacement
depends only on its color. Let
$J:=\{i\ge0:\lambda_{g^i(c_0)}\neq0\}$.

A collision violates the collision-freedom condition of
$\Phi_{\mathrm{OSPO}}$ as a safety violation, so we assume none below.

Any witness $(t_i)$ of $\Phi_{\mathrm{OSPO}}$ has
$m_a(t_i)=m_b(t_i)=i$. Hence
\begin{equation}
  |m_a-m_b|\le1 \text{ at all times, and } m_a,m_b \text{ are unbounded.}
  \tag{$\ast$}
\end{equation}

\emph{Case $J=\emptyset$.} Run fully synchronously forever. No nonzero
$\M$ occurs. So $d$ is constant, and the joint-stop sequence is
$(C_0)$. Both $\Phi_{\mathrm{OSPO}}$ and the perpetuity condition of
$\mathrm{EPG}^{*}$ fail, and this execution witnesses both liveness
violations.

\emph{Case $J\neq\emptyset$.} Let $j:=\min J$, and activate $a$ alone
$j+1$ times (Lemma~\ref{lem:solo}). The last of these activations is
the first nonzero $\M$, and it is solo. So the expansion condition of
$\mathrm{EPG}^{*}$ fails at $i=0$ (Lemma~\ref{lem:solo-drift}(1)); this
is a safety violation.

If $|J|\ge2$, let $j'$ be the second element of $J$. Continue
activating $a$ alone up to its $(j'+1)$-st activation, then append a
fair tail. The robot $a$ performs a second nonzero $\M$ while $b$
performs none. This contradicts $(\ast)$, a safety violation.

If $|J|=1$, append a fully synchronous tail. Each robot performs at
most one nonzero $\M$ in the whole execution, contradicting the
unboundedness in $(\ast)$. This tail witnesses the liveness violation,
and the safety violation of $\mathrm{EPG}^{*}$ precedes it.

Lemma~\ref{lem:concat}(1)(3) is satisfied in every case, and
Lemma~\ref{lem:boundary-common-failure} gives the claim.
\end{proof}

\begin{lemma}
\label{lem:A2-2-OB-R}
$\mathrm{EO}\notin\Sol(\OB^{R})$.
\end{lemma}

\begin{proof}
Let $\mathcal A$ be any deterministic algorithm of $\OB$ under $\RSY$.
An $\OB$ robot has a unique visible state, so in a normalized frame the
rule of $\mathcal A$ is a single pair $(\lambda,\theta)$
(Lemma~\ref{lem:single-rule}(a)).

\emph{Case $(\lambda,\theta)\neq(\lambda^{\star},\theta^{\star})$.}
Since $\RSY$ contains the fully synchronous executions, take the
execution of Lemma~\ref{lem:rigidA1}, with both robots activated in
every round and a normalized frame at every $\LK$. By
Lemma~\ref{lem:osp_not_ob_f} it violates $\Phi_{\mathrm{OSPO}}$, the
degenerate cases $\lambda=0$ and $c_2=0$ being treated there as well.
For $\mathrm{EPG}^{*}$: if $\lambda=0$ the configuration never changes,
the joint-stop sequence has the single term $C_0$ and the perpetuity
condition fails; if $\lambda\neq0$ then $C_1\neq C_0$ and, by
Remark~\ref{rem:epgs-normalized-rule}, $d_1=\sigma_2 d_0$ with
$\sigma_2\neq2$, so the expansion condition fails at $i=0$.

\emph{Case $(\lambda,\theta)=(\lambda^{\star},\theta^{\star})$.} Take
the alternating singleton activation with an empty fully synchronous
prefix (Lemma~\ref{lem:altA}) and a normalized frame at every $\LK$.
Since $\lambda^{\star}\neq0$, the first activation of $a$ is the first
nonzero $\M$ of the execution and is performed while $b$ is not
activated, so by Lemma~\ref{lem:solo-drift}(1) the expansion condition
of $\mathrm{EPG}^{*}$ fails at $i=0$. For $\mathrm{OSPO}$: each
activation moves the active robot directly away from its peer by half
the observed distance, so $d$ is non-decreasing throughout the
execution and in fact $d\mapsto\tfrac32 d$ at every activation, by
Lemma~\ref{lem:ratios}(i) with
$\sigma_1=1-\lambda^{\star}e^{\mathrm{i}\theta^{\star}}=\tfrac32$.
Hence no interval on which $d$ strictly decreases exists and
$\Phi_{\mathrm{OSPO}}$ fails.

In both cases a single execution violates both predicates, so
Lemma~\ref{lem:boundary-common-failure} gives the claim.
\end{proof}

\subsubsection{\texorpdfstring{EPG$^{*}$ $\vee$ S-RDV (ES)}{EPG* or S-RDV}}
\label{ssub:f1-3}

\begin{lemma}
\label{lem:A2-3-FC-A}
$\mathrm{ES}\notin\Sol(\FC^{A})$.
\end{lemma}

\begin{proof}
Let $\mathcal A$ be any deterministic algorithm of $\FC$ under $\ASY$.
Fix distinct initial positions $a_0,b_0$, and put
$A=2a_0-b_0$, $B=2b_0-a_0$, and $P=\{A,a_0,b_0,B\}$.

Run fully synchronous cycles, choosing local frames that give
the two robots identical normalized observations.
Their colors remain equal.
If every computed Move is zero, this fair, time-divergent execution
keeps the robots at their distinct initial positions.
Thus $\Phi_{\SRDV}$ fails, and the joint-stop sequence consists
only of $C_0$, violating the perpetuity condition of
$\mathrm{EPG}^{*}$.

Otherwise, consider the first round with a nonzero output.
Let both robots perform their Looks before either performs its
Compute, and then let both complete their Computes.
They obtain the same normalized displacement $(\lambda,\theta)$,
with $\lambda>0$, and the same new color $\gamma'$.
Delay the Move of $b$ until after $a$ completes its Move.
At the landing of $a$, both robots are stationary, although
$b$ still has a pending Move.
Writing $p\ne a_0$ for the landing point of $a$, this joint stop
has midpoint
\[
  \frac{p+b_0}{2}\ne\frac{a_0+b_0}{2}.
\]
Hence Lemma~\ref{lem:solo-drift}(1) gives a finite violation
of $\Phi_{\mathrm{EPG}^{*}}$, which persists under every extension.

Suppose first that $(\lambda,\theta)\ne(2,0)$.
Since the Look of $a$ observed the initial configuration,
its landing point satisfies $p\ne B$.
No valid exchange time $T_1$ precedes this first completion,
because both cumulative nonzero-move counts were zero.
If $p\notin P$, the stationary-position requirement of
$\SRDV$ is already violated.
Otherwise $p\in\{A,b_0\}$.
The exchange configuration requiring zero Moves by $a$
is now impossible.
The other two require $a$ to reach $B$ with exactly one
completed nonzero Move.
However, any subsequent arrival at $B$ while respecting
the segment constraint requires a second completed nonzero Move:
the exterior endpoint $B$ cannot lie in the interior of
a straight movement contained in $[A,B]$.
Thus no continuation satisfies the exchange phase.
Complete the pending cycle of $b$ and append an arbitrary
admissible fair, time-divergent continuation.
Both predicates fail on the resulting execution.

It remains to consider $(\lambda,\theta)=(2,0)$.
Then $a$ lands at $B$.
Without reactivating $a$, let $b$ complete its already computed
Move, whose destination is $A$.
The robots are now at the distinct positions $(B,A)$,
both have color $\gamma'$, and neither has a pending cycle.

Let $\mathcal A_{\gamma'}$ have exactly the same transition
rule as $\mathcal A$, but designated initial color $\gamma'$.
Apply the general impossibility of stable $\mathrm{RDV}$
in $\FC^{A}$ used in the proof of
Lemma~\ref{lem:srdv_not_fc_a} to $\mathcal A_{\gamma'}$.
It supplies an admissible fair, time-divergent execution
from some distinct initial positions that fails stable rendezvous.
An orientation-preserving similarity maps those positions
to $(B,A)$.
Transforming the local frames accordingly preserves every
local observation and color transition, so variable
disorientation allows this counterexecution to be reproduced
from the reached configuration.

Append that execution to the finite prefix.
The resulting execution is admissible, fair, and time-divergent.
Its failure of stable rendezvous violates the terminal
condition of $\SRDV$, while the earlier violation of
$\Phi_{\mathrm{EPG}^{*}}$ persists.
Thus a single execution violates both predicates, and
Lemma~\ref{lem:boundary-common-failure} gives the claim.
\end{proof}

%
%
%

The remaining obligation of this pair, the class $\FS^{R}$, is
Proposition~\ref{prop:cf-expansion-exchange}.

\subsubsection{\texorpdfstring{EPG$^{*}$ $\vee$ ZCC (EZ)}{EPG* or ZCC}}
\label{ssub:f1-4}

\begin{lemma}
\label{lem:A2-4-FC-S}
$\mathrm{EZ}\notin\Sol(\FC^{S})$.
\end{lemma}
 
\begin{proof}
Let $\mathcal A$ be any deterministic algorithm of $\FC$ under $\SSY$.
The initial colors being common, while both robots are activated
simultaneously their colors follow a common eventually periodic
sequence $c_0,g(c_0),\dots$ (Lemma~\ref{lem:rigidA2}) and the
configuration stays at $\iota$.
 
If $\lambda_{c_j}=0$ for every $j$, the two robots never move and
remain at distinct positions, so the convergence condition of
$\mathrm{ZCC}$ fails, while the joint-stop sequence has the single term
$C_0$ and the perpetuity condition of $\mathrm{EPG}^{*}$ fails. This
fair execution witnesses both liveness violations.
 
Otherwise let $t$ be the least $j$ with $\lambda_{c_j}\neq0$ and
replace round $t$ by an interval in which $a$ alone is activated
repeatedly (Lemma~\ref{lem:solo}); the interval may be made arbitrarily
long, the activations of $b$ being placed far enough in the future to
keep the execution fair. The configuration was unchanged before the
interval, so the first activation of $a$ in it is the first nonzero
$\M$ of the execution and it is solo; by
Lemma~\ref{lem:solo-drift}(1) the expansion condition of
$\mathrm{EPG}^{*}$ fails at $i=0$, a safety violation.
 
The visible state of $a$ is constantly $c_t$, so by
Lemma~\ref{lem:single-rule}(b) its rule is a single pair
$(\lambda,\theta)$ with $\lambda\neq0$, and $b$ is at rest throughout.
Under $\SSY$ every activation of $a$ is a complete cycle, so no
operation is pending at the end of any of them. Put
$\sigma_1:=1-\lambda e^{\mathrm{i}\theta}$ and $c_1:=|\sigma_1|$, and
write $a_0:=a(0)$, $b_0:=b(0)$ and $d:=\|b_0-a_0\|$.
 
If $c_1=0$, the robot $a$ reaches the position of $b$ at the first
activation, and no operation is pending at that time. Use the
coincidence continuation described in the proof of
Lemma~\ref{lem:zcc_not_fc_s}, with $a_0,b_0$ kept fixed. It either
keeps both positions fixed forever, so that the zig-zag condition
fails, or produces a finite violation of the collision-freedom
condition; in either case $\Phi_{\mathrm{ZCC}}$ fails.
 
Assume $c_1\neq0$, and put $\psi:=\arg\sigma_1$
(Lemma~\ref{lem:ratios}(iii)). The vector $w_j:=b_0-a(j)$ satisfies
$w_j=\sigma_1^{\,j}w_0\neq0$, so every $\LK$ of $a$ in the interval
observes distinct positions, and
\[
  f_a(j)=\|a(j)-a_0\|=\|w_0-\sigma_1^{\,j}w_0\|
        =d\,\bigl|1-c_1^{\,j}e^{\mathrm{i}j\psi}\bigr| .
\]
As $b$ stands still, $\|b(j)-a_0\|=d$, so the collision-freedom
condition requires $f_a(j)\le d$, that is $c_1^{\,j}\le2\cos(j\psi)$.
Note that $\sigma_1=1$ would force $\lambda=0$, so $c_1\neq1$ whenever
$\psi\equiv0$.
\begin{itemize}
\item If $\psi\not\equiv0\pmod{2\pi}$, Lemma~\ref{lem:ratios}(iii)
  provides $k\ge1$ with $\cos(k\psi)\le-\tfrac12<0$; making the
  interval long enough to reach the $k$-th activation, the
  collision-freedom condition fails there, a safety violation.
  Append an arbitrary fair, time-divergent continuation.
\item If $\psi\equiv0$ and $c_1>1$, that is $\theta=\pi$, then
  $f_a(j)=(c_1^{\,j}-1)d$ exceeds $d$ as soon as $c_1^{\,j}>2$, and the
  collision-freedom condition again fails after finitely many
  activations. Append an arbitrary fair, time-divergent continuation.
\item In the remaining contraction case $\psi\equiv0$ and $0<c_1<1$,
  that is $\sigma_1\in(0,1)$, end the solo interval after its first
  activation. Then $a(1)=b_0-\sigma_1(b_0-a_0)$ lies in the open
  segment $(a_0,b_0)$ and $b$ is still at $b_0$, so the robots still
  lie in the original initial segment in their original order, and
  neither has a pending operation. Apply the continuation property
  proved in Lemma~\ref{lem:zcc_not_fc_s}, keeping the original initial
  positions fixed. It supplies a fair, time-divergent continuation
  violating $\Phi_{\mathrm{ZCC}}$.
\end{itemize}
 
In every case the safety violation of $\Phi_{\mathrm{EPG}^{*}}$ occurs
at the first activation of the solo interval, which belongs to the
prefix of each construction above, and every construction is a fair,
time-divergent continuation of that prefix. The violation therefore
persists (Lemma~\ref{lem:concat}(1)), one fair infinite execution
violates both predicates, and
Lemma~\ref{lem:boundary-common-failure} gives the claim.
\end{proof}
 
\begin{lemma}
\label{lem:A2-4-FS-AM}
$\mathrm{EZ}\notin\Sol(\FS^{A_M})$.
\end{lemma}
 
\begin{proof}
Let $\mathcal A$ be any deterministic algorithm of $\FS$ under $A_M$,
which admits synchronous activations of one robot at a time
(Lemma~\ref{lem:solo}). Start with alternating singleton activations,
each cycle being completed before the next one begins, with normalized
frames at every $\LK$ that observes distinct positions. This interval
activates one robot at a time, so Lemma~\ref{lem:solo-drift} applies
to it.
 
If no nonzero $\M$ ever occurs, the two robots stay at $\iota$ at
distinct positions, so the convergence condition of $\mathrm{ZCC}$
fails, and the joint-stop sequence has the single term $C_0$, so the
perpetuity condition of $\mathrm{EPG}^{*}$ fails. This fair execution
witnesses both liveness violations.
 
Assume therefore that a nonzero $\M$ occurs. The first one is performed
while the peer is not activated, so by Lemma~\ref{lem:solo-drift}(1)
the expansion condition of $\mathrm{EPG}^{*}$ fails at $i=0$, a safety
violation. End the prefix at the completion of the cycle containing
this $\M$. This is a joint stop, with no pending operation.
 
Every $\LK$ in the prefix observes distinct positions: before the first
nonzero $\M$ the configuration is $\iota$, and that $\M$ is the last
operation of the prefix. Since the visible state of an $\FS$ robot is
its own color, Lemma~\ref{lem:norm}(1) shows that each robot that has
been activated $k$ times carries $g^{k}(c_0)$. Hence the two colors lie
on the same normalized color orbit, the orbit of $c_0$ under $g$.
 
If the robots coincide, use the coincidence continuation from
Lemma~\ref{lem:zcc_not_fc_s}. Otherwise their colors lie on the same
normalized color orbit, so the continuation statement proved in
Lemma~\ref{lem:zcc_not_fs_am} applies. In either case it supplies a
fair, time-divergent continuation violating $\Phi_{\mathrm{ZCC}}$, with
the original initial positions kept fixed. The safety violation
established in the prefix persists (Lemma~\ref{lem:concat}(1)).
 
A single fair infinite execution therefore violates both predicates,
so Lemma~\ref{lem:boundary-common-failure} gives the claim.
\end{proof}

\begin{lemma}
\label{lem:A2-4-OB-R}
$\mathrm{EZ}\notin\Sol(\OB^{R})$.
\end{lemma}

\begin{proof}
Let $\mathcal A$ be any deterministic algorithm of $\OB$ under $\RSY$.
An $\OB$ robot has a unique visible state, so in a normalized frame the
rule is a single pair $(\lambda,\theta)$
(Lemma~\ref{lem:single-rule}(a)).

\emph{Case $(\lambda,\theta)\neq(\lambda^{\star},\theta^{\star})$.}
Since $\RSY$ contains the fully synchronous executions, take the
execution of Lemma~\ref{lem:rigidA1}, with both robots activated in
every round and a normalized frame at every $\LK$. By
Lemma~\ref{lem:zcc_not_ob_f} that execution violates
$\Phi_{\mathrm{ZCC}}$, the degenerate cases $\lambda=0$ and $c_2=0$
being treated there as well. For $\mathrm{EPG}^{*}$: if $\lambda=0$ the
joint-stop sequence has the single term $C_0$ and the perpetuity
condition fails; if $\lambda\neq0$ then $C_1\neq C_0$ and, by
Remark~\ref{rem:epgs-normalized-rule}, $d_1=\sigma_2 d_0$ with
$\sigma_2\neq2$, so the expansion condition fails at $i=0$.

\emph{Case $(\lambda,\theta)=(\lambda^{\star},\theta^{\star})$.} Take
the alternating singleton activation with an empty fully synchronous
prefix (Lemma~\ref{lem:altA}) and a normalized frame at every $\LK$.
Since $\lambda^{\star}\neq0$, the first activation of $a$ is the first
nonzero $\M$ of the execution and is performed while $b$ is not
activated, so by Lemma~\ref{lem:solo-drift}(1) the expansion condition
of $\mathrm{EPG}^{*}$ fails at $i=0$, a safety violation. For ZCC: by Remark~\ref{rem:epgs-normalized-rule}, $c_1=\tfrac32$, so by
Lemma~\ref{lem:ratios}(i) each activation maps
$d\mapsto\tfrac32 d$; hence $d\to\infty$, the two robots do not
converge to a single point, and the convergence condition fails.

In both cases a single execution violates both predicates, the safety
violation being placed first and a fair infinite continuation of the
liveness violation being exhibited, so Lemma~\ref{lem:concat}(1)(3) is
satisfied and Lemma~\ref{lem:boundary-common-failure} gives the claim.
\end{proof}


\subsubsection{\texorpdfstring{EPG$^{*}$ $\vee$ VSS (EV)}{EPG* or VSS}}
\label{ssub:f1-5}

This pair has the single boundary class $\FS^{R}$, and that obligation
is discharged in the main text by
Proposition~\ref{prop:cf-expansion-exchange}, which exhibits one
admissible fair execution violating $\mathrm{EPG}^{*}$,
$\SRDV$ and $\VSS$ simultaneously.

\subsubsection{\texorpdfstring{EPG$^{*}$ $\vee$ VEC (EX)}{EPG* or VEC}}
\label{ssub:f1-6}

\begin{lemma}
\label{lem:A2-6-FC-S}
$\mathrm{EX}\notin\Sol(\FC^{S})$.
\end{lemma}

\begin{proof}
Let $\mathcal A$ be any deterministic algorithm of $\FC$ under $\SSY$.
The initial colors being common, while both robots are activated
simultaneously their colors follow a common eventually periodic
sequence $\gamma_0,g(\gamma_0),\dots$ (Lemma~\ref{lem:rigidA2}) and the
configuration stays at $\iota$.
 
If $\lambda_{c_j}=0$ for every $j$, the two robots never move, the
joint-stop sequence reduces to $C_0$, and $d_0\not\perp d_0$, so the
attainment condition of $\mathrm{VEC}$ fails, while the perpetuity
condition of $\mathrm{EPG}^{*}$ fails as well
(Lemma~\ref{lem:solo-drift}(2)).
 
Otherwise let $t$ be the least $j$ with $\lambda_{c_j}\neq0$, replace
round $t$ by an interval in which $a$ alone is activated $k\ge3$ times
in a row (Lemma~\ref{lem:solo}), and append a fair tail. The
configuration was unchanged before the interval, so the first
activation of $a$ in it is the first nonzero $\M$ of the execution and
it is performed while $b$ is not activated; by
Lemma~\ref{lem:solo-drift}(1) the expansion condition of
$\mathrm{EPG}^{*}$ fails at $i=0$, a safety violation.
 
The visible state of $a$ is constantly $c_t$, so by
Lemma~\ref{lem:single-rule}(b) it applies the same displacement rule at
every activation. With the peer at rest, the only displacements that
satisfy the transition condition of $\mathrm{VEC}$ are those rotating
the segment clockwise by $45^\circ$ or by $90^\circ$; in particular if
$c_1=0$ then $d_i=0$, while all the permitted targets are nonzero, so
the transition condition fails. If $a$ departs from those two
displacements the transition condition fails at that point, and
otherwise $a$ reaches a configuration with $d_i\perp d_0$ in at most
two moves. At that moment both robots are at rest, but $a$ repeats the
same displacement at its next activation, so the stopping condition
fails (Lemma~\ref{lem:vec_not_fc_s},
\cite[Lemma~21]{NNW-SSS2025}). All of these are safety violations and
may be concatenated by Lemma~\ref{lem:concat}(1)(2).
 
In each case a single execution violates both predicates, so
Lemma~\ref{lem:boundary-common-failure} gives the claim.
\end{proof}

\begin{lemma}
\label{lem:A2-6-FS-A}
$\mathrm{EX}\notin\Sol(\FS^{A})$.
\end{lemma}

\begin{proof}
Let $\mathcal A$ be any deterministic algorithm of $\FS$ under $\ASY$.
Begin with a fully synchronous prefix, which is admissible under
$\ASY$; by Lemma~\ref{lem:rigidA2} the two robots always carry the same
state $s_0,s_1,\dots$.

If $\lambda_{s_k}=0$ for every $k$, this fair execution contains no
nonzero $\M$, the joint-stop sequence reduces to $C_0$, and
$d_0\not\perp d_0$, so the attainment condition of $\mathrm{VEC}$ fails
and the perpetuity condition of $\mathrm{EPG}^{*}$ fails as well
(Lemma~\ref{lem:solo-drift}(2)).

Otherwise take the least $k$ with $\lambda_{s_k}>0$. All earlier rounds
left the configuration at $\iota$. Activate $a$ alone in round $k$, so
that it performs the first nonzero $\M$ of the execution, and use the
freedom of $\ASY$ in the choice of $\LK$ times: by
Lemma~\ref{lem:mid} the robot $b$ observes $a$ at a point $p_\mu$,
$\mu\in(0,1)$, of its movement segment. Let $b$ begin its $\M$ before
the $\M$ of $a$ ends, so that no time of $\Sab$ lies between the two
landings and the configuration after both landings is the term $C_1$ of
the joint-stop sequence following $C_0$.

An $\FS$ robot cannot consult the internal state of its peer and
retains no previously observed angle or distance, so the displacement
of $b$ is determined by $p_\mu$ alone, and $C_1$ varies continuously
and non-constantly with $\mu$. The transition condition of
$\mathrm{VEC}$ admits only finitely many successors of $C_0$, and the
expansion condition of $\mathrm{EPG}^{*}$ admits exactly one, namely
the configuration with midpoint $c_0$ and difference vector $2d_0$.
Adding the countably many values of $\mu$ at which the two robots
coincide, Corollary~\ref{cor:continuum} provides a $\mu$ avoiding all
of them, and both the transition condition of $\mathrm{VEC}$
(Lemma~\ref{lem:vec_not_fs_a}, \cite[Lemma~22]{NNW-SSS2025}) and the
expansion condition of $\mathrm{EPG}^{*}$ fail at $i=0$.

Both are safety violations witnessed at a finite time, hence preserved
by any fair tail (Lemma~\ref{lem:concat}(1)), and
Lemma~\ref{lem:boundary-common-failure} gives the claim.
\end{proof}

\begin{lemma}
\label{lem:A2-6-OB-R}
$\mathrm{EX}\notin\Sol(\OB^{R})$.
\end{lemma}

\begin{proof}
Let $\mathcal A$ be any deterministic algorithm of $\OB$ under $\RSY$.
Take the alternating singleton activation of $a$ and $b$ with an empty
fully synchronous prefix (Lemma~\ref{lem:altA}) followed by an
arbitrary fair tail, and supply normalized frames. By
Lemma~\ref{lem:single-rule}(a) the rule of $\mathcal A$ is a single
pair $(\lambda,\theta)$. Every round activates exactly one robot, so
Lemma~\ref{lem:solo-drift} applies to this execution.

If $\lambda=0$, then since $\OB$ has no colors the two robots do not
move under any schedule, the joint-stop sequence reduces to $C_0$, and
$d_0\not\perp d_0$, so the attainment condition of $\mathrm{VEC}$ fails
and the perpetuity condition of $\mathrm{EPG}^{*}$ fails.

Assume $\lambda\neq0$. The activation of $a$ in round~$1$ is then the
first nonzero $\M$ of the execution and $b$ is not activated in that
round, so by Lemma~\ref{lem:solo-drift}(1) the expansion condition of
$\mathrm{EPG}^{*}$ fails at $i=0$, a safety violation.

For $\mathrm{VEC}$: if $c_1=0$, then $a$ lands on the position of $b$ in
round~$1$ and $d_1=0$, while all the targets permitted by the
transition condition are nonzero, so that condition fails. Otherwise
the only single displacements satisfying the transition condition are
those rotating the segment clockwise by $45^\circ$ or by $90^\circ$. If
$a$ departs from them the transition condition fails at that point;
otherwise $a$ reaches a configuration with $d_i\perp d_0$ in at most
two moves, and since the same displacement is applied at every
activation the robots move again immediately afterwards, so the
stopping condition fails. All of these are safety violations.

In every case a single execution violates both predicates, so
Lemma~\ref{lem:boundary-common-failure} gives the claim.
\end{proof}

\subsubsection{\texorpdfstring{EPG$^{*}$ $\vee$ CFSM (EC)}{EPG* or CFSM}}
\label{ssub:f1-7}

\begin{lemma}
\label{lem:A2-7-FC-S}
$\mathrm{EC}\notin\Sol(\FC^{S})$.
\end{lemma}
 
\begin{proof}
Let $\mathcal A$ be any deterministic algorithm of $\FC$ under $\SSY$.
The initial colors being common, while both robots are activated
simultaneously their colors follow a common eventually periodic
sequence $\gamma_0,g(\gamma_0),\dots$ (Lemma~\ref{lem:rigidA2}) and the
configuration stays at $\iota$.
 
If $\lambda_{c_j}=0$ for every $j$, this fair execution contains no
nonzero $\M$, so $\moves{a}=0\neq1$ and the single-move condition of
$\CFSM$ fails, while the joint-stop sequence reduces to $C_0$
and the perpetuity condition of $\mathrm{EPG}^{*}$ fails
(Lemma~\ref{lem:solo-drift}(2)).
 
Otherwise let $t$ be the least $j$ with $\lambda_{c_j}\neq0$, replace
round $t$ by an interval in which $a$ alone is activated twice in a row
(Lemma~\ref{lem:solo}), and append a fair tail. The configuration was
unchanged before the interval, so the first activation of $a$ in it is
the first nonzero $\M$ of the execution and it is performed while $b$
is not activated; by Lemma~\ref{lem:solo-drift}(1) the expansion
condition of $\mathrm{EPG}^{*}$ fails at $i=0$.
 
Since $b$ is not activated, the visible state of $a$ is constantly
$c_t$, so by Lemma~\ref{lem:single-rule}(b) it applies the same pair
$(\lambda,\theta)$, with $\lambda\neq0$, at both activations. If
$c_1=0$ then $a$ lands on the position of $b$ at the first activation,
so $a(t)=b(t)$ at that time and the collision-freedom condition fails.
If $c_1\neq0$ the distance stays positive, so both activations are
nonzero $\M$ operations and $\moves{a}\ge2\neq1$, which violates the
single-move condition. All of these are safety violations
witnessed at a finite time, hence preserved by the fair tail
(Lemma~\ref{lem:concat}(1)(2)).
 
In each case a single execution violates both predicates, so
Lemma~\ref{lem:boundary-common-failure} gives the claim.
\end{proof}

\begin{lemma}
\label{lem:A2-7-OB-R}
$\mathrm{EC}\notin\Sol(\OB^{R})$.
\end{lemma}

\begin{proof}
Let $\mathcal A$ be any deterministic algorithm of $\OB$ under $\RSY$.
Take the alternating singleton activation of $a$ and $b$ with an empty
fully synchronous prefix (Lemma~\ref{lem:altA}) followed by an
arbitrary fair tail, and supply normalized frames. The visible state of
an $\OB$ robot is unique, so the rule is a single pair
$(\lambda,\theta)$ (Lemma~\ref{lem:single-rule}(a)). Every round
activates exactly one robot, so Lemma~\ref{lem:solo-drift} applies.

If $\lambda=0$, then since $\OB$ has no colors the two robots do not
move under any schedule, so $\moves{a}=0\neq1$ and the single-move
condition of $\CFSM$ fails, while the joint-stop sequence
reduces to $C_0$ and the perpetuity condition of $\mathrm{EPG}^{*}$
fails.

Assume $\lambda\neq0$. The activation of $a$ in round~$1$ is the first
nonzero $\M$ of the execution and $b$ is not activated in that round,
so by Lemma~\ref{lem:solo-drift}(1) the expansion condition of
$\mathrm{EPG}^{*}$ fails at $i=0$, a safety violation.

For $\CFSM$: if $c_1=0$, then $a$ lands on the position of $b$
in round~$1$, so $a(t)=b(t)$ at that time and the collision-freedom
condition fails. Otherwise each activation maps $d\mapsto c_1d$
(Lemma~\ref{lem:ratios}(i)) and $d$ stays positive, so the activations
of $a$ in rounds~$1$ and~$3$ are both nonzero $\M$ operations and
$\moves{a}\ge2\neq1$, which violates the single-move condition.

In every case a single execution violates both predicates, so
Lemma~\ref{lem:boundary-common-failure} gives the claim.
\end{proof}

\subsection{Disjunctions Containing SRO}
\label{app:cf-sro-family}

This subsection treats the five antichains $\{\mathrm{SRO},P\}$ with
$P\in\mathcal E\cup\mathcal P_3
=\{\SRDV,\VSS\}\cup\{\mathrm{ZCC},\mathrm{VEC},\CFSM\}$,
which carry eight boundary obligations in total. Every construction
below uses the geometric fact of Lemma~\ref{lem:cf-sro-two-moves}: with
the peer at rest, two nonzero $\M$ operations of the same robot cannot
both satisfy the rotation condition and $\mathbf{C3}$, and a landing on
the peer satisfies neither $\mathbf{C1}$, which requires
$\|d_{i+1}\|=\|d_i\|\neq0$, nor $\mathbf{C2}$, whose right-hand side is
nonzero. Zero $\M$ operations do not advance the sequence $(C_i)$.

Two of the eight obligations are discharged by
Proposition~\ref{prop:cf-sro-exchange-family}, stated in
Section~\ref{subsec:cf-sro-exchange}: for every algorithm in $\FS^{S}$
it exhibits one admissible fair execution violating $\mathrm{SRO}$ and
both exchange tasks simultaneously, which settles the whole maximal
negative boundary of $\mathrm{SRO}\vee\VSS$ and the $\FS^{S}$
obligation of $\mathrm{SRO}\vee\SRDV$. We give its proof first; the
remaining six obligations, including the $\FC^{A}$ obligation of
$\mathrm{SRO}\vee\SRDV$, are proved afterwards.

\begin{proof}[Proof of Proposition~\ref{prop:cf-sro-exchange-family}]
Fix an algorithm and initial positions $a_0=d$, $b_0=2d$, with outer
points $0$ and $3d$. Under normalized observations, write the state
transition as $s\mapsto g(s)$ and the global displacement as
$\Lambda_s w$, where $w$ points to the peer and
$\Lambda_s=\lambda_sR_{\theta_s}$.

If the state orbit from the designated initial state contains no
nonzero output, fair alternating activations leave both robots
stationary. Neither the exchange phase nor the perpetual-motion
requirement of $\mathrm{SRO}$ is satisfied. Otherwise, let $s$ be the
first state in that orbit with a nonzero output. Activate only $a$
until it completes that first nonzero $\M$; all preceding $\M$
operations are zero. We distinguish two cases.

\emph{Case 1: $\Lambda_s=2I$.}
Robot $a$ moves to $3d$. The relative segment reverses direction
without shrinking, violating both permitted rotation types of
$\mathrm{SRO}$. Keep $a$ inactive and activate $b$ until it reaches
state $s$; its preceding normalized outputs are zero, so the
configuration is unchanged and $b$ sees a peer at distance $d$.
By Lemma~\ref{lem:p1_replay} the robot $b$ then lands at $4d$,
outside the interval allowed by $\SRDV$ and outside the endpoint set
allowed by $\VSS$. All three violations are irreversible, so append a
fair suffix.

\emph{Case 2: $\Lambda_s\neq2I$.}
The first nonzero $\M$ makes the common exchange phase unattainable, so
both exchange tasks have already failed. If it also violates
$\mathrm{SRO}$, append a fair suffix. Otherwise the robots remain
distinct. If the normalized state orbit of $a$ after this $\M$ contains
another nonzero output, keep $b$ inactive until that output is
executed; Lemma~\ref{lem:cf-sro-two-moves} then violates $\mathrm{SRO}$,
and a fair suffix again completes the construction.

If no further nonzero output occurs in that orbit, keep $a$ inactive
and let $b$ traverse the same initial state sequence and perform its
first nonzero $\M$. If this violates $\mathrm{SRO}$, we are done.
Otherwise the robots are still distinct and both are now in state
$g(s)$, from which every subsequent normalized output is zero. Fair
alternating activations therefore leave the configuration stationary
forever, violating the perpetual-motion requirement of $\mathrm{SRO}$.
This specifies a fair infinite continuation rather than an infinite
solo execution.
\end{proof}

\subsubsection{\texorpdfstring{SRO $\vee$ S-RDV (RS)}{SRO or S-RDV}}
\label{ssub:f2-1}

\begin{lemma}
\label{lem:A2-10-FC-A}
$\mathrm{RS}\notin\Sol(\FC^{A})$.
\end{lemma}

\begin{proof}
Let $\mathcal A$ be any deterministic algorithm of $\FC$ under $\ASY$.
Run fully synchronous cycles with identical normalized observations
until the first nonzero output.
Throughout this zero-Move prefix, the positions remain unchanged
and the colors remain equal.
If no nonzero output occurs, the resulting fair, time-divergent
execution violates both the perpetuity condition of $\mathrm{SRO}$
and the terminal rendezvous condition of $\SRDV$.

Otherwise, identify the plane with $\mathbb C$ using
orientation-preserving coordinates in which $a_0=0$ and $b_0=1$.
The outer exchange points are then $A=-1$ and $B=2$.
In the first round with a nonzero output, let both robots perform
their Looks before either performs its Compute, and then let both
complete their Computes.
Write $z=\lambda e^{\mathrm i\theta}\ne0$ for their common normalized
displacement coefficient and $\gamma'$ for their common updated color.
Their computed destinations are $z$ for $a$ and $1-z$ for $b$.

Delay the Move of $b$ until strictly after $a$ has landed,
and do not reactivate $a$ in the meantime.
There is therefore a joint stop with positions $(z,1)$,
whose relative-position vector is
\[
  d_1=1-z,\qquad d_0=1.
\]

\emph{Case $z\ne2$.}
The first nonzero Move of $a$ does not land at $B$.
Since its Look observed the initial configuration,
the first-move obstruction of Lemma~\ref{lem:p1_first_move}
applies, so no extension can satisfy the exchange phase of $\SRDV$.

If
\[
  d_1\notin\left\{-\mathrm i,\frac{1-\mathrm i}{2}\right\},
\]
then neither $\mathbf{C1}$ nor $\mathbf{C2}$ holds at the first
transition, and $\mathrm{SRO}$ already fails.
Otherwise
\[
  z\in\left\{1+\mathrm i,\frac{1+\mathrm i}{2}\right\}.
\]
Let $b$ complete its stored Move while $a$ remains at rest.
The next joint-stop vector is $d_2=1-2z$, and
\[
  \frac{d_2}{d_1}=
  \begin{cases}
    2-\mathrm i, & z=1+\mathrm i,\\
    1-\mathrm i, & z=(1+\mathrm i)/2.
  \end{cases}
\]
Its modulus is respectively $\sqrt5$ or $\sqrt2$.
Thus the distance increases, whereas $\mathbf{C1}$ and
$\mathbf{C2}$ permit only factors $1$ and $1/\sqrt2$.
Hence $\mathrm{SRO}$ fails at this second transition.
Once the pending cycle of $b$ is complete, append an arbitrary
admissible fair, time-divergent continuation.
Both finite violations persist.

\emph{Case $z=2$.}
The first joint-stop vector is $d_1=-1$, which satisfies neither
$\mathbf{C1}$ nor $\mathbf{C2}$.
Thus $\mathrm{SRO}$ has already failed.
Without reactivating $a$, let $b$ complete its stored Move.
The robots are now at $(B,A)=(2,-1)$, both have color $\gamma'$,
and neither has a pending cycle.

Let $\mathcal A_{\gamma'}$ have the same transition rule as
$\mathcal A$, but designated initial color $\gamma'$.
Apply the general impossibility of stable $\mathrm{RDV}$
in $\FC^{A}$ used in the proof of
Lemma~\ref{lem:srdv_not_fc_a}.
It gives a fair, time-divergent counterexecution of
$\mathcal A_{\gamma'}$ from distinct initial positions.
As in the proof of Lemma~\ref{lem:A2-3-FC-A},
an orientation-preserving similarity and corresponding choices
of local frames reproduce this counterexecution from $(B,A)$.
The reached colors and absence of pending cycles make it
a valid continuation of our prefix.

This continuation violates the terminal rendezvous condition
of $\SRDV$, while preserving the earlier violation of
$\mathrm{SRO}$.
In every case, a single admissible fair, time-divergent execution
violates both predicates.
Lemma~\ref{lem:boundary-common-failure} gives the claim.
\end{proof}

The remaining obligation of this pair, the class $\FS^{S}$, is
Proposition~\ref{prop:cf-sro-exchange-family}.

\subsubsection{\texorpdfstring{SRO $\vee$ ZCC (RZ)}{SRO or ZCC}}
\label{ssub:f2-2}

\begin{lemma}
\label{lem:A2-11-FC-S}
$\mathrm{RZ}\notin\Sol(\FC^{S})$.
\end{lemma}

\begin{proof}
Let $\mathcal A$ be any deterministic algorithm of $\FC$ under $\SSY$.
As above the colors follow a common eventually periodic sequence
$c_0,g(c_0),\dots$ (Lemma~\ref{lem:rigidA2}) while both robots are
activated simultaneously, and the configuration stays at $\iota$.

If $\lambda_{c_j}=0$ for every $j$, no nonzero $\M$ ever occurs, so the
sequence of configurations reduces to $C_0$ and the perpetuity
condition of $\mathrm{SRO}$ fails, while the robots remain at distinct
positions and the convergence condition of $\mathrm{ZCC}$ fails. This
fair execution witnesses both liveness violations.

Otherwise let $t$ be the least $j$ with $\lambda_{c_j}\neq0$ and
replace round $t$ by an interval in which $a$ alone is activated $k$
times in a row (Lemma~\ref{lem:solo}, with $k$ large enough for what
follows). Since $b$ is not activated, the visible state of $a$ is
constantly $c_t$, so by Lemma~\ref{lem:single-rule}(b) the rule is the
single pair $(\lambda,\theta)$ with $\lambda\neq0$, and $b$ is at rest
throughout.

For $\mathrm{SRO}$: if $c_1=0$ the robot $a$ lands on the position of
$b$ at the first activation, and $d_{i+1}=0$ satisfies neither
$\mathbf{C1}$ nor $\mathbf{C2}$, so the second condition of
$\mathrm{SRO}$ fails. If $c_1\neq0$, taking $k\ge2$ makes $a$ perform
two nonzero $\M$ operations with its peer at rest, so by
Lemma~\ref{lem:cf-sro-two-moves} the second condition of $\mathrm{SRO}$
fails. Either way this is a safety violation in the prefix.

For $\mathrm{ZCC}$: put $\sigma_1:=1-\lambda e^{\mathrm{i}\theta}$ and
$c_1:=|\sigma_1|$, as in the proof of Lemma~\ref{lem:A2-4-FC-S}.
Under $\SSY$ every activation of $a$ is a complete cycle, so neither
robot has a pending operation at the end of any of them.
If $c_1=0$, end the solo interval after its first activation; the
robots then coincide. Use the coincidence continuation described in
the proof of Lemma~\ref{lem:zcc_not_fc_s}. This either keeps both
positions fixed forever, so that the zig-zag condition fails, or
produces a finite violation of the collision-freedom condition.
If $c_1\neq0$ and $\sigma_1\notin(0,1)$, the cases $\psi\not\equiv0$
and $\psi\equiv0$ with $c_1>1$ in the proof of
Lemma~\ref{lem:A2-4-FC-S} apply verbatim: for $k$ large enough the
collision-freedom condition fails at a finite time, and an arbitrary
fair, time-divergent continuation is appended.
In the remaining contraction case $\sigma_1\in(0,1)$, after the $k$
activations we have $a(k)=b(0)-\sigma_1^{\,k}\bigl(b(0)-a(0)\bigr)$
and $b$ is still at $b(0)$. Thus the robots still lie in the original
initial segment in their original order, and neither has a pending
operation. Apply the continuation property proved in
Lemma~\ref{lem:zcc_not_fc_s}, keeping the original initial positions
fixed. It supplies a fair, time-divergent continuation violating
$\Phi_{\mathrm{ZCC}}$, while preserving the safety violation of
$\Phi_{\mathrm{SRO}}$ already obtained in the prefix.

The safety violation of $\Phi_{\mathrm{SRO}}$ occurs in the prefix and
precedes every continuation used above, so
Lemma~\ref{lem:concat}(1)(3) is satisfied and
Lemma~\ref{lem:boundary-common-failure} gives the claim.
\end{proof}

\begin{lemma}
\label{lem:A2-11-FS-AM}
$\mathrm{RZ}\notin\Sol(\FS^{A_M})$.
\end{lemma}

\begin{proof}
Let $\mathcal A$ be any deterministic algorithm of $\FS$ under $A_M$,
which admits synchronous activations of one robot at a time
(Lemma~\ref{lem:solo}). Begin with an interval in which $b$ is never
activated and $a$ alone is activated, one complete cycle at a time,
with normalized frames at every $\LK$ at which the robots are
distinct. Since in $\FS$ activations that only update the color may be
interleaved, the length of this interval is measured in nonzero $\M$
operations and not in activations. Continue it until $a$ lands on the
position of $b$ or completes its second nonzero $\M$, whichever comes
first. As long as $a$ has not landed on $b$, every $\LK$ of $a$
observes distinct positions, so its color follows
$c_0,g(c_0),g^2(c_0),\dots$ (Lemma~\ref{lem:pumpB2}).
 
If in this interval $a$ lands on the position of $b$, then
$d_{i+1}=0$ and neither $\mathbf{C1}$ nor $\mathbf{C2}$ can hold, so
the second condition of $\mathrm{SRO}$ fails at that point, a safety
violation.
 
Suppose next that $a$ never lands on $b$ and performs at most one
nonzero $\M$ along its whole solo orbit, which includes the case in
which no state produces a nonzero $\M$. Then the periodic part of the
orbit of $c_0$ contains no nonzero rule. Stop activating $a$ once its
color has entered that periodic part, and activate $b$ alone the same
number of times. The color of $b$ follows the same finite orbit as
long as its $\LK$s observe distinct positions. If $b$ lands on the
position of $a$, the second condition of $\mathrm{SRO}$ fails as in
the previous case, and the prefix is ended there as described below.
Otherwise both robots now carry colors in the periodic part, at
distinct positions, and neither has a pending operation. Continue with
alternating singleton activations forever. Every $\LK$ observes
distinct positions and every output is zero, so the positions remain
fixed. This fair, time-divergent execution contains only finitely many
nonzero $\M$ operations, so the sequence $(C_i)$ is finite and the
perpetuity condition of $\mathrm{SRO}$ fails, while $f_r$ is eventually
constant and the zig-zag condition of $\mathrm{ZCC}$ fails.
 
In the remaining case $a$ does not land on $b$ and performs two nonzero
$\M$ operations with its peer at rest, so by
Lemma~\ref{lem:cf-sro-two-moves} the rotation condition and
$\mathbf{C3}$ cannot both hold, whatever the internal states, and the
second condition of $\mathrm{SRO}$ fails.
 
In each case with a safety violation of $\Phi_{\mathrm{SRO}}$, the
violation is produced by the last $\M$ of the prefix constructed so
far, and every $\LK$ in that prefix observed distinct positions. By
Lemma~\ref{lem:norm}(1), each robot that has been activated $k$ times
carries $g^{k}(c_0)$.
End the prefix at this joint stop, with no pending operation.
If the robots coincide, use the coincidence continuation
from Lemma~\ref{lem:zcc_not_fc_s}. Otherwise their colors lie on the
same normalized color orbit, so the continuation statement
proved in Lemma~\ref{lem:zcc_not_fs_am} applies.
In either case it supplies a fair, time-divergent continuation
violating $\Phi_{\mathrm{ZCC}}$, with the original initial
positions kept fixed. The safety violation established
in the prefix persists.
 
All the blocks may be concatenated by Lemma~\ref{lem:concat}(1)(2)(3),
and a single execution violates both predicates, so
Lemma~\ref{lem:boundary-common-failure} gives the claim.
\end{proof}

\subsubsection{\texorpdfstring{SRO $\vee$ VSS (RV)}{SRO or VSS}}
\label{ssub:f2-3}

The whole maximal negative boundary of this pair is the single class
$\FS^{S}$, and that obligation is discharged in the main text by
Proposition~\ref{prop:cf-sro-exchange-family}.

\subsubsection{\texorpdfstring{SRO $\vee$ VEC (RX)}{SRO or VEC}}
\label{ssub:f2-4}

\begin{lemma}
\label{lem:A2-13-FC-S}
$\mathrm{RX}\notin\Sol(\FC^{S})$.
\end{lemma}

\begin{proof}
Let $\mathcal A$ be any deterministic algorithm of $\FC$ under $\SSY$.
The colors follow a common sequence $\gamma_0,g(\gamma_0),\dots$
(Lemma~\ref{lem:rigidA2}) while both robots are activated
simultaneously, and the configuration stays at $\iota$.
 
If $\lambda_{c_j}=0$ for every $j$, no nonzero $\M$ ever occurs, so the
sequence of configurations reduces to $C_0$ and both the perpetuity
condition of $\mathrm{SRO}$ and the attainment condition of
$\mathrm{VEC}$, which would require $d_0\perp d_0$, fail.
 
Otherwise let $t$ be the least $j$ with $\lambda_{c_j}\neq0$, replace
round $t$ by an interval in which $a$ alone is activated $k\ge3$ times
in a row (Lemma~\ref{lem:solo}), and append a fair tail. The visible
state of $a$ is constantly $c_t$, so by
Lemma~\ref{lem:single-rule}(b) it applies the same pair
$(\lambda,\theta)$, with $\lambda\neq0$, at every activation.
 
If $c_1=0$, the robot $a$ lands on the position of $b$ at the first
activation and $d_{i+1}=0$; both $\mathbf{C1}$, $\mathbf{C2}$ and all
the targets permitted by the transition condition of $\mathrm{VEC}$ are
nonzero, so the second conditions of both tasks fail simultaneously.
 
If $c_1\neq0$, the same robot performs two nonzero $\M$ operations with
its peer at rest, so by Lemma~\ref{lem:cf-sro-two-moves} the second
condition of $\mathrm{SRO}$ fails. For $\mathrm{VEC}$ the argument of
Lemma~\ref{lem:A2-6-FC-S} applies: the only single displacements
satisfying the transition condition rotate the segment clockwise by
$45^\circ$ or by $90^\circ$, so if $a$ departs from them the transition
condition fails, and otherwise $a$ reaches a configuration with
$d_i\perp d_0$ in at most two moves and then repeats the same
displacement, so the stopping condition fails. All of these are safety
violations.
 
In each case a single execution violates both predicates, so
Lemma~\ref{lem:boundary-common-failure} gives the claim.
\end{proof}

\begin{lemma}
\label{lem:A2-13-FS-A}
$\mathrm{RX}\notin\Sol(\FS^{A})$.
\end{lemma}

\begin{proof}
Let $\mathcal A$ be any deterministic algorithm of $\FS$ under $\ASY$.
Define the execution by starting, as in the first stage of
Lemma~\ref{lem:dmsd_not_fs_a}, with a finite fully synchronous prefix
that freezes the common state of the two robots, which $\ASY$ admits
(Lemma~\ref{lem:rigidA2}), then placing an observation-during-movement
interval (Lemma~\ref{lem:mid}), and finally appending an arbitrary fair
tail.

If $\lambda_s=0$ for every state $s$, the configuration never changes,
so the sequence $(C_i)$ is finite and $d_0\not\perp d_0$; both the
perpetuity condition of $\mathrm{SRO}$ and the attainment condition of
$\mathrm{VEC}$ fail.

Otherwise, in the round of the first common state $s$ with
$\lambda_s>0$, activate $a$ alone so that it performs a nonzero $\M$,
and let the robot $b$, in state $s$, observe a point $p_\mu$,
$\mu\in(0,1)$, of that movement segment (Lemma~\ref{lem:mid}, which is
admissible for $K=A$).
Delay the Move of $b$ until after $a$ has landed, so that the
landing configuration of $a$ is a joint stop.
Do not reactivate $a$ before $b$ completes this Move.
If this $\M$ of $a$ is neither the clockwise
$90^\circ$ rotation about the peer ($\mathbf{C1}$) nor the clockwise
$90^\circ$ rotation about the midpoint ($\mathbf{C2}$), then by
Lemma~\ref{lem:cf-sro-two-moves} the second condition of
$\mathrm{SRO}$ already fails.
If it is one of them, the robot $b$ performs the displacement
it computes from $p_\mu$, and the distance at the time both robots
are at rest is a non-constant continuous function of $\mu$.
The distances permitted by $\mathrm{SRO}$ belong to the
countable set generated from $d$ by the ratios of $\mathbf{C1}$ and
$\mathbf{C2}$, so by Corollary~\ref{cor:continuum} some $\mu$ avoids
that set and $\Phi_{\mathrm{SRO}}$ fails.

In the same observation, an $\FS$ robot $b$ cannot consult the internal
state of its peer and retains no previously observed angle or distance,
so its displacement is determined by $p_\mu$ alone. The segment
obtained after both landings varies continuously and non-constantly
with $\mu$, whereas the transition condition of $\mathrm{VEC}$ admits
only finitely many targets, so by Corollary~\ref{cor:continuum} some
$\mu$ avoids them. Both excluded sets being countable, a single
$\mu\in(0,1)$ avoids both, and the transition condition of
$\mathrm{VEC}$ fails in the same execution.

Both are safety violations, so they may be concatenated by
Lemma~\ref{lem:concat}(1)(2), and
Lemma~\ref{lem:boundary-common-failure} gives the claim.
\end{proof}

\subsubsection{\texorpdfstring{SRO $\vee$ CFSM (RC)}{SRO or CFSM}}
\label{ssub:f2-5}

\begin{lemma}
\label{lem:A2-14-FC-S}
$\mathrm{RC}\notin\Sol(\FC^{S})$.
\end{lemma}

\begin{proof}
Let $\mathcal A$ be any deterministic algorithm of $\FC$ under $\SSY$.
The colors follow a common sequence $\gamma_0,g(\gamma_0),\dots$
(Lemma~\ref{lem:rigidA2}) while both robots are activated
simultaneously, and the configuration stays at $\iota$.
 
If $\lambda_{c_j}=0$ for every $j$, no nonzero $\M$ ever occurs, so
$\moves{a}=0\neq1$ and the single-move condition of $\CFSM$
fails, while the sequence of configurations reduces to $C_0$ and the
perpetuity condition of $\mathrm{SRO}$ fails.
 
Otherwise let $t$ be the least $j$ with $\lambda_{c_j}\neq0$, replace
round $t$ by an interval in which $a$ alone is activated twice in a row
(Lemma~\ref{lem:solo}), and append a fair tail. The visible state of
$a$ is constantly $c_t$, so by Lemma~\ref{lem:single-rule}(b) it
applies the same pair $(\lambda,\theta)$, with $\lambda\neq0$, at both
activations. If $c_1=0$, the robot $a$ lands on the position of $b$ at
the first activation, so $a(t)=b(t)$ at that time and the
collision-freedom condition of $\CFSM$ fails, while
$d_{i+1}=0$ satisfies neither $\mathbf{C1}$ nor $\mathbf{C2}$ and the
second condition of $\mathrm{SRO}$ fails. If $c_1\neq0$, both
activations are nonzero $\M$ operations, so $\moves{a}\ge2\neq1$ and
the single-move condition fails, while
Lemma~\ref{lem:cf-sro-two-moves} makes the second condition of
$\mathrm{SRO}$ fail as well. All of these are safety violations
witnessed at a finite time, hence preserved by the fair tail
(Lemma~\ref{lem:concat}(1)).
 
In each case a single execution violates both predicates, so
Lemma~\ref{lem:boundary-common-failure} gives the claim.
\end{proof}

\subsection{Disjunctions of the Remaining Exchange Tasks}
\label{app:cf-exchange-pairs}

This subsection treats the eight antichains $\{E,P\}$ with
$E\in\mathcal E=\{\SRDV, \VSS\}$ and
$P\in\mathcal P_4=\{\mathrm{OSPO},\mathrm{ZCC},\mathrm{VEC},
\CFSM\}$, which carry eighteen boundary obligations in total.
The exchange component is always supplied by the common first phase of
Appendix~\ref{app:shared-exchange-phase}: either no nonzero $\M$ ever
occurs, so that Lemma~\ref{lem:p1_no_null} is contradicted, or the
first nonzero $\M$ does not land where Lemma~\ref{lem:p1_first_move}
requires, or it does and Lemma~\ref{lem:p1_replay} carries the peer out
of the admissible region. The companion component is obtained from a
solo activation interval, a delayed-move interval, an
observation-during-movement interval, or a violation of a move-count or
stopping requirement, according to the companion task and the boundary
class.

\subsubsection{\texorpdfstring{OSPO $\vee$ S-RDV (OS)}{OSPO or S-RDV}}
\label{ssub:f3-1}

\begin{lemma}
\label{lem:A2-8-FC-A}
$\mathrm{OS}\notin\Sol(\FC^{A})$.
\end{lemma}

\begin{proof}
Let $\mathcal A$ be any deterministic algorithm of $\FC$ under $\ASY$.
Use coordinates $A=0$, $a_0=d$, $b_0=2d$, $B=3d$ on the initial line.
Run fully synchronous cycles with identical normalized observations
(Lemma~\ref{lem:rigidA2}). While every $\M$ is zero, the positions stay
fixed and the colors stay equal.

\emph{No nonzero output.} If no nonzero output occurs, continue
forever. The robots stay distinct at constant distance. So
$\Phi_{\SRDV}$ fails (no rendezvous), and $\Phi_{\mathrm{OSPO}}$ fails
(no decreasing phase). This execution witnesses both liveness
violations.

\emph{A nonzero output occurs.} Otherwise let $\gamma$ be the common
color at the start of the first round with a nonzero output. Let
$(\lambda,\theta)$, with $\lambda>0$, be the normalized displacement
for the initial configuration with peer color $\gamma$. Since
$\mathcal A$ is deterministic, $(\lambda,\theta)$ is known before that
round is scheduled.

\emph{Case $(\lambda,\theta)\neq(2,0)$.} Replace the round by two solo
activations of $a$ (Lemma~\ref{lem:solo}), followed by a fair tail. The
peer color stays $\gamma$, so $a$ applies $(\lambda,\theta)$ twice
(Lemma~\ref{lem:single-rule}(b)). The first of these $\M$ operations
observes $\iota$ and does not land at $B$. So the exchange phase of
$\SRDV$ is unsatisfiable (Lemma~\ref{lem:p1_first_move}).

As in Lemma~\ref{lem:A2-2-FC-S}, either a collision occurs or
$m_a=2,m_b=0$; either way $\Phi_{\mathrm{OSPO}}$ is violated. All these
violations are safety violations.

\emph{Case $(\lambda,\theta)=(2,0)$.} Let both robots Look, and then
both Compute; their colors become a common $\gamma'$. Delay the $\M$ of
$b$ so that it begins strictly after the $\M$ of $a$ ends
(Corollary~\ref{cor:stale}). The robot $a$ moves from $d$ to $3d$ while
$b$ rests at $2d$, so they collide, and $\Phi_{\mathrm{OSPO}}$ fails.

Then let $b$ complete its $\M$ to $A$. The robots are now at $(B,A)$
with common color $\gamma'$, and neither has a pending operation. Let
$\mathcal A_{\gamma'}$ be $\mathcal A$ with initial color $\gamma'$. The
impossibility of stable $\mathrm{RDV}$ in $\FC^{A}$ used in
Lemma~\ref{lem:srdv_not_fc_a} gives a counterexecution of
$\mathcal A_{\gamma'}$. Transport it to $(B,A)$ by an
orientation-preserving similarity, as in Lemma~\ref{lem:A2-3-FC-A}, and
append it. This continuation violates the terminal rendezvous
condition of $\SRDV$ and witnesses that liveness violation. The safety
violation of $\Phi_{\mathrm{OSPO}}$ precedes it.

Lemma~\ref{lem:concat}(1)(2)(3) is satisfied in every case, and
Lemma~\ref{lem:boundary-common-failure} gives the claim.
\end{proof}

\begin{lemma}
\label{lem:A2-8-FS-S}
$\mathrm{OS}\notin\Sol(\FS^{S})$.
\end{lemma}

\begin{proof}
Let $\mathcal A$ be any deterministic algorithm of $\FS$ under $\SSY$.
Use coordinates $A=0$, $a_0=d$, $b_0=2d$, $B=3d$ on the initial line,
and normalized frames at every $\LK$ at which the robots are distinct.
The next color and the displacement of a robot then depend only on its
own color.
 
\emph{First part.} Activate $\{a\},\{b\},\{a\},\dots$, which $\SSY$
admits. By the $\pi$-symmetry of $\iota$ and Lemma~\ref{lem:silent},
both robots follow $c_0,g(c_0),\dots$ until the first nonzero $\M$.
 
If there is no nonzero $\M$, continue forever. The robots stay at
$\iota$, so they never meet; $\Phi_{\SRDV}$ fails by
Lemma~\ref{lem:p1_no_null}, and $\Phi_{\mathrm{OSPO}}$ fails because
$d$ is constant. This execution witnesses both liveness violations.
 
Otherwise let $k$ be least with $\lambda_{g^k(c_0)}\neq0$. The robot
$a$ performs the first nonzero $\M$ at its $(k+1)$-st activation; let
that $\M$ complete, whether or not the two robots meet during it.
\begin{itemize}
  \item If its landing point is not $B$, $\Phi_{\SRDV}$ is
  unsatisfiable (Lemma~\ref{lem:p1_first_move}).
  \item If it is $B$, let $b$ be activated next; it lands at
  $4d\notin[A,B]$ (Lemma~\ref{lem:p1_replay}), so $\Phi_{\SRDV}$
  fails.
\end{itemize}
Either way $\Phi_{\SRDV}$ suffers a safety violation at a finite time.
This argument uses only the landing points, so it does not assume that
the two robots stay distinct.
 
If a collision has occurred in the prefix constructed so far,
the collision-freedom condition of $\mathrm{OSPO}$ has already
failed. Both predicates have then suffered safety violations,
so append an arbitrary fair, time-divergent continuation.
 
Otherwise the first move of $a$ did not land at $B$, since
a move from $d$ to $3d$ would pass through the stationary
robot at $2d$. The safety violation of $\Phi_{\SRDV}$ has
therefore already occurred, and $b$ has not yet performed
its $(k+1)$-st activation.
 
Activate $b$ once. If a collision occurs during this cycle,
both predicates again have safety violations and an arbitrary
fair tail completes the construction. Otherwise let $T$ be
the end of this cycle. Both robots have observed distinct
positions, so they now carry
$\gamma=g^{k+1}(c_0)$, neither has a pending operation,
and $m_a(T)=m_b(T)=1$.
 
Every execution satisfying $\Phi_{\mathrm{OSPO}}$ has
$m_a(t_j)=m_b(t_j)=j$ at its phase boundaries.
Consequently $|m_a-m_b|\le1$ at all times, and both counts
are unbounded.
 
Let $J_\gamma=\{i\ge0:\lambda_{g^i(\gamma)}\ne0\}$.
If $|J_\gamma|\ge2$, let $i'$ be its second element and
activate $a$ alone for $i'+1$ cycles, stopping this
construction at the first collision if one occurs.
A collision violates $\Phi_{\mathrm{OSPO}}$; otherwise
the resulting counts are $m_a=3$, $m_b=1$, which also
violate it. Append an arbitrary fair tail.
 
If $|J_\gamma|\le1$, use a fully synchronous fair tail.
A collision again violates $\Phi_{\mathrm{OSPO}}$.
If no collision occurs, each robot performs at most one
further nonzero move, contradicting the required
unboundedness of the counts.
 
Thus one fair infinite execution violates both predicates,
and Lemma~\ref{lem:boundary-common-failure} gives the claim.
\end{proof}

\begin{lemma}
\label{lem:A2-8-OB-F}
$\mathrm{OS}\notin\Sol(\OB^{F})$.
\end{lemma}

\begin{proof}
Let $\mathcal A$ be any deterministic algorithm of $\OB$ under $\FSY$.
Use coordinates $A=0$, $a_0=d$, $b_0=2d$, $B=3d$ on the initial line
and take the infinite fully synchronous execution of
Lemma~\ref{lem:rigidA1} from this instance, with a normalized frame at
every $\LK$. The visible state of an $\OB$ robot is unique, so the rule
is a single pair $(\lambda,\theta)$ (Lemma~\ref{lem:single-rule}(a)),
and by Lemma~\ref{lem:ratios}(ii) the distance at the end of round~$k$
is $d(k)=c_2^{\,k}d(0)$.
 
\emph{Violation of $\Phi_{\mathrm{OSPO}}$.}
If the two robots ever occupy the same point, the collision-freedom
condition fails at a finite time; this is a safety violation. It covers
$c_2=0$, where they coincide at the end of round~$1$ and the
coincidence is maintained by Lemma~\ref{lem:keep-coincident}, and also
the case in which they meet while travelling inside a round.
 
Otherwise $d(t)>0$ throughout. By Lemma~\ref{lem:rigidA1} every joint
stop of the execution carries one of the round-boundary distances
$d(k)$, so the restriction of $d$ to $\Sab$ is monotone. The turning
times required by $\mathrm{OSPO}$ are joint stops, so by the criterion
stated after Definition~\ref{def:osp} no strictly decreasing phase can
be followed by a strictly increasing one, and $\Phi_{\mathrm{OSPO}}$
fails. This is a liveness violation, and the fully synchronous
execution is itself the fair infinite continuation along which it
persists.
 
\emph{Violation of $\Phi_{\SRDV}$.}
If $\lambda=0$ no nonzero $\M$ ever occurs, so the first nonzero $\M$
demanded by Lemma~\ref{lem:p1_no_null} does not exist. This is a
progress violation, witnessed by the same infinite fully synchronous
execution. If $\lambda\neq0$, both robots perform a nonzero $\M$ in
round~$1$, so an admissible exchange configuration must be
$(B,A)=(3d,0)$ with move counts $(1,1)$. When
$(\lambda,\theta)\neq(2,0)$ the landing of round~$1$ is not this
configuration, so the requirement of Lemma~\ref{lem:p1_first_move} is
unmet. When $(\lambda,\theta)=(2,0)$ round~$1$ does land at $(3d,0)$,
but round~$2$ applies the same rule and sends $a$ to $3d+2(0-3d)=-3d$
and $b$ to $0+2(3d-0)=6d$, both outside $[A,B]=[0,3d]$
(Lemma~\ref{lem:srdv_not_ob_f}). These are safety violations witnessed
at a finite time.
 
In each case a single execution violates both predicates, so
Lemma~\ref{lem:concat}(1)(3) is satisfied and
Lemma~\ref{lem:boundary-common-failure} gives the claim.
\end{proof}

\subsubsection{\texorpdfstring{OSPO $\vee$ VSS (OV)}{OSPO or VSS}}
\label{ssub:f3-2}

\begin{lemma}
\label{lem:A2-9-FS-S}
$\mathrm{OV}\notin\Sol(\FS^{S})$.
\end{lemma}

\begin{proof}
Let $\mathcal A$ be any deterministic algorithm of $\FS$ under $\SSY$.
Use the execution of Lemma~\ref{lem:A2-8-FS-S}. It violates
$\Phi_{\mathrm{OSPO}}$: in that proof the violation comes from a
collision up to $T$, from the move-count requirement $(\ast)$ in the
second part, or, if no nonzero $\M$ occurs, from the constancy of $d$,
and none of these refers to the exchange task. We show that the same
execution also violates $\Phi_{\VSS}$, where $P=\{0,d,2d,3d\}$.

If no nonzero $\M$ occurs, $(\mathrm{P1})$ fails by
Lemma~\ref{lem:p1_no_null}; the execution of
Lemma~\ref{lem:A2-8-FS-S} itself is the fair continuation witnessing
this liveness violation.

Otherwise, as in that lemma, the two robots traverse the same states
until the first nonzero $\M$ and $a$ is activated first, so that $\M$
is performed by $a$ and its $\LK$ observes $\iota$.
\begin{itemize}
  \item If it lands outside $P$, $(\mathrm{P1\text{-}b})$ fails.
  \item If it lands in $P\setminus\{B\}$, $(\mathrm{P1\text{-}c})$ is
  unsatisfiable (Lemma~\ref{lem:p1_first_move}).
  \item If it lands at $B$, then $b$ lands at $4d\notin P$
  (Lemma~\ref{lem:p1_replay}), and $(\mathrm{P1\text{-}b})$ fails.
\end{itemize}
In each case this is a safety violation at a finite time, whether or
not a collision occurs, since $\VSS$ permits passing.

In either case a single execution violates both predicates;
Lemma~\ref{lem:concat}(1)(3) is satisfied, and
Lemma~\ref{lem:boundary-common-failure} gives the claim.
\end{proof}

\begin{lemma}
\label{lem:A2-9-OB-F}
$\mathrm{OV}\notin\Sol(\OB^{F})$.
\end{lemma}

\begin{proof}
Let $\mathcal A$ be any deterministic algorithm of $\OB$ under $\FSY$.
Take the fully synchronous execution of Lemma~\ref{lem:rigidA1} with
normalized frames. By Lemma~\ref{lem:single-rule}(a), both robots apply
one and the same pair $(\lambda,\theta)$ in every round. Use $A=0$,
$a_0=d$, $b_0=2d$, $B=3d$ and $P=\{0,d,2d,3d\}$.

\emph{OSPO.} If a collision occurs, $\Phi_{\mathrm{OSPO}}$ fails.
Otherwise the joint stops lie between rounds when $\lambda>0$, and the
distance there is $c_2^{\,k}d(0)$ (Lemma~\ref{lem:ratios}(ii)). This is
monotone, so no decreasing phase can be followed by an increasing one.
The execution itself witnesses this liveness violation.

\emph{VSS.} We distinguish three cases.
\begin{itemize}
  \item If $\lambda=0$, $(\mathrm{P1})$ fails by
  Lemma~\ref{lem:p1_no_null}. This is a liveness violation of the same
  execution.
  \item If $\lambda>0$ and $(\lambda,\theta)\neq(2,0)$, the round-$1$
  $\M$ of $a$ completes earliest and does not land at $B$. So
  $\Phi_{\VSS}$ is unsatisfiable (Lemma~\ref{lem:p1_first_move}).
  \item If $(\lambda,\theta)=(2,0)$, round $1$ yields $(3d,0)$ and
  round $2$ yields $(-3d,6d)$. Neither point lies in $P$, so
  $(\mathrm{P1\text{-}b})$ fails.
\end{itemize}
The last two violations are safety violations.

Lemma~\ref{lem:concat}(1)(3) is satisfied, and
Lemma~\ref{lem:boundary-common-failure} gives the claim.
\end{proof}

\subsubsection{\texorpdfstring{S-RDV $\vee$ ZCC (SZ)}{S-RDV or ZCC}}
\label{ssub:f3-3}

\begin{lemma}
\label{lem:A2-15-FC-A}
$\mathrm{SZ}\notin\Sol(\FC^{A})$.
\end{lemma}

\begin{proof}
Let $\mathcal A$ be any deterministic algorithm of $\FC$ under $\ASY$.
Use coordinates $A=0$, $a_0=d$, $b_0=2d$, $B=3d$ on the initial line.
Run fully synchronous cycles with identical normalized observations
(Lemma~\ref{lem:rigidA2}). While every Move is zero, the positions stay
fixed and the colors stay equal.
 
\emph{No nonzero output.} Continue forever. The robots stay at $\iota$,
so $\Phi_{\SRDV}$ fails, while $f_r\equiv0$ for both robots, so the
zig-zag condition admits no strictly increasing pair and
$\Phi_{\mathrm{ZCC}}$ fails. This execution witnesses both liveness
violations.
 
\emph{A nonzero output occurs.} Let $\gamma$ be the common color at the
start of the first round with a nonzero output, and let
$(\lambda,\theta)$, with $\lambda>0$, be the normalized displacement
for the initial configuration with peer color $\gamma$. Since
$\mathcal A$ is deterministic, $(\lambda,\theta)$ is known before that
round is scheduled.
 
\emph{Case $(\lambda,\theta)=(2,0)$.} Let both robots Look, and then
both Compute; their colors become a common $\gamma'$. Delay the Move of
$b$ so that it begins strictly after the Move of $a$ ends
(Corollary~\ref{cor:stale}). The robot $a$ moves from $d$ to $3d$ while
$b$ rests at $2d$, so at the landing
\[
  \|a-a_0\|=2d>d=\|b-a_0\| ,
\]
and the collision-freedom condition of $\mathrm{ZCC}$ fails at a finite
time, a safety violation.
 
Then let $b$ complete its Move to $A$. The robots are at $(B,A)$ with
common color $\gamma'$, and neither has a pending operation. As in
Lemma~\ref{lem:A2-8-FC-A}, the impossibility of stable RDV in $\FC^A$
used in Lemma~\ref{lem:srdv_not_fc_a}, applied to $\mathcal A$ with
initial color $\gamma'$, supplies a counterexecution that an
orientation-preserving similarity transports to $(B,A)$. Appending it
violates the terminal rendezvous condition of $\SRDV$. The safety
violation of $\mathrm{ZCC}$ precedes that liveness violation.
 
\emph{Case $(\lambda,\theta)\neq(2,0)$.} Replace the round by a solo
activation interval of $a$ (Lemma~\ref{lem:solo}), during which $b$ is
never activated and keeps the color $\gamma$; by
Lemma~\ref{lem:single-rule}(b) the robot $a$ applies $(\lambda,\theta)$
at every activation and $b$ is at rest throughout. The interval may be
made arbitrarily long, the activations of $b$ being placed far enough
in the future to keep the execution fair. Every activation of $a$ in
it is a complete cycle, and $b$ has no pending operation.
 
The first of these Move operations observes $\iota$ and does not land
at $B$, so the exchange phase of $\SRDV$ is unsatisfiable
(Lemma~\ref{lem:p1_first_move}), a safety violation at a finite time.
 
For $\Phi_{\mathrm{ZCC}}$, put
$\sigma_1:=1-\lambda e^{\mathrm{i}\theta}$ and $c_1:=|\sigma_1|$.
 
If $c_1=0$, the robot $a$ reaches the position of $b$ at the first
activation, and neither robot has a pending operation. The robots
coincide; use the coincidence continuation described in the proof of
Lemma~\ref{lem:zcc_not_fc_s}; it activates one robot at a time with
complete cycles, which $\ASY$ admits. This either keeps both positions
fixed forever, so that the zig-zag condition fails, or produces a
finite violation of the collision-freedom condition.
 
Assume $c_1>0$, and put $\psi:=\arg\sigma_1$
(Lemma~\ref{lem:ratios}(iii)). Writing $w_0$ for the vector from $a$
to $b$ at the beginning of the interval, the vector after $j$
activations is $\sigma_1^{\,j}w_0\neq0$, so every $\LK$ of $a$ in the
interval observes distinct positions, and
\[
  f_a(j)=\|a(j)-a(0)\|=d\,\bigl|1-c_1^{\,j}e^{\mathrm{i}j\psi}\bigr| .
\]
Since $b$ stands still, $\|b-a_0\|=d$, and the collision-freedom
condition requires $f_a(j)\le d$, that is $c_1^{\,j}\le2\cos(j\psi)$.
Note that $\sigma_1=1$ would force $\lambda=0$, so $c_1\neq1$ whenever
$\psi\equiv0$.
\begin{itemize}
\item If $\psi\not\equiv0\pmod{2\pi}$, Lemma~\ref{lem:ratios}(iii)
  provides $k\ge1$ with $\cos(k\psi)\le-\tfrac12<0$. Making the
  interval long enough to reach the $k$-th activation, the
  collision-freedom condition fails there, a safety violation.
\item If $\psi\equiv0$ and $c_1>1$, that is $\theta=\pi$, then
  $f_a(j)=(c_1^{\,j}-1)d$ exceeds $d$ as soon as $c_1^{\,j}>2$, so the
  collision-freedom condition again fails after finitely many
  activations.
\item In the remaining contraction case $\psi\equiv0$ and $0<c_1<1$,
  that is $\sigma_1\in(0,1)$, end the solo interval after its first
  activation. Then $a$ lies in the open segment $(a_0,b_0)$ and $b$ is
  still at $b_0$, so the robots still lie in the original initial
  segment in their original order, and neither has a pending
  operation. Apply the continuation property proved in
  Lemma~\ref{lem:zcc_not_fc_s}, keeping the original initial positions
  fixed; its schedule activates one robot at a time with complete
  cycles, which $\ASY$ admits. It supplies a fair, time-divergent
  continuation violating $\Phi_{\mathrm{ZCC}}$, while preserving the
  safety violation already obtained in the prefix.
\end{itemize}
 
In every case a single execution violates both predicates, the safety
violations preceding the liveness ones, so
Lemma~\ref{lem:concat}(1)(3) is satisfied and
Lemma~\ref{lem:boundary-common-failure} gives the claim.
\end{proof}

\begin{lemma}
\label{lem:A2-15-FS-AM}
$\mathrm{SZ}\notin\Sol(\FS^{A_M})$.
\end{lemma}

\begin{proof}
Let $\mathcal A$ be any deterministic algorithm of $\FS$ under $A_M$.
Synchronous activations of one robot at a time are admissible under
$A_M$ (Lemma~\ref{lem:solo}), so $a$ and $b$ may be activated
alternately, each cycle being completed before the next one begins,
and, under normalized frames, the argument of
Lemma~\ref{lem:p1_replay} applies unchanged.
 
If no nonzero $\M$ ever occurs, which in $\FS$ is equivalent to
$\lambda_{g^j(c_0)}=0$ for every $j$ since the color sequence is
eventually periodic (Lemma~\ref{lem:pumpB2}), the two robots stay at
$\iota$ at distinct positions, so $\Phi_{\SRDV}$ fails
and the convergence condition of $\mathrm{ZCC}$ fails as well.
 
Otherwise the first nonzero $\M$ is performed by $a$. Either it does
not land at $B=3d$, so that the first-phase clauses become
unsatisfiable by the argument of Lemma~\ref{lem:p1_first_move}, or it
does and Lemma~\ref{lem:p1_replay} makes the robot $b$, activated next,
land at $4d\notin[A,B]$; in either case
$\Phi_{\SRDV}$ fails, a safety violation. Let the prefix end with the
cycle that establishes this violation, that of $a$ in the first case
and that of $b$ in the second. Every $\LK$ in this prefix observes
distinct positions: before the first nonzero $\M$ the configuration is
$\iota$, in the second case $b$ observes $a$ at $B\neq2d$, and a
landing on the peer can only be the last $\M$ of the prefix. Since the
visible state of an $\FS$ robot is its own color,
Lemma~\ref{lem:norm}(1) shows that each robot that has been activated
$k$ times carries $g^{k}(c_0)$.
 
End the prefix at this joint stop, with no pending operation.
If the robots coincide, use the coincidence continuation
from Lemma~\ref{lem:zcc_not_fc_s}. Otherwise their colors lie on the
same normalized color orbit, so the continuation statement
proved in Lemma~\ref{lem:zcc_not_fs_am} applies.
In either case it supplies a fair, time-divergent continuation
violating $\Phi_{\mathrm{ZCC}}$, with the original initial
positions kept fixed. The safety violation established
in the prefix persists (Lemma~\ref{lem:concat}(1)).
 
A single fair infinite execution therefore violates both predicates,
and Lemma~\ref{lem:boundary-common-failure} gives the claim.
\end{proof}

\begin{lemma}
\label{lem:A2-15-OB-F}
$\mathrm{SZ}\notin\Sol(\OB^{F})$.
\end{lemma}

\begin{proof}
Let $\mathcal A$ be any deterministic algorithm of $\OB$ under $\FSY$
and take the infinite fully synchronous execution of
Lemma~\ref{lem:rigidA1} with a normalized frame at every $\LK$. The
visible state of an $\OB$ robot is unique, so the rule is a single pair
$(\lambda,\theta)$ (Lemma~\ref{lem:single-rule}(a)).

For this execution $\Phi_{\mathrm{ZCC}}$ fails by
Lemma~\ref{lem:zcc_not_ob_f}, which exhausts the degenerate cases
$\lambda=0$ and $c_2=0$, the violation of the collision-freedom
condition when $\varphi\not\equiv0$, and the violation of the zig-zag
condition when $\varphi\equiv0$.

For $\Phi_{\SRDV}$ the argument is that of
Lemma~\ref{lem:A2-8-OB-F}: if $\lambda=0$ no nonzero $\M$ occurs and
the first nonzero $\M$ demanded by Lemma~\ref{lem:p1_no_null} does not
exist; if $\lambda\neq0$ and $(\lambda,\theta)\neq(2,0)$, the nonzero
$\M$ of round~$1$ does not meet the requirement of
Lemma~\ref{lem:p1_first_move}; and if $(\lambda,\theta)=(2,0)$, then in
round~$2$ the robots land at $-3d$ and $6d$, leaving
$[A,B]=[0,3d]$ (Lemma~\ref{lem:srdv_not_ob_f}).

A single execution violates both predicates, so
Lemma~\ref{lem:boundary-common-failure} gives the claim.
\end{proof}


\subsubsection{\texorpdfstring{S-RDV $\vee$ VEC (SX)}{S-RDV or VEC}}
\label{ssub:f3-4}

\begin{lemma}
\label{lem:A2-16-FC-A}
$\mathrm{SX}\notin\Sol(\FC^{A})$.
\end{lemma}

\begin{proof}
Let $\mathcal A$ be any deterministic algorithm of $\FC$ under $\ASY$.
Use coordinates $A=0$, $a_0=d$, $b_0=2d$, $B=3d$ on the initial line.
The initial colors being common, while both robots are activated
simultaneously their colors follow a common eventually periodic
sequence $c_0,g(c_0),\dots$ (Lemma~\ref{lem:rigidA2}) and the
configuration stays at $\iota$.

If $\lambda_{c_j}=0$ for every $j$, continue forever. This fair
execution contains no nonzero $\M$, so the two robots never occupy one
and the same point permanently and $\Phi_{\SRDV}$ fails, while the
sequence of joint-stop configurations reduces to $C_0$ and
$d_0\not\perp d_0$, so the attainment condition of $\mathrm{VEC}$ fails
as well. This execution witnesses both liveness violations.

Otherwise let $\gamma$ be the common color at the start of the first
round with a nonzero output, and let $(\lambda,\theta)$, with
$\lambda\neq0$, be the normalized displacement for the initial
configuration with peer color $\gamma$. Since $\mathcal A$ is
deterministic, $(\lambda,\theta)$ is known before that round is
scheduled.

\emph{Case $(\lambda,\theta)=(2,0)$.} Let both robots Look, and then
both Compute; their colors become a common $\gamma'$. Delay the Move of
$b$ so that it begins strictly after the Move of $a$ ends
(Corollary~\ref{cor:stale}). At the landing of $a$ neither robot is
inside a $\M$, although $b$ still has a pending one, so that instant is
a joint stop and realizes $C_1=(3d,2d)$, whence $d_1=-d_0$. This is
neither the clockwise rotation by $45^\circ$ nor the one by $90^\circ$,
so the transition condition of $\mathrm{VEC}$ fails at $i=0$, a safety
violation at a finite time.

Then let $b$ complete its Move to $A$. The robots are at $(B,A)$ with
common color $\gamma'$, and neither has a pending operation. As in
Lemma~\ref{lem:A2-3-FC-A}, the impossibility of stable RDV in $\FC^A$
used in Lemma~\ref{lem:srdv_not_fc_a}, applied to $\mathcal A$ with
initial color $\gamma'$, supplies a counterexecution that an
orientation-preserving similarity transports to $(B,A)$. Appending it
violates the terminal rendezvous condition of $\SRDV$. The safety
violation of $\mathrm{VEC}$ precedes that liveness violation.

\emph{Case $(\lambda,\theta)\neq(2,0)$.} Replace the round by an
interval in which $a$ alone is activated $k\ge3$ times in a row
(Lemma~\ref{lem:solo}), followed by an arbitrary admissible fair,
time-divergent tail. The visible state of $a$ is constantly $\gamma$,
so by Lemma~\ref{lem:single-rule}(b) it applies the same pair
$(\lambda,\theta)$ at every activation.

The first of these Move operations observes $\iota$, and the cumulative
nonzero-move counts being $(1,0)$, an admissible exchange configuration
would require $a$ to land at $B$. It does not, so the exchange phase of
$\SRDV$ is unsatisfiable (Lemma~\ref{lem:p1_first_move}), a safety
violation at a finite time.

For $\mathrm{VEC}$: if $c_1=0$, the robot $a$ lands on the position of
$b$ at the first activation and $d_1=0$, while all the targets
permitted by the transition condition are nonzero, so that condition
fails. If $c_1\neq0$, then as in Lemma~\ref{lem:A2-6-FC-S} the only
single displacements satisfying the transition condition rotate the
segment clockwise by $45^\circ$ or by $90^\circ$; if $a$ departs from
them the transition condition fails, and otherwise $a$ reaches a
configuration with $d_i\perp d_0$ in at most two moves and then repeats
the same displacement, so the stopping condition fails. Either way this
is a safety violation in the prefix, preserved by the tail.

In every case a single execution violates both predicates, so
Lemma~\ref{lem:concat}(1)(3) is satisfied and
Lemma~\ref{lem:boundary-common-failure} gives the claim.
\end{proof}

\begin{lemma}
\label{lem:A2-16-FS-A}
$\mathrm{SX}\notin\Sol(\FS^{A})$.
\end{lemma}

\begin{proof}
Let $\mathcal A$ be any deterministic algorithm of $\FS$ under $\ASY$,
and let $E_0$ be the execution of alternating singleton activations
$\{a\},\{b\},\dots$ under normalized frames (Lemma~\ref{lem:solo});
$E_0$ is fair and is determined by $\mathcal A$.

If $E_0$ contains no nonzero $\M$, the two robots stay at $\iota$, so
$\Phi_{\SRDV}$ fails, while the sequence of configurations reduces to
$C_0$ and $d_0\not\perp d_0$, so the attainment condition of
$\mathrm{VEC}$ fails as well; Lemma~\ref{lem:boundary-common-failure}
gives the claim.

Otherwise, as long as no nonzero $\M$ has occurred the two robots
traverse the same sequence of states, so the first nonzero $\M$ is
performed by $a$, which is activated first, from some state $s$ with
$\lambda_s>0$, while $b$ is in the same state $s$ and has no pending
operation. Let $p_E$ be its landing point.

\emph{Case $p_E=B$.} Follow $E_0$ up to the round following this $\M$.
When $a$ has landed, $d_1=b_0-B=-d_0$, which is none of the successors
admitted by the transition condition of $\mathrm{VEC}$; and
Lemma~\ref{lem:p1_replay} makes $b$, activated next, land at
$4d\notin[A,B]$, so $\Phi_{\SRDV}$ fails.

\emph{Case $p_E\neq B$.} Follow $E_0$ up to the start of this $\M$, and
by Lemma~\ref{lem:mid} let $b$ perform its $\LK$ while $a$ is at
$p_\mu=(1-\mu)a_0+\mu p_E$, $\mu\in(0,1)$, delaying the $\M$ of $b$
until $a$ has landed. The $\M$ of $a$ is then the nonzero $\M$ with
earliest completion time and it does not land at $B$, so
$\Phi_{\SRDV}$ fails by Lemma~\ref{lem:p1_first_move}. By
Lemma~\ref{lem:norm}(1) the robot $b$ lands at
$b_0+\lambda_s R_{\theta_s}(p_\mu-b_0)$, an affine function of $\mu$
with nonzero slope $\lambda_s R_{\theta_s}(p_E-a_0)$, so a suitable
coordinate of the segment after both landings is continuous and
non-constant in $\mu$. The transition condition of $\mathrm{VEC}$
admits only finitely many successors, so by
Corollary~\ref{cor:continuum} some $\mu$ avoids them and the
transition condition fails (cf.\ Lemma~\ref{lem:vec_not_fs_a}).

In both cases the two violations are safety violations forced by a
finite prefix, which extends to a fair execution by
Lemma~\ref{lem:concat}(1)(2); Lemma~\ref{lem:boundary-common-failure}
gives the claim.
\end{proof}

\begin{lemma}
\label{lem:A2-16-OB-F}
$\mathrm{SX}\notin\Sol(\OB^{F})$.
\end{lemma}

\begin{proof}
Let $\mathcal A$ be any deterministic algorithm of $\OB$ under $\FSY$
and take the fully synchronous execution in which both robots are
activated in every round (Lemmas~\ref{lem:rigidA1}
and~\ref{lem:rigidA2}), with a normalized frame at every $\LK$. The
violation of $\Phi_{\SRDV}$ is that of
Lemma~\ref{lem:A2-8-OB-F}. For $\Phi_{\mathrm{VEC}}$ the argument of
Lemma~\ref{lem:vec_not_ob_f} applies to the same execution: if
$\lambda=0$ the attainment condition fails, if $c_2=0$ the transition
condition fails, and otherwise the robots move again immediately after
any orthogonal configuration, so the stopping condition fails, or no
such configuration is ever reached and the attainment condition fails.

A single execution violates both predicates, so
Lemma~\ref{lem:boundary-common-failure} gives the claim.
\end{proof}

\subsubsection{\texorpdfstring{S-RDV $\vee$ CFSM (SC)}{S-RDV or CFSM}}
\label{ssub:f3-5}

\begin{lemma}
\label{lem:A2-17-FC-A}
$\mathrm{SC}\notin\Sol(\FC^{A})$.
\end{lemma}

\begin{proof}
Let $\mathcal A$ be any deterministic algorithm of $\FC$ under $\ASY$.
Use coordinates $A=0$, $a_0=d$, $b_0=2d$, $B=3d$ on the initial line.
The colors follow a common sequence $c_0,g(c_0),\dots$
(Lemma~\ref{lem:rigidA2}) while both robots are activated
simultaneously, and the configuration stays at $\iota$.

If $\lambda_{c_j}=0$ for every $j$, continue forever. This fair
execution contains no nonzero $\M$, so $\moves{a}=0\neq1$ and the
single-move condition of $\CFSM$ fails, while the two robots never
occupy one and the same point permanently and $\Phi_{\SRDV}$ fails as
well. This execution witnesses both liveness violations.

Otherwise let $\gamma$ be the common color at the start of the first
round with a nonzero output, and let $(\lambda,\theta)$, with
$\lambda\neq0$, be the normalized displacement for the initial
configuration with peer color $\gamma$.

\emph{Case $(\lambda,\theta)=(2,0)$.} Let both robots Look, and then
both Compute; their colors become a common $\gamma'$. Delay the Move of
$b$ so that it begins strictly after the Move of $a$ ends
(Corollary~\ref{cor:stale}). The robot $a$ moves from $d$ to $3d$ while
$b$ rests at $2d$, so the two robots occupy the same point while $a$
travels and the collision-freedom condition of $\CFSM$ fails, a safety
violation at a finite time.

Then let $b$ complete its Move to $A$. The robots are at $(B,A)$ with
common color $\gamma'$, and neither has a pending operation. As in
Lemma~\ref{lem:A2-3-FC-A}, the impossibility of stable RDV in $\FC^A$
used in Lemma~\ref{lem:srdv_not_fc_a}, applied to $\mathcal A$ with
initial color $\gamma'$, supplies a counterexecution that an
orientation-preserving similarity transports to $(B,A)$. Appending it
violates the terminal rendezvous condition of $\SRDV$. The safety
violation of $\CFSM$ precedes that liveness violation.

\emph{Case $(\lambda,\theta)\neq(2,0)$.} Replace the round by an
interval in which $a$ alone is activated twice in a row
(Lemma~\ref{lem:solo}), followed by an arbitrary admissible fair,
time-divergent tail. The visible state of $a$ is constantly $\gamma$,
so by Lemma~\ref{lem:single-rule}(b) it applies the same pair
$(\lambda,\theta)$ at both activations.

The first of these Move operations observes $\iota$, and the cumulative
nonzero-move counts being $(1,0)$, an admissible exchange configuration
would require $a$ to land at $B$. It does not, so the exchange phase of
$\SRDV$ is unsatisfiable (Lemma~\ref{lem:p1_first_move}), a safety
violation at a finite time.

For $\CFSM$: if $c_1=0$, the robot $a$ lands on the position of $b$ at
the first activation, so $a(t)=b(t)$ there and the collision-freedom
condition fails. If $c_1\neq0$, both activations are nonzero $\M$
operations, so $\moves{a}\ge2\neq1$ and the single-move condition
fails. Either way this is a safety violation in the prefix, preserved
by the tail.

In every case a single execution violates both predicates, so
Lemma~\ref{lem:concat}(1)(3) is satisfied and
Lemma~\ref{lem:boundary-common-failure} gives the claim.
\end{proof}

\begin{lemma}
\label{lem:A2-17-OB-F}
$\mathrm{SC}\notin\Sol(\OB^{F})$.
\end{lemma}

\begin{proof}
Let $\mathcal A$ be any deterministic algorithm of $\OB$ under $\FSY$
and take the fully synchronous execution in which both robots are
activated in every round (Lemmas~\ref{lem:rigidA1}
and~\ref{lem:rigidA2}), with a normalized frame at every $\LK$. The
violation of $\Phi_{\SRDV}$ is that of
Lemma~\ref{lem:A2-8-OB-F}. For $\Phi_{\CFSM}$,
Lemma~\ref{lem:rigidA1} leaves only three possibilities, namely
$\moves{r}=0$, a permanent coincidence of the two robots, and
$\moves{r}=\infty$, each of which violates the single-move condition or
the collision-freedom condition of $\CFSM$.

A single execution violates both predicates, so
Lemma~\ref{lem:boundary-common-failure} gives the claim.
\end{proof}

\subsubsection{\texorpdfstring{ZCC $\vee$ VSS (ZV)}{ZCC or VSS}}
\label{ssub:f3-6}

\begin{lemma}
\label{lem:A2-18-FS-AM}
$\mathrm{ZV}\notin\Sol(\FS^{A_M})$.
\end{lemma}

\begin{proof}
Let $\mathcal A$ be any deterministic algorithm of $\FS$ under $A_M$.
Synchronous activations of one robot at a time are admissible
(Lemma~\ref{lem:solo}), so begin with alternating singleton
activations, each cycle being completed before the next one begins,
under normalized frames.
 
If no nonzero $\M$ ever occurs, condition $(\mathrm{P1})$ of
$\VSS$ fails by Lemma~\ref{lem:p1_no_null}, and the two robots
remain at distinct positions, so the convergence condition of
$\mathrm{ZCC}$ fails as well.
 
Otherwise the first nonzero $\M$ is performed by $a$. Either it does
not land at $B=3d$, so that the first-phase clauses become
unsatisfiable by the argument of Lemma~\ref{lem:p1_first_move}, or it
does and Lemma~\ref{lem:p1_replay} makes the robot $b$, activated next,
land at $4d\notin P$; in either case the endpoint condition of
$\VSS$ fails, a safety violation. Let the prefix end with the cycle
that establishes this violation, that of $a$ in the first case and
that of $b$ in the second. As in the proof of
Lemma~\ref{lem:A2-15-FS-AM}, every $\LK$ in this prefix observes
distinct positions, so by Lemma~\ref{lem:norm}(1) each robot that has
been activated $k$ times carries $g^{k}(c_0)$.
 
End the prefix at this joint stop, with no pending operation.
If the robots coincide, use the coincidence continuation
from Lemma~\ref{lem:zcc_not_fc_s}. Otherwise their colors lie on the
same normalized color orbit, so the continuation statement
proved in Lemma~\ref{lem:zcc_not_fs_am} applies.
In either case it supplies a fair, time-divergent continuation
violating $\Phi_{\mathrm{ZCC}}$, with the original initial
positions kept fixed. The safety violation established
in the prefix persists (Lemma~\ref{lem:concat}(1)).
 
A single fair infinite execution therefore violates both predicates,
and Lemma~\ref{lem:boundary-common-failure} gives the claim.
\end{proof}

\begin{lemma}
\label{lem:A2-18-OB-F}
$\mathrm{ZV}\notin\Sol(\OB^{F})$.
\end{lemma}

\begin{proof}
Let $\mathcal A$ be any deterministic algorithm of $\OB$ under $\FSY$
and take the fully synchronous execution in which both robots are
activated in every round (Lemmas~\ref{lem:rigidA1}
and~\ref{lem:rigidA2}), with a normalized frame at every $\LK$. The
violation of $\Phi_{\mathrm{ZCC}}$ is that of
Lemma~\ref{lem:A2-15-OB-F}, and the violation of $\Phi_{\VSS}$
is that of Lemma~\ref{lem:A2-9-OB-F}, which rests on
Lemma~\ref{lem:vss_not_ob_f}. A single execution violates both
predicates, so Lemma~\ref{lem:boundary-common-failure} gives the claim.
\end{proof}

\subsubsection{\texorpdfstring{VSS $\vee$ VEC (VX)}{VSS or VEC}}
\label{ssub:f3-7}

\begin{lemma}
\label{lem:A2-22-FS-A}
$\mathrm{VX}\notin\Sol(\FS^{A})$.
\end{lemma}

\begin{proof}
Let $\mathcal A$ be any deterministic algorithm of $\FS$ under $\ASY$,
and let $E_0$ be the execution of alternating singleton activations
$\{a\},\{b\},\dots$ under normalized frames (Lemma~\ref{lem:solo});
$E_0$ is fair and is determined by $\mathcal A$.

If $E_0$ contains no nonzero $\M$, condition $(\mathrm{P1})$ of $\VSS$
fails by Lemma~\ref{lem:p1_no_null}, while the sequence of
configurations reduces to $C_0$ and $d_0\not\perp d_0$, so the
attainment condition of $\mathrm{VEC}$ fails as well;
Lemma~\ref{lem:boundary-common-failure} gives the claim.

Otherwise, as in the proof of Lemma~\ref{lem:A2-16-FS-A}, the first
nonzero $\M$ is performed by $a$ from some state $s$ with
$\lambda_s>0$, while $b$ is in the same state $s$ and has no pending
operation. Let $p_E$ be its landing point.

\emph{Case $p_E=B$.} Follow $E_0$ up to the round following this $\M$.
When $a$ has landed, $d_1=b_0-B=-d_0$, which is none of the successors
admitted by the transition condition of $\mathrm{VEC}$; and
Lemma~\ref{lem:p1_replay} makes $b$, activated next, land at
$4d\notin P$, so the endpoint condition of $\VSS$ fails.

\emph{Case $p_E\neq B$.} Construct the execution as in the
corresponding case of Lemma~\ref{lem:A2-16-FS-A}: $b$ performs its
$\LK$ while $a$ is at $p_\mu$, $\mu\in(0,1)$ (Lemma~\ref{lem:mid}), and
its $\M$ is delayed until $a$ has landed. The $\M$ of $a$ is the
nonzero $\M$ with earliest completion time and does not land at $B$,
so by Lemma~\ref{lem:p1_first_move} no valid $T_1$ exists and
$\Phi_{\VSS}$ fails. The landing point of $b$ is affine in $\mu$ with
nonzero slope, while the successors admitted by the transition
condition of $\mathrm{VEC}$ are finitely many and independent of
$\mu$, whether or not the configuration $(p_E,b_0)$ is counted; so by
Corollary~\ref{cor:continuum} some $\mu$ avoids them and the
transition condition fails.

In both cases the two violations are safety violations forced by a
finite prefix, which extends to a fair execution by
Lemma~\ref{lem:concat}(1)(2); Lemma~\ref{lem:boundary-common-failure}
gives the claim.
\end{proof}

\begin{lemma}
\label{lem:A2-22-OB-F}
$\mathrm{VX}\notin\Sol(\OB^{F})$.
\end{lemma}

\begin{proof}
Let $\mathcal A$ be any deterministic algorithm of $\OB$ under $\FSY$
and take the fully synchronous execution in which both robots are
activated in every round (Lemmas~\ref{lem:rigidA1}
and~\ref{lem:rigidA2}), with a normalized frame at every $\LK$. The
violation of $\Phi_{\VSS}$ is that of
Lemma~\ref{lem:A2-9-OB-F} and the violation of $\Phi_{\mathrm{VEC}}$ is
that of Lemma~\ref{lem:A2-16-OB-F}. A single execution violates both
predicates, so Lemma~\ref{lem:boundary-common-failure} gives the claim.
\end{proof}

\subsubsection{\texorpdfstring{VSS $\vee$ CFSM (VC)}{VSS or CFSM}}
\label{ssub:f3-8}

\begin{lemma}
\label{lem:A2-23-OB-F}
$\mathrm{VC}\notin\Sol(\OB^{F})$.
\end{lemma}

\begin{proof}
Let $\mathcal A$ be any deterministic algorithm of $\OB$ under $\FSY$
and take the fully synchronous execution in which both robots are
activated in every round (Lemmas~\ref{lem:rigidA1}
and~\ref{lem:rigidA2}), with a normalized frame at every $\LK$.
 
By Lemma~\ref{lem:rigidA1}, if $\lambda=0$ then $\moves{r}=0$ and the
single-move condition of $\CFSM$ fails, while by
Lemma~\ref{lem:p1_no_null} the first phase of $\VSS$ cannot be
achieved. If $c_2=0$ the two robots coincide in round~$1$, so the
collision-freedom condition of $\CFSM$ fails, and their landing
point is the midpoint $1.5d\notin P$, so $\Phi_{\VSS}$ fails as
well. Otherwise both robots perform a nonzero $\M$ in every round, so
$\moves{r}=\infty$ and the single-move condition fails, while as in
Lemma~\ref{lem:vss_not_ob_f} the landings of round~$2$ leave the set
$P$ of designated points.
 
In each case a single execution violates both predicates, so
Lemma~\ref{lem:boundary-common-failure} gives the claim.
\end{proof}

\subsection{Disjunctions Containing MCv}
\label{app:cf-mcv-pairs}

This subsection treats the three antichains $\{\mathrm{MCv},P\}$ with
$P\in\{\VSS,\mathrm{VEC},\CFSM\}$, which carry four
boundary obligations in total.
For $\FC$, the constructions use the delayed-$\CP$
argument of Lemma~\ref{lem:mcv_not_fc_a}. For $\FS$, after establishing the
companion safety violation, we apply the continuation
statement of Lemma~\ref{lem:mcv_not_fs_am} from a configuration with distinct
positions, a common internal color, and no pending operation.

\subsubsection{\texorpdfstring{MCv $\vee$ VSS (MV)}{MCv or VSS}}
\label{ssub:f4-1}

\begin{lemma}
\label{lem:A2-19-FS-AM}
$\mathrm{MV}\notin\Sol(\FS^{A_M})$.
\end{lemma}

\begin{proof}
Let $\mathcal A$ be any deterministic algorithm of $\FS$ under $A_M$,
which admits both the simultaneous activation of the two robots and
synchronous activations of one robot at a time
(Lemma~\ref{lem:solo}). Use coordinates $A=0$, $a_0=d$, $b_0=2d$,
$B=3d$ on the initial line, with $P=\{0,d,2d,3d\}$, and normalized
frames at every $\LK$ at which the robots are distinct.
 
Activate both robots simultaneously while every output is zero. By
Lemma~\ref{lem:rigidA2} they carry the same color at every round
boundary and the configuration stays at $\iota$.
 
If every output along this execution is zero, continue forever.
Condition $(\mathrm{P1})$ of $\VSS$ fails by
Lemma~\ref{lem:p1_no_null}, and the two robots remain at distinct
positions, so the convergence condition of $\mathrm{MCv}$ fails as
well. This execution witnesses both liveness violations.
 
Otherwise let $\gamma$ be the common color at the start of the first
round with a nonzero output, let $(\lambda,\theta)$, with
$\lambda\neq0$, be the corresponding normalized displacement, and put
$c_1:=c_1(\gamma)$. Since $\mathcal A$ is deterministic, these are
known before that round is scheduled.
 
\emph{Case $c_1=0$, that is $(\lambda,\theta)=(1,0)$.} Activate both
robots in that round. Each moves to the position of its peer, so they
exchange positions and the distance returns to $d$ after passing
through $0$; the monotonicity condition of $\mathrm{MCv}$ therefore
fails, a safety violation at a finite time. The cumulative nonzero-move
counts are $(1,1)$, so an admissible exchange configuration would have
to be $(B,A)=(3d,0)$, whereas the landing is $(2d,d)$; no valid $T_1$
exists and $\VSS$ fails by Lemma~\ref{lem:p1_first_move}. Append an
arbitrary admissible fair, time-divergent tail, which preserves both
violations.
 
\emph{Case $c_1\neq0$.} Define the execution as an interval of
alternating singleton activations carried up to the round following the
first nonzero $\M$, each cycle being completed before the next one
begins, followed by the continuation described below.
 
The first nonzero $\M$ is performed by the robot $a$, which is
activated first. Either it does not land at $B$, so that the
first-phase clauses become unsatisfiable by the argument of
Lemma~\ref{lem:p1_first_move}, or it does and
Lemma~\ref{lem:p1_replay} makes the robot $b$, activated next, land at
$4d\notin P$; in either case the endpoint condition of $\VSS$ fails, a
safety violation. Since $c_1\neq0$, neither of these two Move
operations lands on the peer, so the two robots are at distinct
positions at the end of this interval.
 
Both robots have now been activated equally often, and every $\LK$ so
far has observed distinct positions, so by Lemma~\ref{lem:norm}(1)
they carry the same color.
At this point the robots occupy distinct positions, carry
the same internal color, and have no pending operation.
Apply the continuation statement proved in
Lemma~\ref{lem:mcv_not_fs_am}.
It supplies a fair, time-divergent continuation violating
$\Phi_{\mathrm{MCv}}$, while all safety violations established
in the prefix persist.
 
In every case a single execution violates both predicates, the
violations being finite except where stated, so
Lemma~\ref{lem:concat}(1)(2)(3) is satisfied and
Lemma~\ref{lem:boundary-common-failure} gives the claim.
\end{proof}

\subsubsection{\texorpdfstring{MCv $\vee$ VEC (MX)}{MCv or VEC}}
\label{ssub:f4-2}

\begin{lemma}
\label{lem:A2-20-FC-A}
$\mathrm{MX}\notin\Sol(\FC^{A})$.
\end{lemma}

\begin{proof}
Let $\mathcal A$ be any deterministic algorithm of $\FC$ under $\ASY$.
Run fully synchronous cycles with normalized frames; by
Lemma~\ref{lem:rigidA2} the two robots carry the same color at every
round boundary, and the configuration stays at $\iota$ while every
output is zero.

If every output along this execution is zero, continue forever. No
nonzero $\M$ occurs and the two robots remain at distinct positions, so
the convergence condition of $\mathrm{MCv}$ fails, while the sequence
of joint-stop configurations reduces to $C_0$ and $d_0\not\perp d_0$,
so the attainment condition of $\mathrm{VEC}$ fails as well. This
execution witnesses both liveness violations.

Otherwise stop the zero-Move prefix just before the first round with a
nonzero output, let $\gamma$ be the common color at this point, so that
$\lambda_\gamma>0$, and put $c:=c_1(\gamma)$. Build from here the
delayed-Comp interval of Lemma~\ref{lem:mcv_not_fc_a}: let $b$ perform
its $\LK$ and delay its $\CP$, the saved snapshot determining the
nonzero displacement $v=\lambda_\gamma R_{\theta_\gamma}(w_{\mathrm{old}})$,
where $w_{\mathrm{old}}$ is the vector from $b$ to $a$ at that $\LK$.
Throughout the delay the position of $b$ and its visible color $\gamma$
are unchanged, so by Lemma~\ref{lem:single-rule}(b) the robot $a$ may
be activated alone for any prescribed finite number of complete cycles,
applying the same pair $(\lambda_\gamma,\theta_\gamma)$, with
$\lambda_\gamma\neq0$, at each of them. Activate $a$ alone $k$ times,
where $k=1$ if $c=0$ and otherwise $k\ge3$ is large enough for what
follows, and append an arbitrary admissible fair, time-divergent tail.

\emph{Violation of $\Phi_{\mathrm{VEC}}$.}
If $c=0$, the robot $a$ lands on the position of $b$ at the first
activation and $d_1=0$, while all the targets permitted by the
transition condition are nonzero, so that condition fails. If
$c\neq0$, then as in Lemma~\ref{lem:A2-6-FC-S} the only single
displacements satisfying the transition condition rotate the segment
clockwise by $45^\circ$ or by $90^\circ$; if $a$ departs from them the
transition condition fails, and otherwise $a$ reaches a configuration
with $d_i\perp d_0$ in at most two moves and then repeats the same
displacement, so the stopping condition fails. Either way this is a
safety violation inside the interval.

\emph{Violation of $\Phi_{\mathrm{MCv}}$.}
If $c>1$, the first of these Move operations already increases the
distance between its starting and ending times. If $c=1$, its starting
and ending positions are distinct points at the same distance from the
stationary robot $b$, so the midpoint of the movement segment is
strictly closer to $b$ than either endpoint and the distance increases
between the time $a$ passes that midpoint and the completion of the
Move. In both cases the monotonicity condition fails.

If $c<1$, choose $k$ above so large that the vector $w_k$ from $b$ to
$a$ after $k$ complete cycles satisfies
$\|w_k\|=c^{\,k}\|w_{\mathrm{old}}\|<\|v\|/2$; for $c=0$ this already
holds with $k=1$. Keep $a$ stationary and let $b$ complete its $\CP$
using the saved snapshot, update its color, and perform the resulting
$\M$ with displacement $v$. The distance immediately before that $\M$
is $\|w_k\|$, whereas at its completion it satisfies
\[
  \|w_k-v\|\ \ge\ \|v\|-\|w_k\|\ >\ \|w_k\| ,
\]
so the monotonicity condition fails again.

Both violations are finite and occur inside the interval, so they may
be concatenated by Lemma~\ref{lem:concat}(1) and are preserved by the
fair tail. Lemma~\ref{lem:boundary-common-failure} gives the claim.
\end{proof}

\begin{lemma}
\label{lem:A2-20-FS-A}
$\mathrm{MX}\notin\Sol(\FS^{A})$.
\end{lemma}

\begin{proof}
Let $\mathcal A$ be any deterministic algorithm of $\FS$ under $\ASY$.
Run a fully synchronous prefix, which $\ASY$ admits; by
Lemma~\ref{lem:rigidA2} the two robots carry the same state at every
round boundary.
 
If $\lambda_s=0$ for every state $s$ of that common sequence, this fair
execution contains no nonzero $\M$: the robots stay at their distinct
initial positions, so the convergence condition of $\mathrm{MCv}$
fails, and the sequence of configurations reduces to $C_0$ with
$d_0\not\perp d_0$, so the attainment condition of $\mathrm{VEC}$ fails
as well.
 
Otherwise truncate the prefix at the round in which both robots carry
the first state $s$ with $\lambda_s>0$; the configuration is still
$\iota$. Activate $a$ alone, so that it performs a nonzero $\M$ from
$a_0$ to some $p_E\neq a_0$, and by Lemma~\ref{lem:mid} let $b$, in
state $s$, perform its $\LK$ while $a$ is at
$p_\mu=(1-\mu)a_0+\mu p_E$, $\mu\in(0,1)$, delaying the $\M$ of $b$
until $a$ has landed. By Lemma~\ref{lem:norm}(1) the displacement of $b$ is
$v=\lambda_s R_{\theta_s}(p_\mu-b_0)$, so its landing point is affine
in $\mu$ with nonzero slope $\lambda_s R_{\theta_s}(p_E-a_0)$. We
exclude the at most one $\mu$ with $p_\mu=b_0$, so that $v\neq0$, and
the at most one $\mu$ for which $b$ lands at $p_E$.
 
\emph{VEC.} The successors admitted by the transition condition of
$\mathrm{VEC}$ are finitely many and independent of $\mu$, whether or
not the configuration $(p_E,b_0)$ is counted, so by
Corollary~\ref{cor:continuum} some $\mu$ avoids them as well, and the
transition condition fails at a finite time (cf.\
Lemma~\ref{lem:vec_not_fs_a}).
 
\emph{MCv.} If $p_E=b_0$, then $a$ lands on $b$, and the $\M$ of $b$
raises the distance from $0$ to $\|v\|>0$, so the monotonicity
condition fails. Otherwise the robots are at distinct positions after
both landings, and we use the following dichotomy, applied to every
nonzero $\M$ whose peer is at rest. Along such a $\M$ the distance to
the stationary peer is non-increasing and then non-decreasing. If the
landing point is not the point of the segment closest to the peer, the
distance strictly increases at the end and the monotonicity condition
fails; we may therefore assume that every such $\M$ ends at its closest
approach.
 
Both robots have performed the same number of activations, and every
$\LK$ so far has observed distinct positions, so by
Lemma~\ref{lem:norm}(1) they carry the same color.
At this point the robots occupy distinct positions, carry
the same internal color, and have no pending operation.
Apply the continuation statement proved in
Lemma~\ref{lem:mcv_not_fs_am}.
It supplies a fair, time-divergent continuation violating
$\Phi_{\mathrm{MCv}}$, while all safety violations established
in the prefix persist.
 
In every case a single execution violates both predicates.
Whenever a nonzero output occurs, the violation of
$\Phi_{\mathrm{VEC}}$ is a safety violation at a finite time.
In the case without nonzero $\M$, both liveness violations are
witnessed by the explicit fair execution given there. Otherwise
$\Phi_{\mathrm{VEC}}$ has already suffered a safety violation.
The continuation from Lemma~\ref{lem:mcv_not_fs_am} witnesses either
a finite monotonicity violation or a failure of convergence.
Thus all predicates fail on the same fair infinite execution. Hence
Lemma~\ref{lem:concat}(1)(2)(3) is satisfied, and
Lemma~\ref{lem:boundary-common-failure} gives the claim.
\end{proof}

\subsubsection{\texorpdfstring{MCv $\vee$ CFSM (MC)}{MCv or CFSM}}
\label{ssub:f4-3}

\begin{lemma}
\label{lem:A2-21-FC-A}
$\mathrm{MC}\notin\Sol(\FC^{A})$.
\end{lemma}

\begin{proof}
Let $\mathcal A$ be any deterministic algorithm of $\FC$ under $\ASY$.
Run fully synchronous cycles with normalized frames; by
Lemma~\ref{lem:rigidA2} the two robots carry the same color at every
round boundary, and the configuration stays at $\iota$ while every
output is zero.

If every output along this execution is zero, continue forever. No
nonzero $\M$ occurs, so $\moves{a}=0\neq1$ and the single-move
condition of $\CFSM$ fails, while the two robots remain at distinct
positions and the convergence condition of $\mathrm{MCv}$ fails as
well. This execution witnesses both liveness violations.

Otherwise stop the zero-Move prefix just before the first round with a
nonzero output, let $\gamma$ be the common color at this point, so that
$\lambda_\gamma>0$, and put $c:=c_1(\gamma)$. Build from here the
delayed-Comp interval of Lemma~\ref{lem:mcv_not_fc_a}, exactly as in
Lemma~\ref{lem:A2-20-FC-A}: the robot $b$ performs its $\LK$ and delays
its $\CP$, the saved snapshot determining the nonzero displacement
$v=\lambda_\gamma R_{\theta_\gamma}(w_{\mathrm{old}})$, and during the
delay the robot $a$ is activated alone $k$ times, applying the same
pair $(\lambda_\gamma,\theta_\gamma)$ at each of them
(Lemma~\ref{lem:single-rule}(b)). Take $k=1$ if $c=0$ and otherwise
$k\ge2$ large enough for what follows, and append an arbitrary
admissible fair, time-divergent tail.

\emph{Violation of $\Phi_{\CFSM}$.}
If $c=0$, the robot $a$ lands on the position of $b$ at the first
activation, so $a(t)=b(t)$ at that time and the collision-freedom
condition fails. If $c\neq0$, the $k\ge2$ activations are all nonzero
$\M$ operations, so $\moves{a}\ge2\neq1$ and the single-move condition
fails. Either way this is a safety violation inside the interval.

\emph{Violation of $\Phi_{\mathrm{MCv}}$.}
The three cases $c>1$, $c=1$ and $c<1$ are those of
Lemma~\ref{lem:A2-20-FC-A} and the arguments there apply verbatim: a
single Move of $a$ already breaks monotonicity when $c\ge1$, and when
$c<1$ the choice $\|w_k\|<\|v\|/2$ makes the distance grow from
$\|w_k\|$ to $\|w_k-v\|\ge\|v\|-\|w_k\|>\|w_k\|$ when $b$ releases its
$\M$.

Both violations are finite and occur inside the interval, so they may
be concatenated by Lemma~\ref{lem:concat}(1) and are preserved by the
fair tail. Lemma~\ref{lem:boundary-common-failure} gives the claim.
\end{proof}

\subsection{\texorpdfstring{DMSD $\vee$ CFSM (DC)}{DMSD or CFSM}}
\label{app:cf-dmsd-cfsm}
\label{ssub:f5}

This subsection treats the single antichain
$\{\DMSD,\CFSM\}$, whose maximal negative boundary is
the one class $\OB^{A}$. The construction uses an observation made
during the movement of the peer to produce a non-dyadic distance at a
joint stop, and then one further nonzero $\M$ of a robot that has
already moved.

\begin{lemma}
\label{lem:A2-24-OB-A}
$(\DMSD \vee \CFSM)\notin\Sol(\OB^{A})$.
\end{lemma}

\begin{proof}
Let $\mathcal A$ be any deterministic algorithm of $\OB$ under $\ASY$.
Define the execution as a finite fully synchronous prefix freezing the
common state (Lemma~\ref{lem:rigidA1}), then an
observation-during-movement interval (Lemma~\ref{lem:mid}), then an
interval consisting of one solo activation of $a$, and finally an
arbitrary fair tail. The first prefix is the first stage of
Lemma~\ref{lem:dmsd_not_fs_a}; an $\OB$ algorithm is an $\FS$ algorithm
with $|C|=1$, so that lemma applies here and yields
Corollary~\ref{cor:dmsd_not_ob_a}. Normalize $a(0)=0$, $b(0)=e_1$ and
$D_0=1$, and write $\lambda:=\lambda_s$ and $R:=R_{\theta_s}$ for the
frozen common state $s$.

If $\lambda=0$ for every state, no nonzero $\M$ ever occurs, so
$\moves{a}=0\neq1$ and the single-move condition of $\CFSM$
fails, while the distance stays equal to $D_0$ at every joint stop and
the progress condition (D2) of $\DMSD$ fails as well.

Otherwise, in the observation interval the robot $b$ performs a nonzero
$\M$ from $b(0)$ to $p_E=e_1-\lambda Re_1$, and by
Lemma~\ref{lem:mid} the robot $a$ performs a $\LK$ at a time when $b$
is at $p_\mu=(1-\mu)b(0)+\mu p_E$ for a prescribed $\mu\in(0,1)$. If
$p_\mu$ ever coincides with the position of $a$, the two robots occupy
the same point and the collision-freedom condition of $\CFSM$
fails at once, so we may assume that this does not happen; then
$\lambda>0$ makes the resulting $\M$ of $a$ nonzero, and $a$ lands at
$\lambda Rp_\mu$.

Let both $\M$ operations complete and delay the next $\LK$ of each
robot; throughout that delay neither robot is inside a nonzero $\M$, so
the whole interval is contained in $\Sab$. The distance there is the
non-constant continuous function
\[
  h(\mu)=\|c+\mu w\|,\qquad
  c:=e_1-2\lambda Re_1,\qquad
  w:=\lambda^2R^2e_1\neq0 ,
\]
of $\mu$, as in the third stage of Lemma~\ref{lem:dmsd_not_fs_a}. Since
$\Dyad$ is countable, Corollary~\ref{cor:continuum} provides a
$\mu\in(0,1)$ with $h(\mu)/D_0\notin\Dyad$, so the dyadic condition
(D1) of $\DMSD$ fails. As $0\in\Dyad$, this $\mu$ also keeps
the two robots apart, so the distance stays positive.

At that point $a$ has already performed one nonzero $\M$. In the last
interval $a$ is activated alone; the distance being positive, in the
normalized frame it receives a local input isomorphic to the earlier
one, and $\lambda>0$ makes this $\M$ nonzero as well, so
$\moves{a}\ge2\neq1$ and the single-move condition of $\CFSM$
fails.

Both are safety violations, so they may be concatenated by
Lemma~\ref{lem:concat}(1)(2) and are preserved by the fair tail, and
Lemma~\ref{lem:boundary-common-failure} gives the claim.
\end{proof}

\medskip
This completes the forty-seven boundary obligations of the
two-element antichains: fourteen in
Appendix~\ref{app:cf-expansion-family}, six in
Appendix~\ref{app:cf-sro-family}, eighteen in
Appendix~\ref{app:cf-exchange-pairs}, four in
Appendix~\ref{app:cf-mcv-pairs}, one here, and four in the main text,
namely the two $\FS^{R}$ obligations of
Proposition~\ref{prop:cf-expansion-exchange} and the two $\FS^{S}$
obligations of Proposition~\ref{prop:cf-sro-exchange-family}. Together
they establish the case $|A|=2$ of
Proposition~\ref{prop:all-common-failure-boundaries}.

%% file: CFE-Three-Task.tex

\noindent
Each row lists one three-element antichain $A$, given by the
corresponding disjunction $B_A$ of principal tasks from
Table~\ref{tab:PP13}, together with
$\mathcal N_A=\Max(\mathcal Q\setminus U_A)$ and, for each
$Y\in\mathcal N_A$, the statement exhibiting one admissible fair
execution of an arbitrary algorithm in $Y$ that violates all three task
predicates. Every $Y\in\mathcal N_A$ also satisfies $Y\notin U_{A'}$ for
each two-element antichain $A'\subset A$, so with two exceptions the
execution is the one already constructed for $A'$ in
Table~\ref{tab:PP13-join2}, and the proof only adds the failure of the
third predicate on that execution. The exceptions are the boundary
classes $\OB^{S}$ and $\OB^{A_M}$, which occur in no row of
Table~\ref{tab:PP13-join2}: for them the execution is constructed from
scratch, once in Lemma~\ref{lem:A3-9-OB-S} and once in
Lemma~\ref{lem:A3-15-OB-AM}, and reused by the later rows of the same
family. The column sum is the number $34$ of three-task boundary
obligations certified in
Table~\ref{tab:antichain-enumeration-certificate}
(Appendix~\ref{app:antichain-enumeration}) and recorded in
Table~\ref{tab:cf-family-coverage}; together with
Table~\ref{tab:PP13-join2} this discharges all $47+34=81$ boundary
obligations, and Lemma~\ref{lem:boundary-common-failure} then gives
$\Prof(B_A)=U_A$ for every $A\in\Ant(\mathcal Q)$.

\clearpage
\begingroup
  \small
  \setlength{\tabcolsep}{4pt}
  \renewcommand{\arraystretch}{1.15}
  \captionsetup{hypcap=false}
  \captionof{table}{The sixteen three-element antichains, their maximal
    unsolvable classes $\mathcal N_A$, and the common-failure execution
    that discharges each obligation $(A,Y)$. The rows are grouped by
    the construction families of
    Section~\ref{sec:representative-common-failure} and
    Table~\ref{tab:cf-family-coverage}, and within each group they
    follow the order of Table~\ref{tab:PP13-join2}. All thirty-four
    obligations are proved in
    Appendix~\ref{app:cf-three-task-disjunctions}.}
  \label{tab:PP13-join3}
  \par\medskip
  \noindent
  \begin{tabularx}{\linewidth}{@{}l >{\raggedright\arraybackslash}X c@{}}
    \toprule
    \textbf{Disjunction $B_A$}
      & \textbf{$\mathcal N_A$ and proof reference}
      & \textbf{$|\mathcal N_A|$} \\
    \midrule
    \multicolumn{3}{@{}l}{\emph{Expansion family, $\mathrm{EPG}^{*}\vee E\vee P$ with $E\in\mathcal E$,
      $P\in\mathcal P_4$, Appendix~\ref{app:cf3-expansion}}}\\
    $\mathrm{EPG}^{*}\vee\mathrm{OSPO}\vee \SRDV$
      & $\FC^{A}$ Lem.~\ref{lem:A3-1-FC-A};\quad
        $\FS^{S}$ Lem.~\ref{lem:A3-1-FS-S};\quad
        $\OB^{R}$ Lem.~\ref{lem:A3-1-OB-R}
      & 3 \\
    $\mathrm{EPG}^{*}\vee\mathrm{OSPO}\vee\mathrm{VSS}$
      & $\FS^{S}$ Lem.~\ref{lem:A3-2-FS-S};\quad
        $\OB^{R}$ Lem.~\ref{lem:A3-2-OB-R}
      & 2 \\
    $\mathrm{EPG}^{*}\vee \SRDV \vee\mathrm{ZCC}$
      & $\FC^{A}$ Lem.~\ref{lem:A3-3-FC-A};\quad
        $\FS^{A_M}$ Lem.~\ref{lem:A3-3-FS-AM};\quad
        $\OB^{R}$ Lem.~\ref{lem:A3-3-OB-R}
      & 3 \\
    $\mathrm{EPG}^{*}\vee \SRDV \vee\mathrm{VEC}$
      & $\FC^{A}$ Lem.~\ref{lem:A3-4-FC-A};\quad
        $\FS^{A}$ Lem.~\ref{lem:A3-4-FS-A};\quad
        $\OB^{R}$ Lem.~\ref{lem:A3-4-OB-R}
      & 3 \\
    $\mathrm{EPG}^{*}\vee \SRDV \vee \CFSM$
      & $\FC^{A}$ Lem.~\ref{lem:A3-5-FC-A};\quad
        $\OB^{R}$ Lem.~\ref{lem:A3-5-OB-R}
      & 2 \\
    $\mathrm{EPG}^{*}\vee\mathrm{ZCC}\vee\mathrm{VSS}$
      & $\FS^{A_M}$ Lem.~\ref{lem:A3-6-FS-AM};\quad
        $\OB^{R}$ Lem.~\ref{lem:A3-6-OB-R}
      & 2 \\
    $\mathrm{EPG}^{*}\vee\mathrm{VSS}\vee\mathrm{VEC}$
      & $\FS^{A}$ Lem.~\ref{lem:A3-7-FS-A};\quad
        $\OB^{R}$ Lem.~\ref{lem:A3-7-OB-R}
      & 2 \\
    $\mathrm{EPG}^{*}\vee\mathrm{VSS}\vee \CFSM$
      & $\OB^{R}$ Lem.~\ref{lem:A3-8-OB-R}
      & 1 \\
    \cmidrule(lr){1-3}
    \multicolumn{3}{@{}l}{\emph{Shrinking-rotation family, $\mathrm{SRO}\vee E\vee P$ with $E\in\mathcal E$, $P\in\mathcal P_3$, Appendix~\ref{app:cf3-sro}}}\\
    $\mathrm{SRO}\vee \SRDV \vee\mathrm{ZCC}$
      & $\FC^{A}$ Lem.~\ref{lem:A3-9-FC-A};\quad
        $\FS^{A_M}$ Lem.~\ref{lem:A3-9-FS-AM};\quad
        $\OB^{S}$ Lem.~\ref{lem:A3-9-OB-S}
      & 3 \\
    $\mathrm{SRO}\vee \SRDV \vee\mathrm{VEC}$
      & $\FC^{A}$ Lem.~\ref{lem:A3-10-FC-A};\quad
        $\FS^{A}$ Lem.~\ref{lem:A3-10-FS-A};\quad
        $\OB^{S}$ Lem.~\ref{lem:A3-10-OB-S}
      & 3 \\
    $\mathrm{SRO}\vee \SRDV \vee \CFSM$
      & $\FC^{A}$ Lem.~\ref{lem:A3-11-FC-A};\quad
        $\OB^{S}$ Lem.~\ref{lem:A3-11-OB-S}
      & 2 \\
    $\mathrm{SRO}\vee\mathrm{ZCC}\vee\mathrm{VSS}$
      & $\FS^{A_M}$ Lem.~\ref{lem:A3-12-FS-AM};\quad
        $\OB^{S}$ Lem.~\ref{lem:A3-12-OB-S}
      & 2 \\
    $\mathrm{SRO}\vee\mathrm{VSS}\vee\mathrm{VEC}$
      & $\FS^{A}$ Lem.~\ref{lem:A3-13-FS-A};\quad
        $\OB^{S}$ Lem.~\ref{lem:A3-13-OB-S}
      & 2 \\
    $\mathrm{SRO}\vee\mathrm{VSS}\vee \CFSM$
      & $\OB^{S}$ Lem.~\ref{lem:A3-14-OB-S}
      & 1 \\
    \cmidrule(lr){1-3}
    \multicolumn{3}{@{}l}{\emph{Monotone-convergence family, $\mathrm{MCv}\vee\mathrm{VSS}\vee P$, Appendix~\ref{app:cf3-mcv}}}\\
    $\mathrm{MCv}\vee\mathrm{VSS}\vee\mathrm{VEC}$
      & $\FS^{A}$ Lem.~\ref{lem:A3-15-FS-A};\quad
        $\OB^{A_M}$ Lem.~\ref{lem:A3-15-OB-AM}
      & 2 \\
    $\mathrm{MCv}\vee\mathrm{VSS}\vee \CFSM$
      & $\OB^{A_M}$ Lem.~\ref{lem:A3-16-OB-AM}
      & 1 \\
    \midrule
    \textbf{Total} & & \textbf{34} \\
    \bottomrule
  \end{tabularx}
\endgroup
\clearpage

\subsection{Three-Task Disjunctions Containing \texorpdfstring{EPG$^{*}$}{EPG*}}
\label{app:cf3-expansion}

Throughout Appendix~\ref{app:cf-three-task-disjunctions} we use the
following principle. If $A=\{X_1,X_2,X_3\}$ is a three-element
antichain and $Y\in\mathcal N_A$, then for every two-element antichain
$A'\subset A$ we have $U_{A'}\subseteq U_A$ and hence $Y\notin U_{A'}$.
In every triple treated below, except for the boundary classes
$\OB^{S}$ and $\OB^{A_M}$, some such $A'$ satisfies $Y\in\mathcal
N_{A'}$, so the execution constructed for $A'$ at $Y$ in
Appendix~\ref{app:cf-two-task-disjunctions} applies to the algorithms
of $Y$ and we only have to show that the same execution violates the
third predicate as well. For $\OB^{S}$ and $\OB^{A_M}$ a new execution
is constructed; those cases occur in
Appendices~\ref{app:cf3-sro} and~\ref{app:cf3-mcv}.

This subsection treats the eight antichains
$\{\mathrm{EPG}^{*},E,P\}$ with
$E\in\mathcal E=\{\SRDV,\mathrm{VSS}\}$ and
$P\in\mathcal P_4=\{\mathrm{OSPO},\mathrm{ZCC},\mathrm{VEC},
\CFSM\}$, which carry eighteen boundary obligations in total.
When the reused pair is one of $\mathcal C_{\mathrm{exp}}$, the
$\mathrm{EPG}^{*}$ component is the violation already recorded in
Appendix~\ref{app:cf-expansion-family}. Otherwise no nonzero $\M$ other than the first one begins in the reused execution before the first nonzero $\M$ has been completed. Hence either that execution contains no nonzero $\M$ at all,
in which case the joint-stop sequence is $(C_0)$ and the perpetuity
condition fails, or the completion of the first nonzero $\M$ is a joint
stop realizing $C_1$ at which only one robot has moved, so by
Lemma~\ref{lem:solo-drift}(1) the centre of gravity has shifted and the
expansion condition fails at $i=0$. The peer may still have a pending
$\M$ at that instant; this does not affect the argument, since a robot
that has not begun to move is stationary. We keep the notation
$(\lambda^{\star},\theta^{\star})=(\tfrac12,\pi)$ of
Remark~\ref{rem:epgs-normalized-rule}.

\subsubsection{\texorpdfstring{EPG$^{*}$ $\vee$ OSPO $\vee$ S-RDV (EOS)}{EPG* or OSPO or S-RDV}}
\label{ssub:g1-1}

\begin{lemma}
\label{lem:A3-1-FC-A}
$\mathrm{EOS}\notin\Sol(\FC^{A})$.
\end{lemma}

\begin{proof}
Let $\mathcal A$ be any deterministic algorithm of $\FC$ under $\ASY$.
Use the execution of Lemma~\ref{lem:A2-8-FC-A}, which violates
$\Phi_{\mathrm{OSPO}}$ and $\Phi_{\SRDV}$. We show that it also violates
$\Phi_{\mathrm{EPG}^{*}}$. Recall $c_0=\tfrac32 d$.

\emph{No nonzero output.} The joint-stop sequence is $(C_0)$, so the
perpetuity condition fails. The same infinite execution witnesses this.

\emph{Case $(\lambda,\theta)\neq(2,0)$.} The first solo activation of
$a$ is the first nonzero $\M$. The joint stops before and after it
realize $C_0$ and $C_1$, and $b$ does not move in between. So the
expansion condition fails at $i=0$ (Lemma~\ref{lem:solo-drift}(1)).

\emph{Case $(\lambda,\theta)=(2,0)$.} The $\M$ of $b$ begins strictly
after that of $a$ ends. So the landing time of $a$ is a joint stop,
with configuration $(3d,2d)$ and midpoint $\tfrac52 d\neq c_0$. The
expansion condition fails at $i=0$, as in
Lemma~\ref{lem:A2-3-FC-A}.

In each case a single execution violates all three predicates; in
the last two cases the violation of $\Phi_{\mathrm{EPG}^{*}}$ is a
safety violation in the prefix, so Lemma~\ref{lem:concat}(1)(3) is
satisfied, and Lemma~\ref{lem:boundary-common-failure} gives the claim.
\end{proof}

\begin{lemma}
\label{lem:A3-1-FS-S}
$\mathrm{EOS}\notin\Sol(\FS^{S})$.
\end{lemma}

\begin{proof}
Let $\mathcal A$ be any deterministic algorithm of $\FS$ under $\SSY$.
Use the execution of Lemma~\ref{lem:A2-8-FS-S}, which violates
$\Phi_{\mathrm{OSPO}}$ and $\Phi_{\SRDV}$. To verify the failure of
$\mathrm{EPG}^{*}$, only its first part, the alternating singleton
activation, is needed.

If no nonzero $\M$ occurs, the joint-stop sequence is $(C_0)$ and the
perpetuity condition of $\mathrm{EPG}^{*}$ fails; the execution of
Lemma~\ref{lem:A2-8-FS-S} itself is the fair continuation witnessing
this.

Otherwise, as in that lemma, the two robots traverse the same states
until the first nonzero $\M$, and $a$ is activated first, so that $\M$
is performed by $a$ in a round that activates $a$ alone. The
surrounding round boundaries are joint stops realizing $C_0$ and
$C_1$, and $b$ does not move in between, so the expansion condition
fails at $i=0$ (Lemma~\ref{lem:solo-drift}(1)). This is a safety
violation at a finite time, and it precedes the continuation
witnessing any liveness violation of $\Phi_{\mathrm{OSPO}}$.

In either case a single execution violates all three predicates;
Lemma~\ref{lem:concat}(1)(3) is satisfied, and
Lemma~\ref{lem:boundary-common-failure} gives the claim.
\end{proof}

\begin{lemma}
\label{lem:A3-1-OB-R}
$\mathrm{EOS}\notin\Sol(\OB^{R})$.
\end{lemma}

\begin{proof}
Let $\mathcal A$ be any deterministic algorithm of $\OB$ under $\RSY$
and use the execution of Lemma~\ref{lem:A2-2-OB-R}, which violates
$\Phi_{\mathrm{EPG}^{*}}$ and $\Phi_{\mathrm{OSPO}}$. The visible state
of an $\OB$ robot is unique, so in a normalized frame the rule is a
single pair $(\lambda,\theta)$ (Lemma~\ref{lem:single-rule}(a)).

If $\lambda=0$ no nonzero $\M$ occurs, so the first nonzero $\M$
demanded by Lemma~\ref{lem:p1_no_null} does not exist and
$\Phi_{\SRDV}$ fails.

If $\lambda\neq0$ and
$(\lambda,\theta)\neq(\lambda^{\star},\theta^{\star})$, the execution of
Lemma~\ref{lem:A2-2-OB-R} is the fully synchronous one. When
$(\lambda,\theta)\neq(2,0)$, the nonzero $\M$ of round~$1$ does not meet
the requirement of Lemma~\ref{lem:p1_first_move}, namely a landing of
$a$ at $B=3d$ or of $b$ at $A=0$, so $\Phi_{\SRDV}$
fails at that landing. When $(\lambda,\theta)=(2,0)$, round~$1$ takes
$a$ to $3d$ and $b$ to $0$, and round~$2$ takes $a$ to $-3d$ and $b$ to
$6d$, both outside $[A,B]=[0,3d]$
(Lemma~\ref{lem:srdv_not_ob_f}), so $\Phi_{\SRDV}$
fails.

If $(\lambda,\theta)=(\lambda^{\star},\theta^{\star})$, the execution of
Lemma~\ref{lem:A2-2-OB-R} is the alternating singleton one. Since
$(\tfrac12,\pi)\neq(2,0)$, the robot $a$ moves from $d$ to $d/2$ at its
first activation, which is not the landing at $B=3d$ required by
Lemma~\ref{lem:p1_first_move}, so $\Phi_{\SRDV}$ fails.

In each case, a single execution violates all three predicates, so
Lemma~\ref{lem:boundary-common-failure} gives the claim.
\end{proof}

\subsubsection{\texorpdfstring{EPG$^{*}$ $\vee$ OSPO $\vee$ VSS (EOV)}{EPG* or OSPO or VSS}}
\label{ssub:g1-2}
 
\begin{lemma}
\label{lem:A3-2-FS-S}
$\mathrm{EOV}\notin\Sol(\FS^{S})$.
\end{lemma}

\begin{proof}
The execution of Lemma~\ref{lem:A2-9-FS-S} is that of
Lemma~\ref{lem:A2-8-FS-S} and violates $\Phi_{\mathrm{OSPO}}$ and
$\Phi_{\VSS}$. The argument of Lemma~\ref{lem:A3-1-FS-S} also gives a
violation of $\Phi_{\mathrm{EPG}^{*}}$ on that same execution.
Hence Lemma~\ref{lem:boundary-common-failure} gives the claim.
\end{proof}

\begin{lemma}
\label{lem:A3-2-OB-R}
$\mathrm{EOV}\notin\Sol(\OB^{R})$.
\end{lemma}

\begin{proof}
Let $\mathcal A$ be any deterministic algorithm of $\OB$ under $\RSY$
and use the execution of Lemma~\ref{lem:A2-2-OB-R}, which violates
$\Phi_{\mathrm{EPG}^{*}}$ and $\Phi_{\mathrm{OSPO}}$. By
Lemma~\ref{lem:single-rule}(a) the rule is a single pair
$(\lambda,\theta)$.

If $\lambda=0$ no nonzero $\M$ occurs and condition $(\mathrm{P1})$ of
$\mathrm{VSS}$ fails by Lemma~\ref{lem:p1_no_null}.

If $\lambda\neq0$ and
$(\lambda,\theta)\neq(\lambda^{\star},\theta^{\star})$, the execution is
the fully synchronous one. When $(\lambda,\theta)\neq(2,0)$, the
nonzero $\M$ of round~$1$ does not meet the requirement of
Lemma~\ref{lem:p1_first_move}, so no valid $T_1$ exists and
$\Phi_{\mathrm{VSS}}$ fails. When $(\lambda,\theta)=(2,0)$, round~$2$
takes $a$ to $-3d$ and $b$ to $6d$, neither of which lies in
$P=\{0,d,2d,3d\}$, so the endpoint condition
$(\mathrm{P1}\text{-}b)$ fails.

If $(\lambda,\theta)=(\lambda^{\star},\theta^{\star})$, the execution is
the alternating singleton one and $a$ moves from $d$ to $d/2\notin P$
at its first activation, so $(\mathrm{P1}\text{-}b)$ fails there.

In each case, a single execution violates all three predicates, so
Lemma~\ref{lem:boundary-common-failure} gives the claim.
\end{proof}

\subsubsection{\texorpdfstring{EPG$^{*}$ $\vee$ S-RDV $\vee$ ZCC (ESZ)}{EPG* or S-RDV or ZCC}}
\label{ssub:g1-3}
 
\begin{lemma}
\label{lem:A3-3-FC-A}
$\mathrm{ESZ}\notin\Sol(\FC^{A})$.
\end{lemma}
 
\begin{proof}
Use the execution of Lemma~\ref{lem:A2-15-FC-A}, which violates
$\Phi_{\SRDV}$ and $\Phi_{\mathrm{ZCC}}$.
If it contains no nonzero move, the perpetuity condition
of $\mathrm{EPG}^{*}$ fails.
 
Otherwise consider the landing of the first nonzero move
of $a$. In both branches of that construction, $b$ remains
at its initial position until after this landing, even when
it already has a pending move. Thus the landing time is a
joint stop, and the center of gravity has changed by half
the nonzero displacement of $a$. The expansion condition
of $\mathrm{EPG}^{*}$ therefore fails at this first
transition.
 
This safety violation persists along the continuation
constructed in Lemma~\ref{lem:A2-15-FC-A}. Hence the same fair execution
violates all three predicates, and
Lemma~\ref{lem:boundary-common-failure} gives the claim.
\end{proof}
 
\begin{lemma}
\label{lem:A3-3-FS-AM}
$\mathrm{ESZ}\notin\Sol(\FS^{A_M})$.
\end{lemma}
 
\begin{proof}
Let $\mathcal A$ be any deterministic algorithm of $\FS$ under $A_M$
and use the execution of Lemma~\ref{lem:A2-15-FS-AM}, where
$\Phi_{\SRDV}$ and $\Phi_{\mathrm{ZCC}}$ were shown to
fail. Its first phase is an interval of alternating singleton
activations, which activates one robot at a time, and the configuration
is unchanged before the first nonzero $\M$.
 
In the branch in which no nonzero $\M$ ever occurs the joint-stop
sequence has the single term $C_0$ and the perpetuity condition of
$\mathrm{EPG}^{*}$ fails. Otherwise the first nonzero $\M$ occurs in
that first phase and is performed while the peer is not activated, so
by Lemma~\ref{lem:solo-drift}(1) the expansion condition fails at
$i=0$, a safety violation. The prefix of
Lemma~\ref{lem:A2-15-FS-AM} is then followed by the continuation used
there, namely the coincidence continuation from
Lemma~\ref{lem:zcc_not_fc_s} or the continuation from
Lemma~\ref{lem:zcc_not_fs_am}, and the violation persists along it
(Lemma~\ref{lem:concat}(1)).
 
In each branch a single fair execution violates all three predicates,
so Lemma~\ref{lem:boundary-common-failure} gives the claim.
\end{proof}
 
\begin{lemma}
\label{lem:A3-3-OB-R}
$\mathrm{ESZ}\notin\Sol(\OB^{R})$.
\end{lemma}
 
\begin{proof}
Let $\mathcal A$ be any deterministic algorithm of $\OB$ under $\RSY$
and use the execution of Lemma~\ref{lem:A2-4-OB-R}, which violates
$\Phi_{\mathrm{EPG}^{*}}$ and $\Phi_{\mathrm{ZCC}}$. By
Lemma~\ref{lem:single-rule}(a) the rule is a single pair
$(\lambda,\theta)$.
 
If $\lambda=0$ no nonzero $\M$ occurs, so the first nonzero $\M$
demanded by Lemma~\ref{lem:p1_no_null} does not exist and
$\Phi_{\SRDV}$ fails.
 
If $\lambda\neq0$ and
$(\lambda,\theta)\neq(\lambda^{\star},\theta^{\star})$, the execution of
Lemma~\ref{lem:A2-4-OB-R} is the fully synchronous one. When
$(\lambda,\theta)\neq(2,0)$, the nonzero $\M$ of round~$1$ does not meet
the requirement of Lemma~\ref{lem:p1_first_move} and
$\Phi_{\SRDV}$ fails at that landing. When
$(\lambda,\theta)=(2,0)$, the robots land at $3d$ and $0$ in round~$1$
and at $-3d$ and $6d$ in round~$2$, leaving $[A,B]=[0,3d]$
(Lemma~\ref{lem:A2-8-OB-F}, Lemma~\ref{lem:srdv_not_ob_f}), so
$\Phi_{\SRDV}$ fails.
 
If $(\lambda,\theta)=(\lambda^{\star},\theta^{\star})$, the execution is
the alternating singleton one and $a$ moves from $d$ to $d/2$ at its
first activation, which is not the landing at $B=3d$ required by
Lemma~\ref{lem:p1_first_move}, so $\Phi_{\SRDV}$ fails.
 
In each case a single execution violates all three predicates, so
Lemma~\ref{lem:boundary-common-failure} gives the claim.
\end{proof}


\subsubsection{\texorpdfstring{EPG$^{*}$ $\vee$ S-RDV $\vee$ VEC (ESX)}{EPG* or S-RDV or VEC}}
\label{ssub:g1-4}
 
\begin{lemma}
\label{lem:A3-4-FC-A}
$\mathrm{ESX}\notin\Sol(\FC^{A})$.
\end{lemma}
 
\begin{proof}
Let $\mathcal A$ be any deterministic algorithm of $\FC$ under $\ASY$
and use the execution of Lemma~\ref{lem:A2-16-FC-A}, which violates
$\Phi_{\SRDV}$ and $\Phi_{\mathrm{VEC}}$. We show that it also violates
$\Phi_{\mathrm{EPG}^{*}}$. The initial centre of gravity is
$\tfrac12(a_0+b_0)=\tfrac32 d$.

\emph{No nonzero output.} The execution contains no nonzero $\M$, so
the joint-stop sequence has the single term $C_0$ and the perpetuity
condition of $\mathrm{EPG}^{*}$ fails. The same infinite execution
witnesses this, together with the two liveness violations of
Lemma~\ref{lem:A2-16-FC-A}.

\emph{Case $(\lambda,\theta)=(2,0)$.} The $\M$ of $b$ begins strictly
after that of $a$ ends, so at the landing of $a$ neither robot is
inside a $\M$, although $b$ still has a pending one. That instant is a
joint stop and realizes $C_1=(3d,2d)$, of centre of gravity
$\tfrac52 d\neq\tfrac32 d$; equivalently, $b$ completes no nonzero $\M$
between $C_0$ and $C_1$, so Lemma~\ref{lem:solo-drift}(1) applies. The
expansion condition of $\mathrm{EPG}^{*}$ fails at $i=0$, a safety
violation in the prefix.

\emph{Case $(\lambda,\theta)\neq(2,0)$.} The configuration is unchanged
before the solo activation interval of $a$, so the first activation of
$a$ in it is the first nonzero $\M$ of the execution and $b$ is not
activated during it. The maximal joint-stop intervals before and after
it realize $C_0$ and $C_1$, so by Lemma~\ref{lem:solo-drift}(1) the
expansion condition fails at $i=0$, a safety violation in the prefix.

In the first case all three violations are liveness violations of one
infinite execution. In the other two the violation of
$\Phi_{\mathrm{EPG}^{*}}$ is finite and precedes the continuations
supplied by Lemma~\ref{lem:A2-16-FC-A}. In either case
Lemma~\ref{lem:concat}(1)(3) is satisfied, and
Lemma~\ref{lem:boundary-common-failure} gives the claim.
\end{proof}

\begin{lemma}
\label{lem:A3-4-FS-A}
$\mathrm{ESX}\notin\Sol(\FS^{A})$.
\end{lemma}
 
\begin{proof}
Let $\mathcal A$ be any deterministic algorithm of $\FS$ under $\ASY$
and use the execution of Lemma~\ref{lem:A2-16-FS-A}, which violates
$\Phi_{\SRDV}$ and $\Phi_{\mathrm{VEC}}$.

If no nonzero $\M$ occurs, the joint-stop sequence has the single term
$C_0$ and the perpetuity condition of $\mathrm{EPG}^{*}$ fails.
Otherwise, in both cases of that proof the first nonzero $\M$ is
performed by $a$ from the initial configuration, and $b$ completes no
nonzero $\M$ before $a$ has landed: for $p_E=B$ it is not activated,
and for $p_E\neq B$ its $\M$ is delayed until after the landing. The
landing of $a$ is therefore a joint stop realizing $C_1$, and by
Lemma~\ref{lem:solo-drift}(1) the expansion condition of
$\mathrm{EPG}^{*}$ fails at $i=0$, a safety violation preserved by
the rest of the execution (Lemma~\ref{lem:concat}(1)(2)).

In each case a single execution violates all three predicates, so
Lemma~\ref{lem:boundary-common-failure} gives the claim.
\end{proof}
 
\begin{lemma}
\label{lem:A3-4-OB-R}
$\mathrm{ESX}\notin\Sol(\OB^{R})$.
\end{lemma}
 
\begin{proof}
Let $\mathcal A$ be any deterministic algorithm of $\OB$ under $\RSY$
and use the execution of Lemma~\ref{lem:A2-6-OB-R}, which violates
$\Phi_{\mathrm{EPG}^{*}}$ and $\Phi_{\mathrm{VEC}}$ and which is in
every case the alternating singleton one with an empty fully
synchronous prefix and normalized frames. By
Lemma~\ref{lem:single-rule}(a) the rule is a single pair
$(\lambda,\theta)$.
 
If $\lambda=0$ no nonzero $\M$ occurs, so the first nonzero $\M$
demanded by Lemma~\ref{lem:p1_no_null} does not exist and
$\Phi_{\SRDV}$ fails. If $\lambda\neq0$ and
$(\lambda,\theta)\neq(2,0)$, the landing of round~$1$ is not the
$B=3d$ required by Lemma~\ref{lem:p1_first_move}, so
$\Phi_{\SRDV}$ fails at that landing. If
$(\lambda,\theta)=(2,0)$, then $c_1=|1-2|=1\neq0$ and
Lemma~\ref{lem:p1_replay} applies: the robot $b$, activated immediately
after $a$ lands at $3d$, receives an isomorphic local input and lands
at $4d\notin[A,B]$, so $\Phi_{\SRDV}$ fails.
 
In each case a single execution violates all three predicates, so
Lemma~\ref{lem:boundary-common-failure} gives the claim.
\end{proof}

\subsubsection{\texorpdfstring{EPG$^{*}$ $\vee$ S-RDV $\vee$ CFSM (ESC)}{EPG* or S-RDV or CFSM}}
\label{ssub:g1-5}
 
\begin{lemma}
\label{lem:A3-5-FC-A}
$\mathrm{ESC}\notin\Sol(\FC^{A})$.
\end{lemma}
 
\begin{proof}
Let $\mathcal A$ be any deterministic algorithm of $\FC$ under $\ASY$
and use the execution of Lemma~\ref{lem:A2-17-FC-A}, which violates
$\Phi_{\CFSM}$ and $\Phi_{\SRDV}$. That execution has the same three
cases as the one of Lemma~\ref{lem:A2-16-FC-A}, and the argument of
Lemma~\ref{lem:A3-4-FC-A} applies verbatim to each of them: with no
nonzero output the joint-stop sequence has the single term $C_0$ and
the perpetuity condition of $\mathrm{EPG}^{*}$ fails; when
$(\lambda,\theta)=(2,0)$ the landing of $a$ is a joint stop realizing
$C_1=(3d,2d)$ and Lemma~\ref{lem:solo-drift}(1) makes the expansion
condition fail at $i=0$; and otherwise the first nonzero $\M$ is the
first activation of $a$ in its solo interval, with $b$ not activated,
so Lemma~\ref{lem:solo-drift}(1) applies again.

In the first case all three violations are liveness violations of one
infinite execution; in the other two the violation of
$\Phi_{\mathrm{EPG}^{*}}$ is finite and precedes the continuations
supplied by Lemma~\ref{lem:A2-17-FC-A}. In either case
Lemma~\ref{lem:concat}(1)(3) is satisfied, and
Lemma~\ref{lem:boundary-common-failure} gives the claim.
\end{proof}
 
\begin{lemma}
\label{lem:A3-5-OB-R}
$\mathrm{ESC}\notin\Sol(\OB^{R})$.
\end{lemma}
 
\begin{proof}
Let $\mathcal A$ be any deterministic algorithm of $\OB$ under $\RSY$
and use the execution of Lemma~\ref{lem:A2-7-OB-R}, which violates
$\Phi_{\mathrm{EPG}^{*}}$ and $\Phi_{\CFSM}$ and which is the
alternating singleton one with an empty fully synchronous prefix and
normalized frames. By Lemma~\ref{lem:single-rule}(a) the rule is a
single pair $(\lambda,\theta)$.
 
If $\lambda=0$ no nonzero $\M$ occurs, so the first nonzero $\M$
demanded by Lemma~\ref{lem:p1_no_null} does not exist. If
$\lambda\neq0$ and $(\lambda,\theta)\neq(2,0)$, the landing of
round~$1$ is not the $B=3d$ required by
Lemma~\ref{lem:p1_first_move}. If $(\lambda,\theta)=(2,0)$, then
$c_1=1\neq0$, which is the third case of
Lemma~\ref{lem:A2-7-OB-R}, and Lemma~\ref{lem:p1_replay} makes $b$ land
at $4d\notin[A,B]$. In every case $\Phi_{\SRDV}$ fails.
 
In each case a single execution violates all three predicates, so
Lemma~\ref{lem:boundary-common-failure} gives the claim.
\end{proof}

\subsubsection{\texorpdfstring{EPG$^{*}$ $\vee$ ZCC $\vee$ VSS (EZV)}{EPG* or ZCC or VSS}}
\label{ssub:g1-6}
 
\begin{lemma}
\label{lem:A3-6-FS-AM}
$\mathrm{EZV}\notin\Sol(\FS^{A_M})$.
\end{lemma}
 
\begin{proof}
Let $\mathcal A$ be any deterministic algorithm of $\FS$ under $A_M$
and use the execution of Lemma~\ref{lem:A2-18-FS-AM}, which violates
$\Phi_{\mathrm{VSS}}$ and $\Phi_{\mathrm{ZCC}}$. Its first phase is an
interval of alternating singleton activations, and the configuration is
unchanged before the first nonzero $\M$.
 
In the branch in which no nonzero $\M$ ever occurs, the joint-stop
sequence has the single term $C_0$ and the perpetuity condition of
$\mathrm{EPG}^{*}$ fails. Otherwise the first nonzero $\M$ occurs in
that first phase and is performed while the peer is not activated, so
by Lemma~\ref{lem:solo-drift}(1) the expansion condition fails at
$i=0$, a safety violation. It persists along the continuation from
Lemma~\ref{lem:zcc_not_fc_s} or Lemma~\ref{lem:zcc_not_fs_am} used in
Lemma~\ref{lem:A2-18-FS-AM} (Lemma~\ref{lem:concat}(1)).
 
In each case a single execution violates all three predicates, so
Lemma~\ref{lem:boundary-common-failure} gives the claim.
\end{proof}
 
\begin{lemma}
\label{lem:A3-6-OB-R}
$\mathrm{EZV}\notin\Sol(\OB^{R})$.
\end{lemma}
 
\begin{proof}
Let $\mathcal A$ be any deterministic algorithm of $\OB$ under $\RSY$
and use the execution of Lemma~\ref{lem:A2-4-OB-R}, which violates
$\Phi_{\mathrm{EPG}^{*}}$ and $\Phi_{\mathrm{ZCC}}$. By
Lemma~\ref{lem:single-rule}(a) the rule is a single pair
$(\lambda,\theta)$.
 
If $\lambda=0$ no nonzero $\M$ occurs and condition $(\mathrm{P1})$ of
$\mathrm{VSS}$ fails by Lemma~\ref{lem:p1_no_null}.
 
If $\lambda\neq0$ and
$(\lambda,\theta)\neq(\lambda^{\star},\theta^{\star})$, the execution of
Lemma~\ref{lem:A2-4-OB-R} is the fully synchronous one. When
$(\lambda,\theta)\neq(2,0)$, the nonzero $\M$ of round~$1$ does not meet
the requirement of Lemma~\ref{lem:p1_first_move}, so no valid $T_1$
exists and $\Phi_{\mathrm{VSS}}$ fails. When $(\lambda,\theta)=(2,0)$,
the landings $3d$ and $0$ of round~$1$ both belong to
$P=\{0,d,2d,3d\}$, but in round~$2$ the robots land at $-3d$ and $6d$,
leaving $P$, so the endpoint condition $(\mathrm{P1}\text{-}b)$ fails.
 
If $(\lambda,\theta)=(\lambda^{\star},\theta^{\star})$, the execution is
the alternating singleton one and $a$ moves from $d$ to $d/2\notin P$
at its first activation, so $(\mathrm{P1}\text{-}b)$ fails there.
 
In each case a single execution violates all three predicates, so
Lemma~\ref{lem:boundary-common-failure} gives the claim.
\end{proof}

\subsubsection{\texorpdfstring{EPG$^{*}$ $\vee$ VSS $\vee$ VEC (EVX)}{EPG* or VSS or VEC}}
\label{ssub:g1-7}
 
\begin{lemma}
\label{lem:A3-7-FS-A}
$\mathrm{EVX}\notin\Sol(\FS^{A})$.
\end{lemma}
 
\begin{proof}
Let $\mathcal A$ be any deterministic algorithm of $\FS$ under $\ASY$
and use the execution of Lemma~\ref{lem:A2-22-FS-A}, which violates
$\Phi_{\mathrm{VSS}}$ and $\Phi_{\mathrm{VEC}}$.
 
If no nonzero $\M$ occurs, the joint-stop sequence has the single term
$C_0$ and the perpetuity condition of $\mathrm{EPG}^{*}$ fails.
Otherwise, in both cases of that proof the first nonzero $\M$ is
performed by $a$ from the initial configuration, and $b$ completes no
nonzero $\M$ before $a$ has landed: for $p_E=B$ it is not activated,
and for $p_E\neq B$ its $\M$ is delayed until after the landing. The
landing of $a$ is therefore a joint stop realizing $C_1$, and by
Lemma~\ref{lem:solo-drift}(1) the expansion condition of
$\mathrm{EPG}^{*}$ fails at $i=0$, a safety violation preserved by
the rest of the execution (Lemma~\ref{lem:concat}(1)(2)).
 
In each case a single execution violates all three predicates, so
Lemma~\ref{lem:boundary-common-failure} gives the claim.
\end{proof}
 
\begin{lemma}
\label{lem:A3-7-OB-R}
$\mathrm{EVX}\notin\Sol(\OB^{R})$.
\end{lemma}
 
\begin{proof}
Let $\mathcal A$ be any deterministic algorithm of $\OB$ under $\RSY$
and use the execution of Lemma~\ref{lem:A2-6-OB-R}, which violates
$\Phi_{\mathrm{EPG}^{*}}$ and $\Phi_{\mathrm{VEC}}$ and which is the
alternating singleton one with normalized frames. By
Lemma~\ref{lem:single-rule}(a) the rule is a single pair
$(\lambda,\theta)$.
 
If $\lambda=0$ no nonzero $\M$ occurs and $(\mathrm{P1})$ fails by
Lemma~\ref{lem:p1_no_null}. If $\lambda\neq0$ and
$(\lambda,\theta)\neq(2,0)$, the landing of round~$1$ is not the
$B=3d$ required by Lemma~\ref{lem:p1_first_move} and no valid $T_1$
exists. If $(\lambda,\theta)=(2,0)$, then $c_1=|1-2|=1\neq0$ and
Lemma~\ref{lem:p1_replay} applies: the robot $b$, activated immediately
after $a$ lands at $3d$, lands at $4d\notin P=\{0,d,2d,3d\}$. In every
case $\Phi_{\mathrm{VSS}}$ fails.
 
In each case a single execution violates all three predicates, so
Lemma~\ref{lem:boundary-common-failure} gives the claim.
\end{proof}

\subsubsection{\texorpdfstring{EPG$^{*}$ $\vee$ VSS $\vee$ CFSM (EVC)}{EPG* or VSS or CFSM}}
\label{ssub:g1-8}
 
\begin{lemma}
\label{lem:A3-8-OB-R}
$\mathrm{EVC}\notin\Sol(\OB^{R})$.
\end{lemma}
 
\begin{proof}
Let $\mathcal A$ be any deterministic algorithm of $\OB$ under $\RSY$
and use the execution of Lemma~\ref{lem:A2-7-OB-R}, which violates
$\Phi_{\mathrm{EPG}^{*}}$ and $\Phi_{\CFSM}$ and which is the
alternating singleton one with normalized frames. By
Lemma~\ref{lem:single-rule}(a) the rule is a single pair
$(\lambda,\theta)$.
 
If $\lambda=0$ no nonzero $\M$ occurs and $(\mathrm{P1})$ fails by
Lemma~\ref{lem:p1_no_null}. If $\lambda\neq0$ and
$(\lambda,\theta)\neq(2,0)$, the landing of round~$1$ is not the
$B=3d$ required by Lemma~\ref{lem:p1_first_move} and no valid $T_1$
exists. If $(\lambda,\theta)=(2,0)$, then $c_1=1\neq0$ and
Lemma~\ref{lem:p1_replay} makes $b$ land at $4d\notin P=\{0,d,2d,3d\}$.
In every case $\Phi_{\mathrm{VSS}}$ fails.
 
In each case a single execution violates all three predicates, so
Lemma~\ref{lem:boundary-common-failure} gives the claim.
\end{proof}


\subsection{Three-Task Disjunctions Containing SRO}
\label{app:cf3-sro}
 
This subsection treats the six antichains $\{\mathrm{SRO},E,P\}$ with
$E\in\mathcal E$ and $P\in\mathcal P_3=\{\mathrm{ZCC},\allowbreak\mathrm{VEC},\allowbreak\CFSM\}$, which carry thirteen boundary obligations in total.
The class $\OB^{S}$ occurs here as a boundary class for five of these
antichains, and no corresponding lemma exists in
Appendix~\ref{app:cf-two-task-disjunctions}; the execution for it is
constructed from scratch in Lemma~\ref{lem:A3-9-OB-S} and reused by the
four later $\OB^{S}$ lemmas.
 
\subsubsection{\texorpdfstring{SRO $\vee$ S-RDV $\vee$ ZCC (RSZ)}{SRO or S-RDV or ZCC}}
\label{ssub:g2-1}
 
\begin{lemma}
\label{lem:A3-9-FC-A}
$\mathrm{RSZ}\notin\Sol(\FC^{A})$.
\end{lemma}
 
\begin{proof}
Use the construction of Lemma~\ref{lem:A2-15-FC-A}.
If no nonzero output occurs, the perpetuity condition
of $\mathrm{SRO}$ fails as well.
 
If the first nonzero rule is $(2,0)$, the first solo move
of $a$ reverses the segment and satisfies neither
$\mathbf{C1}$ nor $\mathbf{C2}$.
Thus $\Phi_{\mathrm{SRO}}$ has a safety violation before
the continuation used in Lemma~\ref{lem:A2-15-FC-A}.
 
Otherwise keep the peer inactive during the initial
solo interval of $a$. If the first move lands on the peer,
the resulting zero segment violates the rotation condition
of $\mathrm{SRO}$. If it does not, choose the interval
to contain at least two nonzero moves of $a$.
Lemma~\ref{lem:cf-sro-two-moves} then gives a safety
violation of $\Phi_{\mathrm{SRO}}$.
 
For the $\mathrm{ZCC}$ component, use the finite safety
violation or the continuation property of Lemma~\ref{lem:zcc_not_fc_s},
as in the proof of Lemma~\ref{lem:A2-15-FC-A}.
In the contraction case $\sigma_1\in(0,1)$ the robots remain in the
original initial segment in their original order after any number of
solo activations of $a$, so the continuation property applies at the
end of the interval chosen above.
The earlier safety violations persist.
 
Thus one fair infinite execution violates all three
predicates, and Lemma~\ref{lem:boundary-common-failure}
gives the claim.
\end{proof}
 
\begin{lemma}
\label{lem:A3-9-FS-AM}
$\mathrm{RSZ}\notin\Sol(\FS^{A_M})$.
\end{lemma}
 
\begin{proof}
Let $\mathcal A$ be any deterministic algorithm of $\FS$ under $A_M$
and modify the execution of Lemma~\ref{lem:A2-11-FS-AM} as follows.
 
If no nonzero $\M$ ever occurs, $\Phi_{\SRDV}$ fails by
Lemma~\ref{lem:p1_no_null}, and the first condition of $\mathrm{SRO}$
and the third condition of $\mathrm{ZCC}$ fail as shown in
Lemma~\ref{lem:A2-11-FS-AM}.
 
If the first nonzero $\M$ of $a$ does not land at the $B=3d$ required
by Lemma~\ref{lem:p1_first_move}, then
$\Phi_{\SRDV}$ fails at that point, and the execution
of Lemma~\ref{lem:A2-11-FS-AM} may be used unchanged, so
$\Phi_{\mathrm{SRO}}$ and $\Phi_{\mathrm{ZCC}}$ fail as shown there.
 
If it does land at $B$, then $(\lambda_s,\theta_s)=(2,0)$, and this
displacement is neither the $90^\circ$ clockwise rotation about the
peer required by $\mathbf{C1}$ nor the one about the midpoint required
by $\mathbf{C2}$, so by Lemma~\ref{lem:cf-sro-two-moves} the second
condition of $\mathrm{SRO}$ fails already at this move and no second
nonzero $\M$ of $a$ is needed. Insert, immediately after this move, a
solo activation interval of $b$: activating $b$ alone until it reaches
the state $s$ leaves the configuration unchanged, since in all earlier
states $\lambda=0$ and the state transition of an $\FS$ robot depends
only on its own state; then $b$ in state $s$ observes a peer at
distance $d$ and, in the normalized frame, receives a local input
isomorphic to that of $a$, so it applies the same $(2,0)$ and, by
Lemma~\ref{lem:p1_replay}, lands at $4d\notin[A,B]$, which violates
$\Phi_{\SRDV}$. Moreover, the first move of $a$ alone takes it
from $d$ to $3d$ while $b$ rests at $b_0=2d$, so
\[
  \|a-a_0\|=2d>d=\|b-a_0\|
\]
and the collision-freedom condition of $\mathrm{ZCC}$ has already
failed. No additional delayed $\M$ is needed; append an arbitrary fair,
time-divergent tail.
 
In every case a single fair execution violates all three predicates,
and the safety violations established in its prefix persist
(Lemma~\ref{lem:concat}(1)), so
Lemma~\ref{lem:boundary-common-failure} gives the claim.
\end{proof}
 
\begin{lemma}
\label{lem:A3-9-OB-S}
$\mathrm{RSZ}\notin\Sol(\OB^{S})$.
\end{lemma}
 
\begin{proof}
Let $\mathcal A$ be any deterministic algorithm of $\OB$ under $\SSY$.
No lemma of Appendix~\ref{app:cf-two-task-disjunctions} has $\OB^{S}$
as a boundary class, so we construct the execution directly. The
visible state of an $\OB$ robot is unique, so in a normalized frame the
rule is a single pair $(\lambda,\theta)$
(Lemma~\ref{lem:single-rule}(a)).
 
\emph{Case $\lambda=0$.} Since $\OB$ has no colors, the two robots do
not move under any schedule. Hence the sequence of configurations
reduces to $C_0$ and the perpetuity condition of $\mathrm{SRO}$ fails,
the robots remain at distinct positions and the convergence condition
of $\mathrm{ZCC}$ fails, and no nonzero $\M$ exists, so
$\Phi_{\SRDV}$ fails by Lemma~\ref{lem:p1_no_null}. All
three are independent of the schedule, hence persist along any fair
infinite continuation.
 
\emph{Case $(\lambda,\theta)=(2,0)$.} Define the execution by one solo
activation of $a$, then one solo activation of $b$, then alternating
singleton activations and an arbitrary fair tail, all admissible under
$\SSY$. The displacement of $a$ is neither the $90^\circ$ clockwise
rotation about the peer ($\mathbf{C1}$) nor the one about the midpoint
($\mathbf{C2}$), so by Lemma~\ref{lem:cf-sro-two-moves} the second
condition of $\mathrm{SRO}$ fails at this move. Moreover
$\sigma_1=1-2=-1$, so $c_1=1$ and $\psi\equiv\pi$
(Lemma~\ref{lem:ratios}(iii)), whence
$f_a(1)=d\,|1-e^{\mathrm i\pi}|=2d>d=\|b(0)-a(0)\|$ and the
collision-freedom condition of $\mathrm{ZCC}$ fails. Finally $a$ has
landed at $B=3d$, so the robot $b$, activated immediately afterwards,
receives in the normalized frame an isomorphic local input, applies the
same $(2,0)$ and, by Lemma~\ref{lem:p1_replay}, lands at
$4d\notin[A,B]$, so $\Phi_{\SRDV}$ fails.
 
\emph{Case $\lambda\neq0$ and $(\lambda,\theta)\neq(2,0)$.} Define the
execution by an interval in which $a$ alone is activated $k$ times in a
row (Lemma~\ref{lem:solo}, with $k$ large enough for what follows),
then alternating singleton activations and an arbitrary fair tail. The
first nonzero $\M$ of $a$ is not the landing at $B$ required by
Lemma~\ref{lem:p1_first_move}, so
$\Phi_{\SRDV}$ fails at that point. For
$\Phi_{\mathrm{SRO}}$: if that move is neither of type $\mathbf{C1}$
nor of type $\mathbf{C2}$, the second condition fails at once, which
includes the case $c_1=0$, since then $d_{i+1}=0$ while both types
require a nonzero value; if it is of one of the two types, the second
solo nonzero $\M$ of $a$ makes the second condition fail by
Lemma~\ref{lem:cf-sro-two-moves}. For $\Phi_{\mathrm{ZCC}}$, put $\sigma_1:=1-\lambda e^{\mathrm{i}\theta}$
and regard $\mathcal A$ as an $\FC$ algorithm with a single color, so
that the constructions in the proof of Lemma~\ref{lem:zcc_not_fc_s}
apply. Every activation is a complete cycle under $\SSY$, so no
operation is pending at the end of the solo interval. If $c_1=0$,
the robots coincide after the first activation; end the solo interval
there and use the coincidence continuation described in the proof of
Lemma~\ref{lem:zcc_not_fc_s}. This either keeps both positions fixed
forever or produces a finite violation of the collision-freedom
condition. If $c_1\neq0$ and $\sigma_1\notin(0,1)$, the
collision-freedom condition fails at a finite time for $k$ large
enough, as in the proof of Lemma~\ref{lem:A2-15-FC-A}. In the remaining
contraction case $\sigma_1\in(0,1)$, the robots still lie in the
original initial segment in their original order after the $k$
activations, and neither has a pending operation; the continuation
property proved in Lemma~\ref{lem:zcc_not_fc_s} supplies a fair,
time-divergent continuation violating $\Phi_{\mathrm{ZCC}}$, while
preserving the safety violations already obtained in the prefix.
 
In every case a single execution violates all three predicates, so
Lemma~\ref{lem:boundary-common-failure} gives the claim.
\end{proof}

\subsubsection{\texorpdfstring{SRO $\vee$ S-RDV $\vee$ VEC (RSX)}{SRO or S-RDV or VEC}}
\label{ssub:g2-2}
 
\begin{lemma}
\label{lem:A3-10-FC-A}
$\mathrm{RSX}\notin\Sol(\FC^{A})$.
\end{lemma}
 
\begin{proof}
Let $\mathcal A$ be any deterministic algorithm of $\FC$ under $\ASY$
and use the execution of Lemma~\ref{lem:A2-16-FC-A}, which violates
$\Phi_{\SRDV}$ and $\Phi_{\mathrm{VEC}}$. We show that it also violates
$\Phi_{\mathrm{SRO}}$, under the same case distinction.

\emph{No nonzero output.} No nonzero $\M$ ever occurs, so the sequence
of joint-stop configurations reduces to $C_0$ and the perpetuity
condition of $\mathrm{SRO}$ fails. The same infinite execution
witnesses this, together with the two liveness violations of
Lemma~\ref{lem:A2-16-FC-A}.

\emph{Case $(\lambda,\theta)=(2,0)$.} The landing of $a$ is a joint
stop realizing $C_1=(3d,2d)$, so $d_1=-d_0$. Neither $\mathbf{C1}$,
which requires $d_1=R_{-90^\circ}d_0$, nor $\mathbf{C2}$, which
requires $d_1=2^{-1/2}R_{-45^\circ}d_0$, is satisfied by a rotation by
$180^\circ$ of the same norm. The second condition of $\mathrm{SRO}$
therefore fails, a safety violation in the prefix.

\emph{Case $(\lambda,\theta)\neq(2,0)$.} In the solo activation
interval of $a$ the robot applies the same pair $(\lambda,\theta)$,
with $\lambda\neq0$, at every activation
(Lemma~\ref{lem:single-rule}(b)). If $c_1=0$ then $a$ lands on the
position of $b$ at the first activation and $d_1=0$, while
$\mathbf{C1}$ preserves the norm and $\mathbf{C2}$ scales it by
$2^{-1/2}$, so neither can hold with $d_0\neq0$; the second condition
of $\mathrm{SRO}$ fails. If $c_1\neq0$, the interval has length
$k\ge3$, so $a$ performs two nonzero $\M$ operations with its peer at
rest, and by Lemma~\ref{lem:cf-sro-two-moves} the rotation condition
and $\mathbf{C3}$ cannot both hold, so the second condition fails
again. Either way this is a safety violation in the prefix.

In the first case all three violations are liveness violations of one
infinite execution. In the other two the violation of
$\Phi_{\mathrm{SRO}}$ is finite and precedes the continuations
supplied by Lemma~\ref{lem:A2-16-FC-A}. In either case
Lemma~\ref{lem:concat}(1)(3) is satisfied, and
Lemma~\ref{lem:boundary-common-failure} gives the claim.
\end{proof}
 
\begin{lemma}
\label{lem:A3-10-FS-A}
$\mathrm{RSX}\notin\Sol(\FS^{A})$.
\end{lemma}
 
\begin{proof}
Let $\mathcal A$ be any deterministic algorithm of $\FS$ under $\ASY$
and use the execution of Lemma~\ref{lem:A2-13-FS-A}, which violates
$\Phi_{\mathrm{SRO}}$ and $\Phi_{\mathrm{VEC}}$.
 
If $\lambda_s=0$ for every state $s$, no nonzero $\M$ occurs and
$\Phi_{\SRDV}$ fails by Lemma~\ref{lem:p1_no_null},
while $\Phi_{\mathrm{SRO}}$ and $\Phi_{\mathrm{VEC}}$ fail in the first
case of Lemma~\ref{lem:A2-13-FS-A}.
 
Otherwise the robot $a$ performs the first nonzero $\M$ in the round of
the first common state $s$ with $\lambda_s>0$. If its landing is not
the $B=3d$ required by Lemma~\ref{lem:p1_first_move}, then
$\Phi_{\SRDV}$ fails at that point. If it is, then
$(\lambda,\theta)=(2,0)$, which is neither the $\mathbf{C1}$ rotation
about the peer nor the $\mathbf{C2}$ rotation about the midpoint, so by
Lemma~\ref{lem:cf-sro-two-moves} $\Phi_{\mathrm{SRO}}$ fails at once.
In this case the points of the movement segment of $a$ are
$p_\mu=(1+2\mu)d$ with $\mu\in(0,1)$, and the robot $b$, observing
$p_\mu$ in state $s$, applies the same $(2,0)$ in the normalized frame
and lands at $2d+2\bigl((1+2\mu)d-2d\bigr)=4\mu d$. Taking
$\mu\in(3/4,1)$ gives $4\mu d>3d$, so the robots leave
$[A,B]=[0,3d]$ and $\Phi_{\SRDV}$ fails. The set of
parameters to be excluded for the $\mathrm{VEC}$ violation is
countable, so such a $\mu$ can be chosen inside $(3/4,1)$ and all three
violations occur in the same execution.
 
All three are safety violations, so they may be concatenated by
Lemma~\ref{lem:concat}(1)(2), and
Lemma~\ref{lem:boundary-common-failure} gives the claim.
\end{proof}
 
\begin{lemma}
\label{lem:A3-10-OB-S}
$\mathrm{RSX}\notin\Sol(\OB^{S})$.
\end{lemma}
 
\begin{proof}
Let $\mathcal A$ be any deterministic algorithm of $\OB$ under $\SSY$
and use the execution of Lemma~\ref{lem:A3-9-OB-S}, in which $a$ alone
is activated and then $a$ and $b$ are activated alternately. The
violations of $\Phi_{\mathrm{SRO}}$ and
$\Phi_{\SRDV}$ were shown there. By
Lemma~\ref{lem:single-rule}(a) the rule is a single pair
$(\lambda,\theta)$.
 
If $\lambda=0$ the robots do not move, the sequence of configurations
reduces to $C_0$, and $d_0\not\perp d_0$, so the attainment condition
of $\mathrm{VEC}$ fails. If $c_1=0$, the robot $a$ lands on the
position of $b$ at its first activation and $d_{i+1}=0$, while all the
targets permitted by the transition condition of $\mathrm{VEC}$ are
nonzero, so that condition fails. Otherwise the only single
displacements satisfying the transition condition rotate the segment
clockwise by $45^\circ$ or by $90^\circ$; if $a$ departs from them the
transition condition fails at that point, and otherwise $a$ reaches a
configuration with $d_i\perp d_0$ within its two activations, at which
moment both robots are at rest, but the same displacement is repeated
in the alternating singleton activations that follow, so the robots
move again and the stopping condition fails.
 
All of these are safety violations, and a single execution violates all
three predicates, so Lemma~\ref{lem:boundary-common-failure} gives the
claim.
\end{proof}

\subsubsection{\texorpdfstring{SRO $\vee$ S-RDV $\vee$ CFSM (RSC)}{SRO or S-RDV or CFSM}}
\label{ssub:g2-3}
 
\begin{lemma}
\label{lem:A3-11-FC-A}
$\mathrm{RSC}\notin\Sol(\FC^{A})$.
\end{lemma}
 
\begin{proof}
Let $\mathcal A$ be any deterministic algorithm of $\FC$ under $\ASY$
and use the execution of Lemma~\ref{lem:A2-17-FC-A}, which violates
$\Phi_{\CFSM}$ and $\Phi_{\SRDV}$. The three cases are those of
Lemma~\ref{lem:A3-10-FC-A}, and the argument there applies verbatim,
the solo activation interval of $a$ now having length two, which still
supplies the two nonzero $\M$ operations required by
Lemma~\ref{lem:cf-sro-two-moves}.

In the first case all three violations are liveness violations of one
infinite execution; in the other two the violation of
$\Phi_{\mathrm{SRO}}$ is finite and precedes the continuations supplied
by Lemma~\ref{lem:A2-17-FC-A}. In either case
Lemma~\ref{lem:concat}(1)(3) is satisfied, and
Lemma~\ref{lem:boundary-common-failure} gives the claim.
\end{proof}
 
\begin{lemma}
\label{lem:A3-11-OB-S}
$\mathrm{RSC}\notin\Sol(\OB^{S})$.
\end{lemma}
 
\begin{proof}
Let $\mathcal A$ be any deterministic algorithm of $\OB$ under $\SSY$
and use the execution of Lemma~\ref{lem:A3-9-OB-S}, where the
violations of $\Phi_{\mathrm{SRO}}$ and
$\Phi_{\SRDV}$ were shown; $\Phi_{\CFSM}$ fails
under the same case distinction.
 
If $\lambda=0$, the two robots do not move under any schedule, so
$\moves{a}=0\neq1$ and the single-move condition fails.
 
If $(\lambda,\theta)=(2,0)$, then $c_1=\sqrt{1-4+4}=1\neq0$, so the
distance remains $d$ and $a$ performs a nonzero $\M$ both in round~$1$
and at its second activation in the alternating singleton part; hence
$\moves{a}\ge2\neq1$ and the single-move condition fails.
 
If $\lambda\neq0$ and $(\lambda,\theta)\neq(2,0)$, then for $c_1=0$ the
robot $a$ lands on the position of $b$ at its first activation, so
$a(t)=b(t)$ and the collision-freedom condition fails, while for
$c_1\neq0$ both of its two consecutive solo activations are nonzero
$\M$ operations and $\moves{a}\ge2$ violates the single-move condition.
All of these are safety violations.
 
In each case a single execution violates all three predicates, so
Lemma~\ref{lem:boundary-common-failure} gives the claim.
\end{proof}

\subsubsection{\texorpdfstring{SRO $\vee$ ZCC $\vee$ VSS (RZV)}{SRO or ZCC or VSS}}
\label{ssub:g2-4}
 
\begin{lemma}
\label{lem:A3-12-FS-AM}
$\mathrm{RZV}\notin\Sol(\FS^{A_M})$.
\end{lemma}
 
\begin{proof}
Let $\mathcal A$ be any deterministic algorithm of $\FS$ under $A_M$
and use the execution of Lemma~\ref{lem:A3-9-FS-AM} with
$\Phi_{\SRDV}$ replaced by $\Phi_{\mathrm{VSS}}$.
 
If no nonzero $\M$ ever occurs, condition $(\mathrm{P1})$ of
$\mathrm{VSS}$ fails, and the first condition of $\mathrm{SRO}$ and the
third condition of $\mathrm{ZCC}$ fail as in
Lemma~\ref{lem:A2-11-FS-AM}.
 
If the first nonzero $\M$ of $a$ does not land at the $B=3d$ required
by Lemma~\ref{lem:p1_first_move}, no valid $T_1$ exists and
$\Phi_{\mathrm{VSS}}$ fails at that point; the execution of
Lemma~\ref{lem:A2-11-FS-AM} may then be used unchanged.
 
If it does land at $B$, then $(\lambda_s,\theta_s)=(2,0)$, which is
neither the $\mathbf{C1}$ rotation about the peer nor the
$\mathbf{C2}$ rotation about the midpoint, so by
Lemma~\ref{lem:cf-sro-two-moves} the second condition of
$\mathrm{SRO}$ fails at this move. Insert immediately afterwards a
solo activation interval of $b$: activating $b$ alone until it reaches
the state $s$ leaves the configuration unchanged, since
$\lambda=0$ in all earlier states, and then $b$ in state $s$ receives,
in the normalized frame, a local input isomorphic to that of $a$,
applies the same $(2,0)$ and by Lemma~\ref{lem:p1_replay} lands at
$4d\notin P=\{0,d,2d,3d\}$, so $\Phi_{\mathrm{VSS}}$ fails.
Moreover, the first move of $a$ alone takes it from $d$ to $3d$ while
$b$ rests at $b_0=2d$, so
\[
  \|a-a_0\|=2d>d=\|b-a_0\|
\]
and the collision-freedom condition of $\mathrm{ZCC}$ has already
failed. No additional delayed $\M$ is needed; append an arbitrary fair,
time-divergent tail.
 
In every case a single fair execution violates all three predicates,
and the safety violations established in its prefix persist
(Lemma~\ref{lem:concat}(1)), so
Lemma~\ref{lem:boundary-common-failure} gives the claim.
\end{proof}
 
\begin{lemma}
\label{lem:A3-12-OB-S}
$\mathrm{RZV}\notin\Sol(\OB^{S})$.
\end{lemma}
 
\begin{proof}
Let $\mathcal A$ be any deterministic algorithm of $\OB$ under $\SSY$
and use the execution of Lemma~\ref{lem:A3-9-OB-S}, where the
violations of $\Phi_{\mathrm{SRO}}$ and $\Phi_{\mathrm{ZCC}}$ were
shown; $\Phi_{\mathrm{VSS}}$ fails under the same case distinction.
 
If $\lambda=0$, the two robots do not move, so no nonzero $\M$ exists
and condition $(\mathrm{P1})$ of $\mathrm{VSS}$ fails.
 
If $(\lambda,\theta)=(2,0)$, then in that execution $a$ is activated
alone in round~$1$ and $b$ in round~$2$, both under normalized frames.
The robot $a$ lands at $B=3d$, and $c_1=1\neq0$, so
Lemma~\ref{lem:p1_replay} applies and the robot $b$, activated
immediately afterwards, receives an isomorphic local input and lands at
$4d\notin P=\{0,d,2d,3d\}$, so $\Phi_{\mathrm{VSS}}$ fails.
 
If $\lambda\neq0$ and $(\lambda,\theta)\neq(2,0)$, the first nonzero
$\M$ of $a$ is not the landing at $B$ required by
Lemma~\ref{lem:p1_first_move}, so no valid $T_1$ exists and
$\Phi_{\mathrm{VSS}}$ fails.
 
In each case a single execution violates all three predicates, so
Lemma~\ref{lem:boundary-common-failure} gives the claim.
\end{proof}

\subsubsection{\texorpdfstring{SRO $\vee$ VSS $\vee$ VEC (RVX)}{SRO or VSS or VEC}}
\label{ssub:g2-5}
 
\begin{lemma}
\label{lem:A3-13-FS-A}
$\mathrm{RVX}\notin\Sol(\FS^{A})$.
\end{lemma}
 
\begin{proof}
Let $\mathcal A$ be any deterministic algorithm of $\FS$ under $\ASY$
and use the execution of Lemma~\ref{lem:A2-13-FS-A}, which violates
$\Phi_{\mathrm{SRO}}$ and $\Phi_{\mathrm{VEC}}$.
 
If $\lambda_s=0$ for every state $s$, no nonzero $\M$ occurs and
condition $(\mathrm{P1})$ of $\mathrm{VSS}$ fails, while
$\Phi_{\mathrm{SRO}}$ and $\Phi_{\mathrm{VEC}}$ fail in the first case
of Lemma~\ref{lem:A2-13-FS-A}.
 
Otherwise the robot $a$ performs the first nonzero $\M$ in the round of
the first common state $s$ with $\lambda_s>0$. If its landing is not
the $B=3d$ required by Lemma~\ref{lem:p1_first_move}, no valid $T_1$
exists and $\Phi_{\mathrm{VSS}}$ fails at that point. If it is, then
$(\lambda,\theta)=(2,0)$ and, as in Lemma~\ref{lem:A3-10-FS-A}, the
points of the movement segment of $a$ are $p_\mu=(1+2\mu)d$ with
$\mu\in(0,1)$, and the robot $b$ observing $p_\mu$ in state $s$ applies
the same $(2,0)$ in the normalized frame and lands at $4\mu d$. The set
$P=\{0,d,2d,3d\}$ being finite, only finitely many $\mu$ give
$4\mu d\in P$, so together with the countable set to be excluded for
the $\mathrm{VEC}$ violation a $\mu\in(0,1)$ avoiding all of them can
be chosen. For that $\mu$ we have $4\mu d\notin P$, so
$\Phi_{\mathrm{VSS}}$ fails, and all three violations occur in the same
execution.
 
All three are safety violations, so they may be concatenated by
Lemma~\ref{lem:concat}(1)(2), and
Lemma~\ref{lem:boundary-common-failure} gives the claim.
\end{proof}
 
\begin{lemma}
\label{lem:A3-13-OB-S}
$\mathrm{RVX}\notin\Sol(\OB^{S})$.
\end{lemma}
 
\begin{proof}
Let $\mathcal A$ be any deterministic algorithm of $\OB$ under $\SSY$
and use the execution of Lemma~\ref{lem:A3-9-OB-S} with the same case
distinction $\lambda=0$ / $(\lambda,\theta)=(2,0)$ / otherwise. The
violation of $\Phi_{\mathrm{SRO}}$ is that of
Lemma~\ref{lem:A3-9-OB-S} and the violation of $\Phi_{\mathrm{VSS}}$ is
that of Lemma~\ref{lem:A3-12-OB-S}.
 
For $\Phi_{\mathrm{VEC}}$: if $\lambda=0$ the robots do not move, the
sequence of configurations reduces to $C_0$, and $d_0\not\perp d_0$, so
the attainment condition fails. If $(\lambda,\theta)=(2,0)$ then
$d_{i+1}=-d_i$, which is none of $d_i$, $R_{-90^\circ}d_i$ and
$2^{\mp1/2}R_{-45^\circ}d_i$, so the transition condition fails.
Otherwise, if $c_1=0$ then $d_{i+1}=0$ while all the permitted targets
are nonzero and the transition condition fails, and if $c_1\neq0$ then,
as in Lemma~\ref{lem:A3-10-OB-S}, the only single displacements
satisfying the transition condition rotate the segment clockwise by
$45^\circ$ or by $90^\circ$; departing from them makes the transition
condition fail, and otherwise an orthogonal configuration is reached
within the two activations of $a$ and the same displacement is repeated
afterwards, so the stopping condition fails.
 
In each case a single execution violates all three predicates, so
Lemma~\ref{lem:boundary-common-failure} gives the claim.
\end{proof}

\subsubsection{\texorpdfstring{SRO $\vee$ VSS $\vee$ CFSM (RVC)}{SRO or VSS or CFSM}}
\label{ssub:g2-6}
 
\begin{lemma}
\label{lem:A3-14-OB-S}
$\mathrm{RVC}\notin\Sol(\OB^{S})$.
\end{lemma}
 
\begin{proof}
Let $\mathcal A$ be any deterministic algorithm of $\OB$ under $\SSY$
and use the execution of Lemma~\ref{lem:A3-9-OB-S}. All three
violations follow the same case distinction
$\lambda=0$ / $(\lambda,\theta)=(2,0)$ / otherwise: the violation of
$\Phi_{\mathrm{SRO}}$ is that of Lemma~\ref{lem:A3-9-OB-S}, the
violation of $\Phi_{\mathrm{VSS}}$ is that of
Lemma~\ref{lem:A3-12-OB-S}, and the violation of
$\Phi_{\CFSM}$ is that of Lemma~\ref{lem:A3-11-OB-S}. A single
execution violates all three predicates, so
Lemma~\ref{lem:boundary-common-failure} gives the claim.
\end{proof}


\subsection{Three-Task Disjunctions Containing MCv}
\label{app:cf3-mcv}

This subsection treats the two antichains
$\{\mathrm{MCv},\mathrm{VSS},P\}$ with $P\in\{\mathrm{VEC},
\CFSM\}$, which carry three boundary obligations in total. The
class $\OB^{A_M}$ occurs here as a boundary class for both antichains,
and no lemma of Appendix~\ref{app:cf-two-task-disjunctions} has it as a
boundary class; the execution for it is constructed from scratch in
Lemma~\ref{lem:A3-15-OB-AM} and reused in Lemma~\ref{lem:A3-16-OB-AM}.

We record once the two activation patterns used there. Under $A_M$ the
adversary may activate one robot at a time
(Lemma~\ref{lem:solo}, whose statement includes $K=A_M$), and it may
also activate both robots in lockstep: giving the two robots identical
phase times $t_L<t_C<t_B<t_E$ in each round places every $\Look$ before
both $\M$ intervals, so no $\Look$ falls inside the $\M$ interval of the
peer and $M$\nobreakdash-atomicity is respected. In such a round both
robots compute from the same snapshot, so
Lemma~\ref{lem:ratios}(ii) applies and gives $d\mapsto c_2 d$ with the
midpoint fixed. We refer to the latter pattern as a \emph{fully
synchronous} round; note that Lemma~\ref{lem:rigidA1} is stated for
$K=F$ and is not invoked below.

\subsubsection{\texorpdfstring{MCv $\vee$ VSS $\vee$ VEC (MVX)}{MCv or VSS or VEC}}
\label{ssub:g3-1}
 
\begin{lemma}
\label{lem:A3-15-FS-A}
$\mathrm{MVX}\notin\Sol(\FS^{A})$.
\end{lemma}

\begin{proof}
Let $\mathcal A$ be any deterministic algorithm of $\FS$ under $\ASY$.
Use coordinates $A=0$, $a_0=d$, $b_0=2d$, $B=3d$ on the initial line,
with $P=\{0,d,2d,3d\}$, and normalized frames at every $\LK$ at which
the robots are distinct. Activate both robots simultaneously while
every output is zero; they then carry the same state at every round
boundary and the configuration stays at $\iota$.
 
\emph{No nonzero output.} Continue forever. Condition $(\mathrm{P1})$
fails by Lemma~\ref{lem:p1_no_null}, so $\Phi_{\mathrm{VSS}}$ fails;
the sequence of joint-stop configurations reduces to $C_0$ and
$d_0\not\perp d_0$, so the attainment condition of $\mathrm{VEC}$
fails; and the two robots remain at distinct positions, so the
convergence condition of $\mathrm{MCv}$ fails as well. This execution
witnesses all three liveness violations.
 
Otherwise let $\sigma$ be the common state at the start of the first
round with a nonzero output, let $(\lambda,\theta)$, with
$\lambda\neq0$, be the corresponding normalized displacement, and put
$c_1:=c_1(\sigma)$.
 
\emph{Case $c_1=0$, that is $(\lambda,\theta)=(1,0)$.} Activate both
robots in that round. Each moves to the position of its peer, so they
exchange positions and the distance returns to $d$ after passing
through $0$; the monotonicity condition of $\mathrm{MCv}$ fails. The
cumulative nonzero-move counts are $(1,1)$, so an admissible exchange
configuration would have to be $(B,A)=(3d,0)$, whereas the landing is
$(2d,d)$; no valid $T_1$ exists and $\Phi_{\mathrm{VSS}}$ fails by
Lemma~\ref{lem:p1_first_move}. Finally $C_1=(2d,d)$ gives $d_1=-d_0$,
which is neither the clockwise rotation by $45^\circ$ nor the one by
$90^\circ$, so the transition condition of $\mathrm{VEC}$ fails at
$i=0$. All three are safety violations at a finite time; append an
arbitrary admissible fair, time-divergent tail.
 
\emph{Case $c_1\neq0$.} At the round boundary at which both robots
carry $\sigma$ the configuration is still $\iota$; from there continue
as in Lemma~\ref{lem:A2-22-FS-A}, activating $a$ first, so that $a$
performs the first nonzero $\M$ from $\sigma$ while $b$ is in $\sigma$
with no pending operation. Let $p_E$ be its landing point.
 
If $p_E=B$, then as in that lemma $d_1=-d_0$ violates the transition
condition of $\mathrm{VEC}$, and $b$ lands at $4d\notin P$ by
Lemma~\ref{lem:p1_replay}, so $\Phi_{\mathrm{VSS}}$ fails. Moreover
the $\M$ of $a$ from $d$ to $3d$ passes through $b_0=2d$, so the
distance to the stationary peer drops to $0$ and returns to $d$, and
the monotonicity condition of $\mathrm{MCv}$ fails. All three are
safety violations; append an arbitrary fair tail.
 
If $p_E\neq B$, then as in that lemma $b$ observes $p_\mu$ and its
$\M$ is delayed until $a$ has landed; $\Phi_{\mathrm{VSS}}$ fails by
Lemma~\ref{lem:p1_first_move}, and the transition condition of
$\mathrm{VEC}$ fails for all $\mu$ outside a finite set. The landing
point of $b$ is affine in $\mu$ with nonzero slope, so it equals $p_E$
for at most one $\mu$; excluding it as well, and since $c_1\neq0$
gives $p_E\neq b_0$, the robots are at distinct positions after both
landings. It remains to violate $\mathrm{MCv}$ in this case.
 
For $\mathrm{MCv}$ we use the following dichotomy, applied to every
nonzero $\M$ of the execution whose peer is at rest. Along such a $\M$
the distance to the stationary peer is the distance from a fixed point
to a point travelling along a segment, hence non-increasing and then
non-decreasing. If the landing point is not the point of that segment
closest to the peer, the distance strictly increases on a nonempty
final part of the $\M$, so the monotonicity condition of
$\mathrm{MCv}$ fails at a finite time and we are done. We may therefore
assume that every such $\M$ ends exactly at its closest approach, so
that the distance is non-increasing along the execution so far.
 
Both robots have now performed the same number of activations, and
every $\LK$ so far has observed distinct positions (excluding also the
at most one $\mu$ with $p_\mu=b_0$), so by Lemma~\ref{lem:norm}(1)
they carry the same color.
At this point the robots occupy distinct positions, carry
the same internal color, and have no pending operation.
Apply the continuation statement proved in
Lemma~\ref{lem:mcv_not_fs_am}; its schedules are admissible under
$A_M$, hence under $\ASY$.
It supplies a fair, time-divergent continuation violating
$\Phi_{\mathrm{MCv}}$, while all safety violations established
in the prefix persist.
 
In the case without nonzero output, all three liveness violations are
witnessed by the explicit fair execution given there. In the cases $c_1=0$ and $p_E=B$, and whenever the
dichotomy above yields a monotonicity violation, all three violations
are safety violations at a finite time. In the remaining case, the
companion predicates have already suffered safety violations. The continuation from Lemma~\ref{lem:mcv_not_fs_am} witnesses either
a finite monotonicity violation or a failure of convergence.
Thus all predicates fail on the same fair infinite execution.
Hence Lemma~\ref{lem:concat}(1)(2)(3) is satisfied and
$\{\mathrm{MCv},\mathrm{VSS},\mathrm{VEC}\}$ has the common failure
property on $\FS^{A}$. Lemma~\ref{lem:boundary-common-failure} gives
the claim.
\end{proof}
 
\begin{lemma}
\label{lem:A3-15-OB-AM}
$\mathrm{MVX}\notin\Sol(\OB^{A_M})$.
\end{lemma}

\begin{proof}
Let $\mathcal A$ be any deterministic algorithm of $\OB$ under $A_M$.
No lemma of Appendix~\ref{app:cf-two-task-disjunctions} has $\OB^{A_M}$
as a boundary class, so we construct the execution directly, using the
two activation patterns recorded above. Use coordinates $A=0$,
$a_0=d$, $b_0=2d$, $B=3d$ on the initial line, with
$P=\{0,d,2d,3d\}$. The visible state of an $\OB$ robot is unique, so in
a normalized frame the rule is a single pair $(\lambda,\theta)$
(Lemma~\ref{lem:single-rule}(a)). Put $c_1:=c_1(\lambda,\theta)$ and
$c_2:=c_2(\lambda,\theta)$ as in Lemma~\ref{lem:ratios}.

\emph{Case $\lambda=0$.} Neither robot ever moves, under any schedule.
Hence $(\mathrm{P1})$ fails by Lemma~\ref{lem:p1_no_null} and
$\Phi_{\mathrm{VSS}}$ fails; the sequence of joint-stop configurations
reduces to $C_0$ and $d_0\not\perp d_0$, so the attainment condition of
$\mathrm{VEC}$ fails; and the robots remain at distinct positions, so
the convergence condition of $\mathrm{MCv}$ fails. All three are
independent of the schedule and persist along any fair infinite
continuation.

\emph{Case $c_1=0$, that is $(\lambda,\theta)=(1,0)$.} Activate both
robots simultaneously in round~$1$, which $A_M$ permits. Each moves to
the position of its peer, so they exchange positions and the distance
returns to $d$ after passing through $0$; the monotonicity condition of
$\mathrm{MCv}$ therefore fails. The cumulative nonzero-move counts are
$(1,1)$, so an admissible exchange configuration would have to be
$(B,A)=(3d,0)$, whereas the landing is $(2d,d)$; no valid $T_1$ exists
and $\Phi_{\mathrm{VSS}}$ fails by Lemma~\ref{lem:p1_first_move}.
Finally $C_1=(2d,d)$ gives $d_1=-d_0$, which is neither the clockwise
rotation by $45^\circ$ nor the one by $90^\circ$, so the transition
condition of $\mathrm{VEC}$ fails at $i=0$. All three are safety
violations at a finite time; append an arbitrary admissible fair,
time-divergent tail, which preserves them.

\emph{Case $0<c_1<1$.} Define the execution by an interval of
alternating singleton activations of $a$ and $b$, carried up to the
round following the first nonzero $\M$ and given normalized frames
(Lemma~\ref{lem:solo}), followed by a delayed-$\M$ interval of $b$
(Corollary~\ref{cor:stale}; the suspension is placed between
$t_C(b,i)$ and $t_B(b,i)$ and is therefore admissible under $A_M$),
followed by an arbitrary fair tail.

In the first interval, if the landing point of the first nonzero $\M$
is not the point required by Lemma~\ref{lem:p1_first_move}, no valid
$T_1$ exists; and if it is that point, the robot activated in the next
round receives an isomorphic local input and lands at $4d\notin P$ by
Lemma~\ref{lem:p1_replay}. Either way $\Phi_{\mathrm{VSS}}$ fails.
For $\Phi_{\mathrm{VEC}}$ the argument of Lemma~\ref{lem:A2-6-OB-R}
applies verbatim: the only single displacements that the transition
condition admits while the peer stands still are the $45^\circ$ and
$90^\circ$ clockwise rotations of the segment, so if $a$ departs from
them that condition fails at that move, and if it does not, a
perpendicular configuration is reached within at most two moves, after
which the same displacement is applied again and the stopping condition
fails. Since $c_1\neq0$, neither of the two Move operations of this
interval lands on the peer, so the robots are at distinct positions at
its end.

In the second interval, if no robot performs a nonzero $\M$ from that
point on, the two robots stay at these distinct positions and the
convergence condition of $\mathrm{MCv}$ fails. Otherwise $b$ computes a
destination with displacement $v\neq0$ approaching $a$, suspends its
$\M$, and in the meantime $a$ alone is activated repeatedly. Since
$\OB$ has no colors, every such activation applies the same
$(\lambda,\theta)$, so $d(j)=c_1^{\,j}d\to0$ by
Lemma~\ref{lem:ratios}(i) and the vector $w$ from $b$ to $a$ at the
release can be made to satisfy $\|w\|<\|v\|/2$. The distance
immediately before the $\M$ of $b$ is $\|w\|$, whereas at its
completion it satisfies
\[
  \|w-v\|\ \ge\ \|v\|-\|w\|\ >\ \|w\| ,
\]
so the monotonicity condition of $\mathrm{MCv}$ fails
\cite[Lemma~4]{CP}.

\emph{Case $c_1\ge1$.} The solo activations of $a$ no longer bring the
distance below $\|v\|/2$, so define the execution to be an infinite
sequence of fully synchronous rounds. For $\lambda>0$ the inequality
$c_1\ge1$ is equivalent to $\lambda\ge2\cos\theta$, which implies
$\lambda\ge\cos\theta$, that is $c_2\ge1$. Hence
Lemma~\ref{lem:ratios}(ii) gives $d(k)=c_2^{\,k}d(0)\ge d(0)>0$ for
every $k$, so the two robots do not converge to a single point and the
convergence condition of $\mathrm{MCv}$ fails. For
$\Phi_{\mathrm{VEC}}$ the argument of Lemma~\ref{lem:vec_not_ob_f}
applies verbatim to this execution, since it uses only that both robots
are activated simultaneously in every round and that the displacement
rule is the single pair $(\lambda,\theta)$. For
$\Phi_{\mathrm{VSS}}$: if $(\lambda,\theta)\neq(2,0)$, the nonzero $\M$
of round~$1$ does not meet the requirement of
Lemma~\ref{lem:p1_first_move}; and if $(\lambda,\theta)=(2,0)$, then
round~$1$ takes $a$ to $3d$ and $b$ to $0$, so in round~$2$ the robot
$a$ lands at $-3d$ and $b$ at $6d$, both outside $P$, which violates
$(\mathrm{P1}\text{-}b)$.

In every case a single execution violates all three predicates, so
$\{\mathrm{MCv},\mathrm{VSS},\mathrm{VEC}\}$ has the common failure
property on $\OB^{A_M}$ and
Lemma~\ref{lem:boundary-common-failure} gives the claim.
\end{proof}
 
\subsubsection{\texorpdfstring{MCv $\vee$ VSS $\vee$ CFSM (MVC)}{MCv or VSS or CFSM}}
\label{ssub:g3-2}
 
\begin{lemma}
\label{lem:A3-16-OB-AM}
$\mathrm{MVC}\notin\Sol(\OB^{A_M})$.
\end{lemma}

\begin{proof}
Let $\mathcal A$ be any deterministic algorithm of $\OB$ under $A_M$
and use the execution of Lemma~\ref{lem:A3-15-OB-AM}, where the
violations of $\Phi_{\mathrm{MCv}}$ and $\Phi_{\mathrm{VSS}}$ were
established. We show that $\Phi_{\CFSM}$ fails under the same case
distinction.

\emph{Case $\lambda=0$.} Neither robot ever moves, so
$\moves{a}=0\neq1$ and the single-move condition fails.

\emph{Case $c_1=0$.} The two robots exchange positions in round~$1$ and
therefore occupy the same point while they travel, so the
collision-freedom condition fails.

\emph{Case $0<c_1<1$.} The robot $a$ performs its first nonzero $\M$ in
the interval of alternating singleton activations, and in the
delayed-$\M$ interval of $b$ the distance stays positive, so each of
the repeated solo activations of $a$ is again a nonzero $\M$; hence
$\moves{a}\ge2\neq1$ and the single-move condition fails.

\emph{Case $c_1\ge1$.} The execution consists of fully synchronous
rounds, and $c_2\ge1$ gives $d(k)\ge d(0)>0$ for every $k$, as shown in
Lemma~\ref{lem:A3-15-OB-AM}. Since $\lambda\neq0$, the displacement of
each robot in round~$k$ has norm $\lambda\,d(k)>0$, so both robots
perform a nonzero $\M$ in every round, and these are counted as
distinct $\M$ operations because $t_E(r,i)<t_L(r,i+1)$. Hence
$\moves{a}=\infty\neq1$ and the single-move condition fails.

In every case a single execution violates all three predicates, so
$\{\mathrm{MCv},\mathrm{VSS},\CFSM\}$ has the common failure property
on $\OB^{A_M}$ and Lemma~\ref{lem:boundary-common-failure} gives the
claim.
\end{proof}

%% file: arxiv.bbl
\newcommand{\etalchar}[1]{$^{#1}$}
\begin{thebibliography}{DDFN18}

\bibitem[BFK{\etalchar{+}}21]{BuFKPSW21}
K.~Buchin, P.~Flocchini, I.~Kostitsyna, T.~Peters, N.~Santoro, and K.~Wada.
\newblock Autonomous mobile robots: Refining the computational landscape.
\newblock In {\em 23rd IPDPS Workshop on Advances in Parallel and Distributed Computational Models (APDCM)}, pages 576--585, 2021.

\bibitem[BFK{\etalchar{+}}25]{BuFKPSW22}
K.~Buchin, P.~Flocchini, I.~Kostitsyna, T.~Peters, N.~Santoro, and K.~Wada.
\newblock On the computational power of energy-constrained mobile robots.
\newblock {\em Information and Computation}, 303:105280, 2025.

\bibitem[CP05]{CP}
R.~Cohen and D.~Peleg.
\newblock Convergence properties of the gravitational algorithms in asynchronous robot systems.
\newblock {\em {SIAM} Journal on Computing}, 34(6):1516--1528, 2005.

\bibitem[DDFN18]{GraphLandscape}
M.~D'Emidio, G.~{Di Stefano}, D.~Frigioni, and A.~Navarra.
\newblock Characterizing the computational power of mobile robots on graphs and implications for the {Euclidean} plane.
\newblock {\em Information and Computation}, 263:57--74, 2018.

\bibitem[DFP{\etalchar{+}}16]{DFPSY}
S.~Das, P.~Flocchini, G.~Prencipe, N.~Santoro, and M.~Yamashita.
\newblock Autonomous mobile robots with lights.
\newblock {\em Theoretical Computer Science}, 609:171--184, 2016.

\bibitem[DGS{\etalchar{+}}24]{VisibilityLandscape}
A.~Das, S.~Ghosh, A.~Sharma, P.~Goswami, and B.~Sau.
\newblock The computational landscape of autonomous mobile robots: The visibility perspective.
\newblock In {\em Distributed Computing and Intelligent Technology (ICDCIT)}, volume 14501 of {\em Lecture Notes in Computer Science}, pages 85--100, 2024.

\bibitem[FFS26a]{SequentialMFCS}
C.~Feletti, P.~Flocchini, and N.~Santoro.
\newblock Computational power of energy-constrained autonomous robots under sequential schedulers.
\newblock In {\em 51st International Symposium on Mathematical Foundations of Computer Science (MFCS)}, volume 386 of {\em Leibniz International Proceedings in Informatics (LIPIcs)}, pages 22:1--22:17, 2026.

\bibitem[FFS26b]{SequentialTCS}
C.~Feletti, P.~Flocchini, and N.~Santoro.
\newblock On the computational power of mobile robots under sequential schedulers.
\newblock {\em Theoretical Computer Science}, 1083:116172, 2026.

\bibitem[FSSW24]{FSSW23}
P.~Flocchini, N.~Santoro, Y.~Sudo, and K.~Wada.
\newblock On asynchrony, memory, and communication: Separations and landscapes.
\newblock In {\em 27th International Conference on Principles of Distributed Systems (OPODIS)}, volume 286 of {\em Leibniz International Proceedings in Informatics (LIPIcs)}, pages 28:1--28:23, 2024.

\bibitem[FSW20]{FSW19}
P.~Flocchini, N.~Santoro, and K.~Wada.
\newblock On memory, communication, and synchronous schedulers when moving and computing.
\newblock In {\em 23rd International Conference on Principles of Distributed Systems (OPODIS)}, volume 153 of {\em Leibniz International Proceedings in Informatics (LIPIcs)}, pages 25:1--25:17, 2020.

\bibitem[KSW25]{KSW-2-robot}
N.~Kitamura, Y.~Sudo, and K.~Wada.
\newblock Two-robot computational landscape: A complete characterization of model power in minimal mobile robot systems.
\newblock {\em CoRR}, abs/2512.22770, 2025.

\bibitem[MSW26]{MSW-rendezvous}
T.~Masuzawa, Y.~Sudo, and K.~Wada.
\newblock Rendezvous under variable disorientation: The algorithmic power of fixed unit distance.
\newblock {\em CoRR}, abs/2610.09713, 2026.

\bibitem[NNW26]{NNW-SSS2025}
S.~Naito, T.~Ninomiya, and K.~Wada.
\newblock Beyond pairwise comparisons: Unveiling structural landscape of mobile robot models.
\newblock In {\em Stabilization, Safety, and Security of Distributed Systems (SSS 2025)}, Lecture Notes in Computer Science, pages 408--424, 2026.

\bibitem[SY99]{SY}
I.~Suzuki and M.~Yamashita.
\newblock Distributed anonymous mobile robots: Formation of geometric patterns.
\newblock {\em {SIAM} Journal on Computing}, 28(4):1347--1363, 1999.

\end{thebibliography}
